\documentclass[a4paper,11pt]{article}
\pdfoutput=1
\usepackage{amsmath}
\usepackage{amssymb}
\usepackage{amsthm}
\usepackage{fullpage}
\usepackage{tikz}
\usetikzlibrary{cd}
\usetikzlibrary{shapes}
\usetikzlibrary{calc}
\usepackage{mathtools}
\usetikzlibrary{decorations.pathreplacing}
\usepackage{graphicx}
\usepackage{cite}
\usepackage{float}
\usepackage{enumitem}
\usepackage{caption}
\usepackage{subcaption}

\usepackage{tocloft}

\usepackage{hyperref}

\newcommand{\Z}{\mathbb{Z}}
\newcommand{\N}{\mathbb{N}}

\newcommand{\C}{\mathbb{C}}

\newcommand{\TL}{\mathrm{TL}}
\newcommand{\BTL}{\mathrm{TL}^{1}}
\newcommand{\BBTL}{\mathrm{TL}^{2}}
\newcommand{\dTL}{\mathrm{dTL}}
\newcommand{\Gh}{\mathrm{Gh}^{2}}
\newcommand{\Gho}{\mathrm{Gh}^{1}}
\newcommand{\ncGh}{\widetilde{\mathrm{Gh}}{}^{2}}
\newcommand{\ncGho}{\widetilde{\mathrm{Gh}}{}^{1}}
\newcommand{\dGh}{\mathrm{dGh}^{2}}
\newcommand{\dGho}{\mathrm{dGh}^{1}}
\newcommand{\ncdGh}{\mathrm{d}\widetilde{\mathrm{Gh}}{}^{2}}
\newcommand{\ncdGho}{\mathrm{d}\widetilde{\mathrm{Gh}}{}^{1}}

\newcommand{\M}{\mathrm{M}}
\newcommand{\blob}{\mathrm{b}^1}

\newcommand{\out}[2]{\left\lvert {#1}\ \ {#2}\right\rvert}

\definecolor{paleblue}{HTML}{E9E9F3}
\definecolor{richblue}{HTML}{3333B2}
\definecolor{dgreen}{HTML}{1DBF48}
\definecolor{dred}{HTML}{A00022}
\definecolor{link}{HTML}{11CDFC}
\definecolor{bdylink}{HTML}{F411FC}
\definecolor{defect}{HTML}{7F11FC}
\definecolor{amber}{HTML}{FCA505}
\definecolor{cyanish}{HTML}{11CDFC}
\definecolor{magentaprint}{HTML}{F04BFC}
\definecolor{highlight}{HTML}{0000FF}

\tikzstyle{thicc}=[line width=1.6pt]
\tikzstyle{string}=[line width=0.8pt]
\tikzstyle{dotty}=[dotted,black!40,line width=0.8pt]
\tikzstyle{lstring}=[line width=1.6pt,draw=black!20!richblue]
\tikzstyle{lfill}=[fill=richblue!15]
\tikzstyle{ledge}=[draw=richblue!70]
\tikzstyle{ldotty}=[dotted,richblue!70]
\tikzstyle{lsq}=[lfill,ledge]
\tikzstyle{lgh}=[black!20!richblue]
\tikzstyle{ltri}=[lfill]

\pgfdeclarelayer{background}
\pgfdeclarelayer{foreground}
\pgfsetlayers{background,main,foreground}

\newcommand{\ld}[1]{
\begin{scope}[shift={#1}]
\begin{pgfonlayer}{foreground}
\filldraw [ledge,fill=white] (0,0) circle (0.1);
\end{pgfonlayer}
\end{scope}
}

\newcommand{\lup}[1]{
\begin{scope}[shift={#1}]
\clip (0,0) rectangle (1,-1);
\begin{pgfonlayer}{background}
\filldraw[lsq] (0,0) rectangle (1,-1);
\end{pgfonlayer}
\draw[lstring] (0,-0.5) arc (-90:0:0.5cm);
\draw[lstring] (1,-0.5) arc (90:180:0.5cm);
\end{scope}
}

\newcommand{\rup}[1]{
\begin{scope}[shift={#1}]
\clip (0,0) rectangle (1,-1);
\begin{pgfonlayer}{background}
\filldraw[lsq] (0,0) rectangle (1,-1);
\end{pgfonlayer}
\draw[lstring] (0,-0.5) arc (90:0:0.5cm);
\draw[lstring] (1,-0.5) arc (270:180:0.5cm);
\end{scope}
}

\newcommand{\llup}[1]{
\begin{scope}[shift={#1}]
\clip (0,0) rectangle (1,-1);
\begin{pgfonlayer}{background}
\filldraw[lsq] (0,0) rectangle (1,-1);
\end{pgfonlayer}
\draw[lstring] (0,-0.5) arc (-90:0:0.5);
\ld{(1,-0.5)}
\ld{(0.5,-1)}
\end{scope}
}

\newcommand{\rrup}[1]{
\begin{scope}[shift={#1}]
\clip (0,0) rectangle (1,-1);
\begin{pgfonlayer}{background}
\filldraw[lsq] (0,0) rectangle (1,-1);
\end{pgfonlayer}
\draw[lstring] (1,-0.5) arc (-90:-180:0.5);
\ld{(0,-0.5)}
\ld{(0.5,-1)}
\end{scope}
}

\newcommand{\lld}[1]{
\begin{scope}[shift={#1}]
\clip (0,0) rectangle (1,-1);
\begin{pgfonlayer}{background}
\filldraw[lsq] (0,0) rectangle (1,-1);
\end{pgfonlayer}
\draw[lstring] (0,-0.5) arc (90:0:0.5);
\ld{(1,-0.5)}
\ld{(0.5,0)}
\end{scope}
}

\newcommand{\rrd}[1]{
\begin{scope}[shift={#1}]
\clip (0,0) rectangle (1,-1);
\begin{pgfonlayer}{background}
\filldraw[lsq] (0,0) rectangle (1,-1);
\end{pgfonlayer}
\draw[lstring] (1,-0.5) arc (90:180:0.5);
\ld{(0,-0.5)}
\ld{(0.5,0)}
\end{scope}
}

\newcommand{\vv}[1]{
\begin{scope}[shift={#1}]
\clip (0,0) rectangle (1,-1);
\begin{pgfonlayer}{background}
\filldraw[lsq] (0,0) rectangle (1,-1);
\end{pgfonlayer}
\draw[lstring] (0.5,0)--(0.5,-1);
\ld{(1,-0.5)}
\ld{(0,-0.5)}
\end{scope}
}

\newcommand{\hh}[1]{
\begin{scope}[shift={#1}]
\clip (0,0) rectangle (1,-1);
\begin{pgfonlayer}{background}
\filldraw[lsq] (0,0) rectangle (1,-1);
\end{pgfonlayer}
\draw[lstring] (0,-0.5)--(1,-0.5);
\ld{(0.5,0)}
\ld{(0.5,-1)}
\end{scope}
}

\newcommand{\emp}[1]{
\begin{scope}[shift={#1}]
\clip (0,0) rectangle (1,-1);
\begin{pgfonlayer}{background}
\filldraw[lsq] (0,0) rectangle (1,-1);
\end{pgfonlayer}
\ld{(1,-0.5)}
\ld{(0.5,-1)}
\ld{(0,-0.5)}
\ld{(0.5,0)}
\end{scope}
}

\newcommand{\faceop}[2]{
\begin{scope}[shift={#1}]
\filldraw[lsq] (0,0) rectangle (1,-1);
\fill[fill=richblue!70] (0,0)--(0.25,0)--(0,-0.25)--cycle;
\node at (0.5,-0.5) [anchor=center] {\scriptsize #2};
\end{scope}
}

\newcommand{\facetop}[2]{
\begin{scope}[shift={#1}]
\fill[ltri] (0,0)--(2,0)--(1,-1)--cycle;
\draw[ledge] (0,0)--(1,-1)--(2,0);
\draw[ldotty] (0,0)--(2,0);
\node at (1,-0.4) [anchor=center] {\scriptsize #2};
\end{scope}
}

\newcommand{\facebot}[2]{
\begin{scope}[shift={#1}]
\fill[ltri] (0,0)--(2,0)--(1,1)--cycle;
\draw[ledge] (0,0)--(1,1)--(2,0);
\draw[ldotty] (0,0)--(2,0);
\node at (1,0.4) [anchor=center] {\scriptsize #2};
\end{scope}
}

\newcommand{\vdotsq}[1]{
\begin{scope}[shift={#1}]
\filldraw[lsq] (0,0) rectangle (1,-1);
\draw (0.5,-0.4) node[anchor=center]{$\vdots$};
\end{scope}
}

\newcommand{\utri}[2]{
\begin{scope}[shift={#2}]
\fill[ltri] (2,0)--(1,-1)--(0,0)--cycle;
\draw [ledge] (2,0)--(1,-1)--(0,0);
\draw (1,-0.4) node[anchor=center]{#1};
\draw[ldotty] (0,0)--(2,0);
\end{scope}
}

\newcommand{\bulk}[2]{
\begin{scope}[shift={#2}]
\draw [lsq] (0,0)--(1,1)--(2,0)--(1,-1)--cycle;
\draw (1,0) node[anchor=center]{#1};
\filldraw[richblue!70] (0.2,0.2)--(0.2,-0.2)--(0,0)--cycle;
\end{scope}
}

\newcommand{\aatop}[1]{
\begin{scope}[shift={#1}]
\begin{pgfonlayer}{background}
\fill[ltri] (2,0)--(1,-1)--(0,0)--cycle;
\draw[ldotty] (0,0)--(2,0);
\end{pgfonlayer}
\begin{scope}
\clip (2,0)--(1,-1)--(0,0)--cycle;
\draw[lstring] (0.5,-0.5) to[out=60,in=120] (1.5,-0.5);
\end{scope}
\draw[ledge] (2,0)--(1,-1)--(0,0);
\end{scope}
}

\newcommand{\aabot}[1]{
\begin{scope}[shift={#1}]
\begin{pgfonlayer}{background}
\fill[ltri] (2,0)--(1,1)--(0,0)--cycle;
\draw[ldotty] (0,0)--(2,0);
\end{pgfonlayer}
\begin{scope}
\clip (2,0)--(1,1)--(0,0)--cycle;
\draw[lstring] (0.5,0.5) to[out=-60,in=-120] (1.5,0.5);
\end{scope}
\draw[ledge] (2,0)--(1,1)--(0,0);
\end{scope}
}

\newcommand{\batop}[1]{
\begin{scope}[shift={#1}]
\begin{pgfonlayer}{background}
\fill[ltri] (2,0)--(1,-1)--(0,0)--cycle;
\draw[ldotty] (0,0)--(2,0);
\end{pgfonlayer}
\begin{scope}
\clip (2,0)--(1,-1)--(0,0)--cycle;
\draw[lstring] (0.5,-0.5) to[out=60,in=-90] (0.7,0);
\draw[lstring] (1.5,-0.5) to[out=120,in=-90] (1.3,0);
\end{scope}
\draw[ledge] (2,0)--(1,-1)--(0,0);
\end{scope}
}

\newcommand{\babot}[1]{
\begin{scope}[shift={#1}]
\begin{pgfonlayer}{background}
\fill[ltri] (2,0)--(1,1)--(0,0)--cycle;
\draw[ldotty] (0,0)--(2,0);
\end{pgfonlayer}
\begin{scope}
\clip (2,0)--(1,1)--(0,0)--cycle;
\draw[lstring] (0.5,0.5) to[out=-60,in=90] (0.7,0);
\draw[lstring] (1.5,0.5) to[out=-120,in=90] (1.3,0);
\end{scope}
\draw[ledge] (2,0)--(1,1)--(0,0);
\end{scope}
}

\newcommand{\bbtop}[1]{
\begin{scope}[shift={#1}]
\begin{pgfonlayer}{background}
\fill[ltri] (2,0)--(1,-1)--(0,0)--cycle;
\draw[ldotty] (0,0)--(2,0);
\end{pgfonlayer}
\begin{scope}
\clip (2,0)--(1,-1)--(0,0)--cycle;
\draw[lstring] (0.5,-0.5) to[out=60,in=-90] (0.7,0);
\draw[lstring] (1.5,-0.5) to[out=120,in=-90] (1.3,0);
\end{scope}
\draw[ledge] (2,0)--(1,-1)--(0,0);
\filldraw[lgh] (0.35,0) circle (0.09);
\filldraw[lgh] (1,0) circle (0.09);
\end{scope}
}

\newcommand{\bbbot}[1]{
\begin{scope}[shift={#1}]
\begin{pgfonlayer}{background}
\fill[ltri] (2,0)--(1,1)--(0,0)--cycle;
\draw[ldotty] (0,0)--(2,0);
\end{pgfonlayer}
\begin{scope}
\clip (2,0)--(1,1)--(0,0)--cycle;
\draw[lstring] (0.5,0.5) to[out=-60,in=90] (0.7,0);
\draw[lstring] (1.5,0.5) to[out=-120,in=90] (1.3,0);
\end{scope}
\draw[ledge] (2,0)--(1,1)--(0,0);
\filldraw[lgh] (0.35,0) circle (0.09);
\filldraw[lgh] (1,0) circle (0.09);
\end{scope}
}

\newcommand{\bctop}[1]{
\begin{scope}[shift={#1}]
\begin{pgfonlayer}{background}
\fill[ltri] (2,0)--(1,-1)--(0,0)--cycle;
\draw[ldotty] (0,0)--(2,0);
\end{pgfonlayer}
\begin{scope}
\clip (2,0)--(1,-1)--(0,0)--cycle;
\draw[lstring] (0.5,-0.5) to[out=60,in=-90] (0.7,0);
\draw[lstring] (1.5,-0.5) to[out=120,in=-90] (1.3,0);
\end{scope}
\draw[ledge] (2,0)--(1,-1)--(0,0);
\filldraw[lgh] (0.35,0) circle (0.09);
\filldraw[lgh] (1.65,0) circle (0.09);
\end{scope}
}

\newcommand{\bcbot}[1]{
\begin{scope}[shift={#1}]
\begin{pgfonlayer}{background}
\fill[ltri] (2,0)--(1,1)--(0,0)--cycle;
\draw[ldotty] (0,0)--(2,0);
\end{pgfonlayer}
\begin{scope}
\clip (2,0)--(1,1)--(0,0)--cycle;
\draw[lstring] (0.5,0.5) to[out=-60,in=90] (0.7,0);
\draw[lstring] (1.5,0.5) to[out=-120,in=90] (1.3,0);
\end{scope}
\draw[ledge] (2,0)--(1,1)--(0,0);
\filldraw[lgh] (0.35,0) circle (0.09);
\filldraw[lgh] (1.65,0) circle (0.09);
\end{scope}
}

\newcommand{\bdtop}[1]{
\begin{scope}[shift={#1}]
\begin{pgfonlayer}{background}
\fill[ltri] (2,0)--(1,-1)--(0,0)--cycle;
\draw[ldotty] (0,0)--(2,0);
\end{pgfonlayer}
\begin{scope}
\clip (2,0)--(1,-1)--(0,0)--cycle;
\draw[lstring] (0.5,-0.5) to[out=60,in=-90] (0.7,0);
\draw[lstring] (1.5,-0.5) to[out=120,in=-90] (1.3,0);
\end{scope}
\draw[ledge] (2,0)--(1,-1)--(0,0);
\filldraw[lgh] (1,0) circle (0.09);
\filldraw[lgh] (1.65,0) circle (0.09);
\end{scope}
}

\newcommand{\bdbot}[1]{
\begin{scope}[shift={#1}]
\begin{pgfonlayer}{background}
\fill[ltri] (2,0)--(1,1)--(0,0)--cycle;
\draw[ldotty] (0,0)--(2,0);
\end{pgfonlayer}
\begin{scope}
\clip (2,0)--(1,1)--(0,0)--cycle;
\draw[lstring] (0.5,0.5) to[out=-60,in=90] (0.7,0);
\draw[lstring] (1.5,0.5) to[out=-120,in=90] (1.3,0);
\end{scope}
\draw[ledge] (2,0)--(1,1)--(0,0);
\filldraw[lgh] (1,0) circle (0.09);
\filldraw[lgh] (1.65,0) circle (0.09);
\end{scope}
}

\newcommand{\betop}[1]{
\begin{scope}[shift={#1}]
\begin{pgfonlayer}{background}
\fill[ltri] (2,0)--(1,-1)--(0,0)--cycle;
\draw[ldotty] (0,0)--(2,0);
\end{pgfonlayer}
\begin{scope}
\clip (2,0)--(1,-1)--(0,0)--cycle;
\draw[lstring] (0.5,-0.5) to[out=60,in=260] (1.35,0.03);
\end{scope}
\draw[ledge] (2,0)--(1,-1)--(0,0);
\filldraw[lgh] (0.65,0) circle (0.09);
\ld{(1.5,-0.5)}
\end{scope}
}

\newcommand{\bebot}[1]{
\begin{scope}[shift={#1}]
\begin{pgfonlayer}{background}
\fill[ltri] (2,0)--(1,1)--(0,0)--cycle;
\draw[ldotty] (0,0)--(2,0);
\end{pgfonlayer}
\begin{scope}
\clip (2,0)--(1,1)--(0,0)--cycle;
\draw[lstring] (0.5,0.5) to[out=-60,in=-260] (1.35,-0.03);
\end{scope}
\draw[ledge] (2,0)--(1,1)--(0,0);
\filldraw[lgh] (0.65,0) circle (0.09);
\ld{(1.5,0.5)}
\end{scope}
}

\newcommand{\bftop}[1]{
\begin{scope}[shift={#1}]
\begin{pgfonlayer}{background}
\fill[ltri] (2,0)--(1,-1)--(0,0)--cycle;
\draw[ldotty] (0,0)--(2,0);
\end{pgfonlayer}
\begin{scope}
\clip (2,0)--(1,-1)--(0,0)--cycle;
\draw[lstring] (0.5,-0.5) to[out=60,in=270] (0.65,0);
\end{scope}
\draw[ledge] (2,0)--(1,-1)--(0,0);
\filldraw[lgh] (1.35,0) circle (0.09);
\ld{(1.5,-0.5)}
\end{scope}
}

\newcommand{\bfbot}[1]{
\begin{scope}[shift={#1}]
\begin{pgfonlayer}{background}
\fill[ltri] (2,0)--(1,1)--(0,0)--cycle;
\draw[ldotty] (0,0)--(2,0);
\end{pgfonlayer}
\begin{scope}
\clip (2,0)--(1,1)--(0,0)--cycle;
\draw[lstring] (0.5,0.5) to[out=-60,in=-270] (0.65,0);
\end{scope}
\draw[ledge] (2,0)--(1,1)--(0,0);
\filldraw[lgh] (1.35,0) circle (0.09);
\ld{(1.5,0.5)}
\end{scope}
}

\newcommand{\bgtop}[1]{
\begin{scope}[shift={#1}]
\begin{pgfonlayer}{background}
\fill[ltri] (2,0)--(1,-1)--(0,0)--cycle;
\draw[ldotty] (0,0)--(2,0);
\end{pgfonlayer}
\begin{scope}
\clip (2,0)--(1,-1)--(0,0)--cycle;
\draw[lstring] (1.5,-0.5) to[out=120,in=280] (0.65,0.03);
\end{scope}
\draw[ledge] (2,0)--(1,-1)--(0,0);
\filldraw[lgh] (1.35,0) circle (0.09);
\ld{(0.5,-0.5)}
\end{scope}
}

\newcommand{\bgbot}[1]{
\begin{scope}[shift={#1}]
\begin{pgfonlayer}{background}
\fill[ltri] (2,0)--(1,1)--(0,0)--cycle;
\draw[ldotty] (0,0)--(2,0);
\end{pgfonlayer}
\begin{scope}
\clip (2,0)--(1,1)--(0,0)--cycle;
\draw[lstring] (1.5,0.5) to[out=-120,in=-280] (0.65,-0.03);
\end{scope}
\draw[ledge] (2,0)--(1,1)--(0,0);
\filldraw[lgh] (1.35,0) circle (0.09);
\ld{(0.5,0.5)}
\end{scope}
}

\newcommand{\bhtop}[1]{
\begin{scope}[shift={#1}]
\begin{pgfonlayer}{background}
\fill[ltri] (2,0)--(1,-1)--(0,0)--cycle;
\draw[ldotty] (0,0)--(2,0);
\end{pgfonlayer}
\begin{scope}
\clip (2,0)--(1,-1)--(0,0)--cycle;
\draw[lstring] (1.5,-0.5) to[out=120,in=-90] (1.35,0);
\end{scope}
\draw[ledge] (2,0)--(1,-1)--(0,0);
\filldraw[lgh] (0.65,0) circle (0.09);
\ld{(0.5,-0.5)}
\end{scope}
}

\newcommand{\bhbot}[1]{
\begin{scope}[shift={#1}]
\begin{pgfonlayer}{background}
\fill[ltri] (2,0)--(1,1)--(0,0)--cycle;
\draw[ldotty] (0,0)--(2,0);
\end{pgfonlayer}
\begin{scope}
\clip (2,0)--(1,1)--(0,0)--cycle;
\draw[lstring] (1.5,0.5) to[out=-120,in=90] (1.35,0);
\end{scope}
\draw[ledge] (2,0)--(1,1)--(0,0);
\filldraw[lgh] (0.65,0) circle (0.09);
\ld{(0.5,0.5)}
\end{scope}
}

\newcommand{\emtop}[1]{
\begin{scope}[shift={#1}]
\begin{pgfonlayer}{background}
\fill[ltri] (2,0)--(1,-1)--(0,0)--cycle;
\draw[ldotty] (0,0)--(2,0);
\draw[ledge] (2,0)--(1,-1)--(0,0);
\end{pgfonlayer}
\ld{(0.5,-0.5)}
\ld{(1.5,-0.5)}
\end{scope}
}

\newcommand{\embot}[1]{
\begin{scope}[shift={#1}]
\begin{pgfonlayer}{background}
\fill[ltri] (2,0)--(1,1)--(0,0)--cycle;
\draw[ldotty] (0,0)--(2,0);
\draw[ledge] (2,0)--(1,1)--(0,0);
\end{pgfonlayer}
\ld{(0.5,0.5)}
\ld{(1.5,0.5)}
\end{scope}
}

\newcommand{\idsq}[1]{
\begin{scope}[shift={#1}]
\filldraw[lsq] (0,0)--(1,1)--(2,0)--(1,-1)--cycle;
\draw[lstring] (0.5,0.5) to[out=-45,in=225] (1.5,0.5);
\draw[lstring] (0.5,-0.5) to[out=45,in=135] (1.5,-0.5);
\end{scope}
}

\newcommand{\esq}[1]{
\begin{scope}[shift={#1}]
\filldraw[lsq] (0,0)--(1,1)--(2,0)--(1,-1)--cycle;
\draw[lstring] (0.5,0.5) to[out=-45,in=45] (0.5,-0.5);
\draw[lstring] (1.5,0.5) to[out=-135,in=135] (1.5,-0.5);
\end{scope}
}

\newcommand{\emsq}[1]{
\begin{scope}[shift={#1}]
\filldraw[lsq] (0,0)--(1,1)--(2,0)--(1,-1)--cycle;
\ld{(0.5,0.5)}
\ld{(0.5,-0.5)}
\ld{(1.5,0.5)}
\ld{(1.5,-0.5)}
\end{scope}
}

\newcommand{\tsq}[1]{
\begin{scope}[shift={#1}]
\filldraw[lsq] (0,0)--(1,1)--(2,0)--(1,-1)--cycle;
\draw[lstring] (0.5,0.5) to[out=-45,in=225] (1.5,0.5);
\ld{(0.5,-0.5)}
\ld{(1.5,-0.5)}
\end{scope}
}

\newcommand{\bsq}[1]{
\begin{scope}[shift={#1}]
\filldraw[lsq] (0,0)--(1,1)--(2,0)--(1,-1)--cycle;
\draw[lstring] (0.5,-0.5) to[out=45,in=135] (1.5,-0.5);
\ld{(0.5,0.5)}
\ld{(1.5,0.5)}
\end{scope}
}

\newcommand{\lsq}[1]{
\begin{scope}[shift={#1}]
\filldraw[lsq] (0,0)--(1,1)--(2,0)--(1,-1)--cycle;
\draw[lstring] (0.5,0.5) to[out=-45,in=45] (0.5,-0.5);
\ld{(1.5,0.5)}
\ld{(1.5,-0.5)}
\end{scope}
}

\newcommand{\rsq}[1]{
\begin{scope}[shift={#1}]
\filldraw[lsq] (0,0)--(1,1)--(2,0)--(1,-1)--cycle;
\draw[lstring] (1.5,0.5) to[out=-135,in=135] (1.5,-0.5);
\ld{(0.5,0.5)}
\ld{(0.5,-0.5)}
\end{scope}
}

\newcommand{\usq}[1]{
\begin{scope}[shift={#1}]
\filldraw[lsq] (0,0)--(1,1)--(2,0)--(1,-1)--cycle;
\draw[lstring] (0.5,-0.5) -- (1.5,0.5);
\ld{(0.5,0.5)}
\ld{(1.5,-0.5)}
\end{scope}
}

\newcommand{\dsq}[1]{
\begin{scope}[shift={#1}]
\filldraw[lsq] (0,0)--(1,1)--(2,0)--(1,-1)--cycle;
\draw[lstring] (0.5,0.5) -- (1.5,-0.5);
\ld{(1.5,0.5)}
\ld{(0.5,-0.5)}
\end{scope}
}

\newcommand{\aatri}[1]{
\begin{scope}[shift={#1}]
\begin{pgfonlayer}{background}
\fill[ltri] (2,0)--(1,-1)--(0,0)--cycle;
\draw[ldotty] (0,0)--(2,0);
\end{pgfonlayer}
\begin{scope}
\clip (2,0)--(1,-1)--(0,0)--cycle;
\draw[lstring] (0.5,-0.5) to[out=45,in=135] (1.5,-0.5);
\end{scope}
\draw[ledge] (2,0)--(1,-1)--(0,0);
\end{scope}
}

\newcommand{\batri}[1]{
\begin{scope}[shift={#1}]
\begin{pgfonlayer}{background}
\fill[ltri] (2,0)--(1,-1)--(0,0)--cycle;
\draw[ldotty] (0,0)--(2,0);
\end{pgfonlayer}
\begin{scope}
\clip (2,0)--(1,-1)--(0,0)--cycle;
\draw[lstring] (0.5,-0.5) to[out=45,in=-90] (0.7,0);
\draw[lstring] (1.5,-0.5) to[out=135,in=-90] (1.3,0);
\end{scope}
\draw[ledge] (2,0)--(1,-1)--(0,0);
\end{scope}
}

\newcommand{\bbtri}[1]{
\begin{scope}[shift={#1}]
\begin{pgfonlayer}{background}
\fill[ltri] (2,0)--(1,-1)--(0,0)--cycle;
\draw[ldotty] (0,0)--(2,0);
\end{pgfonlayer}
\begin{scope}
\clip (2,0)--(1,-1)--(0,0)--cycle;
\draw[lstring] (0.5,-0.5) to[out=45,in=-90] (0.7,0);
\draw[lstring] (1.5,-0.5) to[out=135,in=-90] (1.3,0);
\end{scope}
\draw[ledge] (2,0)--(1,-1)--(0,0);
\filldraw[lgh] (0.35,0) circle (0.09);
\filldraw[lgh] (1,0) circle (0.09);
\end{scope}
}

\newcommand{\bctri}[1]{
\begin{scope}[shift={#1}]
\begin{pgfonlayer}{background}
\fill[ltri] (2,0)--(1,-1)--(0,0)--cycle;
\draw[ldotty] (0,0)--(2,0);
\end{pgfonlayer}
\begin{scope}
\clip (2,0)--(1,-1)--(0,0)--cycle;
\draw[lstring] (0.5,-0.5) to[out=45,in=-90] (0.7,0);
\draw[lstring] (1.5,-0.5) to[out=135,in=-90] (1.3,0);
\end{scope}
\draw[ledge] (2,0)--(1,-1)--(0,0);
\filldraw[lgh] (0.35,0) circle (0.09);
\filldraw[lgh] (1.65,0) circle (0.09);
\end{scope}
}

\newcommand{\bdtri}[1]{
\begin{scope}[shift={#1}]
\begin{pgfonlayer}{background}
\fill[ltri] (2,0)--(1,-1)--(0,0)--cycle;
\draw[ldotty] (0,0)--(2,0);
\end{pgfonlayer}
\begin{scope}
\clip (2,0)--(1,-1)--(0,0)--cycle;
\draw[lstring] (0.5,-0.5) to[out=45,in=-90] (0.7,0);
\draw[lstring] (1.5,-0.5) to[out=135,in=-90] (1.3,0);
\end{scope}
\draw[ledge] (2,0)--(1,-1)--(0,0);
\filldraw[lgh] (1,0) circle (0.09);
\filldraw[lgh] (1.65,0) circle (0.09);
\end{scope}
}

\newcommand{\betri}[1]{
\begin{scope}[shift={#1}]
\begin{pgfonlayer}{background}
\fill[ltri] (2,0)--(1,-1)--(0,0)--cycle;
\draw[ldotty] (0,0)--(2,0);
\end{pgfonlayer}
\begin{scope}
\clip (2,0)--(1,-1)--(0,0)--cycle;
\draw[lstring] (0.5,-0.5) to[out=45,in=270] (1.35,0);
\end{scope}
\draw[ledge] (2,0)--(1,-1)--(0,0);
\filldraw[lgh] (0.65,0) circle (0.09);
\ld{(1.5,-0.5)}
\end{scope}
}

\newcommand{\bftri}[1]{
\begin{scope}[shift={#1}]
\begin{pgfonlayer}{background}
\fill[ltri] (2,0)--(1,-1)--(0,0)--cycle;
\draw[ldotty] (0,0)--(2,0);
\end{pgfonlayer}
\begin{scope}
\clip (2,0)--(1,-1)--(0,0)--cycle;
\draw[lstring] (0.5,-0.5) to[out=45,in=270] (0.65,0);
\end{scope}
\draw[ledge] (2,0)--(1,-1)--(0,0);
\filldraw[lgh] (1.35,0) circle (0.09);
\ld{(1.5,-0.5)}
\end{scope}
}

\newcommand{\bgtri}[1]{
\begin{scope}[shift={#1}]
\begin{pgfonlayer}{background}
\fill[ltri] (2,0)--(1,-1)--(0,0)--cycle;
\draw[ldotty] (0,0)--(2,0);
\end{pgfonlayer}
\begin{scope}
\clip (2,0)--(1,-1)--(0,0)--cycle;
\draw[lstring] (1.5,-0.5) to[out=135,in=270] (0.65,0);
\end{scope}
\draw[ledge] (2,0)--(1,-1)--(0,0);
\filldraw[lgh] (1.35,0) circle (0.09);
\ld{(0.5,-0.5)}
\end{scope}
}

\newcommand{\bhtri}[1]{
\begin{scope}[shift={#1}]
\begin{pgfonlayer}{background}
\fill[ltri] (2,0)--(1,-1)--(0,0)--cycle;
\draw[ldotty] (0,0)--(2,0);
\end{pgfonlayer}
\begin{scope}
\clip (2,0)--(1,-1)--(0,0)--cycle;
\draw[lstring] (1.5,-0.5) to[out=135,in=-90] (1.35,0);
\end{scope}
\draw[ledge] (2,0)--(1,-1)--(0,0);
\filldraw[lgh] (0.65,0) circle (0.09);
\ld{(0.5,-0.5)}
\end{scope}
}

\newcommand{\emtri}[1]{
\begin{scope}[shift={#1}]
\begin{pgfonlayer}{background}
\fill[ltri] (2,0)--(1,-1)--(0,0)--cycle;
\draw[ldotty] (0,0)--(2,0);
\draw[ledge] (2,0)--(1,-1)--(0,0);
\end{pgfonlayer}
\ld{(0.5,-0.5)}
\ld{(1.5,-0.5)}
\end{scope}
}

\newcommand{\e}[1]{\left\langle {#1} \right\rangle}

\newcommand{\floor}[1]{{\left\lfloor {#1} \right\rfloor}}

\newtheorem{theorem}{Theorem}[section]
\newtheorem{proposition}[theorem]{Proposition}

\newtheorem{corollary}[theorem]{Corollary}

\title{\vspace{-15mm}The ghost algebra and the dilute ghost algebra}
\author{\vspace{-3mm}Madeline Nurcombe}
\date{}

\numberwithin{equation}{section}

\begin{document}
\maketitle

\vspace{-5mm}
\begin{abstract}
We introduce the \textit{ghost algebra}, a two-boundary generalisation of the Temperley-Lieb (TL) algebra, using a diagrammatic presentation. The existing two-boundary TL algebra has a basis of string diagrams with two boundaries, and the number of strings connected to each boundary must be even; in the ghost algebra, this number may be odd. To preserve associativity while allowing boundary-to-boundary strings to have distinct parameters according to the parity of their endpoints, as seen in the one-boundary TL algebra, we decorate the boundaries with bookkeeping dots called \textit{ghosts}. We also introduce the \textit{dilute ghost algebra}, an analogous two-boundary generalisation of the dilute TL algebra. We then present loop models associated with these algebras, and classify solutions to their boundary Yang-Baxter equations, given existing solutions to the Yang-Baxter equations for the TL and dilute TL models. This facilitates the construction of a one-parameter family of commuting transfer tangles, making these models Yang-Baxter integrable.
\end{abstract}

\vspace{-3mm}

\tableofcontents

\section{Introduction} \label{s:intro}
Lattice loop models are of interest in statistical mechanics, where they provide a tractable approach to modelling physical systems with non-local interactions. A key step towards solving such models is to construct a family of commuting operators called \textit{transfer tangles}, dependent on one parameter $u$, often called the spectral parameter. There is a known method for constructing commuting transfer tangles for a two-dimensional square lattice using local operators called \textit{face operators} satisfying certain relations. Indeed, if the face operators satisfy the \textit{Yang-Baxter equation} (YBE), the \textit{boundary Yang-Baxter equations} (BYBEs), and a \textit{local inversion relation}, then they can be composed into a transfer tangle $T(u)$ such that $T(u)T(v) = T(v)T(u)$ for all $u$ and $v$; see, for example, \cite{BehrendPearceOBrien}. The resulting model is then called \textit{Yang-Baxter integrable}.

Lattice loop models are often encoded algebraically using diagram algebras. These algebras have a set of diagrams for a basis, on which multiplication is defined based on concatenation of diagrams. Perhaps the most famous is the Temperley-Lieb (TL) algebra \cite{TL,Jones}. This algebra was first discussed by Temperley and Lieb in \cite{TL}, where it was found to connect transfer matrices in ice-type models to certain graph colouring problems, and to the enumeration of weighted trees. It has since been applied to a range of physical models---from polymers and percolation \cite{LMM} to quantum spin chains \cite{PasquierSaleur}---and has also been used in pure mathematics for knot theory \cite{JonesPnom,Kauffman}.

In our loop model context, the TL algebra $\TL_n(\beta)$ describes a fully-packed square lattice of non-crossing strings with closed boundary conditions. The one- and two-boundary TL algebras $\BTL_n(\beta;\alpha_1,\alpha_2)$ and $\BBTL_n(\beta;\alpha_1,\alpha_2,\delta_1,\delta_2)$ generalise this to allow open boundary conditions on one or both sides of the lattice, respectively. As diagram algebras, each of these has a basis of rectangular string diagrams, considered up to \textit{connectivity}, or more precisely, frame-preserving ambient isotopy. Some examples are given in Figure \ref{fig:diagexamples}. Each diagram has $n$ nodes on each of two opposing sides of the rectangle, between which non-crossing strings are drawn. In the one- and two-boundary algebras, one or both of the remaining sides become dotted boundary lines, to which strings may be connected at distinct points. Each node must have exactly one string endpoint attached to it, and each string must connect two distinct nodes, a node to a boundary, or two distinct boundaries. In the two-boundary TL algebra, however, the number of strings connected to each boundary must be even. This is already implied by the other rules in the one-boundary case, but must be separately imposed for the two-boundary case, as seen in Figure \ref{fig:exnotBBTL}. We will soon see that this evenness condition is essential for $\BBTL_n$ to be associative. 

\begin{figure}[h]
\centering
\begin{subfigure}[b]{0.18\textwidth}
\begin{align*}
\begin{tikzpicture}[baseline={([yshift=-1mm]current bounding box.center)},scale=0.5]
{
\draw [very thick](0,-4.5) -- (0,0.5);
\draw [very thick](3,-4.5)--(3,0.5);
\draw (0,0) arc (90:-90:0.5);
\draw (0,-2) to[out=0,in=180] (3,0);
\draw (0,-3) arc(90:-90:0.5);
\draw (3,-1) arc(90:270:1.5);
\draw (3,-2) arc(90:270:0.5);
}
\end{tikzpicture}
\end{align*}
\caption{$\TL_5$}
\label{fig:exTL}
\end{subfigure}
\begin{subfigure}[b]{0.18\textwidth}
\begin{align*}
\begin{tikzpicture}[baseline={([yshift=-1mm]current bounding box.center)},scale=0.5]
{
\draw[dotted] (0,0.5)--(3,0.5);
\draw [very thick](0,-4.5) -- (0,0.5);
\draw [very thick](3,-4.5)--(3,0.5);
\draw (0,0) arc (-90:0:0.5);
\draw (0,-1) arc (-90:0:1.5);
\draw (0,-2) to[out=0,in=180] (3,0);
\draw (0,-3) to[out=0,in=180] (3,-1);
\draw (0,-4) -- (3,-4);
\draw (3,-2) arc(90:270:0.5);
}
\end{tikzpicture}
\end{align*}
\caption{$\BTL_5$}
\label{fig:exBTL}
\end{subfigure}
\begin{subfigure}[b]{0.18\textwidth}
\begin{align*}
\begin{tikzpicture}[baseline={([yshift=-1mm]current bounding box.center)},scale=0.5]
{
\draw[dotted] (0,0.5)--(3,0.5);
\draw[dotted] (0,-4.5)--(3,-4.5);
\draw [very thick](0,-4.5) -- (0,0.5);
\draw [very thick](3,-4.5)--(3,0.5);
\draw (0,0) to[out=0,in=0] (0,-3);
\draw (0,-1) arc (90:-90:0.5);
\draw (0,-4) arc(90:0:0.5);
\draw (3,0) arc (90:270:0.5);
\draw (3,-2) to[out=180,in=270] (2,0.5);
\draw (3,-3) arc(90:270:0.5);
\draw (1.25,0.5)--(1.25,-4.5);
\draw (1.5,0.5)--(1.5,-4.5);
\draw (1.75,0.5)--(1.75,-4.5);
}
\end{tikzpicture}
\end{align*}
\caption{$\BBTL_5$}
\label{fig:exBBTL}
\end{subfigure}
\begin{subfigure}[b]{0.18\textwidth}
\begin{align*}
\begin{tikzpicture}[baseline={([yshift=-1mm]current bounding box.center)},scale=0.5]
{
\draw[dotted] (0,0.5)--(3,0.5);
\draw[dotted] (0,-4.5)--(3,-4.5);
\draw [very thick](0,-4.5) -- (0,0.5);
\draw [very thick](3,-4.5)--(3,0.5);
\draw (0,0) to[out=0,in=180] (3,-3);
\draw (0,-1) to[out=0,in=90] (1.2,-4.5);
\draw (0,-2) arc(90:-90:0.5);
\draw (0,-4) arc(90:0:0.5);
\draw (3,0) arc(270:180:0.5);
\draw (3,-1) arc (90:270:0.5);
\draw (3,-4) arc(90:180:0.5);
}
\end{tikzpicture}
\end{align*}
\caption{$\notin \BBTL_5$}
\label{fig:exnotBBTL}
\end{subfigure}
\begin{subfigure}[b]{0.18\textwidth}
\begin{align*}
\begin{tikzpicture}[baseline={([yshift=-1mm]current bounding box.center)},scale=0.5]
{
\draw [very thick](0,-4.5) -- (0,0.5);
\draw [very thick](3,-4.5)--(3,0.5);
\draw (0,0) to[out=0,in=180] (3,-1);
\draw (0,-1) arc (90:-90:1);
\draw (3,-2) arc(90:270:1);
\filldraw[fill=white] (0,-2) circle (0.12);
\filldraw[fill=white] (0,-4) circle (0.12);
\filldraw[fill=white] (3,0) circle (0.12);
\filldraw[fill=white] (3,-3) circle (0.12);
}
\end{tikzpicture}
\end{align*}
\caption{$\dTL_5$}
\label{fig:exdTL}
\end{subfigure}
\caption{Examples of basis diagrams from different diagram algebras. Figure \ref{fig:exnotBBTL} is included as an example of a two-boundary diagram that is \textit{not} a basis diagram of $\BBTL_5$, because it has an odd number of strings connected to each boundary.}
    \label{fig:diagexamples}
\end{figure}
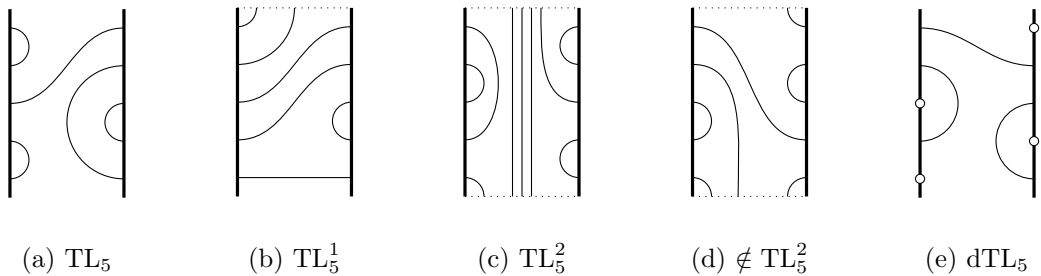

The primary purpose of the ghost algebra constructed in this paper is to provide a two-boundary generalisation of the TL and one-boundary TL algebras that allows basis diagrams with an odd number of strings connected to either boundary. In particular, we want this algebra to be associative, like $\TL_n$ and $\BTL_n$, and to have subalgebras isomorphic to $\BTL_n$ associated with each boundary, whose intersection is isomorphic to $\TL_n$. The specifics of multiplication in $\BTL_n$ cause this to be nontrivial, however.

The one-boundary TL algebra was introduced as the blob algebra $\blob_n(\beta;\alpha)$ by Martin and Saleur in \cite{MartinSaleur}. We use the slightly generalised three-parameter definition and diagrammatic presentation from Pearce, Rasmussen and Tipunin in \cite{Tipunin}, denoted by $\BTL_n(\beta; \alpha_1,\alpha_2)$. The isomorphism $\BTL_n(\beta;\alpha_1,\alpha_2) \cong \blob_n(\beta; \frac{\alpha_1}{\alpha_2})$ for $\alpha_2 \neq 0$ is discussed in \cite[App.~A]{Tipunin}.

The three parameters in $\BTL_n(\beta;\alpha_1,\alpha_2)$ arise during multiplication. First, the diagrams are concatenated. This may produce strings with both ends on the boundary, called \textit{boundary arcs}, and loops. Each string endpoint on the boundary is assigned a parity by numbering them from left to right, starting from 1. Any boundary arcs are then removed and each replaced by a factor of a \textit{boundary parameter} $\alpha_1$ or $\alpha_2$, if the left endpoint is odd or even, respectively. Each loop is removed and replaced by a factor of the \textit{loop parameter} $\beta$. For example, in $\BTL_7(\beta;\alpha_1,\alpha_2)$, 
\begin{align}
\begin{tikzpicture}[baseline={([yshift=-3mm]current bounding box.center)},scale=0.55]
{
\draw[dotted] (0,0.5)--(6,0.5);
\draw [very thick](0,-6.5) -- (0,0.5);
\draw [very thick](3,-6.5)--(3,0.5);
\draw [very thick](6,-6.5)--(6,0.5);
\draw (0,0) arc(-90:0:0.6 and 0.5);
\draw (0,-1) arc (90:-90:0.5);
\draw (0,-3) to[out=0,in=-90] (1.2,0.5);
\draw (0,-4) to[out=0,in=180] (3,-6);
\draw (0,-5) arc(90:-90:0.5);
\draw (3,0) arc(270:180:0.6 and 0.5);
\draw (3,-1) to[out=180,in=180] (3,-4);
\draw (3,-2) arc(90:270:0.5);
\draw (3,-5) ..controls (1,-5) and (1.8,-1.5).. (1.8,0.5);
\draw (3,0) arc(-90:0:0.6 and 0.5);
\draw (3,-1) ..controls (4,-1) and (4.2,0).. (4.2,0.5);
\draw (3,-2) arc(90:-90:0.5);
\draw (3,-4) ..controls (4.2,-4) and (4.8,-0.5).. (4.8,0.5);
\draw (3,-5) ..controls (4.5,-5) and (5.3,-1).. (5.4,0.5);
\draw (3,-6) ..controls (4.7,-6) and (5,-2).. (6,-2);
\draw (6,0) arc(90:270:0.45 and 0.5);
\draw (6,-3) to[out=180,in=180] (6,-6);
\draw (6,-4) arc(90:270:0.5);
\node at (0.6,0.35) [anchor=south] {\scriptsize 1};
\node at (1.2,0.35) [anchor=south] {\scriptsize 2};
\node at (1.8,0.35) [anchor=south] {\scriptsize 3};
\node at (2.4,0.35) [anchor=south] {\scriptsize 4};
\node at (3.6,0.35) [anchor=south] {\scriptsize 5};
\node at (4.2,0.35) [anchor=south] {\scriptsize 6};
\node at (4.8,0.35) [anchor=south] {\scriptsize 7};
\node at (5.4,0.35) [anchor=south] {\scriptsize 8};
}
\end{tikzpicture}
\;= \,\beta \alpha_1 \alpha_2^2\;
\begin{tikzpicture}[baseline={([yshift=-3mm]current bounding box.center)},scale=0.55]
{
\draw[dotted] (0,0.5)--(3,0.5);
\draw [very thick](0,-6.5) -- (0,0.5);
\draw [very thick](3,-6.5)--(3,0.5);
\draw (0,0) arc(-90:0:0.6 and 0.5);
\draw (0,-1) arc (90:-90:0.5);
\draw (0,-3) to[out=0,in=-90] (1.2,0.5);
\draw (0,-4) to[out=0,in=180] (3,-2);
\draw (0,-5) arc(90:-90:0.5);
\draw (3,0) arc(90:270:0.45 and 0.5);
\draw (3,-3) to[out=180,in=180] (3,-6);
\draw (3,-4) arc(90:270:0.5);
\node at (0.6,0.35) [anchor=south] {\vphantom{\scriptsize 1}};
}
\end{tikzpicture}\; .
\end{align}

The two-boundary TL algebra has a similar diagrammatic multiplication. This algebra was introduced by Mitra, Nienhuis, de Gier and Batchelor in \cite{MNdGB}, with one parameter associated with each boundary, similar to the blob algebra. Its structure, representation theory and applications to spin chains have been studied widely, along with a number of very similar algebras. 

In \cite{dGN}, de Gier and Nichols related the two-boundary TL algebra to the affine Hecke algebra of type C, using the action of its centre to show that all irreducible representations factor through a finite-dimensional quotient, later called the \textit{symplectic blob algebra}. For generic values of a parameter associated with the action of the centre, they found a representation of dimension $2^n$ that can be expressed as a tensor product of representations, and thus interpreted as a spin chain. They then established under what circumstances this representation is irreducible. 

The symplectic blob algebra was further studied by Green, King, Martin and Parker in a series of papers covering various topics in its representation theory \cite{TowersI,NonGenRepThySympBlob,SympBlobBlocks,QHSympBlob}, as well as a rigorous treatment of its presentation \cite{SympBlobPres}. 

In \cite{DaughertyRam}, Daugherty and Ram analysed the \textit{calibrated} representations of the two-boundary TL algebra, wherein the central elements from the affine Hecke algebra are simultaneously diagonalisable. More recently, Chernyak, Gainutdinov and Saleur investigated boundary conditions for an XXZ spin chain in \cite{ChernyakGainutdinovSaleur}, and in the two-boundary case, produced representations of the \textit{universal two-boundary TL algebra}, a central extension of the symplectic blob algebra.

For consistency with our conventions for the one-boundary TL algebra, in this paper, we use the version of the two-boundary TL algebra with two parameters associated with each boundary, as is done in \cite{SympBlobPres}. We denote this by $\BBTL_n(\beta;\alpha_1,\alpha_2,\delta_1,\delta_2)$; the parameters $\delta_1$ and $\delta_2$ are assigned to bottom boundary arcs with odd and even left endpoint, respectively.

Suppose we allow two-boundary diagrams with an odd number of strings attached to each boundary, and use the diagram multiplication rules from $\BBTL_n(\beta;\alpha_1,\alpha_2,\delta_1,\delta_2)$. Then with $n=2$, we have
\vspace{-3mm}
\begin{align}
\scalebox{1.4}{\bigg(}\,
\begin{tikzpicture}[baseline={([yshift=5mm]current bounding box.south)},scale=0.6]
{
\draw[dotted] (0,0.5)--(4,0.5);
\draw[dotted] (0,-1.5)--(4,-1.5);
\draw[very thick] (0,0.5)--(0,-1.5);
\draw[very thick] (2,0.5)--(2,-1.5);
\draw[very thick] (4,0.5)--(4,-1.5);
\draw (0,0) arc (-90:0:0.5);
\draw (0,-1) to[out=0,in=180] (2,0);
\draw (2,-1) arc (90:180:0.5);
\draw (2,0) arc (-90:0:0.5);
\draw (2,-1) -- (4,-1);
\draw (4,0) arc (270:180:0.5);
}
\end{tikzpicture}\,
\scalebox{1.4}{\bigg)}\,
\begin{tikzpicture}[baseline={([yshift=5mm]current bounding box.south)},scale=0.6]
{
\draw[dotted] (0,0.5)--(2,0.5);
\draw[dotted] (0,-1.5)--(2,-1.5);
\draw[very thick] (0,0.5)--(0,-1.5);
\draw[very thick] (2,0.5)--(2,-1.5);
\draw (0,0) arc (-90:0:0.5);
\draw (0,-1) -- (2,-1);
\draw (2,0) arc (270:180:0.5);
}
\end{tikzpicture}
\, &= \,
\begin{tikzpicture}[baseline={([yshift=5mm]current bounding box.south)},scale=0.6]
{
\draw[dotted] (0,0.5)--(4,0.5);
\draw[dotted] (0,-1.5)--(4,-1.5);
\draw (0,0) arc (-90:0:0.5);
\draw (0,-1) to[out=0,in=-90] (1,0.5);
\draw (2,-1) arc (90:180:0.5);
\draw (2,0) arc (270:180:0.5);
\draw (2,0) arc (-90:0:0.5);
\draw (2,-1) -- (4,-1);
\draw (4,0) arc (270:180:0.5);
\node at (0.5,0.35) [anchor=south] {\scriptsize 1};
\node at (1,0.35) [anchor=south] {\scriptsize 2};
\node at (1.5,0.35) [anchor=south] {\scriptsize 3};
\node at (2.5,0.35) [anchor=south] {\scriptsize 4};
\node at (3.5,0.35) [anchor=south] {\scriptsize 5};
\draw[very thick] (0,0.5)--(0,-1.5);
\draw[very thick] (2,0.5)--(2,-1.5);
\draw[very thick] (4,0.5)--(4,-1.5);
}
\end{tikzpicture}
\, = 
\alpha_1\,
\begin{tikzpicture}[baseline={([yshift=5mm]current bounding box.south)},scale=0.6]
{
\draw[dotted] (0,0.5)--(2,0.5);
\draw[dotted] (0,-1.5)--(2,-1.5);
\draw[very thick] (0,0.5)--(0,-1.5);
\draw[very thick] (2,0.5)--(2,-1.5);
\draw (0,0) arc (-90:0:0.5);
\draw (0,-1) to[out=0,in=-90] (1,0.5);
\draw (2,-1) arc (90:180:0.5);
\draw (2,0) arc (270:180:0.5);
}
\end{tikzpicture}\; ,
\end{align}
but
\begin{align}
\begin{tikzpicture}[baseline={([yshift=5mm]current bounding box.south)},scale=0.6]
{
\draw[dotted] (0,0.5)--(2,0.5);
\draw[dotted] (0,-1.5)--(2,-1.5);
\draw[very thick] (0,0.5)--(0,-1.5);
\draw[very thick] (2,0.5)--(2,-1.5);
\draw (0,0) arc (-90:0:0.5);
\draw (0,-1) to[out=0,in=180] (2,0);
\draw (2,-1) arc (90:180:0.5);
}
\end{tikzpicture}
\,\scalebox{1.4}{\bigg(}\,
\begin{tikzpicture}[baseline={([yshift=5mm]current bounding box.south)},scale=0.6]
{
\draw[dotted] (0,0.5)--(4,0.5);
\draw[dotted] (0,-1.5)--(4,-1.5);
\draw[very thick] (0,0.5)--(0,-1.5);
\draw[very thick] (4,0.5)--(4,-1.5);
\draw (0,0) arc (-90:0:0.5);
\draw (0,-1) -- (2,-1);
\draw (2,0) arc (270:180:0.5);
\draw (2,0) arc (-90:0:0.5);
\draw (2,-1) -- (4,-1);
\draw (4,0) arc (270:180:0.5);
\node at (0.5,0.35) [anchor=south] {\scriptsize 1};
\node at (1.5,0.35) [anchor=south] {\scriptsize 2};
\node at (2.5,0.35) [anchor=south] {\scriptsize 3};
\node at (3.5,0.35) [anchor=south] {\scriptsize 4};
\draw[very thick] (2,0.5)--(2,-1.5);
}
\end{tikzpicture}\,
\scalebox{1.4}{\bigg)}
 = 
\alpha_2\,
\begin{tikzpicture}[baseline={([yshift=5mm]current bounding box.south)},scale=0.6]
{
\draw[dotted] (0,0.5)--(4,0.5);
\draw[dotted] (0,-1.5)--(4,-1.5);
\draw[very thick] (0,0.5)--(0,-1.5);
\draw[very thick] (2,0.5)--(2,-1.5);
\draw[very thick] (4,0.5)--(4,-1.5);
\draw (0,0) arc (-90:0:0.5);
\draw (0,-1) to[out=0,in=180] (2,0);
\draw (2,-1) arc (90:180:0.5);
\draw (2,0) arc (-90:0:0.5);
\draw (2,-1) -- (4,-1);
\draw (4,0) arc (270:180:0.5);
}
\end{tikzpicture}
\, =
\alpha_2\,
\begin{tikzpicture}[baseline={([yshift=5mm]current bounding box.south)},scale=0.6]
{
\draw[dotted] (0,0.5)--(2,0.5);
\draw[dotted] (0,-1.5)--(2,-1.5);
\draw[very thick] (0,0.5)--(0,-1.5);
\draw[very thick] (2,0.5)--(2,-1.5);
\draw (0,0) arc (-90:0:0.5);
\draw (0,-1) to[out=0,in=-90] (1,0.5);
\draw (2,-1) arc (90:180:0.5);
\draw (2,0) arc (270:180:0.5);
}
\end{tikzpicture} \; .
\end{align}
Associativity requires that these products are equal. As they are scalar multiples of the same basis diagram, this means that their coefficients $\alpha_1$ and $\alpha_2$ must be equal. However, much of the structural richness of the $\BTL_n(\beta;\alpha_1,\alpha_2)$ subalgebra would be lost by considering only the cases where $\alpha_1=\alpha_2$, or, equivalently for $\alpha_2 \neq 0$, $\blob_n(\beta;1)$. For example, as seen in \cite{MartinSaleur,MartinWoodcock}, the representation theory of $\blob_n(\beta; \alpha)$ differs when $\alpha$ is specialised to different complex numbers. Hence we wish to allow $\alpha_1 \neq \alpha_2$.

As seen above, naively including two-boundary diagrams with an odd number of strings connected to either boundary does not yield the desired associative algebra with potentially distinct $\alpha_1$ and $\alpha_2$. Indeed, left-multiplication by any diagram with an odd number of strings on a boundary changes the parity of all the strings to its right on that boundary, because we count string parity starting from the left. Starting from the right gives the same problem with right-multiplication. Hence we seek an alternative method of keeping track of the parity of boundary connections that is robust to multiplication by such diagrams. 

This motivates the introduction of \textit{ghosts}: dots on the boundaries that allow us to keep track of the parity of strings. A ghost is left at each endpoint of each boundary arc removed during multiplication. When determining the parity of string endpoints on each boundary, we number the ghosts as well as the string endpoints. We also require the sum of the number of ghosts and the number of string endpoints on each boundary to be even. This allows us to have diagrams with an odd number of strings connected to each boundary, and parity-dependent boundary parameters, while preserving associativity.

Similar problems with associativity occur when trying to produce generalisations of the dilute TL algebra with boundary, and these can also be resolved by introducing ghosts. The dilute TL algebra $\dTL_n(\beta)$ describes \textit{dilute} lattice configurations, where the strings need not be fully-packed. The basis diagrams for $\dTL_n$ follow the same rules as the basis diagrams for $\TL_n$, except there may be \textit{empty} nodes, with no attached strings, as seen in Figure \ref{fig:exdTL}. Naturally, we would like one- and two-boundary generalisations of $\dTL_n(\beta)$ to describe dilute systems with open boundary conditions on one or both boundaries, but existing versions lack some of our desired properties. In \cite{DJS}, Dubail, Jacobsen and Saleur use a dilute version of the blob algebra; after translating their diagrams from the blob picture to the boundary picture as in \cite[App.~A]{Tipunin}, this does effectively have distinct boundary parameters, but all diagrams have an even number of strings connected to the boundary. De Gier, Lee and Rasmussen study a dilute one-boundary lattice model in \cite{deGierLeeRasmussen} that allows diagrams with an odd number of strings connected to the boundary, with a subalgebra isomorphic to $\BTL_n(\beta;\alpha_1,\alpha_2)$, but they found that nontrivial solutions to the BYBE exist only if all boundary parameters are equal. Alluqmani constructs a dilute blob algebra in \cite{Alluqmani} that also allows diagrams with an odd number of strings connected to the boundary, after converting to the boundary picture, but has only one boundary parameter. This algebra is actually based on the Motzkin algebra $\M_n(\beta,\eta)$ (see e.g. \cite{Motzkin}), but there is an isomorphism $\M_n(\beta, \eta) \cong \dTL_n(\beta-1)$ as long as $\eta \neq 0$, to be discussed in future work.

The dilute ghost algebra $\dGh_n$ constructed in this paper is an associative two-boundary generalisation of $\dTL_n$ that allows diagrams with odd numbers of strings connected to the boundaries, has parity-dependent parameters at each boundary, and admits solutions to the BYBE for generic values of the parameters. It also has a one-boundary subalgebra $\dGho_n$ with these desired properties.

The layout of this paper is as follows. In Section \ref{s:ghostalg}, we define the ghost algebra. We first define the diagram basis in Section \ref{ss:ndiags}, and explain the multiplication in Section \ref{ss:alg}. In Section \ref{s:dilute}, we recall the definition of the dilute TL algebra, and then define the dilute ghost algebra diagrammatically. Section \ref{s:relations} discusses key subalgebras of the ghost algebra and the dilute ghost algebra, and how they fit together. In Section \ref{s:loopmodels}, we present lattice loop models associated with the ghost algebra and the dilute ghost algebra. We construct transfer tangles from bulk and boundary face operators satisfying the YBE, crossing symmetry, the local inversion relation, and the BYBEs. We use existing bulk face operators from the TL and dilute TL algebras, and give a classification of the new boundary face operators with ghosts that satisfy the BYBEs for generic parameter values. For the dilute case, this involves solving 110 functional equations for ten unknown functions, and takes up approximately a third of the paper. Mathematica code was used to multiply diagrams and simplify their coefficients.

The dimensions of the algebras introduced in this paper are given in Appendix \ref{app:dims}, and associativity is proven in Appendix \ref{app:assoc}. In Appendix \ref{app:cellular}, we show that the ghost algebra and its dilute and one-boundary counterparts are \textit{cellular}, as defined by Graham and Lehrer in \cite{GL}. We note that the TL and one-boundary TL algebras are cellular, and that this is closely tied to their diagrammatic presentations, using an anti-involution defined by reflection of basis diagrams. Accordingly, the ghost algebra and the dilute ghost algebra are constructed to be cellular with respect to this reflection anti-involution. However, relaxing this requirement leads to some natural generalisations of these algebras that are still associative and have parity-dependent boundary parameters; the generalised algebras are discussed in Appendix \ref{app:nc}. The classification of BYBE solutions for the ghost algebra includes many diagrammatic details; the proof for the dilute ghost algebra is more streamlined, but some of the missing details are given in Appendix \ref{app:configs}.

\subsection{Terminology and notation}
Let $A$ be an algebra. We say $A$ is \textit{unital} if it has a multiplicative identity. Let $B$ be a vector subspace of $A$ that is closed under multiplication. If $B$ contains the multiplicative identity of $A$, then it is a \textit{unital subalgebra} of $A$; if not, it is a \textit{non-unital subalgebra} of $A$. Note that a non-unital subalgebra $B$ of $A$ may be a unital algebra in its own right; in this case, the multiplicative identity of $B$ does not coincide with that of $A$, but is instead some other idempotent.

Let $A_1$ and $A_2$ be unital algebras with identities $I_1$ and $I_2$ respectively, and let $\psi: A_1 \to A_2$ be a linear map that preserves multiplication. That is, $\psi(xy) = \psi(x)\psi(y)$ for all $x,y \in A_1$. If $\psi(I_1) = I_2$, then $\psi$ is a \textit{unital} homomorphism; otherwise, it is a \textit{non-unital} homomorphism.

All subalgebras and homomorphisms in this paper are unital unless otherwise specified. All algebras are over $\C$, and all parameters are indeterminates, though they may be specialised to elements of $\C$.

Each diagram algebra in this paper has a superscript indicating how many boundaries it has, or no superscript, if it has no boundary.

\section{The ghost algebra} \label{s:ghostalg}
In this section, we define the ghost algebra using a diagrammatic presentation, with a basis given in Section \ref{ss:ndiags} and multiplication given in Section \ref{ss:alg}.

\subsection{Diagram basis}\label{ss:ndiags}
We define the \textit{ghost algebra} $\Gh_n$ in terms of a diagram basis. A \textit{$\Gh_n$-diagram} is  drawn on a rectangle where the top and bottom sides are dotted \textit{boundary} lines, and $n$ nodes are placed on each of the remaining sides. Non-crossing strings are drawn within the rectangle such that each node has exactly one string attached to it, and each string connects a node to another node, or to a boundary. Hence there are no loops, and no strings with each end connected to a boundary. There is no restriction on the number of strings attached to either boundary.

A finite number of filled black circles called \textit{ghosts} may be drawn along the boundaries, except at the ends of any strings. At each boundary, we require that the sum of the number of ghosts and the number of strings attached to that boundary is even; this weaker evenness condition will ensure associativity of the algebra.

For example, some $\Gh_6$-diagrams are
\begin{align}
\begin{tikzpicture}[baseline={([yshift=-1mm]current bounding box.center)},xscale=0.5,yscale=0.5]
{
\draw[dotted] (0,0.5)--(3,0.5);
\draw[dotted] (0,-5.5)--(3,-5.5);
\draw [very thick](0,-5.5) -- (0,0.5);
\draw [very thick](3,-5.5)--(3,0.5);
\draw (0,0) arc (-90:0:0.5);
\draw (0,-1) arc (90:-90:0.5);
\draw (0,-3) to[out=0,in=-90] (1.3,0.5);
\draw (0,-4) to[out=0,in=180] (3,-1);
\draw (0,-5) to[out=0,in=180] (3,-2);
\draw (3,0) arc (270:180:0.5);
\draw (3,-3) to[out=180,in=90] (1.7,-5.5);
\draw (3,-4) arc (90:270:0.5);
\filldraw (0.9,0.5) circle (0.08);
\filldraw (1.7,0.5) circle (0.08);
\filldraw (2.1,0.5) circle (0.08);
\filldraw (0.85,-5.5) circle (0.08);
}
\end{tikzpicture}\; ,
\quad\ \ 
\begin{tikzpicture}[baseline={([yshift=-1mm]current bounding box.center)},xscale=0.5,yscale=0.5]
{
\draw[dotted] (0,0.5)--(3,0.5);
\draw[dotted] (0,-5.5)--(3,-5.5);
\draw [very thick](0,-5.5) -- (0,0.5);
\draw [very thick](3,-5.5)--(3,0.5);
\draw (0,0) to[out=0,in=180] (3,-4);
\draw (0,-1) to[out=0,in=0] (0,-4);
\draw (0,-2) arc (90:-90:0.5);
\draw (0,-5) -- (3,-5);
\draw (3,0) arc (90:270:0.5);
\draw (3,-2) arc (90:270:0.5);
}
\end{tikzpicture}\; ,
\quad\ \ 
\begin{tikzpicture}[baseline={([yshift=-1mm]current bounding box.center)},xscale=0.5,yscale=0.5]
{
\draw[dotted] (0,0.5)--(3,0.5);
\draw[dotted] (0,-5.5)--(3,-5.5);
\draw [very thick](0,-5.5) -- (0,0.5);
\draw [very thick](3,-5.5)--(3,0.5);
\draw (0,0) arc (90:0:1.65 and 5.5);
\draw (0,-1) arc (90:0:1.375 and 4.5);
\draw (0,-2) arc (90:0:1.1 and 3.5);
\draw (0,-3) arc (90:0:0.825 and 2.5);
\draw (0,-4) arc (90:0:0.55 and 1.5);
\draw (0,-5) arc (90:0:0.27 and 0.5);
\draw (3,-5) arc (270:180:1.65 and 5.5);
\draw (3,-4) arc (270:180:1.375 and 4.5);
\draw (3,-3) arc (270:180:1.1 and 3.5);
\draw (3,-2) arc (270:180:0.825 and 2.5);
\draw (3,-1) arc (270:180:0.55 and 1.5);
\draw (3,-0) arc (270:180:0.27 and 0.5);
}
\end{tikzpicture}\; ,
\quad\ \ 
\begin{tikzpicture}[baseline={([yshift=-1mm]current bounding box.center)},xscale=0.5,yscale=0.5]
{
\draw[dotted] (0,0.5)--(3,0.5);
\draw[dotted] (0,-5.5)--(3,-5.5);
\draw [very thick](0,-5.5) -- (0,0.5);
\draw [very thick](3,-5.5)--(3,0.5);
\draw (0,0) arc (-90:0:0.5);
\draw (0,-1) arc (90:-90:0.5);
\draw (0,-3) to[out=0,in=-90] (1.3,0.5);
\draw (0,-4) to[out=0,in=180] (3,-1);
\draw (0,-5) to[out=0,in=180] (3,-2);
\draw (3,0) arc (270:180:0.5);
\draw (3,-3) to[out=180,in=90] (1.7,-5.5);
\draw (3,-4) arc (90:270:0.5);
\filldraw (0.9,0.5) circle (0.08);
\filldraw (0.85,-5.5) circle (0.08);
}
\end{tikzpicture}\; ,
\quad\ \ 
\begin{tikzpicture}[baseline={([yshift=-1mm]current bounding box.center)},xscale=0.5,yscale=0.5]
{
\draw[dotted] (0,0.5)--(3,0.5);
\draw[dotted] (0,-5.5)--(3,-5.5);
\draw [very thick](0,-5.5) -- (0,0.5);
\draw [very thick](3,-5.5)--(3,0.5);
\draw (0,0) arc (-90:0:0.8 and 0.5);
\draw (0,-1) to[out=0,in=180] (3,0);
\draw (0,-2) to[out=0,in=180] (3,-1);
\draw (0,-3) to[out=0,in=180] (3,-2);
\draw (0,-4) to[out=0,in=180] (3,-3);
\draw (0,-5) to[out=0,in=180] (3,-4);
\draw (3,-5) arc (90:180:0.8 and 0.5);
\filldraw (1.9,0.5) circle (0.08);
\filldraw (1.1,-5.5) circle (0.08);
}
\end{tikzpicture}\; ,
\quad\ \ 
\begin{tikzpicture}[baseline={([yshift=-1mm]current bounding box.center)},xscale=0.5,yscale=0.5]
{
\draw[dotted] (0,0.5)--(3,0.5);
\draw[dotted] (0,-5.5)--(3,-5.5);
\draw [very thick](0,-5.5) -- (0,0.5);
\draw [very thick](3,-5.5)--(3,0.5);
\draw (0,0) arc (-90:0:0.5);
\draw (0,-1) arc (90:-90:0.5);
\draw (0,-3) to[out=0,in=-90] (1.3,0.5);
\draw (0,-4) to[out=0,in=180] (3,-1);
\draw (0,-5) to[out=0,in=180] (3,-2);
\draw (3,0) arc (270:180:0.5);
\draw (3,-3) to[out=180,in=90] (1.7,-5.5);
\draw (3,-4) arc (90:270:0.5);
\filldraw (0.9,0.5) circle (0.08);
\filldraw (2.35,-5.5) circle (0.08);
}
\end{tikzpicture}\; 
. \label{eq:egdiags}
\end{align}

A \textit{domain} is the part of a boundary between a pair of adjacent strings, between a corner and its nearest string, or, if there are no strings connected to that boundary, the whole boundary. For example, the first $\Gh_6$-diagram above has four domains at the top boundary, and two at the bottom.

Two $\Gh_n$-diagrams are considered equal if the same nodes are connected to each other or to the same boundary, and in each domain, the number of ghosts is equivalent modulo 2. For example, the first and fourth diagrams in \eqref{eq:egdiags} are equal, while the first and sixth are not. It follows that each $\Gh_n$-diagram can be drawn with at most one ghost in each of its domains.

Note that strings connecting the top boundary to the bottom boundary are disallowed in $\Gh_n$, while there may be arbitrarily many of these in $\BBTL_n$-diagrams. This means $\Gh_n$ is finite-dimensional, like $\TL_n$ and $\BTL_n$, while $\BBTL_n$ is infinite-dimensional. One can consider a quotient of $\BBTL_n$ obtained by removing top-to-bottom boundary arcs at the cost of an additional parameter, but to keep an even number of strings connected to each boundary, these must be removed in pairs. This is often called the \textit{symplectic blob algebra}; see, for example, \cite{SympBlobPres}. 

Readers more comfortable with the blob diagrams of \cite{MartinSaleur} might wish to redraw the $\Gh_n$-diagrams with blobs, instead of boundary connections. However, the usual process to obtain a boundary diagram from a blob diagram is to cut the strings at each blob and connect both ends to the appropriate boundary (see \cite[App.~A]{Tipunin}), leading to an even number of strings connected to each boundary. Hence this does not work for diagrams with an odd number of strings connected to either boundary.

Formulas for the dimensions of the ghost algebra and all other algebras constructed in this paper are found in Appendix \ref{app:dims}, as well as a table of their dimensions for small $n$.

\subsection{Multiplication}\label{ss:alg}
Let $\Gh_n$ be the complex vector space with the set of all $\Gh_n$-diagrams as its basis. To turn this into an algebra, we define multiplication on $\Gh_n$-diagrams, and extend this bilinearly to the whole space. This algebra is dependent on nine parameters, $\beta$, $\alpha_1$, $\alpha_2$, $\alpha_3$, $\gamma_{12}$, $\gamma_3$, $\delta_1$, $\delta_2$ and $\delta_3$. We take these to be indeterminates, but they may be specialised to fixed values in $\C$.

To multiply two $\Gh_n$-diagrams, we first concatenate them. Each loop that appears is removed and replaced by a factor of $\beta$. The ghosts and strings are then numbered along each boundary, starting from 1 on the left. Any string that is attached to a boundary at each end is called a \textit{boundary arc}. Each boundary arc is removed, leaving a ghost at each of its endpoints, and replaced by a parameter according to which boundary each end was attached to, and the parity of the connection points. This is summarised in Table \ref{tab:params}. The extra vertical line segment in the middle is removed, and what remains is a scalar multiple of a $\Gh_n$-diagram. If desired, this may be neatened by continuously deforming the strings and letting each domain have at most one ghost.

\begin{table}[H]
\centering
\caption{Table of boundary arcs and their associated parameters in the ghost algebra.}\label{tab:params}
\vspace{2mm}
\begin{tabular}{|c|c||c|c|}
\hline
Parameter & Boundary arc & Parameter & Boundary arc\\
\hline
$\alpha_1$ & $\begin{tikzpicture}[baseline={([yshift=-1mm]current bounding box.center)},xscale=0.5,yscale=0.5]
{
\draw[dotted] (0,0.5)--(2,0.5);
\draw (0.2,0.5) arc (-180:0:0.8);
\node at (0.2,0.5) [anchor=south] {\scriptsize odd};
\node at (1.8,0.5) [anchor=south] {\scriptsize even};
\node at (1,-0.5) {};
}
\end{tikzpicture}$
&$\delta_1$ &$\begin{tikzpicture}[baseline={([yshift=-1mm]current bounding box.center)},xscale=0.5,yscale=0.5]
{
\draw[dotted] (0,0)--(2,0);
\draw (0.2,0) arc (180:0:0.8);
\node at (0.2,-0.9) [anchor=south] {\scriptsize odd};
\node at (1.8,-0.9) [anchor=south] {\scriptsize even};
\node at (1,1) {};
}
\end{tikzpicture}$\\
\hline
$\alpha_2$ &$\begin{tikzpicture}[baseline={([yshift=-1mm]current bounding box.center)},xscale=0.5,yscale=0.5]
{
\draw[dotted] (0,0.5)--(2,0.5);
\draw (0.2,0.5) arc (-180:0:0.8);
\node at (0.2,0.5) [anchor=south] {\scriptsize even};
\node at (1.8,0.5) [anchor=south] {\scriptsize odd};
\node at (1,-0.5) {};
}
\end{tikzpicture}$ 
&$\delta_2$& $\begin{tikzpicture}[baseline={([yshift=-1mm]current bounding box.center)},xscale=0.5,yscale=0.5]
{
\draw[dotted] (0,0)--(2,0);
\draw (0.2,0) arc (180:0:0.8);
\node at (0.2,-0.9) [anchor=south] {\scriptsize even};
\node at (1.8,-0.9) [anchor=south] {\scriptsize odd};
\node at (1,1) {};
}
\end{tikzpicture}$\\
\hline
$\alpha_3$ & $\begin{tikzpicture}[baseline={([yshift=4mm]current bounding box.south)},xscale=0.5,yscale=0.5]
{
\draw[dotted] (0,0.5)--(2,0.5);
\draw (0.2,0.5) arc (-180:0:0.8);
\node at (0.2,0.5) [anchor=south] {\scriptsize odd};
\node at (1.8,0.5) [anchor=south] {\scriptsize odd};
\node at (1,-0.5) {};
}
\end{tikzpicture}\, , \  \begin{tikzpicture}[baseline={([yshift=4mm]current bounding box.south)},xscale=0.5,yscale=0.5]
{
\draw[dotted] (0,0.5)--(2,0.5);
\draw (0.2,0.5) arc (-180:0:0.8);
\node at (0.2,0.5) [anchor=south] {\scriptsize even};
\node at (1.8,0.5) [anchor=south] {\scriptsize even};
\node at (1,-0.5) {};
}
\end{tikzpicture}$
&$\delta_3$ &$\begin{tikzpicture}[baseline={([yshift=-1mm]current bounding box.center)},xscale=0.5,yscale=0.5]
{
\draw[dotted] (0,0)--(2,0);
\draw (0.2,0) arc (180:0:0.8);
\node at (0.2,-0.9) [anchor=south] {\scriptsize odd};
\node at (1.8,-0.9) [anchor=south] {\scriptsize odd};
\node at (1,1) {};
}
\end{tikzpicture}\, , \ \begin{tikzpicture}[baseline={([yshift=-1mm]current bounding box.center)},xscale=0.5,yscale=0.5]
{
\draw[dotted] (0,0)--(2,0);
\draw (0.2,0) arc (180:0:0.8);
\node at (0.2,-0.9) [anchor=south] {\scriptsize even};
\node at (1.8,-0.9) [anchor=south] {\scriptsize even};
\node at (1,1) {};
}
\end{tikzpicture}$\\
\hline \hline
$\gamma_{12}$ & $\begin{tikzpicture}[baseline={([yshift=5mm]current bounding box.south)},xscale=0.5,yscale=0.5]
{
\draw[dotted] (0,0)--(2,0);
\draw[dotted] (0,1)--(2,1);
\draw (1,1)--(1,0);
\node at (1,1) [anchor=south] {\scriptsize odd};
\node at (1,-0.9) [anchor=south] {\scriptsize even};
}
\end{tikzpicture}\, , \ \  \begin{tikzpicture}[baseline={([yshift=5mm]current bounding box.south)},xscale=0.5,yscale=0.5]
{
\draw[dotted] (0,0)--(2,0);
\draw[dotted] (0,1)--(2,1);
\draw (1,1)--(1,0);
\node at (1,1) [anchor=south] {\scriptsize even};
\node at (1,-0.9) [anchor=south] {\scriptsize odd};
}
\end{tikzpicture}$
&$\gamma_3$ & $\begin{tikzpicture}[baseline={([yshift=5mm]current bounding box.south)},xscale=0.5,yscale=0.5]
{
\draw[dotted] (0,0)--(2,0);
\draw[dotted] (0,1)--(2,1);
\draw (1,1)--(1,0);
\node at (1,1) [anchor=south] {\scriptsize odd};
\node at (1,-0.9) [anchor=south] {\scriptsize odd};
}
\end{tikzpicture}\,  , \ \  \begin{tikzpicture}[baseline={([yshift=5mm]current bounding box.south)},xscale=0.5,yscale=0.5]
{
\draw[dotted] (0,0)--(2,0);
\draw[dotted] (0,1)--(2,1);
\draw (1,1)--(1,0);
\node at (1,1) [anchor=south] {\scriptsize even};
\node at (1,-0.9) [anchor=south] {\scriptsize even};
}
\end{tikzpicture}$
\\
\hline
\end{tabular}
\end{table}

For example, with $n=8$, we have
\begin{align}
\begin{tikzpicture}[baseline={([yshift=-1mm]current bounding box.center)},xscale=0.6,yscale=0.6]
{
\draw[dotted] (0,0.5)--(6,0.5);
\draw[dotted] (0,-7.5)--(6,-7.5);
\draw [very thick] (0,-7.5)--(0,0.5);
\draw [very thick] (3,-7.5)--(3,0.5);
\draw [very thick] (6,-7.5)--(6,0.5);
\draw (0,0) arc (-90:0:0.6 and 0.5);
\draw (0,-1) arc (90:-90:0.5);
\draw (0,-3) arc (-90:0:1.2 and 3.5);
\draw (0,-4) arc (90:-90:1.2 and 1.5);
\draw (0,-5) arc (90:-90:0.5);
\node at (0.6,0.5) [anchor=south] {\scriptsize 1};
\node at (1.2,0.5) [anchor=south] {\scriptsize 2};
\filldraw (0.75,-7.5) circle (0.08);
\node at (0.75,-7.5) [anchor=north] {\scriptsize 1};
\draw (3,0) arc (270:180:0.6 and 0.5);
\draw (3,-1) arc (270:180:1.2 and 1.5);
\draw (3,-2) arc (90:180:1.5 and 5.5);
\draw (3,-3) arc (90:270:0.5);
\draw (3,-5) arc (90:180:0.9 and 2.5);
\draw (3,-6) arc (90:270:0.5);
\node at (2.4,0.5) [anchor=south] {\scriptsize 4};
\node at (1.8,0.5) [anchor=south] {\scriptsize 3};
\filldraw (2.55,-7.5) circle (0.08);
\node at (1.5,-7.5) [anchor=north] {\scriptsize 2};
\node at (2.1,-7.5) [anchor=north] {\scriptsize 3};
\node at (2.55,-7.5) [anchor=north] {\scriptsize 4};
\draw (3,0) arc (90:-90:0.5);
\draw (3,-2) to[out=0,in=-90] (4,0.5);
\draw (3,-3) arc (90:-90:0.5);
\draw (3,-5) arc (-90:0:1.8 and 5.5);
\draw (3,-6) to[out=0,in=180] (6,-7);
\draw (3,-7) arc (90:0:0.6 and 0.5);
\filldraw (3.5,0.5) circle (0.08);
\filldraw (4.4,0.5) circle (0.08);
\node at (3.5,0.5) [anchor=south] {\scriptsize 5};
\node at (4,0.5) [anchor=south] {\scriptsize 6};
\node at (4.4,0.5) [anchor=south] {\scriptsize 7};
\node at (4.8,0.5) [anchor=south] {\scriptsize 8};
\node at (3.6,-7.5) [anchor=north] {\scriptsize 5};
\draw (6,0) arc (270:180:0.8 and 0.5);
\draw (6,-1) arc (90:270:1.2 and 1.5);
\draw (6,-2) arc (90:270:0.5);
\draw (6,-5) arc (90:270:0.5);
\filldraw (5.6,0.5) circle (0.08);
\node at (5.2,0.5) [anchor=south] {\scriptsize 9};
\node at (5.6,0.5) [anchor=south] {\scriptsize 10};
\filldraw (4.8,-7.5) circle (0.08);
\node at (4.8,-7.5) [anchor=north] {\scriptsize 6};
}
\end{tikzpicture}
\ &=
\beta\alpha_1\gamma_{12}\gamma_3\ 
\begin{tikzpicture}[baseline={([yshift=-1mm]current bounding box.center)},xscale=0.6,yscale=0.6]
{
\draw[dotted] (0,0.5)--(6,0.5);
\draw[dotted] (0,-7.5)--(6,-7.5);
\draw [very thick] (0,-7.5)--(0,0.5);
\draw [very thick] (3,-7.5)--(3,0.5);
\draw [very thick] (6,-7.5)--(6,0.5);
\draw (0,0) arc (-90:0:0.6 and 0.5);
\draw (0,-1) arc (90:-90:0.5);
\draw (0,-3) arc (-90:0:1.2 and 3.5);
\draw (0,-4) arc (90:-90:1.2 and 1.5);
\draw (0,-5) arc (90:-90:0.5);
\node at (0.6,0.5) [anchor=south] {\scriptsize 1};
\node at (1.2,0.5) [anchor=south] {\scriptsize 2};
\filldraw (0.75,-7.5) circle (0.08);
\node at (0.75,-7.5) [anchor=north] {\scriptsize 1};
\draw (3,-6) arc (90:270:0.5);
\filldraw (2.4,0.5) circle (0.08);
\node at (2.4,0.5) [anchor=south] {\scriptsize 4};
\filldraw (1.8,0.5) circle (0.08);
\node at (1.8,0.5) [anchor=south] {\scriptsize 3};
\filldraw (2.55,-7.5) circle (0.08);
\filldraw (1.5,-7.5) circle (0.08);
\node at (1.5,-7.5) [anchor=north] {\scriptsize 2};
\filldraw (2.1,-7.5) circle (0.08);
\node at (2.1,-7.5) [anchor=north] {\scriptsize 3};
\filldraw (2.55,-7.5) circle (0.08);
\node at (2.55,-7.5) [anchor=north] {\scriptsize 4};
\draw (3,-6) to[out=0,in=180] (6,-7);
\draw (3,-7) arc (90:0:0.6 and 0.5);
\filldraw (3.5,0.5) circle (0.08);
\filldraw (4.4,0.5) circle (0.08);
\node at (3.5,0.5) [anchor=south] {\scriptsize 5};
\filldraw (4,0.5) circle (0.08);
\node at (4,0.5) [anchor=south] {\scriptsize 6};
\node at (4.4,0.5) [anchor=south] {\scriptsize 7};
\filldraw (4.8,0.5) circle (0.08);
\node at (4.8,0.5) [anchor=south] {\scriptsize 8};
\node at (3.6,-7.5) [anchor=north] {\scriptsize 5};
\draw (6,0) arc (270:180:0.8 and 0.5);
\draw (6,-1) arc (90:270:1.2 and 1.5);
\draw (6,-2) arc (90:270:0.5);
\draw (6,-5) arc (90:270:0.5);
\filldraw (5.6,0.5) circle (0.08);
\node at (5.2,0.5) [anchor=south] {\scriptsize 9};
\node at (5.6,0.5) [anchor=south] {\scriptsize 10};
\filldraw (4.8,-7.5) circle (0.08);
\node at (4.8,-7.5) [anchor=north] {\scriptsize 6};
}
\end{tikzpicture}
\ =
\beta \alpha_1 \gamma_{12}\gamma_3\ 
\begin{tikzpicture}[baseline={([yshift=-1mm]current bounding box.center)},xscale=0.6,yscale=0.6]
{
\draw[dotted] (0,0.5)--(3,0.5);
\draw[dotted] (0,-7.5)--(3,-7.5);
\draw [very thick] (0,-7.5)--(0,0.5);
\draw [very thick] (3,-7.5)--(3,0.5);
\draw (0,0) arc (-90:0:0.6 and 0.5);
\draw (0,-1) arc (90:-90:0.5);
\draw (0,-3) arc (-90:0:1.2 and 3.5);
\draw (0,-4) arc (90:-90:1.2 and 1.5);
\draw (0,-5) arc (90:-90:0.5);
\draw (3,0) arc (270:180:0.8 and 0.5);
\draw (3,-1) arc (90:270:1.2 and 1.5);
\draw (3,-2) arc (90:270:0.5);
\draw (3,-5) arc (90:270:0.5);
\draw (3,-7) arc (90:180:0.8 and 0.5);
\filldraw (2.7,0.5) circle (0.08);
\filldraw (2.6,-7.5) circle (0.08);
\node at (1,0.5) [anchor=south] {\phantom{\scriptsize 1}};
\node at (0.6,-7.5) [anchor=north] {\phantom{\scriptsize 1}};
}
\end{tikzpicture}\; ,
\end{align}
and
\begin{align}
\begin{tikzpicture}[baseline={([yshift=-1mm]current bounding box.center)},xscale=0.6,yscale=0.6]
{
\draw[dotted] (0,0.5)--(6,0.5);
\draw[dotted] (0,-7.5)--(6,-7.5);
\draw [very thick] (0,-7.5)--(0,0.5);
\draw [very thick] (3,-7.5)--(3,0.5);
\draw [very thick] (6,-7.5)--(6,0.5);
\draw (0,0) arc (90:-90:0.5);
\draw (0,-2) to[out=0,in=-90] (1.2,0.5);
\draw (0,-3) to[out=0,in=180] (3,-2);
\draw (0,-4) arc (90:-90:0.5);
\draw (0,-6) arc (90:-90:0.5);
\filldraw (0.6,0.5) circle (0.08);
\node at (0.6,0.5) [anchor=south] {\scriptsize 1};
\node at (1.2,0.5) [anchor=south] {\scriptsize 2};
\node at (0.6,-7.5) [anchor=north] {\scriptsize 1};
\draw (3,0) arc (270:180:0.6 and 0.5);
\draw (3,-1) arc (270:180:1.2 and 1.5);
\draw (3,-3) arc (90:270:0.5);
\draw (3,-5) arc (90:180:1.8 and 2.5);
\draw (3,-6) arc (90:180:1.2 and 1.5);
\draw (3,-7) arc (90:180:0.6 and 0.5);
\filldraw (0.6,-7.5) circle (0.08);
\node at (1.8,0.5) [anchor=south] {\scriptsize 3};
\node at (2.4,0.5) [anchor=south] {\scriptsize 4};
\node at (1.2,-7.5) [anchor=north] {\scriptsize 2};
\node at (1.8,-7.5) [anchor=north] {\scriptsize 3};
\node at (2.4,-7.5) [anchor=north] {\scriptsize 4};
\draw (3,0) arc (-90:0:0.6 and 0.5);
\draw (3,-1) arc (-90:0:1.8 and 1.5);
\draw (3,-2) to[out=0,in=180] (6,0);
\draw (3,-3) to[out=0,in=180] (6,-1);
\draw (3,-4) to[out=0,in=180] (6,-6);
\draw (3,-5) arc (90:0:1.725 and 2.5);
\draw (3,-6) arc (90:0:0.825 and 1.5);
\draw (3,-7) arc (90:0:0.375 and 0.5);
\filldraw (4.2,0.5) circle (0.08);
\filldraw (5.4,0.5) circle (0.08);
\filldraw (4.275,-7.5) circle (0.08);
\node at (3.6,0.5) [anchor=south] {\scriptsize 5};
\node at (4.2,0.5) [anchor=south] {\scriptsize 6};
\node at (4.8,0.5) [anchor=south] {\scriptsize 7};
\node at (5.4,0.5) [anchor=south] {\scriptsize 8};
\node at (3.375,-7.5) [anchor=north] {\scriptsize 5};
\node at (3.825,-7.5) [anchor=north] {\scriptsize 6};
\node at (4.275,-7.5) [anchor=north] {\scriptsize 7};
\node at (4.725,-7.5) [anchor=north] {\scriptsize 8};
\draw (6,-2) arc (90:270:0.5);
\draw (6,-4) arc (90:270:0.5);
\draw (6,-7) arc (90:180:0.825 and 0.5);
\filldraw (5.625,-7.5) circle (0.08);
\node at (5.175,-7.5) [anchor=north] {\scriptsize 9};
\node at (5.625,-7.5) [anchor=north] {\scriptsize 10};
}
\end{tikzpicture}
\ = \alpha_2 \alpha_3 \delta_1\delta_2\delta_3\,
\begin{tikzpicture}[baseline={([yshift=-1mm]current bounding box.center)},xscale=0.6,yscale=0.6]
{
\draw[dotted] (0,0.5)--(6,0.5);
\draw[dotted] (0,-7.5)--(6,-7.5);
\draw [very thick] (0,-7.5)--(0,0.5);
\draw [very thick] (3,-7.5)--(3,0.5);
\draw [very thick] (6,-7.5)--(6,0.5);
\draw (0,0) arc (90:-90:0.5);
\draw (0,-2) to[out=0,in=-90] (1.2,0.5);
\draw (0,-3) to[out=0,in=180] (3,-2);
\draw (0,-4) arc (90:-90:0.5);
\draw (0,-6) arc (90:-90:0.5);
\filldraw (0.6,0.5) circle (0.08);
\node at (0.6,0.5) [anchor=south] {\scriptsize 1};
\node at (1.2,0.5) [anchor=south] {\scriptsize 2};
\node at (0.6,-7.5) [anchor=north] {\scriptsize 1};
\draw (3,-3) arc (90:270:0.5);
\filldraw (0.6,-7.5) circle (0.08);
\node at (1.8,0.5) [anchor=south] {\scriptsize 3};
\node at (2.4,0.5) [anchor=south] {\scriptsize 4};
\node at (1.2,-7.5) [anchor=north] {\scriptsize 2};
\node at (1.8,-7.5) [anchor=north] {\scriptsize 3};
\node at (2.4,-7.5) [anchor=north] {\scriptsize 4};
\filldraw (1.8,0.5) circle (0.08);
\filldraw (2.4,0.5) circle (0.08);
\filldraw (1.2,-7.5) circle (0.08);
\filldraw (1.8,-7.5) circle (0.08);
\filldraw (2.4,-7.5) circle (0.08);
\draw (3,-2) to[out=0,in=180] (6,0);
\draw (3,-3) to[out=0,in=180] (6,-1);
\draw (3,-4) to[out=0,in=180] (6,-6);
\filldraw (4.2,0.5) circle (0.08);
\filldraw (5.4,0.5) circle (0.08);
\filldraw (4.725,-7.5) circle (0.08);
\node at (3.6,0.5) [anchor=south] {\scriptsize 5};
\node at (4.2,0.5) [anchor=south] {\scriptsize 6};
\node at (4.8,0.5) [anchor=south] {\scriptsize 7};
\node at (5.4,0.5) [anchor=south] {\scriptsize 8};
\node at (3.375,-7.5) [anchor=north] {\scriptsize 5};
\node at (3.825,-7.5) [anchor=north] {\scriptsize 6};
\node at (4.275,-7.5) [anchor=north] {\scriptsize 7};
\node at (4.725,-7.5) [anchor=north] {\scriptsize 8};
\filldraw (3.375,-7.5) circle (0.08);
\filldraw (3.825,-7.5) circle (0.08);
\filldraw (4.275,-7.5) circle (0.08);
\filldraw (3.6,0.5) circle (0.08);
\filldraw (4.8,0.5) circle (0.08);
\draw (6,-2) arc (90:270:0.5);
\draw (6,-4) arc (90:270:0.5);
\draw (6,-7) arc (90:180:0.825 and 0.5);
\filldraw (5.625,-7.5) circle (0.08);
\node at (5.175,-7.5) [anchor=north] {\scriptsize 9};
\node at (5.625,-7.5) [anchor=north] {\scriptsize 10};
}
\end{tikzpicture}
\ = \alpha_2 \alpha_3 \delta_1\delta_2\delta_3\,
\begin{tikzpicture}[baseline={([yshift=-1mm]current bounding box.center)},xscale=0.6,yscale=0.6]
{
\draw[dotted] (0,0.5)--(3,0.5);
\draw[dotted] (0,-7.5)--(3,-7.5);
\draw [very thick] (0,-7.5)--(0,0.5);
\draw [very thick] (3,-7.5)--(3,0.5);
\draw (0,0) arc (90:-90:0.5);
\draw (0,-2) to[out=0,in=-90] (1.2,0.5);
\draw (0,-3) to[out=0,in=180] (3,0);
\draw (0,-4) arc (90:-90:0.5);
\draw (0,-6) arc (90:-90:0.5);
\filldraw (0.6,0.5) circle (0.08);
\draw (3,-1) to[out=180,in=180] (3,-6);
\draw (3,-2) arc (90:270:0.5);
\draw (3,-4) arc (90:270:0.5);
\draw (3,-7) arc (90:180:0.825 and 0.5);
\filldraw (2.625,-7.5) circle (0.08);
\node at (1,0.5) [anchor=south] {\phantom{\scriptsize 1}};
\node at (0.6,-7.5) [anchor=north] {\phantom{\scriptsize 1}};
}
\end{tikzpicture}\; .
\end{align}

With this multiplication, $\Gh_n$ is an associative unital algebra. Associativity is proven in Appendix \ref{app:assoc}. The identity is the diagram with $n$ horizontal strings linking the nodes on either side of the rectangle; this is the same as for $\TL_n$, $\BTL_n$ and $\BBTL_n$. We note that $\Gh_n$ has subalgebras isomorphic to $\TL_n(\beta)$, $\BTL_n(\beta;\alpha_1,\alpha_2)$ and $\BTL_n(\beta;\delta_1,\delta_2)$; these are spanned by the diagrams with no boundary connections, no ghosts and no bottom boundary connections, and no ghosts and no top boundary connections, respectively. There is also a subalgebra spanned by the diagrams that have no bottom boundary connections, but may have ghosts; we call this the \textit{one-boundary ghost algebra} $\Gho_n(\beta;\alpha_1,\alpha_2,\alpha_3)$. The diagrams with no top boundary connections span a subalgebra isomorphic to $\Gho_n(\beta;\delta_1,\delta_2,\delta_3)$.

We also note that some boundary arcs of different parities are assigned the same parameters; see the arcs listed for the parameters $\alpha_3$, $\delta_3$, $\gamma_{12}$ and $\gamma_3$ in Table \ref{tab:params}. This is done for another point of consistency with the $\TL_n$ and $\BTL_n$ subalgebras. The algebras $\TL_n$ and $\BTL_n$ are \textit{cellular} algebras, in the sense of Graham and Lehrer \cite{GL}. The definition of cellularity requires an anti-involution on the algebra, and for $\TL_n$ and $\BTL_n$, the anti-involution is given by reflecting basis diagrams about a vertical line. Assigning the parameters $\alpha_3$, $\delta_3$, $\gamma_{12}$ and $\gamma_3$ to certain pairs of boundary arcs of different parities means the ghost algebra is also cellular with respect to the reflection anti-involution; this is proven in Appendix \ref{app:cellular}.

If these boundary arcs were assigned distinct parameters, the resulting algebra is still associative and unital, and has the desired $\TL_n$ and $\BTL_n$ subalgebras. We call this the \textit{generalised ghost algebra} $\ncGh_n$, and it is discussed in Appendix \ref{app:nc}, alongside its dilute counterpart, the \textit{generalised dilute ghost algebra} $\ncdGh_n$.

\section{The dilute ghost algebra}\label{s:dilute}
Recall from Section \ref{s:intro} that each basis diagram in the two-boundary TL algebra $\BBTL_n$ is required to have an even number of strings connected to each boundary. In fact, each basis diagram of the one-boundary TL algebra $\BTL_n$ also has an even number of strings attached to the boundary. Indeed, each string has two ends, and each of the $2n$ nodes must have exactly one end of a string attached to it, so there are an even number of string endpoints remaining, and these are all attached to the boundary. This is why the $\BBTL_n$-diagrams with no strings connected to the bottom boundary are precisely the basis diagrams of $\BTL_n$, and thus $\BBTL_n(\beta;\alpha_1,\alpha_2,\delta_1,\delta_2)$ has a subalgebra isomorphic to $\BTL_n(\beta;\alpha_1,\alpha_2)$ spanned by these diagrams. 

However, this relies on the assumption that every node has exactly one string endpoint attached to it. Weakening this requirement to allow \textit{empty} nodes, with no attached strings, we enter the realm of \textit{dilute} algebras.

\subsection{The dilute Temperley-Lieb algebra}
The \textit{dilute Temperley-Lieb algebra} $\dTL_n(\beta)$ is another diagram algebra, with a basis of \textit{$\dTL_n$-diagrams} drawn on the familiar rectangle of $2n$ nodes. Non-crossing strings are drawn within the rectangle, and each must connect a node to a different node. Each node must have at most one attached string, but may have none; empty nodes are drawn as unfilled circles. Strings must not connect nodes to the boundaries, however, so we remove the boundary lines to indicate this. Some example $\dTL_5$-diagrams are
\begin{align}
\begin{tikzpicture}[baseline={([yshift=-1mm]current bounding box.center)},scale=0.55]
{
\draw[very thick] (0,0.5)--(0,-4.5);
\draw[very thick] (3,0.5)--(3,-4.5);
\draw (0,0) to[out=0,in=180] (3,-2);
\draw (0,-1) arc(90:-90:1);
\draw (0,-4) to[out=0,in=180] (3,-3);
\draw (3,0) arc(90:270:0.5);
\filldraw[fill=white] (0,-2) circle (0.12);
\filldraw[fill=white] (3,-4) circle (0.12);
}
\end{tikzpicture}
\; , \qquad
\begin{tikzpicture}[baseline={([yshift=-1mm]current bounding box.center)},scale=0.55]
{
\draw[very thick] (0,0.5)--(0,-4.5);
\draw[very thick] (3,0.5)--(3,-4.5);
\filldraw[fill=white] (0,0) circle (0.12);
\filldraw[fill=white] (0,-1) circle (0.12);
\filldraw[fill=white] (0,-2) circle (0.12);
\filldraw[fill=white] (0,-3) circle (0.12);
\filldraw[fill=white] (0,-4) circle (0.12);
\filldraw[fill=white] (3,0) circle (0.12);
\filldraw[fill=white] (3,-1) circle (0.12);
\filldraw[fill=white] (3,-2) circle (0.12);
\filldraw[fill=white] (3,-3) circle (0.12);
\filldraw[fill=white] (3,-4) circle (0.12);
}
\end{tikzpicture}
\; , \qquad
\begin{tikzpicture}[baseline={([yshift=-1mm]current bounding box.center)},xscale=-0.5,yscale=0.55]
{
\draw[very thick] (0,0.5)--(0,-4.5);
\draw[very thick] (3,0.5)--(3,-4.5);
\draw (0,0) to[out=0,in=0] (0,-3);
\draw (0,-1) arc(90:-90:0.5);
\draw (0,-4) to[out=0,in=180] (3,-2);
\draw (3,0) arc(90:270:0.5);
\draw (3,-3) arc(90:270:0.5);
}
\end{tikzpicture}
\; , \qquad 
\begin{tikzpicture}[baseline={([yshift=-1mm]current bounding box.center)},scale=0.55]
{
\draw[very thick] (0,0.5)--(0,-4.5);
\draw[very thick] (3,0.5)--(3,-4.5);
\draw (0,-3)--(3,-3);
\draw (0,-4)--(3,-4);
\draw (3,0) arc(90:270:1);
\filldraw[fill=white] (0,0) circle (0.12);
\filldraw[fill=white] (0,-1) circle (0.12);
\filldraw[fill=white] (0,-2) circle (0.12);
\filldraw[fill=white] (3,-1) circle (0.12);
}
\end{tikzpicture}
\; .
\end{align}
Two $\dTL_n$-diagrams are equal if the strings connect the same nodes in each.

Multiplication is similarly defined on pairs of diagrams by concatenation, and extended bilinearly to the whole space. Any loops produced are removed and replaced by a factor of $\beta$, as for $\TL_n$, but now if a string meets an empty node during concatenation, the product of those diagrams is set to zero. For example, in $\dTL_6$,
\begin{align}
\begin{tikzpicture}[baseline={([yshift=-1mm]current bounding box.center)},xscale=0.55,yscale=0.55]
{
\draw[very thick] (0,0.5)--(0,-5.5);
\draw[very thick] (3,0.5)--(3,-5.5);
\draw[very thick] (6,0.5)--(6,-5.5);
\draw (0,-1) to[out=0,in=180] (3,0);
\draw (0,-3) arc (90:-90:1);
\draw (3,-1) to[out=180,in=180] (3,-5);
\draw (3,-2) arc (90:270:0.5);
\draw (3,0) arc (90:-90:1) (3,-2);
\draw (3,-3) to[out=0,in=180] (6,-4);
\draw (3,-5) to[out=0,in=180] (6,-5);
\draw (6,0) arc (90:270:0.5);
\filldraw[fill=white] (0,0) circle (0.12);
\filldraw[fill=white] (0,-2) circle (0.12);
\filldraw[fill=white] (0,-4) circle (0.12);
\filldraw[fill=white] (3,-1) circle (0.12);
\filldraw[fill=white] (3,-4) circle (0.12);
\filldraw[fill=white] (6,-2) circle (0.12);
\filldraw[fill=white] (6,-3) circle (0.12);
}
\end{tikzpicture}
\ &= 0,
&\begin{tikzpicture}[baseline={([yshift=-1mm]current bounding box.center)},xscale=0.55,yscale=0.55]
{
\draw (0,0) arc (90:-90:0.5);
\draw (0,-2) to[out=0,in=180] (3,0);
\draw (0,-4) arc (90:-90:0.5);
\draw (3,-1) to[out=180,in=180] (3,-5);
\draw (3,-3) arc (90:270:0.5);
\draw (3,0) to[out=0,in=180] (6,-1);
\draw (3,-1) arc (90:-90:1);
\draw (3,-4) arc (90:-90:0.5);
\draw (6,-2) arc (90:270:1);
\draw[very thick] (0,0.5)--(0,-5.5);
\draw[very thick] (3,0.5)--(3,-5.5);
\draw[very thick] (6,0.5)--(6,-5.5);
\filldraw[fill=white] (0,-3) circle (0.12);
\filldraw[fill=white] (3,-2) circle (0.12);
\filldraw[fill=white] (6,0) circle (0.12);
\filldraw[fill=white] (6,-3) circle (0.12);
\filldraw[fill=white] (6,-5) circle (0.12);
}
\end{tikzpicture}
\ &=\beta\  \begin{tikzpicture}[baseline={([yshift=-1mm]current bounding box.center)},xscale=0.55,yscale=0.55]
{
\draw (0,0) arc (90:-90:0.5);
\draw (0,-2) to[out=0,in=180] (3,-1);
\draw (0,-4) arc (90:-90:0.5);
\draw (3,-2) arc (90:270:1);
\draw[very thick] (0,0.5)--(0,-5.5);
\draw[very thick] (3,0.5)--(3,-5.5);
\filldraw[fill=white] (0,-3) circle (0.12);
\filldraw[fill=white] (3,0) circle (0.12);
\filldraw[fill=white] (3,-3) circle (0.12);
\filldraw[fill=white] (3,-5) circle (0.12);
}
\end{tikzpicture} \; .
\end{align}

This algebra is associative and unital. To state the identity, let diagrams with $k$ dashed strings represent the sum of the $2^k$ $\dTL_n$-diagrams obtained by drawing or not drawing each dashed string. For example, in $\dTL_5$,
\begin{align}
\begin{tikzpicture}[baseline={([yshift=-1mm]current bounding box.center)},xscale=0.55,yscale=0.55]
{
\draw[dashed] (0,0) to[out=0,in=0] (0,-2);
\draw[dashed] (3,-2) arc (90:270:0.5);
\draw[dashed] (0,-4) -- (3,-4);
\draw (0,-3) to[out=0,in=180] (3,-1);
\draw[very thick] (0,0.5)--(0,-4.5);
\draw[very thick] (3,0.5)--(3,-4.5);
\filldraw[fill=white] (0,-1) circle (0.12);
\filldraw[fill=white] (3,0) circle (0.12);
}
\end{tikzpicture} 
\; &=\; 
\begin{tikzpicture}[baseline={([yshift=-1mm]current bounding box.center)},xscale=0.55,yscale=0.55]
{
\draw (0,0) to[out=0,in=0] (0,-2);
\draw (3,-2) arc (90:270:0.5);
\draw (0,-4) -- (3,-4);
\draw (0,-3) to[out=0,in=180] (3,-1);
\draw[very thick] (0,0.5)--(0,-4.5);
\draw[very thick] (3,0.5)--(3,-4.5);
\filldraw[fill=white] (0,-1) circle (0.12);
\filldraw[fill=white] (3,0) circle (0.12);
}
\end{tikzpicture} 
\; +\; 
\begin{tikzpicture}[baseline={([yshift=-1mm]current bounding box.center)},xscale=0.55,yscale=0.55]
{
\draw (0,0) to[out=0,in=0] (0,-2);
\draw (3,-2) arc (90:270:0.5);
\draw (0,-3) to[out=0,in=180] (3,-1);
\draw[very thick] (0,0.5)--(0,-4.5);
\draw[very thick] (3,0.5)--(3,-4.5);
\filldraw[fill=white] (0,-1) circle (0.12);
\filldraw[fill=white] (3,0) circle (0.12);
\filldraw[fill=white] (0,-4) circle (0.12);
\filldraw[fill=white] (3,-4) circle (0.12);
}
\end{tikzpicture} 
\; +\;
\begin{tikzpicture}[baseline={([yshift=-1mm]current bounding box.center)},xscale=0.55,yscale=0.55]
{
\draw (0,0) to[out=0,in=0] (0,-2);
\draw (0,-4) -- (3,-4);
\draw (0,-3) to[out=0,in=180] (3,-1);
\draw[very thick] (0,0.5)--(0,-4.5);
\draw[very thick] (3,0.5)--(3,-4.5);
\filldraw[fill=white] (0,-1) circle (0.12);
\filldraw[fill=white] (3,0) circle (0.12);
\filldraw[fill=white] (3,-2) circle (0.12);
\filldraw[fill=white] (3,-3) circle (0.12);
}
\end{tikzpicture} 
\; +\; 
\begin{tikzpicture}[baseline={([yshift=-1mm]current bounding box.center)},xscale=0.55,yscale=0.55]
{
\draw (3,-2) arc (90:270:0.5);
\draw (0,-4) -- (3,-4);
\draw (0,-3) to[out=0,in=180] (3,-1);
\draw[very thick] (0,0.5)--(0,-4.5);
\draw[very thick] (3,0.5)--(3,-4.5);
\filldraw[fill=white] (0,-1) circle (0.12);
\filldraw[fill=white] (3,0) circle (0.12);
\filldraw[fill=white] (0,0) circle (0.12);
\filldraw[fill=white] (0,-2) circle (0.12);
}
\end{tikzpicture} \nonumber
\\[3mm]
&\quad +\; 
\begin{tikzpicture}[baseline={([yshift=-1mm]current bounding box.center)},xscale=0.55,yscale=0.55]
{
\draw (0,0) to[out=0,in=0] (0,-2);
\draw (0,-3) to[out=0,in=180] (3,-1);
\draw[very thick] (0,0.5)--(0,-4.5);
\draw[very thick] (3,0.5)--(3,-4.5);
\filldraw[fill=white] (0,-1) circle (0.12);
\filldraw[fill=white] (3,0) circle (0.12);
\filldraw[fill=white] (0,-4) circle (0.12);
\filldraw[fill=white] (3,-2) circle (0.12);
\filldraw[fill=white] (3,-3) circle (0.12);
\filldraw[fill=white] (3,-4) circle (0.12);
}
\end{tikzpicture} 
\; +\; 
\begin{tikzpicture}[baseline={([yshift=-1mm]current bounding box.center)},xscale=0.55,yscale=0.55]
{
\draw (3,-2) arc (90:270:0.5);
\draw (0,-3) to[out=0,in=180] (3,-1);
\draw[very thick] (0,0.5)--(0,-4.5);
\draw[very thick] (3,0.5)--(3,-4.5);
\filldraw[fill=white] (0,-1) circle (0.12);
\filldraw[fill=white] (3,0) circle (0.12);
\filldraw[fill=white] (0,0) circle (0.12);
\filldraw[fill=white] (0,-2) circle (0.12);
\filldraw[fill=white] (0,-4) circle (0.12);
\filldraw[fill=white] (3,-4) circle (0.12);
}
\end{tikzpicture} 
\; + \; 
\begin{tikzpicture}[baseline={([yshift=-1mm]current bounding box.center)},xscale=0.55,yscale=0.55]
{
\draw (0,-4) -- (3,-4);
\draw (0,-3) to[out=0,in=180] (3,-1);
\draw[very thick] (0,0.5)--(0,-4.5);
\draw[very thick] (3,0.5)--(3,-4.5);
\filldraw[fill=white] (0,-1) circle (0.12);
\filldraw[fill=white] (3,0) circle (0.12);
\filldraw[fill=white] (0,0) circle (0.12);
\filldraw[fill=white] (0,-2) circle (0.12);
\filldraw[fill=white] (3,-2) circle (0.12);
\filldraw[fill=white] (3,-3) circle (0.12);
}
\end{tikzpicture}
\; + \; 
\begin{tikzpicture}[baseline={([yshift=-1mm]current bounding box.center)},xscale=0.55,yscale=0.55]
{
\draw (0,-3) to[out=0,in=180] (3,-1);
\draw[very thick] (0,0.5)--(0,-4.5);
\draw[very thick] (3,0.5)--(3,-4.5);
\filldraw[fill=white] (0,-1) circle (0.12);
\filldraw[fill=white] (3,0) circle (0.12);
\filldraw[fill=white] (0,0) circle (0.12);
\filldraw[fill=white] (0,-2) circle (0.12);
\filldraw[fill=white] (3,-2) circle (0.12);
\filldraw[fill=white] (3,-3) circle (0.12);
\filldraw[fill=white] (3,-4) circle (0.12);
\filldraw[fill=white] (0,-4) circle (0.12);
}
\end{tikzpicture} \; .
\end{align}
The identity in $\dTL_n$ is then
\begin{align}
I_{\dTL_n} = \begin{tikzpicture}[baseline={([yshift=-1mm]current bounding box.center)},scale=0.55]
{
\draw [very thick](0,-5) -- (0,0.5);
\draw[dashed] (0,0) -- (3,0);
\draw[dashed] (0,-1) -- (3,-1);
\draw[dashed] (0,-2) -- (3,-2);
\draw (1.5,-2.5) node[anchor=center]{$\vdots$};
\draw[dashed] (0,-3.5) -- (3,-3.5);
\draw[dashed] (0,-4.5) -- (3,-4.5);
\draw [very thick](3,-5) -- (3,0.5);
\draw (0,0) node[font=\scriptsize,anchor=east]{$1$};
\draw (0,-1) node[font=\scriptsize,anchor=east]{$2$};
\draw (0,-2) node[font=\scriptsize,anchor=east]{$3$};
\draw (0,-3.5) node[font=\scriptsize,anchor=east]{$n-1$};
\draw (0,-4.5) node[font=\scriptsize,anchor=east]{$n$};
}
\end{tikzpicture}\; . \label{eq:diluteidentity}
\end{align}

Observe that the set of all $\TL_n$-diagrams is a subset of the set of all $\dTL_n$-diagrams, and since the $\TL_n$-diagrams have no empty nodes, their products are the same in $\dTL_n$ as in $\TL_n$. The span of the $\TL_n$-diagrams is closed under multiplication, but it does not contain the identity of $\dTL_n$, since this is a linear combination of $\dTL_n$-diagrams, including some with empty nodes. Hence the $\TL_n$-diagrams in $\dTL_n(\beta)$ span a non-unital subalgebra isomorphic to $\TL_n(\beta)$.

There is a unital subalgebra of $\dTL_n(\beta)$ isomorphic to $\TL_n(\beta+1)$, however. This is generated by the $\dTL_n$ identity, and diagrams of the form
\begin{align}
&\begin{tikzpicture}[baseline={([yshift=-1mm]current bounding box.center)},scale=0.55]
{
\draw [very thick](0,-5.5) -- (0,0.5);
\draw[dashed] (0,0) -- (3,0);
\draw (1.5,-0.3) node[anchor=center]{$\vdots$};
\draw[dashed] (0,-1) -- (3,-1);
\draw (0,-1) node[font=\scriptsize,anchor=east]{$j-1$};
\draw[dashed,dash phase=1.25pt] (0,-2) arc (90:-90:0.5);
\draw[dashed,dash phase=1.25pt] (3,-2) arc (90:270:0.5);
\draw (0,-2) node[font=\scriptsize,anchor=east]{$j$};
\draw (0,-3) node[font=\scriptsize,anchor=east]{$j+1$};
\draw[dashed] (0,-4) -- (3,-4);
\draw (0,-4) node[font=\scriptsize,anchor=east]{$j+2$};
\draw (1.5,-4.3) node[anchor=center]{$\vdots$};
\draw[dashed] (0,-5) -- (3,-5);
\draw (0,-5) node[font=\scriptsize,anchor=east]{$n$};
\draw [very thick](3,-5.5) -- (3,0.5);
\draw (0,0) node[font=\scriptsize,anchor=east]{$1$};
}
\end{tikzpicture}\; , &1\leq j \leq n-1.
\end{align}

\subsection{Diagram basis and multiplication}
The \textit{dilute ghost algebra} $\dGh_n$ is defined in terms of a basis of \textit{$\dGh_n$-diagrams}. A $\dGh_n$-diagram consists of non-crossing strings drawn within a two-boundary rectangle with $2n$ nodes. Each string must connect a node to another node, or to a boundary, and each node must have at most one string attached to it. Ghosts may be drawn in the domains of the boundaries, and we require that the number of ghosts plus the number of strings on each boundary is even.

Two $\dGh_n$-diagrams are equal if the strings connect the same nodes to each other, or to the same boundary, and the numbers of ghosts in corresponding domains are equivalent modulo 2.

These $\dGh_n$-diagrams are multiplied by concatenation, where each loop is removed and replaced by a factor of $\beta$, and each boundary arc is removed and replaced by a factor of the corresponding parameter from Table \ref{tab:params}, leaving a ghost at each endpoint. The parity of boundary arc endpoints is again determined by counting the strings and ghosts along each boundary, starting from 1 on the left. If, during concatenation, a string meets an empty node, the product is set to zero. If not, we remove the extra vertical line in the middle to produce a scalar multiple of a $\dGh_n$-diagram, tightening the strings and removing extra pairs of ghosts from each domain, if desired.

For example, in $\dGh_5$, we have 
\begin{align}
\begin{tikzpicture}[baseline={([yshift=-1mm]current bounding box.center)},xscale=0.6,yscale=0.6]
{
\draw[very thick] (0,0.5)--(0,-4.5);
\draw[very thick] (3,0.5)--(3,-4.5);
\draw[very thick] (6,0.5)--(6,-4.5);
\draw[dotted] (0,0.5)--(6,0.5);
\draw[dotted] (0,-4.5)--(6,-4.5);
\draw (0,0) arc (-90:0:0.6 and 0.5);
\draw (0,-2) to[out=0,in=180] (3,-4);
\draw (0,-3) arc (90:-90:0.5);
\draw (3,-1) to[out=180,in=270] (2,0.5);
\draw (3,-2) arc (90:270:0.5);
\draw (3,-1) arc (90:-90:1);
\draw (3,-4) to[out=0,in=180] (6,0);
\draw (6,-1) to[out=180,in=90] (4.8,-4.5);
\draw (6,-2) arc (90:270:0.5);
\filldraw[fill=white] (0,-1) circle (0.12);
\filldraw[fill=white] (3,0) circle (0.12);
\filldraw[fill=white] (3,-2) circle (0.12);
\filldraw[fill=white] (6,-4) circle (0.12);
\filldraw (4,-4.5) circle (0.09);
}
\end{tikzpicture}
\ = 0,
\end{align}
and
\begin{align}
\begin{tikzpicture}[baseline={([yshift=-1mm]current bounding box.center)},xscale=0.6,yscale=0.6]
{
\draw[dotted] (0,0.5)--(6,0.5);
\draw[dotted] (0,-4.5)--(6,-4.5);
\draw (0,0) arc (90:-90:0.5);
\draw (0,-3) to[out=0,in=180] (3,-1);
\draw (3,0) to[out=180,in=270] (2,0.5);
\draw (3,-2) to[out=180,in=90] (1.2,-4.5);
\draw (3,-4) ..controls (2.2,-3.95) and (1.9,-4.2) .. (1.8,-4.5);
\draw (3,0) arc (90:-90:0.5);
\draw (3,-2) to[out=0,in=90] (4.8,-4.5);
\draw (3,-4) to[out=0,in=90] (3.6,-4.5);
\draw (6,0) arc (270:180:0.6 and 0.5);
\draw (6,-1) arc (90:270:0.5);
\draw (6,-4) arc (90:180:0.6 and 0.5);
\draw[very thick] (0,0.5)--(0,-4.5);
\draw[very thick] (3,0.5)--(3,-4.5);
\draw[very thick] (6,0.5)--(6,-4.5);
\node at (2,0.5) [anchor=south] {\footnotesize 1};
\node at (2.5,0.5) [anchor=south] {\footnotesize 2};
\node at (4.1,0.5) [anchor=south] {\footnotesize 3};
\node at (5.4,0.5) [anchor=south] {\footnotesize 4};
\node at (0.6,-4.5)[anchor=north] {\footnotesize 1};
\node at (1.2,-4.5)[anchor=north] {\footnotesize 2};
\node at (1.8,-4.5)[anchor=north] {\footnotesize 3};
\node at (2.4,-4.5)[anchor=north] {\footnotesize 4};
\node at (3.6,-4.5)[anchor=north] {\footnotesize 5};
\node at (4.2,-4.5)[anchor=north] {\footnotesize 6};
\node at (4.8,-4.5)[anchor=north] {\footnotesize 7};
\node at (5.4,-4.5)[anchor=north] {\footnotesize 8};
\filldraw[fill=white] (0,-2) circle (0.12);
\filldraw[fill=white] (0,-4) circle (0.12);
\filldraw[fill=white] (3,-3) circle (0.12);
\filldraw[fill=white] (6,-3) circle (0.12);
\filldraw (2.5,0.5) circle (0.09);
\filldraw (4.1,0.5) circle (0.09);
\filldraw (0.6,-4.5) circle (0.09);
\filldraw (2.4,-4.5) circle (0.09);
\filldraw (4.2,-4.5) circle (0.09);
}
\end{tikzpicture}
\ &=\delta_2\delta_3 \  \begin{tikzpicture}[baseline={([yshift=-1mm]current bounding box.center)},xscale=0.6,yscale=0.6]
{
\draw[dotted] (0,0.5)--(3,0.5);
\draw[dotted] (0,-4.5)--(3,-4.5);
\draw[dotted] (3,0.5)--(6,0.5);
\draw[dotted] (3,-4.5)--(6,-4.5);
%
\draw (0,0) arc (90:-90:0.5);
\draw (0,-3) to[out=0,in=180] (3,-1);
%
\draw (3,0) to[out=180,in=270] (2,0.5);
%
\draw (3,0) arc (90:-90:0.5);
%
\draw (6,0) arc (270:180:0.6 and 0.5);
\draw (6,-1) arc (90:270:0.5);
\draw (6,-4) arc (90:180:0.6 and 0.5);
\node at (2,0.5) [anchor=south] {\footnotesize 1};
\node at (2.5,0.5) [anchor=south] {\footnotesize 2};
\node at (4.1,0.5) [anchor=south] {\footnotesize 3};
\node at (5.4,0.5) [anchor=south] {\footnotesize 4};
\node at (0.6,-4.5)[anchor=north] {\footnotesize 1};
\node at (1.2,-4.5)[anchor=north] {\footnotesize 2};
\node at (1.8,-4.5)[anchor=north] {\footnotesize 3};
\node at (2.4,-4.5)[anchor=north] {\footnotesize 4};
\node at (3.6,-4.5)[anchor=north] {\footnotesize 5};
\node at (4.2,-4.5)[anchor=north] {\footnotesize 6};
\node at (4.8,-4.5)[anchor=north] {\footnotesize 7};
\node at (5.4,-4.5)[anchor=north] {\footnotesize 8};
\filldraw (2.5,0.5) circle (0.09);
\filldraw (4.1,0.5) circle (0.09);
\filldraw (0.6,-4.5) circle (0.09);
\filldraw (1.2,-4.5) circle (0.09);
\filldraw (1.8,-4.5) circle (0.09);
\filldraw (2.4,-4.5) circle (0.09);
\filldraw (3.6,-4.5) circle (0.09);
\filldraw (4.2,-4.5) circle (0.09);
\filldraw (4.8,-4.5) circle (0.09);
\draw[very thick] (0,0.5)--(0,-4.5);
\draw[very thick] (3,0.5)--(3,-4.5);
\draw[very thick] (6,0.5)--(6,-4.5);
\filldraw[fill=white] (0,-2) circle (0.12);
\filldraw[fill=white] (0,-4) circle (0.12);
\filldraw[fill=white] (3,-3) circle (0.12);
\filldraw[fill=white] (6,-3) circle (0.12);
}
\end{tikzpicture}
\ =\delta_2\delta_3\  
\begin{tikzpicture}[baseline={([yshift=-1mm]current bounding box.center)},xscale=0.6,yscale=0.6]
{
\draw[dotted] (0,0.5)--(3,0.5);
\draw[dotted] (0,-4.5)--(3,-4.5);
%
%
\draw (0,0) arc(90:-90:0.5);
\draw (0,-3) to[out=0,in=-90] (1.5,0.5);
\draw (3,0) arc (270:180:0.6 and 0.5);
\draw (3,-1) arc (90:270:0.5);
\draw (3,-4) arc (90:180:0.6 and 0.5);
\filldraw (1.2,-4.5) circle (0.09);
\node at (2,0.5) [anchor=south] {\vphantom{\footnotesize 1234}};
\node at (0.75,-4.5)[anchor=north] {\vphantom{\footnotesize 12345678}};
\draw[very thick] (0,0.5)--(0,-4.5);
\draw[very thick] (3,0.5)--(3,-4.5);
\filldraw[fill=white] (0,-2) circle (0.12);
\filldraw[fill=white] (0,-4) circle (0.12);
\filldraw[fill=white] (3,-3) circle (0.12);
}
\end{tikzpicture}
\; .
\end{align}

This algebra is associative, as proven in Appendix \ref{app:assoc}, and unital; the identity is the same as in \eqref{eq:diluteidentity} for $\dTL_n$, but with dotted boundary lines added to the top and bottom, to be consistent with the other $\dGh_n$-diagrams.

The $\dGh_n$-diagrams without any strings attached to the bottom boundary span a subalgebra that we call the \textit{one-boundary dilute ghost algebra} $\dGho_n(\beta;\alpha_1,\alpha_2,\alpha_3)$. The $\dGh_n$-diagrams without any strings connected to the top boundary span a subalgebra isomorphic to $\dGho_n(\beta;\delta_1,\delta_2,\delta_3)$, and the intersection of these two subalgebras is isomorphic to $\dTL_n(\beta)$.

\section{Relationships between new and existing algebras}\label{s:relations}

In this section, we describe some notable unital and non-unital subalgebras of the ghost and dilute ghost algebras, culminating in the commutative diagrams in Figures \ref{fig:nonunital} and \ref{fig:unital}.

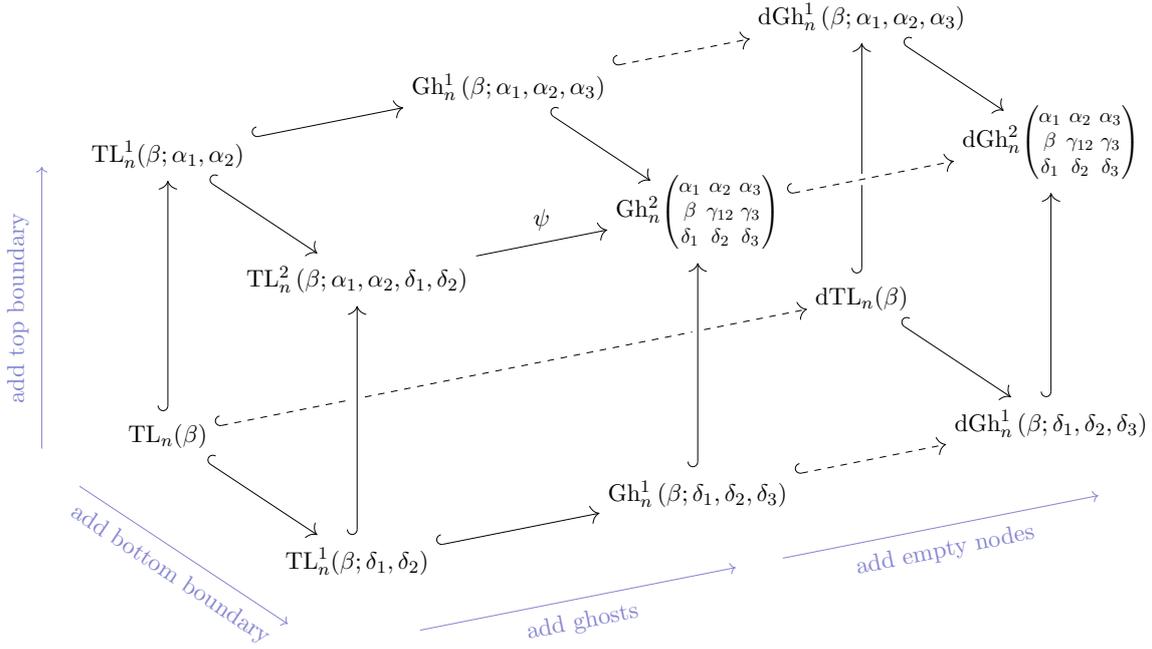
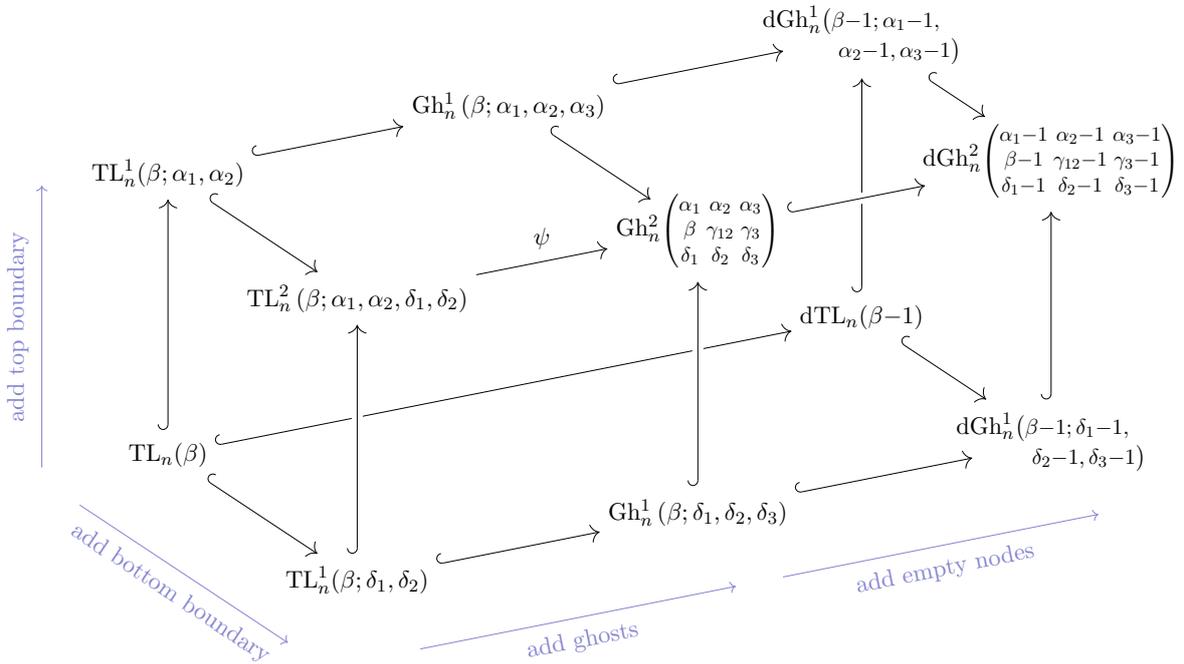
\begin{figure}[p]
\centering
\begin{subfigure}[b]{\textwidth}
\begin{align*}
\scalebox{0.83}{
\begin{tikzpicture}[inner ysep=2mm,minimum size=5mm,baseline={([yshift=-1mm]current bounding box.north)}]
\node (TL) at (0,0) {$\TL_n(\beta)$};
\node (1BTL) at (0,4.5) {$\BTL_n(\beta; \alpha_1,\alpha_2)$};
\node (2BTL) at (3,2.5) {$\BBTL_n\left(\beta;\alpha_1,\alpha_2, \delta_1,\delta_2\right)$};
\node (Gh1) at (5.4,5.58) {$\Gho_n\left(\beta; \alpha_1,\alpha_2,\alpha_3\right)$};
\node (Gh2) at (8.4,3.58) {$\Gh_n$\scalebox{0.85}{$\begingroup
\setlength\arraycolsep{0.5mm}
\begin{pmatrix}
\alpha_1 & \alpha_2 & \alpha_3 \\
\beta & \gamma_{12} & \gamma_3 \\
\delta_1 & \delta_2 & \delta_3
\end{pmatrix}\endgroup$}};
\node (dGh2) at (14,4.7) {$\dGh_n$\scalebox{0.85}{$\begingroup
\setlength\arraycolsep{0.5mm}
\begin{pmatrix}
\alpha_1 & \alpha_2 & \alpha_3 \\
\beta & \gamma_{12} & \gamma_3 \\
\delta_1 & \delta_2 & \delta_3
\end{pmatrix}\endgroup$}};
\node (dGh1b) at (14,0.2) {$\dGho_n\left(\beta;\delta_1,\delta_2,\delta_3\right)$};
\node (1BTLb) at (3,-2) {$\BTL_n(\beta; \delta_1,\delta_2)$};
\node (Gh1b) at (8.4,-0.92) {$\Gho_n\left(\beta; \delta_1,\delta_2,\delta_3\right)$};
\node (dTL) at (11,2.2) {$\dTL_n(\beta)$};
\node (dGh1) at (11,6.7) {$\dGho_n\left(\beta;\alpha_1,\alpha_2,\alpha_3\right)$};
\draw[{Hooks[right,scale=2]}-{>[scale=1.6]}] (TL) to (1BTL);
\draw[{Hooks[right,scale=2]}-{>[scale=1.6]},dashed] (TL) to (dTL);
\draw[{Hooks[right,scale=2]}-{>[scale=1.6]}] (1BTL) to (Gh1);
\draw[{Hooks[right,scale=2]}-{>[scale=1.6]}] (1BTL) to (2BTL);
\draw[{Hooks[right,scale=2]}-{>[scale=1.6]},dashed] (Gh1) to (dGh1);
\draw[{Hooks[right,scale=2]}-{>[scale=1.6]}] (Gh1) to (7.65,4.08);
\draw[{Hooks[right,scale=2]}-{>[scale=1.6]}] (dGh1) -- (13.25,5.2);
\draw[{Hooks[right,scale=2]}-{>[scale=1.6]}] (dTL) to (dGh1);
\draw[line width=5pt,white] (Gh2) to (dGh2);
\draw[{Hooks[right,scale=2]}-{>[scale=1.6]},dashed] (Gh2) to (dGh2);
\draw[{Hooks[right,scale=2]}-{>[scale=1.6]}] (TL) to (1BTLb);
\draw[{Hooks[right,scale=2]}-{>[scale=1.6]}] (1BTLb) to (Gh1b);
\draw[line width=5pt,white] (1BTLb) to (2BTL);
\draw[{Hooks[right,scale=2]}-{>[scale=1.6]}] (1BTLb) to (2BTL);
\draw[{Hooks[right,scale=2]}-{>[scale=1.6]},dashed] (Gh1b) to (dGh1b);
\draw[line width=5pt,white] (Gh1b) to (Gh2);
\draw[{Hooks[right,scale=2]}-{>[scale=1.6]}] (Gh1b) to (Gh2);
\draw[{Hooks[right,scale=2]}-{>[scale=1.6]}] (dGh1b) to (dGh2);
\draw[{Hooks[right,scale=2]}-{>[scale=1.6]}] (dTL) to (dGh1b);
\draw[-{>[scale=1.6]}] (2BTL) -- (Gh2) node[midway,above,align=center] {$\psi$};
\draw[->,richblue!60] (-2,-0.2)--(-2,4.3) node[midway,rotate=90,align=center,above] {add top boundary};
\draw[->,richblue!60] (-1.4,-0.8)--(1.9,-3) node[midway,rotate=-33.69,align=center,below] {add bottom boundary};
\draw[->,richblue!60] (4,-3.1)--(9,-2.1) node[midway,rotate=11.31,align=center,below] {add ghosts};
\draw[->,richblue!60] (9.75,-1.95)--(14.75,-0.95) node[midway,rotate=11.31,align=center,below] {add empty nodes};
\end{tikzpicture}
}
\end{align*}
    \caption{Commutative diagram relating notable unital and non-unital subalgebras of $\Gh_n$ and $\dGh_n$.}
    \label{fig:nonunital}
\end{subfigure}
\begin{subfigure}[b]{\textwidth}
    \begin{align*}
\scalebox{0.83}{
\begin{tikzpicture}[inner ysep=2mm,minimum size=5mm,baseline={([yshift=-1mm]current bounding box.north)}]
\node (TL) at (0,0) {$\TL_n(\beta)$};
\node (1BTL) at (0,4.5) {$\BTL_n(\beta; \alpha_1,\alpha_2)$};
\node (2BTL) at (3,2.5) {$\BBTL_n\left(\beta;\alpha_1,\alpha_2, \delta_1,\delta_2\right)$};
\node (Gh1) at (5.4,5.58) {$\Gho_n\left(\beta; \alpha_1,\alpha_2,\alpha_3\right)$};
\node (Gh2) at (8.4,3.58) {$\Gh_n$\scalebox{0.85}{$\begingroup
\setlength\arraycolsep{0.5mm}
\begin{pmatrix}
\alpha_1 & \alpha_2 & \alpha_3 \\
\beta & \gamma_{12} & \gamma_3 \\
\delta_1 & \delta_2 & \delta_3
\end{pmatrix}\endgroup$}};
\node (dGh2) at (14,4.7) {\scalebox{1}{$\dGh_n$}\scalebox{0.85}{$\begingroup
\setlength\arraycolsep{0.5mm}
\begin{pmatrix}
\alpha_1\textcolor{black}{-}1 & \alpha_2\textcolor{black}{-}1 & \alpha_3\textcolor{black}{-}1 \\
\beta\textcolor{black}{-}1 & \gamma_{12}\textcolor{black}{-}1 & \gamma_3\textcolor{black}{-}1 \\
\delta_1\textcolor{black}{-}1 & \delta_2\textcolor{black}{-}1 & \delta_3\textcolor{black}{-}1
\end{pmatrix}\endgroup$}};
\node (dGh1b) at (14,0.2) {\begin{tabular}{l}\scalebox{1}{$\dGho_n$}\scalebox{0.9}{$\big(\beta\textcolor{black}{-}1;\delta_1\textcolor{black}{-}1,$}\\ \hspace{12mm}\scalebox{0.9}{$\delta_2\textcolor{black}{-}1,\delta_3\textcolor{black}{-}1\big)$}\end{tabular}};
\node (1BTLb) at (3,-2) {$\BTL_n(\beta; \delta_1,\delta_2)$};
\node (Gh1b) at (8.4,-0.92) {$\Gho_n\left(\beta; \delta_1,\delta_2,\delta_3\right)$};
\node (dTL) at (11,2.2) {\scalebox{1}{$\dTL_n(\beta\textcolor{black}{-}1)$}};
\node (dGh1) at (11,6.7) {\begin{tabular}{l}\scalebox{1}{$\dGho_n$}\scalebox{0.9}{$\big(\beta\textcolor{black}{-}1;\alpha_1\textcolor{black}{-}1,$}\\ \hspace{12mm}\scalebox{0.9}{$\alpha_2\textcolor{black}{-}1,\alpha_3\textcolor{black}{-}1\big)$}\end{tabular}};
\draw[{Hooks[right,scale=2]}-{>[scale=1.6]}] (TL) to (1BTL);
\draw[{Hooks[right,scale=2]}-{>[scale=1.6]}] (TL) to (dTL);
\draw[{Hooks[right,scale=2]}-{>[scale=1.6]}] (1BTL) to (Gh1);
\draw[{Hooks[right,scale=2]}-{>[scale=1.6]}] (1BTL) to (2BTL);
\draw[{Hooks[right,scale=2]}-{>[scale=1.6]}] (Gh1) to (9.75,6.45);
\draw[{Hooks[right,scale=2]}-{>[scale=1.6]}] (Gh1) to (7.65,4.08);
\draw[{Hooks[right,scale=2]}-{>[scale=1.6]}] (dGh1) -- (12.95,5.4);
\draw[{Hooks[right,scale=2]}-{>[scale=1.6]}] (dTL) to (dGh1);
\draw[line width=5pt,white] (Gh2) to (dGh2);
\draw[{Hooks[right,scale=2]}-{>[scale=1.6]}] (Gh2) to (12,4.3);
\draw[{Hooks[right,scale=2]}-{>[scale=1.6]}] (TL) to (1BTLb);
\draw[{Hooks[right,scale=2]}-{>[scale=1.6]}] (1BTLb) to (Gh1b);
\draw[line width=5pt,white] (1BTLb) to (2BTL);
\draw[{Hooks[right,scale=2]}-{>[scale=1.6]}] (1BTLb) to (2BTL);
\draw[{Hooks[right,scale=2]}-{>[scale=1.6]}] (Gh1b) to (12.75,-0.05);
\draw[line width=5pt,white] (Gh1b) to (Gh2);
\draw[{Hooks[right,scale=2]}-{>[scale=1.6]}] (Gh1b) to (Gh2);
\draw[{Hooks[right,scale=2]}-{>[scale=1.6]}] (dGh1b) to (dGh2);
\draw[{Hooks[right,scale=2]}-{>[scale=1.6]}] (dTL) to (dGh1b);
\draw[-{>[scale=1.6]}] (2BTL) -- (Gh2) node[midway,above,align=center] {$\psi$};
\draw[->,richblue!60] (-2,-0.2)--(-2,4.3) node[midway,rotate=90,align=center,above] {add top boundary};
\draw[->,richblue!60] (-1.4,-0.8)--(1.9,-3) node[midway,rotate=-33.69,align=center,below] {add bottom boundary};
\draw[->,richblue!60] (4,-3.1)--(9,-2.1) node[midway,rotate=11.31,align=center,below] {add ghosts};
\draw[->,richblue!60] (9.75,-1.95)--(14.75,-0.95) node[midway,rotate=11.31,align=center,below] {add empty nodes};
\end{tikzpicture}
}
\end{align*}
    \caption{Commutative diagram relating notable unital subalgebras of $\Gh_n$ and $\dGh_n$.}
    \label{fig:unital}
\end{subfigure}
\caption{Commutative diagrams summarising the maps from Section \ref{s:relations}, with blue notes indicating their 3D layout. Solid and dashed arrows are unital and non-unital homomorphisms, respectively. Hooked arrows are injections. The map $\psi$ is defined in Section \ref{s:relations}, paragraph three.}
\end{figure}

As discussed in Section \ref{s:ghostalg}, the ghost algebra $\Gh_n$ has one-boundary subalgebras isomorphic to $\Gho_n(\beta;\alpha_1,\alpha_2,\alpha_3)$ and $\Gho_n(\beta;\delta_1,\delta_2,\delta_3)$, spanned by the $\Gh_n$-diagrams with no strings attached to the bottom and top boundaries, respectively. Within these, the $\Gh_n$-diagrams without ghosts span subalgebras isomorphic to $\BTL_n(\beta;\alpha_1,\alpha_2)$ and $\BTL_n(\beta;\delta_1,\delta_2)$, respectively. Their intersection, spanned by the diagrams with no boundary connections, is isomorphic to $\TL_n(\beta)$.

The two-boundary TL algebra $\BBTL_n(\beta;\alpha_1,\alpha_2,\delta_1,\delta_2)$ similarly has subalgebras isomorphic to $\BTL_n(\beta;\alpha_1,\alpha_2)$ and $\BTL_n(\beta;\delta_1,\delta_2)$, whose intersection is isomorphic to $\TL_n(\beta)$. Further, we can construct a homomorphism $\psi: \BBTL_n(\beta;\alpha_1,\alpha_2,\delta_1,\delta_2) \to \Gh_n$ by mapping each $\BBTL_n$-diagram to the $\Gh_n$-diagram obtained by replacing any top-to-bottom boundary arcs with factors of $\gamma_{12}$ or $\gamma_3$ and leaving ghosts at their endpoints, as is done in $\Gh_n$ multiplication. This is not injective, since $\BBTL_n$ is infinite-dimensional, while $\Gh_n$ is finite-dimensional.

From Section \ref{s:dilute}, $\dGh_n$ has one-boundary subalgebras isomorphic to $\dGho_n(\beta;\alpha_1,\alpha_2,\alpha_3)$ and $\dGho_n(\beta;\delta_1,\delta_2,\delta_3)$, spanned by the $\dGh_n$-diagrams with no strings attached to the bottom and top boundaries, respectively. Within these subalgebras, the diagrams without any boundary connections span a subalgebra isomorphic to $\dTL_n(\beta)$, and the diagrams without empty nodes span subalgebras isomorphic to $\Gho_n(\beta;\alpha_1,\alpha_2,\alpha_3)$ and $\Gho_n(\beta;\delta_1,\delta_2,\delta_3)$, respectively. Recall that $\dTL_n$ has a non-unital subalgebra isomorphic to $\TL_n(\beta)$, spanned by the diagrams without empty nodes. The dilute ghost algebra similarly has a non-unital subalgebra isomorphic to $\Gh_n$, spanned by the $\dGh_n$-diagrams without empty nodes. 

These relationships are summarised in a commutative diagram in Figure \ref{fig:nonunital}.

Since unital subalgebras are often more useful than non-unital subalgebras, it is fortunate that we also have similar unital subalgebras of these dilute algebras. Recall that a dashed string in a dilute diagram is the sum of the diagrams with a solid string and no string in that position. When a dashed string is connected to a dashed string in multiplication, the result is a dashed string. A dashed loop is then the sum of a loop and no loop, so the parameter associated with a dashed loop is effectively the loop parameter, plus one. Similarly, the parameter associated with a dashed boundary arc is effectively the parameter associated with the corresponding solid boundary arc, plus one. It follows that we have an injective unital homomorphism from $\Gho_n(\beta;\alpha_1,\alpha_2,\alpha_3)$ to $\dGho_n(\beta-1;\alpha_1-1,\alpha_2-1,\alpha_3-1)$, obtained by turning the solid strings in each $\Gho_n$-diagram into dashed strings. There are analogous injective unital homomorphisms between the ghost and dilute ghost algebras, and between $\TL_n(\beta)$ and $\dTL_n(\beta-1)$. This is summarised in Figure \ref{fig:unital}.

We note that the zero- and one-boundary algebras $\TL_n$, $\dTL_n$, $\BTL_n$, $\Gho_n$ and $\dGho_n$ also have injective unital homomorphisms into the corresponding algebras with the same parameters and subscript $n+1$. This is given by adding a string linking the $(n+1)$th nodes on each side of the diagram, or a dashed string in the dilute cases. This process does not work for the two-boundary algebras, since the added string would intersect any strings connected to the bottom boundary.

\section{Loop models}\label{s:loopmodels}
In Sections \ref{ss:dense} and \ref{ss:dilute}, we present lattice loop models associated with the ghost algebra and the dilute ghost algebra, respectively. For each algebra, the goal is to produce a one-parameter family of commuting \textit{transfer tangles} $T(u)$. These are built from the \textit{bulk face operators} discussed in Sections \ref{sss:densebulk} and \ref{sss:dilutebulk}, and the \textit{boundary face operators} discussed in Sections \ref{sss:densebdy} and \ref{sss:dilutebdy}. We follow a standard construction of the transfer tangles such that, if the face operators satisfy the \textit{Yang-Baxter equation (YBE)}, the local inversion relation and the \textit{boundary Yang-Baxter equations (BYBEs)}, then the resulting transfer tangles commute, as proven in \cite[\S 2.4]{BehrendPearceOBrien}. Since the bulk of these lattices behaves identically to the TL or dilute TL lattices, we use known bulk face operators from each that satisfy the YBE and local inversion relation. Any new boundary behaviour due to the ghosts is captured by the boundary face operators; much of Sections \ref{sss:densebdy} and \ref{sss:dilutebdy} is spent classifying boundary face operators that satisfy the BYBEs. The final transfer tangles are given in Sections \ref{sss:densetransfer} and \ref{sss:dilutetransfer}.

\subsection{Ghost algebra}\label{ss:dense}
The ghost algebra can be used to describe a fully-packed lattice loop model with two boundaries. An example of such a lattice is given in Figure \ref{fig:denselattice}. These lattices are constructed from two possible \textit{bulk squares}
\begin{align}
\begin{tikzpicture}[baseline={([yshift=-1mm]current bounding box.center)},scale=0.75]
{
\lup{(0,0)}
}
\end{tikzpicture}
\; , 
\hspace{30mm}
\begin{tikzpicture}[baseline={([yshift=-1mm]current bounding box.center)},scale=0.75]
{
\rup{(0,0)}
}
\end{tikzpicture} \; ,
\end{align}
and five possible \textit{boundary triangles} for each boundary,
\begin{align}
\begin{tikzpicture}[baseline={([yshift=-1mm]current bounding box.center)},scale=0.75]
{
\aatop{(0,0)}
}
\end{tikzpicture}
\; , \qquad
\begin{tikzpicture}[baseline={([yshift=-1mm]current bounding box.center)},scale=0.75]
{
\batop{(0,0)}
}
\end{tikzpicture}
\; , \qquad
\begin{tikzpicture}[baseline={([yshift=-1mm]current bounding box.center)},scale=0.75]
{
\bbtop{(0,0)}
}
\end{tikzpicture}
\; , \qquad
\begin{tikzpicture}[baseline={([yshift=-1mm]current bounding box.center)},scale=0.75]
{
\bctop{(0,0)}
}
\end{tikzpicture}
\;, \qquad
\begin{tikzpicture}[baseline={([yshift=-1mm]current bounding box.center)},scale=0.75]
{
\bdtop{(0,0)}
}
\end{tikzpicture} 
\; , \nonumber
\\[4mm]
\begin{tikzpicture}[baseline={([yshift=-1mm]current bounding box.center)},scale=0.75]
{
\aabot{(0,0)}
}
\end{tikzpicture}
\; , \qquad
\begin{tikzpicture}[baseline={([yshift=-1mm]current bounding box.center)},scale=0.75]
{
\babot{(0,0)}
}
\end{tikzpicture}
\; , \qquad
\begin{tikzpicture}[baseline={([yshift=-1mm]current bounding box.center)},scale=0.75]
{
\bbbot{(0,0)}
}
\end{tikzpicture}
\; , \qquad
\begin{tikzpicture}[baseline={([yshift=-1mm]current bounding box.center)},scale=0.75]
{
\bcbot{(0,0)}
}
\end{tikzpicture}
\;, \qquad
\begin{tikzpicture}[baseline={([yshift=-1mm]current bounding box.center)},scale=0.75]
{
\bdbot{(0,0)}
}
\end{tikzpicture} 
\; . \label{eq:bdytris}
\end{align}

\begin{figure}[bt]
\centering
\begin{align*}
\begin{tikzpicture}[baseline={([yshift=-1mm]current bounding box.center)},scale=0.95]
{
\aatop{(0,0)}
\batop{(2,0)}
\bbtop{(4,0)}
\bctop{(6,0)}
\bdtop{(8,0)}
\batop{(10,0)}
\bbtop{(12,0)}
\bdtop{(14,0)}
\aabot{(6,-5)}
\bcbot{(4,-5)}
\bbbot{(10,-5)}
\bdbot{(0,-5)}
\aabot{(2,-5)}
\babot{(8,-5)}
\bbbot{(12,-5)}
\bcbot{(14,-5)}
\lup{(0,-1)}
\lup{(1,-1)}
\lup{(2,-1)}
\rup{(3,-1)}
\lup{(4,-1)}
\lup{(5,-1)}
\rup{(6,-1)}
\lup{(7,-1)}
\lup{(8,-1)}
\rup{(9,-1)}
\rup{(10,-1)}
\rup{(11,-1)}
\lup{(12,-1)}
\rup{(13,-1)}
\rup{(14,-1)}
\rup{(15,-1)}
\rup{(0,-2)}
\lup{(1,-2)}
\rup{(2,-2)}
\rup{(3,-2)}
\lup{(4,-2)}
\lup{(5,-2)}
\lup{(6,-2)}
\rup{(7,-2)}
\lup{(8,-2)}
\lup{(9,-2)}
\rup{(10,-2)}
\lup{(11,-2)}
\lup{(12,-2)}
\lup{(13,-2)}
\lup{(14,-2)}
\lup{(15,-2)}
\lup{(0,-3)}
\rup{(1,-3)}
\lup{(2,-3)}
\rup{(3,-3)}
\rup{(4,-3)}
\lup{(5,-3)}
\rup{(6,-3)}
\rup{(7,-3)}
\lup{(8,-3)}
\rup{(9,-3)}
\lup{(10,-3)}
\lup{(11,-3)}
\lup{(12,-3)}
\lup{(13,-3)}
\rup{(14,-3)}
\lup{(15,-3)}
\begin{pgfonlayer}{main}
\foreach \x in {0.5,2.5,4.5,6.5,8.5,10.5,12.5,14.5}
{\draw[lstring] (\x,-1) to[out=90,in=240] (\x,-0.5);
\draw[lstring] (\x,-4) to[out=-90,in=120] (\x,-4.5);
\foreach \y in {-0.5,-1,-2,-3,-4,-4.5}
{\fill[lgh] (\x,\y) circle (0.8pt);}}
\foreach \xx in {1.5,3.5,5.5,7.5,9.5,11.5,13.5,15.5}
{\draw[lstring] (\xx,-1) to[out=90,in=300] (\xx,-0.5);
\draw[lstring] (\xx,-4) to[out=-90,in=60] (\xx,-4.5);
\foreach \yy in {-0.5,-1,-2,-3,-4,-4.5}
{\fill[lgh] (\xx,\yy) circle (0.8pt);}}
\foreach \w in {1,2,3,4,5,6,7,8,9,10,11,12,13,14,15}
{\foreach \z in {-1.5,-2.5,-3.5}
{\fill[lgh] (\w,\z) circle (0.8 pt);}}
\end{pgfonlayer}
}
\end{tikzpicture}
\end{align*}
\vspace*{-5mm}
\caption{Fully-packed lattice that can be described using the ghost algebra $\Gh_3$.}\label{fig:denselattice}
\end{figure}
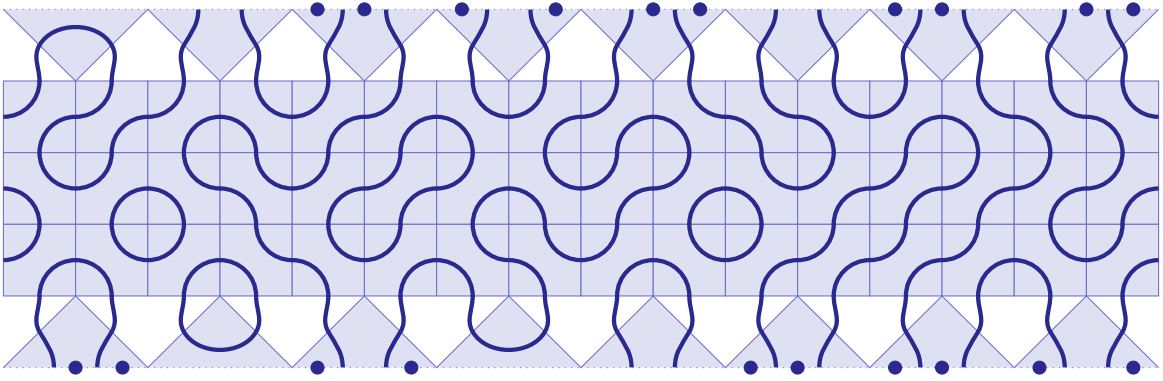

\subsubsection{Bulk face operator}\label{sss:densebulk}
The bulk squares of the ghost lattice model are the same as those in the Temperley-Lieb lattice model, so we can use the well-known Temperley-Lieb bulk face operator,
\begin{align}
\begin{tikzpicture}[baseline={([yshift=-1mm]current bounding box.center)},scale=0.75]
{
\bulk{$u$}{(0,0)}
}
\end{tikzpicture}
\ &=\  
\frac{\sin(\lambda-u)}{\sin(\lambda)}
\begin{tikzpicture}[baseline={([yshift=-1mm]current bounding box.center)},scale=0.75]
{
\filldraw[lsq] (0,0)--(1,1)--(2,0)--(1,-1)--cycle;
\draw[lstring] (0.5,0.5) to[out=-45,in=225] (1.5,0.5);
\draw[lstring] (0.5,-0.5) to[out=45,in=135] (1.5,-0.5);
}
\end{tikzpicture}
+
\frac{\sin(u)}{\sin(\lambda)}
\begin{tikzpicture}[baseline={([yshift=-1mm]current bounding box.center)},scale=0.75]
{
\filldraw[lsq] (0,0)--(1,1)--(2,0)--(1,-1)--cycle;
\draw[lstring] (0.5,0.5) to[out=-45,in=45] (0.5,-0.5);
\draw[lstring] (1.5,0.5) to[out=-135,in=135] (1.5,-0.5);
}
\end{tikzpicture}\; , \label{eq:densefaceop}
\end{align}
where $\beta = 2\cos(\lambda)$.

This satisfies the \textit{Yang-Baxter equation} (YBE)
\begin{align}
\begin{tikzpicture}[baseline={([yshift=-1mm]current bounding box.center)},scale=0.75]
{
\draw[lstring] (1.5,0.5)--(2,0.5);
\draw[lstring] (1.5,-0.5)--(2,-0.5);
\filldraw [lsq] (0,0)--(1,1)--(2,0)--(1,-1)--cycle;
\draw (1,0) node[anchor=center]{\footnotesize $u-v$};
\filldraw[richblue!70] (0.2,0.2)--(0.2,-0.2)--(0,0)--cycle;
\filldraw[lsq] (2,1) rectangle (3,0);
\fill[richblue!70] (2,0.3)--(2.3,0)--(2,0)--cycle;
\draw (2.5,0.5) node[anchor=center]{$v$};
\filldraw[lsq] (2,0) rectangle (3,-1);
\fill[richblue!70] (2,-0.7)--(2.3,-1)--(2,-1)--cycle;
\draw (2.5,-0.5) node[anchor=center]{$u$};
}
\end{tikzpicture}
\ =\ 
\begin{tikzpicture}[baseline={([yshift=-1mm]current bounding box.center)},scale=0.75]
{
\draw[lstring] (0,0.5)--(0.5,0.5);
\draw[lstring] (0,-0.5)--(0.5,-0.5);
\filldraw [lsq] (0,0)--(1,1)--(2,0)--(1,-1)--cycle;
\draw (1,0) node[anchor=center]{\footnotesize $u-v$};
\filldraw[richblue!70] (0.2,0.2)--(0.2,-0.2)--(0,0)--cycle;
\filldraw[lsq] (-1,0) rectangle (0,1);
\fill[richblue!70] (-1,0.3)--(-0.7,0)--(-1,0)--cycle;
\draw (-0.5,0.5) node[anchor=center]{$u$};
\filldraw[lsq] (-1,-1) rectangle (0,0);
\fill[richblue!70] (-1,-0.7)--(-0.7,-1)--(-1,-1)--cycle;
\draw (-0.5,-0.5) node[anchor=center]{$v$};
}
\end{tikzpicture} \; .\label{eq:denseYBE}
\end{align}

We also have \textit{crossing symmetry}
\begin{align}
\begin{tikzpicture}[baseline={([yshift=-1mm]current bounding box.center)},scale=0.75]
{
\bulk{$u$}{(0,0)}
}
\end{tikzpicture}
\; = \;
\begin{tikzpicture}[baseline={([yshift=-1mm]current bounding box.center)},scale=0.75,rotate=90]
{
\bulk{$\lambda - u$}{(0,0)}
}
\end{tikzpicture}
\; = \;
\begin{tikzpicture}[baseline={([yshift=-1mm]current bounding box.center)},scale=0.75, rotate=180]
{
\bulk{$u$}{(0,0)}
}
\end{tikzpicture}
\; = \;
\begin{tikzpicture}[baseline={([yshift=-1mm]current bounding box.center)},scale=0.75,rotate=270]
{
\bulk{$\lambda-u$}{(0,0)}
}
\end{tikzpicture}\; , \label{eq:densecrossing}
\end{align}
and the \textit{local inversion relation}
\begin{align}
\begin{tikzpicture}[baseline={([yshift=-1mm]current bounding box.center)},scale=0.75]
{
\bulk{$u$}{(0,0)}
\bulk{$-u$}{(2,0)}
\draw[lstring] (1.5,0.5) to[out=45,in=135] (2.5,0.5);
\draw[lstring] (1.5,-0.5) to[out=-45,in=-135] (2.5,-0.5);
}
\end{tikzpicture}
\; =\, f(u)\;
\begin{tikzpicture}[baseline={([yshift=-1mm]current bounding box.center)},scale=0.75]
{
\filldraw[lsq] (0,0)--(1,1)--(2,0)--(1,-1)--cycle;
\draw[lstring] (0.5,0.5) to[out=-45,in=-135] (1.5,0.5);
\draw[lstring] (0.5,-0.5) to[out=45,in=135] (1.5,-0.5);
}
\end{tikzpicture}\; , \qquad f \neq 0. \label{eq:denseLIR}
\end{align}

\subsubsection{Boundary face operators}\label{sss:densebdy}
Unlike the YBE and the local inversion relation, the BYBEs involve the boundary face operators. The boundary face operators in the TL lattice models use only the first triangle for each boundary from \eqref{eq:bdytris}, while the one- and two-boundary TL lattice models also use the second triangle for the top or both boundaries, respectively. Thus, for our lattice model to have any new boundary behaviour due to the ghosts, we cannot simply use the existing BYBE solutions from the zero-, one- and two-boundary TL lattice models. Hence in this section, we classify solutions to the BYBEs involving our new boundary triangles with ghosts.

The BYBE for the top boundary of the ghost lattice model is
\begin{align}
\begin{tikzpicture}[baseline={([yshift=-1mm]current bounding box.center)},scale=0.75]
{
\bulk{$u-v$}{(0,0)}
\utri{$u$}{(1,1)}
\bulk{$u+v$}{(2,0)}
\utri{$v$}{(3,1)}
\draw[lstring] (1.5,-0.5) to[out=-45,in=-135] (2.5,-0.5);
}
\end{tikzpicture}
\ = \ 
\begin{tikzpicture}[baseline={([yshift=-1mm]current bounding box.center)},scale=0.75]
{
\utri{$v$}{(0,0)}
\bulk{$u+v$}{(1,-1)}
\utri{$u$}{(2,0)}
\bulk{$u-v$}{(3,-1)}
\draw[lstring] (2.5,-1.5) to [out=-45,in=-135] (3.5,-1.5);
}
\end{tikzpicture}\; , \label{eq:denseBYBE}
\end{align}
and the BYBE for the bottom boundary is
\begin{align}
\begin{tikzpicture}[baseline={([yshift=-1mm]current bounding box.center)},xscale=0.75,yscale=-0.75]
{
\bulk{$u-v$}{(0,0)}
\utri{$u$}{(1,1)}
\bulk{$u+v$}{(2,0)}
\utri{$v$}{(3,1)}
\draw[lstring] (1.5,-0.5) to[out=-45,in=-135] (2.5,-0.5);
}
\end{tikzpicture}
\ = \ 
\begin{tikzpicture}[baseline={([yshift=-1mm]current bounding box.center)},xscale=0.75,yscale=-0.75]
{
\utri{$v$}{(0,0)}
\bulk{$u+v$}{(1,-1)}
\utri{$u$}{(2,0)}
\bulk{$u-v$}{(3,-1)}
\draw[lstring] (2.5,-1.5) to [out=-45,in=-135] (3.5,-1.5);
}
\end{tikzpicture}\; . \label{eq:denseBYBEbot}
\end{align}
Given the bulk face operator \eqref{eq:densefaceop}, we seek solutions for the boundary face operators satisfying \eqref{eq:denseBYBE} and \eqref{eq:denseBYBEbot}. Since the bottom boundary equation is just a reflection of the top boundary equation about a horizontal line, it suffices to solve the top boundary equation, and reflect the solution for the bottom boundary, as long as we also swap the parameters $\alpha_i \leftrightarrow \delta_i$. Note also that the BYBEs are homogeneous, so our solutions for the boundary face operator with parameter $u$ will have an overall scaling factor, that may depend on $u$.

Note that the top boundary BYBE involves no bottom boundary triangles, and has two places on each side that strings could be connected to. This means the top boundary BYBE is contained within the one-boundary ghost algebra on two nodes, $\Gho_2$. Hence we will solve this BYBE within $\Gho_2$, so our diagrams will have no bottom boundary.

Since the top boundary BYBE is exactly the same for the one-boundary ghost algebra $\Gho_n$, and the bottom triangle with no boundary connections solves the bottom boundary BYBE, this means the one-boundary ghost lattice model is also Yang-Baxter integrable, with the transfer tangle given in Section \ref{sss:densetransfer}, using that bottom triangle alone as the bottom boundary face operator.

Let 
\begin{align}
\begin{tikzpicture}[baseline={([yshift=-1mm]current bounding box.center)},scale=0.75]
{
\utri{$u$}{(0,0)}
}
\end{tikzpicture}
&=\, 
a(u)
\begin{tikzpicture}[baseline={([yshift=-1mm]current bounding box.center)},scale=0.75]
{
\aatop{(0,0)}
}
\end{tikzpicture}
+b_1(u)
\begin{tikzpicture}[baseline={([yshift=-1mm]current bounding box.center)},scale=0.75]
{
\batop{(0,0)}
}
\end{tikzpicture}
+b_2(u)
\begin{tikzpicture}[baseline={([yshift=-1mm]current bounding box.center)},scale=0.75]
{
\bbtop{(0,0)}
}
\end{tikzpicture}
+b_3(u)
\begin{tikzpicture}[baseline={([yshift=-1mm]current bounding box.center)},scale=0.75]
{
\bctop{(0,0)}
}
\end{tikzpicture}
+b_4(u)
\begin{tikzpicture}[baseline={([yshift=-1mm]current bounding box.center)},scale=0.75]
{
\bdtop{(0,0)}
}
\end{tikzpicture} \label{eq:densetri}
\end{align}
and let
\begin{align}
A_1 &= \alpha_1^2 + \alpha_2^2 - 2\alpha_3^2, &A_2 &= \alpha_1\alpha_2 - \alpha_3^2. \label{eq:A1A2def}
\end{align}

\begin{theorem}\label{thm:denseBYBE}
Let the face operator for $\Gh_n$ be given by \eqref{eq:densefaceop}. Then the BYBE \eqref{eq:denseBYBE} admits two solutions for the boundary face operator \eqref{eq:densetri}. Solution I is
\begin{align}
a(u) &= \frac{b(u)}{2\sin(2u)\sin (\lambda)} \bigg(\kappa +  \Big(\!\left(c_1\alpha_1 + c_3\alpha_2 + \left(c_2+c_4\right) \alpha_3\right) \cos(2u) \nonumber \\
&\hspace{46mm} - \left(c_3 \alpha_1 + c_1\alpha_2 + \left(c_2+c_4\right)\alpha_3\right) \cos(2u-\lambda) \Big)\!\bigg), \\
b_i(u) &= c_i\, b(u), \qquad i \in \{1,2,3,4\},
\end{align}
for any $\kappa, c_1,c_2,c_3,c_4 \in \C$ with $c_1c_3 = c_2c_4$ and $b: \C \to \C$. Solution II is
\begin{align}
a(u) &= \frac{b(u)}{2\sin(2u)\sin(\lambda)} \left(\kappa -A_1\cos(2u) + A_2 \cos(2u-\lambda)  \right), \\
b_1(u) &= -\alpha_1\, b(u), \\
b_2(u) &= b_4(u) = \alpha_3\, b(u), \\
b_3(u) &= -\alpha_2\, b(u),
\end{align}
for any $\kappa \in \C$ and $b:\C \to \C$. 
\end{theorem}

\begin{proof}
We consider the expression
\begin{align}
\begin{tikzpicture}[baseline={([yshift=-1mm]current bounding box.center)},scale=0.7]
{
\utri{$v$}{(0,0)}
\bulk{$u+v$}{(1,-1)}
\utri{$u$}{(2,0)}
\bulk{$u-v$}{(3,-1)}
\draw[lstring] let \n1 = {sin(45)} in (2.5,-1.5) arc (225:315:\n1);
}
\end{tikzpicture}
\ - \ \begin{tikzpicture}[baseline={([yshift=-1mm]current bounding box.center)},scale=0.7]
{
\bulk{$u-v$}{(0,0)}
\utri{$u$}{(1,1)}
\bulk{$u+v$}{(2,0)}
\utri{$v$}{(3,1)}
\draw[lstring] let \n1 = {sin(45)} in (1.5,-0.5) arc (225:315:\n1);
}
\end{tikzpicture}\;  \label{eq:denseBYBEequiv}
\end{align}
as a linear combination of the basis diagrams, and note that the BYBE \eqref{eq:denseBYBE} holds if and only if the coefficient of each of these diagrams is zero. 

Now, the bulk and boundary face operators in \eqref{eq:denseBYBEequiv} are linear combinations of two squares and five triangles, respectively. This means \eqref{eq:denseBYBEequiv} is a linear combination of $5\times 2\times 5\times 2+2\times 5\times 2\times 5 = 200$ \textit{configurations} of these squares and triangles. However, many of these configurations correspond to the same diagrams, and some of their contributions to the coefficients of these diagrams cancel out. As a result of this, only twenty-four of the thirty possible $\Gho_2$-diagrams have coefficients that are not trivially zero.

We first consider the diagrams with four strings connected to the boundary. Each of these diagrams arises from two possible configurations; for example,
\begin{align}
\begin{tikzpicture}[baseline={([yshift=-1mm]current bounding box.center)},scale=0.6]
{
\draw[dotted] (0,0.5)--(3,0.5);
\draw[very thick] (0,0.5)--(0,-1.5);
\draw[very thick] (3,0.5)--(3,-1.5);
\draw (0,0) to[out=0,in=-90] (0.6,0.5);
\draw (0,-1) to[out=0,in=-90] (1.2,0.5);
\draw (3,-1) to[out=180,in=270] (1.8,0.5);
\draw (3,0) to[out=180,in=270] (2.4,0.5);
\filldraw (0.3,0.5) circle (0.08);
\filldraw (2.1,0.5) circle (0.08);
}
\end{tikzpicture}\; :
\qquad
\begin{tikzpicture}[baseline={([yshift=-1mm]current bounding box.center)},scale=0.6]
{
\bctri{(0,0)}
\esq{(1,-1)}
\bbtri{(2,0)}
\idsq{(3,-1)}
\draw[lstring] let \n1 = {sin(45)} in (2.5,-1.5) arc (225:315:\n1);
\foreach \x in {2.5,3.5}
{\foreach \y in {-0.5,-1.5}
{\fill[lgh] (\x,\y) circle (1.6 pt);}}
\fill[lgh] (1.5,-0.5) circle (1.6 pt);
}
\end{tikzpicture}\; , \qquad
\begin{tikzpicture}[baseline={([yshift=-1mm]current bounding box.center)},scale=0.6]
{
\idsq{(1,-1)}
\bctri{(2,0)}
\esq{(3,-1)}
\bbtri{(4,0)}
\draw[lstring] let \n1 = {sin(45)} in (2.5,-1.5) arc (225:315:\n1);
\foreach \x in {2.5,3.5}
{\foreach \y in {-0.5,-1.5}
{\fill[lgh] (\x,\y) circle (1.6 pt);}}
\fill[lgh] (4.5,-0.5) circle (1.6 pt);
}
\end{tikzpicture}\; .
\end{align}
The coefficients of the diagrams with four strings connected to the boundary thus have the form
\begin{align}
\Big( b_i(u)b_j(v)-b_i(v)b_j(u)\Big)\frac{\sin(u+v)\sin(u-v-\lambda)}{\sin^2(\lambda)} \label{eq:bibj}
\end{align}
with $i,j \in \{1,2,3,4\}$. For these to be zero for all $u,v$, the expression in the larger brackets must be zero. If $b_i\neq 0$ for some $i \in \{1,2,3,4\}$, then for each $j$, we have
\begin{align}
\frac{b_j(v)}{b_i(v)} = \frac{b_j(u)}{b_i(u)}.
\end{align}
Since one side of this equation depends only on $u$ and the other only on $v$, both sides are independent of $u$ and $v$. So each $b_j$ is a scalar multiple of $b_i$. Hence there exists a nonzero function $b$ such that, for each $k \in \{1,2,3,4\}$, we can write
\begin{align}
b_k(u) = c_k b(u) \label{eq:introb}
\end{align}
for some $c_k \in \C$. This also holds if $b_i=0$ for all $i \in \{1,2,3,4\}$, in which case $c_k = 0$ for all $k \in \{1,2,3,4\}$.

Next, we consider the diagrams
\begin{align}
\begin{tikzpicture}[baseline={([yshift=-1mm]current bounding box.center)},scale=0.6]
{
\draw[dotted] (0,0.5)--(3,0.5);
\draw [very thick](0,-1.5) -- (0,0.5);
\draw [very thick](3,-1.5)--(3,0.5);
\draw (0,0) to[out=0,in=-90] (0.8,0.5);
\draw (0,-1) to[out=0,in=-90] (1.6,0.5);
\draw (3,0) arc (90:270:0.5);
}
\end{tikzpicture}
\; , \qquad
\begin{tikzpicture}[baseline={([yshift=-1mm]current bounding box.center)},scale=0.6]
{
\draw[dotted] (0,0.5)--(3,0.5);
\draw [very thick](0,-1.5) -- (0,0.5);
\draw [very thick](3,-1.5)--(3,0.5);
\draw (0,0) to[out=0,in=-90] (0.8,0.5);
\draw (0,-1) to[out=0,in=-90] (1.6,0.5);
\draw (3,0) arc (90:270:0.5);
\filldraw (0.4,0.5) circle (0.08);
\filldraw (1.2,0.5) circle (0.08);
}
\end{tikzpicture}
\; , \qquad
\begin{tikzpicture}[baseline={([yshift=-1mm]current bounding box.center)},scale=0.6]
{
\draw[dotted] (0,0.5)--(3,0.5);
\draw [very thick](0,-1.5) -- (0,0.5);
\draw [very thick](3,-1.5)--(3,0.5);
\draw (0,0) to[out=0,in=-90] (0.8,0.5);
\draw (0,-1) to[out=0,in=-90] (1.6,0.5);
\draw (3,0) arc (90:270:0.5);
\filldraw (0.4,0.5) circle (0.08);
\filldraw (2,0.5) circle (0.08);
}
\end{tikzpicture}
\; , \qquad
\begin{tikzpicture}[baseline={([yshift=-1mm]current bounding box.center)},scale=0.6]
{
\draw[dotted] (0,0.5)--(3,0.5);
\draw [very thick](0,-1.5) -- (0,0.5);
\draw [very thick](3,-1.5)--(3,0.5);
\draw (0,0) to[out=0,in=-90] (0.8,0.5);
\draw (0,-1) to[out=0,in=-90] (1.6,0.5);
\draw (3,0) arc (90:270:0.5);
\filldraw (2,0.5) circle (0.08);
\filldraw (1.2,0.5) circle (0.08);
}
\end{tikzpicture}
\; . \label{eq:densediags}
\end{align}
There are thirteen configurations that contribute to the coefficient of the first diagram. These configurations and their contributions are as follows. Note that the minus sign in the contribution from the first configuration comes from the subtraction in \eqref{eq:denseBYBEequiv}, and the factors of $\beta=2\cos(\lambda)$ and each $\alpha_i$ come from the loops and boundary arcs formed, respectively.
\begin{align}
\begin{tikzpicture}[baseline={([yshift=-1mm]current bounding box.center)},scale=0.6]
{
\idsq{(1,-1)}
\batri{(2,0)}
\esq{(3,-1)}
\aatri{(4,0)}
\draw[lstring] let \n1 = {sin(45)} in (2.5,-1.5) arc (225:315:\n1);
\foreach \x in {2.5,3.5}
{\foreach \y in {-0.5,-1.5}
{\fill[lgh] (\x,\y) circle (1.6 pt);}}
\fill[lgh] (4.5,-0.5) circle (1.6 pt);
}
\end{tikzpicture}
&&
-c_1\, a(v)b(u)\sin(\lambda-u+v)\sin(u+v)
\\[2.5mm]
\begin{tikzpicture}[baseline={([yshift=-1mm]current bounding box.center)},scale=0.6]
{
\batri{(0,0)}
\esq{(1,-1)}
\aatri{(2,0)}
\idsq{(3,-1)}
\draw[lstring] let \n1 = {sin(45)} in (2.5,-1.5) arc (225:315:\n1);
\foreach \x in {2.5,3.5}
{\foreach \y in {-0.5,-1.5}
{\fill[lgh] (\x,\y) circle (1.6 pt);}}
\fill[lgh] (1.5,-0.5) circle (1.6 pt);
}
\end{tikzpicture}
&&
c_1\, a(u)b(v) \sin(u+v)\sin(\lambda-u+v)
\\[2.5mm]
\begin{tikzpicture}[baseline={([yshift=-1mm]current bounding box.center)},scale=0.6]
{
\batri{(0,0)}
\idsq{(1,-1)}
\aatri{(2,0)}
\esq{(3,-1)}
\draw[lstring] let \n1 = {sin(45)} in (2.5,-1.5) arc (225:315:\n1);
\foreach \x in {2.5,3.5}
{\foreach \y in {-0.5,-1.5}
{\fill[lgh] (\x,\y) circle (1.6 pt);}}
\fill[lgh] (1.5,-0.5) circle (1.6 pt);
}
\end{tikzpicture}
&&
c_1\, a(u) b(v) \sin(\lambda-u-v)\sin(u-v)
\\[2.5mm]
\begin{tikzpicture}[baseline={([yshift=-1mm]current bounding box.center)},scale=0.6]
{
\batri{(0,0)}
\esq{(1,-1)}
\aatri{(2,0)}
\esq{(3,-1)}
\draw[lstring] let \n1 = {sin(45)} in (2.5,-1.5) arc (225:315:\n1);
\foreach \x in {2.5,3.5}
{\foreach \y in {-0.5,-1.5}
{\fill[lgh] (\x,\y) circle (1.6 pt);}}
\fill[lgh] (1.5,-0.5) circle (1.6 pt);
}
\end{tikzpicture}
&&
2\cos(\lambda)\, c_1\,   a(u)b(v)\sin(u+v)\sin(u-v)
\\[2.5mm]
\begin{tikzpicture}[baseline={([yshift=-1mm]current bounding box.center)},scale=0.6]
{
\aatri{(0,0)}
\idsq{(1,-1)}
\batri{(2,0)}
\esq{(3,-1)}
\draw[lstring] let \n1 = {sin(45)} in (2.5,-1.5) arc (225:315:\n1);
\foreach \x in {2.5,3.5}
{\foreach \y in {-0.5,-1.5}
{\fill[lgh] (\x,\y) circle (1.6 pt);}}
\fill[lgh] (1.5,-0.5) circle (1.6 pt);
}
\end{tikzpicture}
&&
c_1\, a(v)b(u)\sin(\lambda-u-v)\sin(u-v)
\\[2.5mm]
\begin{tikzpicture}[baseline={([yshift=-1mm]current bounding box.center)},scale=0.6]
{
\batri{(0,0)}
\esq{(1,-1)}
\batri{(2,0)}
\esq{(3,-1)}
\draw[lstring] let \n1 = {sin(45)} in (2.5,-1.5) arc (225:315:\n1);
\foreach \x in {2.5,3.5}
{\foreach \y in {-0.5,-1.5}
{\fill[lgh] (\x,\y) circle (1.6 pt);}}
\fill[lgh] (1.5,-0.5) circle (1.6 pt);
}
\end{tikzpicture}
&&
\alpha_1\, c_1^2\, b(u)b(v)\sin(u+v)\sin(u-v)
\\[2.5mm]
\begin{tikzpicture}[baseline={([yshift=-1mm]current bounding box.center)},scale=0.6]
{
\batri{(0,0)}
\esq{(1,-1)}
\bbtri{(2,0)}
\esq{(3,-1)}
\draw[lstring] let \n1 = {sin(45)} in (2.5,-1.5) arc (225:315:\n1);
\foreach \x in {2.5,3.5}
{\foreach \y in {-0.5,-1.5}
{\fill[lgh] (\x,\y) circle (1.6 pt);}}
\fill[lgh] (1.5,-0.5) circle (1.6 pt);
}
\end{tikzpicture}
&&
\alpha_3\, c_1c_2\, b(u)b(v) \sin(u+v)\sin(u-v)
\\[2.5mm]
\begin{tikzpicture}[baseline={([yshift=-1mm]current bounding box.center)},scale=0.6]
{
\batri{(0,0)}
\esq{(1,-1)}
\bctri{(2,0)}
\esq{(3,-1)}
\draw[lstring] let \n1 = {sin(45)} in (2.5,-1.5) arc (225:315:\n1);
\foreach \x in {2.5,3.5}
{\foreach \y in {-0.5,-1.5}
{\fill[lgh] (\x,\y) circle (1.6 pt);}}
\fill[lgh] (1.5,-0.5) circle (1.6 pt);
}
\end{tikzpicture}
&&
\alpha_2\, c_1 c_3\, b(u)b(v) \sin(u+v)\sin(u-v)
\\[2.5mm]
\begin{tikzpicture}[baseline={([yshift=-1mm]current bounding box.center)},scale=0.6]
{
\batri{(0,0)}
\esq{(1,-1)}
\bdtri{(2,0)}
\esq{(3,-1)}
\draw[lstring] let \n1 = {sin(45)} in (2.5,-1.5) arc (225:315:\n1);
\foreach \x in {2.5,3.5}
{\foreach \y in {-0.5,-1.5}
{\fill[lgh] (\x,\y) circle (1.6 pt);}}
\fill[lgh] (1.5,-0.5) circle (1.6 pt);
}
\end{tikzpicture}
&&
\alpha_3\, c_1 c_4\, b(u)b(v) \sin(u+v)\sin(u-v)
\\[2.5mm]
\begin{tikzpicture}[baseline={([yshift=-1mm]current bounding box.center)},scale=0.6]
{
\batri{(0,0)}
\idsq{(1,-1)}
\batri{(2,0)}
\esq{(3,-1)}
\draw[lstring] let \n1 = {sin(45)} in (2.5,-1.5) arc (225:315:\n1);
\foreach \x in {2.5,3.5}
{\foreach \y in {-0.5,-1.5}
{\fill[lgh] (\x,\y) circle (1.6 pt);}}
\fill[lgh] (1.5,-0.5) circle (1.6 pt);
}
\end{tikzpicture}
&&
\alpha_2\, c_1^2\, b(u)b(v) \sin(\lambda-u-v) \sin(u-v)
\end{align}
\begin{align}
\begin{tikzpicture}[baseline={([yshift=-1mm]current bounding box.center)},scale=0.6]
{
\batri{(0,0)}
\idsq{(1,-1)}
\bbtri{(2,0)}
\esq{(3,-1)}
\draw[lstring] let \n1 = {sin(45)} in (2.5,-1.5) arc (225:315:\n1);
\foreach \x in {2.5,3.5}
{\foreach \y in {-0.5,-1.5}
{\fill[lgh] (\x,\y) circle (1.6 pt);}}
\fill[lgh] (1.5,-0.5) circle (1.6 pt);
}
\end{tikzpicture}
&&
\alpha_3\, c_1c_2\, b(u)b(v) \sin(\lambda-u-v)\sin(u-v)
\\[2.5mm]
\begin{tikzpicture}[baseline={([yshift=-1mm]current bounding box.center)},scale=0.6]
{
\bdtri{(0,0)}
\idsq{(1,-1)}
\batri{(2,0)}
\esq{(3,-1)}
\draw[lstring] let \n1 = {sin(45)} in (2.5,-1.5) arc (225:315:\n1);
\foreach \x in {2.5,3.5}
{\foreach \y in {-0.5,-1.5}
{\fill[lgh] (\x,\y) circle (1.6 pt);}}
\fill[lgh] (1.5,-0.5) circle (1.6 pt);
}
\end{tikzpicture}
&&
\alpha_3\, c_1c_4\, b(u)b(v) \sin(\lambda-u-v) \sin(u-v)
\\[2.5mm]
\begin{tikzpicture}[baseline={([yshift=-1mm]current bounding box.center)},scale=0.6]
{
\bdtri{(0,0)}
\idsq{(1,-1)}
\bbtri{(2,0)}
\esq{(3,-1)}
\draw[lstring] let \n1 = {sin(45)} in (2.5,-1.5) arc (225:315:\n1);
\foreach \x in {2.5,3.5}
{\foreach \y in {-0.5,-1.5}
{\fill[lgh] (\x,\y) circle (1.6 pt);}}
\fill[lgh] (1.5,-0.5) circle (1.6 pt);
}
\end{tikzpicture}
&& \alpha_1\, c_2c_4\, b(u)b(v) \sin(\lambda-u-v)\sin(u-v)
\end{align}

The sum of these contributions gives the coefficient of the first diagram. Simplifying it and setting it equal to zero, we find
\begin{align}
0 &= c_1\left( a(u)b(v) \sin(2u) - a(v)b(u) \sin(2v) \right)\nonumber \\
&\hspace{7mm}+ b(u)b(v) \csc(\lambda)\sin(u-v) \Big(c_1 (c_1\alpha_1 +c_3\alpha_2 + (c_2+c_4)\alpha_3) \sin (u+v)  \nonumber \\
&\hspace{57mm}-\left(c_2 c_4\alpha_1 +c_1^2\alpha_2+ c_1 (c_2+c_4)\alpha_3\right)\sin (u+v-\lambda) \Big)\, .
\end{align}
If $c_1 \neq 0$, this can be rearranged to find
\begin{align}
&\quad 2\,\frac{a(u)}{b(u)} \sin (\lambda ) \sin (2 u)-\cos(2 u) (c_1\alpha_1+c_3\alpha_2  + (c_2+c_4)\alpha_3) \nonumber \\
&\qquad\quad  +\cos(2 u-\lambda ) \left(\frac{c_2 c_4}{c_1}\alpha_1+c_1\alpha_2 + (c_2+c_4)\alpha_3\!\right) \nonumber \\
&= 2\,\frac{a(v) }{b(v)}\sin (\lambda ) \sin (2 v)-\cos(2 v) (c_1\alpha_1+c_3 \alpha_2 +(c_2+c_4)\alpha_3)\nonumber \\
&\qquad\quad  + \cos(2 v-\lambda ) \left(c_1\alpha_2+\frac{c_2 c_4}{c_1}\alpha_1 + (c_2+c_4)\alpha_3 \!\right),
\end{align}
and therefore 
\begin{align}
\kappa &= 2\,\frac{a(u)}{b(u)} \sin (\lambda ) \sin (2 u)-\cos(2 u) (c_1\alpha_1+c_3\alpha_2  + (c_2+c_4)\alpha_3)  \nonumber \\
&\qquad +\cos(2 u-\lambda ) \left(\frac{c_2 c_4}{c_1}\alpha_1+c_1\alpha_2 + (c_2+c_4)\alpha_3\!\right)
\end{align}
for some $\kappa \in\C$. Then
\begin{align}
a(u) &= \frac{1}{2}\, b(u) \csc(2u)\csc(\lambda) \bigg( \kappa + \cos(2 u) (c_1\alpha_1+c_3\alpha_2  + (c_2+c_4)\alpha_3) \nonumber \\
&\hspace{45mm}-\cos(2 u-\lambda ) \left(\frac{c_2 c_4}{c_1}\alpha_1+c_1\alpha_2 + (c_2+c_4)\alpha_3\right) \!\!\bigg).
\end{align}
Substituting this into the coefficients of the remaining three diagrams from \eqref{eq:densediags}, and setting them equal to zero, we find
\begin{align}
0&= \frac{-1}{c_1}\ (c_1c_3-c_2c_4)(c_2\alpha_1 + c_1 \alpha_3)\ b(u)b(v)\csc^2(\lambda)\sin(u-v)\sin(u+v-\lambda), \\
0&= \frac{-1}{c_1}\ (c_1c_3-c_2c_4)(c_3\alpha_1 - c_1 \alpha_2)\ b(u)b(v)\csc^2(\lambda)\sin(u-v)\sin(u+v-\lambda),\\
0&= \frac{-1}{c_1}\ (c_1c_3-c_2c_4)(c_4\alpha_1 + c_1 \alpha_3)\ b(u)b(v)\csc^2(\lambda)\sin(u-v)\sin(u+v-\lambda).
\end{align}
Hence we must have 
\begin{align}
c_1c_3 &= c_2c_4, \label{eq:c1c3eqc2c4}
\end{align}
or
\begin{align}
c_1&=-\alpha_1\, \rho, & c_2&=c_4 = \alpha_3\, \rho, &c_3 &= -\alpha_2\, \rho, \label{eq:c1c3NEQc2c4}
\end{align}
for some $\rho \in \C$. 

If $c_1=0$, but $c_i \neq 0$ for some $i \in \{2,3,4\}$, then setting the coefficients of all four diagrams from \eqref{eq:densediags} to zero analogously implies that \eqref{eq:c1c3eqc2c4} or \eqref{eq:c1c3NEQc2c4} must hold. If $c_i = 0$ for all $i$, both are true. Applying either of \eqref{eq:c1c3eqc2c4} and \eqref{eq:c1c3NEQc2c4} is sufficient for the coefficients of all diagrams to be zero, and after some rescaling, this yields the two solutions from the theorem.
\end{proof}

\subsubsection{Transfer tangle}\label{sss:densetransfer}
The \textit{transfer tangle} $T(u)$ is constructed from the bulk and boundary face operators as
\begin{align}
    T(u) = \begin{tikzpicture}[baseline={([yshift=-1mm]current bounding box.center)},scale=1]
{
\draw[lstring] (0.5,-1)--(0.5,-0.5);
\draw[lstring] (0.5,-5)--(0.5,-5.5);
\draw[lstring] (1.5,-1)--(1.5,-0.5);
\draw[lstring] (1.5,-5)--(1.5,-5.5);
\facetop{(0,0)}{$\lambda -u$}
\faceop{(0,-1)}{$u$}
\faceop{(1,-1)}{$\lambda - u$}
\faceop{(0,-2)}{$u$}
\faceop{(1,-2)}{$\lambda -u$}
\vdotsq{(0,-3)}
\vdotsq{(1,-3)}
\faceop{(0,-4)}{$u$}
\faceop{(1,-4)}{$\lambda - u$}
\facebot{(0,-6)}{$u$}
}
\end{tikzpicture}\ \; .
\end{align}
Using the YBE, local inversion relation, BYBEs and crossing symmetry, it can be shown that $T(u)T(v) = T(v)T(u)$ for any $u,v \in \C$; see, for example, \cite[\S 2.4]{BehrendPearceOBrien}.

\subsection{Dilute ghost algebra}\label{ss:dilute}
The dilute ghost algebra can be used to describe a dilute lattice loop model with two boundaries. An example of such a lattice is given in Figure \ref{fig:dilutelattice}. The dilute lattice is constructed from nine possible bulk squares
\begin{align}
\begin{tikzpicture}[baseline={([yshift=-1mm]current bounding box.center)},scale=0.75]
{
\lup{(0,0)}
}
\end{tikzpicture}\; ,\quad
\begin{tikzpicture}[baseline={([yshift=-1mm]current bounding box.center)},scale=0.75]
{
\rup{(0,0)}
}
\end{tikzpicture}\; ,\quad
\begin{tikzpicture}[baseline={([yshift=-1mm]current bounding box.center)},scale=0.75]
{
\llup{(0,0)}
}
\end{tikzpicture}\; ,\quad
\begin{tikzpicture}[baseline={([yshift=-1mm]current bounding box.center)},scale=0.75]
{
\rrup{(0,0)}
}
\end{tikzpicture}\; ,\quad
\begin{tikzpicture}[baseline={([yshift=-1mm]current bounding box.center)},scale=0.75]
{
\lld{(0,0)}
}
\end{tikzpicture}\; ,\quad
\begin{tikzpicture}[baseline={([yshift=-1mm]current bounding box.center)},scale=0.75]
{
\rrd{(0,0)}
}
\end{tikzpicture}\; ,\quad
\begin{tikzpicture}[baseline={([yshift=-1mm]current bounding box.center)},scale=0.75]
{
\hh{(0,0)}
}
\end{tikzpicture}\; ,\quad
\begin{tikzpicture}[baseline={([yshift=-1mm]current bounding box.center)},scale=0.75]
{
\vv{(0,0)}
}
\end{tikzpicture}\; ,\quad
\begin{tikzpicture}[baseline={([yshift=-1mm]current bounding box.center)},scale=0.75]
{
\emp{(0,0)}
}
\end{tikzpicture}\; ,
\end{align}
and ten possible boundary triangles for each boundary, 
\begin{align}
\begin{tikzpicture}[baseline={([yshift=-1mm]current bounding box.center)},scale=0.75]
{
\aatop{(0,0)}
}
\end{tikzpicture}
\; , \qquad
\begin{tikzpicture}[baseline={([yshift=-1mm]current bounding box.center)},scale=0.75]
{
\emtop{(0,0)}
}
\end{tikzpicture} 
\; , \qquad
\begin{tikzpicture}[baseline={([yshift=-1mm]current bounding box.center)},scale=0.75]
{
\batop{(0,0)}
}
\end{tikzpicture}
\; , \qquad
\begin{tikzpicture}[baseline={([yshift=-1mm]current bounding box.center)},scale=0.75]
{
\bbtop{(0,0)}
}
\end{tikzpicture}
\; , \qquad
\begin{tikzpicture}[baseline={([yshift=-1mm]current bounding box.center)},scale=0.75]
{
\bctop{(0,0)}
}
\end{tikzpicture}
\; , \nonumber \\[3mm]
\begin{tikzpicture}[baseline={([yshift=-1mm]current bounding box.center)},scale=0.75]
{
\bdtop{(0,0)}
}
\end{tikzpicture}
\; , \qquad
\begin{tikzpicture}[baseline={([yshift=-1mm]current bounding box.center)},scale=0.75]
{
\betop{(0,0)}
}
\end{tikzpicture}
\; , \qquad
\begin{tikzpicture}[baseline={([yshift=-1mm]current bounding box.center)},scale=0.75]
{
\bftop{(0,0)}
}
\end{tikzpicture}
\; , \qquad
\begin{tikzpicture}[baseline={([yshift=-1mm]current bounding box.center)},scale=0.75]
{
\bgtop{(0,0)}
}
\end{tikzpicture}
\; , \qquad
\begin{tikzpicture}[baseline={([yshift=-1mm]current bounding box.center)},scale=0.75]
{
\bhtop{(0,0)}
}
\end{tikzpicture}
\; , \nonumber \\[4mm]
\begin{tikzpicture}[baseline={([yshift=-1mm]current bounding box.center)},scale=0.75]
{
\aabot{(0,0)}
}
\end{tikzpicture}
\; , \qquad
\begin{tikzpicture}[baseline={([yshift=-1mm]current bounding box.center)},scale=0.75]
{
\embot{(0,0)}
}
\end{tikzpicture} 
\; , \qquad
\begin{tikzpicture}[baseline={([yshift=-1mm]current bounding box.center)},scale=0.75]
{
\babot{(0,0)}
}
\end{tikzpicture}
\; , \qquad
\begin{tikzpicture}[baseline={([yshift=-1mm]current bounding box.center)},scale=0.75]
{
\bbbot{(0,0)}
}
\end{tikzpicture}
\; , \qquad
\begin{tikzpicture}[baseline={([yshift=-1mm]current bounding box.center)},scale=0.75]
{
\bcbot{(0,0)}
}
\end{tikzpicture}
\; , \nonumber \\[3mm]
\begin{tikzpicture}[baseline={([yshift=-1mm]current bounding box.center)},scale=0.75]
{
\bdbot{(0,0)}
}
\end{tikzpicture}
\; , \qquad
\begin{tikzpicture}[baseline={([yshift=-1mm]current bounding box.center)},scale=0.75]
{
\bebot{(0,0)}
}
\end{tikzpicture}
\; , \qquad
\begin{tikzpicture}[baseline={([yshift=-1mm]current bounding box.center)},scale=0.75]
{
\bfbot{(0,0)}
}
\end{tikzpicture}
\; , \qquad
\begin{tikzpicture}[baseline={([yshift=-1mm]current bounding box.center)},scale=0.75]
{
\bgbot{(0,0)}
}
\end{tikzpicture}
\; , \qquad
\begin{tikzpicture}[baseline={([yshift=-1mm]current bounding box.center)},scale=0.75]
{
\bhbot{(0,0)}
}
\end{tikzpicture}\; .
\end{align}

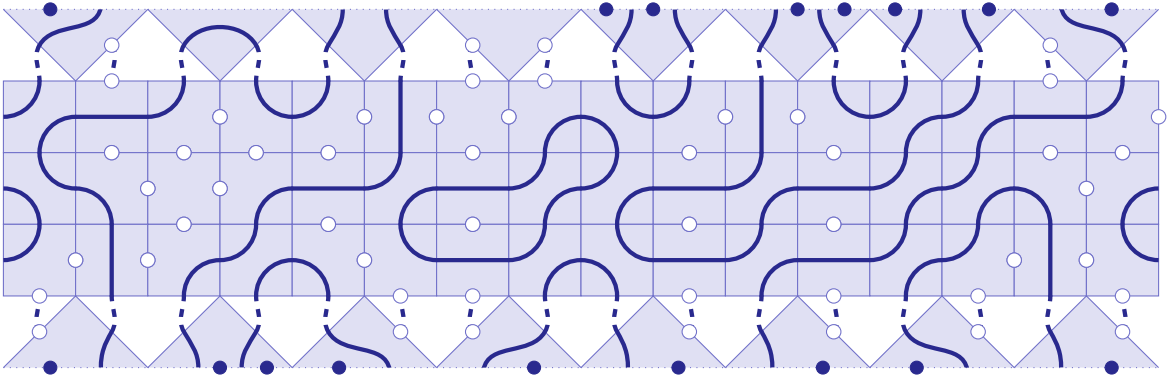
\begin{figure}
\centering
\begin{align*}
\begin{tikzpicture}[baseline={([yshift=-1mm]current bounding box.center)},scale=0.95]
{
\betop{(0,0)}
\aatop{(2,0)}
\batop{(4,0)}
\emtop{(6,0)}
\bbtop{(8,0)}
\bdtop{(10,0)}
\bctop{(12,0)}
\bgtop{(14,0)}
\bgbot{(6,-5)}
\bebot{(4,-5)}
\bfbot{(10,-5)}
\bhbot{(0,-5)}
\bdbot{(2,-5)}
\bfbot{(8,-5)}
\bebot{(12,-5)}
\bfbot{(14,-5)}
\lup{(0,-1)}
\hh{(1,-1)}
\llup{(2,-1)}
\rrup{(3,-1)}
\llup{(4,-1)}
\vv{(5,-1)}
\emp{(6,-1)}
\rrd{(7,-1)}
\rup{(8,-1)}
\llup{(9,-1)}
\vv{(10,-1)}
\rrup{(11,-1)}
\lup{(12,-1)}
\lup{(13,-1)}
\hh{(14,-1)}
\llup{(15,-1)}
\rup{(0,-2)}
\lld{(1,-2)}
\emp{(2,-2)}
\rrd{(3,-2)}
\hh{(4,-2)}
\lup{(5,-2)}
\hh{(6,-2)}
\lup{(7,-2)}
\lup{(8,-2)}
\hh{(9,-2)}
\lup{(10,-2)}
\hh{(11,-2)}
\lup{(12,-2)}
\lup{(13,-2)}
\lld{(14,-2)}
\rrd{(15,-2)}
\llup{(0,-3)}
\vv{(1,-3)}
\rrd{(2,-3)}
\lup{(3,-3)}
\lld{(4,-3)}
\rrup{(5,-3)}
\hh{(6,-3)}
\lup{(7,-3)}
\rup{(8,-3)}
\hh{(9,-3)}
\lup{(10,-3)}
\hh{(11,-3)}
\lup{(12,-3)}
\llup{(13,-3)}
\vv{(14,-3)}
\rrup{(15,-3)}
\begin{pgfonlayer}{main}
\foreach \x in {0.5,2.5,4.5,6.5,8.5,10.5,12.5,14.5}
{\draw[lstring,dashed,dash phase=1pt] (\x,-1) to[out=90,in=240] (\x,-0.5);
\draw[lstring,dashed,dash phase=1pt] (\x,-4) to[out=-90,in=120] (\x,-4.5);
\foreach \y in {-0.5,-1,-2,-3,-4,-4.5}
{\fill[lgh] (\x,\y) circle (0.8pt);}}
\foreach \xx in {1.5,3.5,5.5,7.5,9.5,11.5,13.5,15.5}
{\draw[lstring,dashed,dash phase=1pt] (\xx,-1) to[out=90,in=300] (\xx,-0.5);
\draw[lstring,dashed, dash phase=1pt] (\xx,-4) to[out=-90,in=60] (\xx,-4.5);
\foreach \yy in {-0.5,-1,-2,-3,-4,-4.5}
{\fill[lgh] (\xx,\yy) circle (0.8pt);}}
\foreach \w in {1,2,3,4,5,6,7,8,9,10,11,12,13,14,15}
{\foreach \z in {-1.5,-2.5,-3.5}
{\fill[lgh] (\w,\z) circle (0.8 pt);}}
\end{pgfonlayer}
}
\end{tikzpicture}
\end{align*}
\vspace*{-5mm}

\caption{Dilute lattice that can be described using the dilute ghost algebra $\dGh_3$.}\label{fig:dilutelattice}
\end{figure}

\subsubsection{Bulk face operator}\label{sss:dilutebulk}
The bulk squares of the dilute ghost lattice model are the same as those in the dilute Temperley-Lieb lattice model, so we can use an existing face operator from that model. We use a face operator due to Nienhuis \cite{Nienhuis} and Izergin and Korepin \cite{IzerginKorepin}, given by
\begin{align}
\begin{tikzpicture}[baseline={([yshift=-1mm]current bounding box.center)},scale=0.7]
{
\bulk{$u$}{(0,0)}
}
\end{tikzpicture}
\ &=\  
w_1(u)
\begin{tikzpicture}[baseline={([yshift=-1mm]current bounding box.center)},scale=0.7]
{
\filldraw [lsq] (0,0)--(1,1)--(2,0)--(1,-1)--cycle;
\ld{(0.5,-0.5)}
\ld{(1.5,-0.5)}
\ld{(0.5,0.5)}
\ld{(1.5,0.5)}
}
\end{tikzpicture}
+ w_2(u)\left(\begin{tikzpicture}[baseline={([yshift=-1mm]current bounding box.center)},scale=0.7]
{
\filldraw[lsq] (0,0)--(1,1)--(2,0)--(1,-1)--cycle;
\draw[lstring] let \n2 = {sin(45)} in (0.5,-0.5) arc (135:45:\n2);
\ld{(0.5,0.5)}
\ld{(1.5,0.5)}
}
\end{tikzpicture}
+
\begin{tikzpicture}[baseline={([yshift=-1mm]current bounding box.center)},scale=0.7]
{
\filldraw[lsq] (0,0)--(1,1)--(2,0)--(1,-1)--cycle;
\draw[lstring] let \n1 = {sin(45)} in (0.5,0.5) arc (225:315:\n1);
\ld{(0.5,-0.5)}
\ld{(1.5,-0.5)}
}
\end{tikzpicture}\right)
+ w_3(u)\left(\begin{tikzpicture}[baseline={([yshift=-1mm]current bounding box.center)},scale=0.7]
{
\filldraw[lsq] (0,0)--(1,1)--(2,0)--(1,-1)--cycle;
\draw[lstring] let \n2 = {sin(45)} in (1.5,0.5) arc (135:225:\n2);
\ld{(0.5,-0.5)}
\ld{(0.5,0.5)}
}
\end{tikzpicture}
+
\begin{tikzpicture}[baseline={([yshift=-1mm]current bounding box.center)},scale=0.7]
{
\filldraw[lsq] (0,0)--(1,1)--(2,0)--(1,-1)--cycle;
\draw[lstring] let \n1 = {sin(45)} in (0.5,0.5) arc (45:-45:\n1);
\ld{(1.5,-0.5)}
\ld{(1.5,0.5)}
}
\end{tikzpicture}\right) \nonumber
\\[3mm]
&\quad + w_4(u)\left(\begin{tikzpicture}[baseline={([yshift=-1mm]current bounding box.center)},scale=0.7]
{
\filldraw[lsq] (0,0)--(1,1)--(2,0)--(1,-1)--cycle;
\draw[lstring] (0.5,-0.5)--(1.5,0.5);
\ld{(1.5,-0.5)}
\ld{(0.5,0.5)}
}
\end{tikzpicture}
+
\begin{tikzpicture}[baseline={([yshift=-1mm]current bounding box.center)},scale=0.7]
{
\filldraw[lsq] (0,0)--(1,1)--(2,0)--(1,-1)--cycle;
\draw[lstring] (0.5,0.5)--(1.5,-0.5);
\ld{(0.5,-0.5)}
\ld{(1.5,0.5)}
}
\end{tikzpicture}\right)
+ w_5(u)\;
\begin{tikzpicture}[baseline={([yshift=-1mm]current bounding box.center)},scale=0.7]
{
\filldraw[lsq] (0,0)--(1,1)--(2,0)--(1,-1)--cycle;
\draw[lstring] let \n1 = {sin(45)} in (0.5,0.5) arc (225:315:\n1);
\draw[lstring] let \n2 = {sin(45)} in (0.5,-0.5) arc (135:45:\n2);
}
\end{tikzpicture}
+w_6(u)\;
\begin{tikzpicture}[baseline={([yshift=-1mm]current bounding box.center)},scale=0.7]
{
\filldraw[lsq] (0,0)--(1,1)--(2,0)--(1,-1)--cycle;
\draw[lstring] let \n1 = {sin(45)} in (0.5,0.5) arc (45:-45:\n1);
\draw[lstring] let \n2 = {sin(45)} in (1.5,0.5) arc (135:225:\n2);
}
\end{tikzpicture}\; , \label{eq:dilutesq}
\end{align}
where
\begin{align}
w_1(u) &= \sin(2\phi) \sin(3\phi) + \sin (u) \sin(3\phi-u),\\
w_2(u) &= \sin(2\phi) \sin(3\phi-u),\\
w_3(u) &= \sin(2\phi) \sin(u),\\
w_4(u) &= \sin(u) \sin(3\phi-u),\\
w_5(u) &= \sin(2\phi-u)\sin(3\phi-u),\\
w_6(u) &= -\sin (u) \sin(\phi-u),
\end{align}
and $\beta = -2\cos(4\phi)$.

These bulk face operators satisfy the YBE
\begin{align}
\begin{tikzpicture}[baseline={([yshift=-1mm]current bounding box.center)},scale=0.75]
{
\draw[lstring,dashed] (1.5,0.5)--(2,0.5);
\draw[lstring,dashed] (1.5,-0.5)--(2,-0.5);
\filldraw [lsq] (0,0)--(1,1)--(2,0)--(1,-1)--cycle;
\draw (1,0) node[anchor=center]{\footnotesize $u-v$};
\filldraw[richblue!70] (0.2,0.2)--(0.2,-0.2)--(0,0)--cycle;
\filldraw[lsq] (2,1) rectangle (3,0);
\fill[richblue!70] (2,0.3)--(2.3,0)--(2,0)--cycle;
\draw (2.5,0.5) node[anchor=center]{$v$};
\filldraw[lsq] (2,0) rectangle (3,-1);
\fill[richblue!70] (2,-0.7)--(2.3,-1)--(2,-1)--cycle;
\draw (2.5,-0.5) node[anchor=center]{$u$};
}
\end{tikzpicture}
\ =\ 
\begin{tikzpicture}[baseline={([yshift=-1mm]current bounding box.center)},scale=0.75]
{
\draw[lstring,dashed] (0,0.5)--(0.5,0.5);
\draw[lstring,dashed] (0,-0.5)--(0.5,-0.5);
\filldraw [lsq] (0,0)--(1,1)--(2,0)--(1,-1)--cycle;
\draw (1,0) node[anchor=center]{\footnotesize $u-v$};
\filldraw[richblue!70] (0.2,0.2)--(0.2,-0.2)--(0,0)--cycle;
\filldraw[lsq] (-1,0) rectangle (0,1);
\fill[richblue!70] (-1,0.3)--(-0.7,0)--(-1,0)--cycle;
\draw (-0.5,0.5) node[anchor=center]{$u$};
\filldraw[lsq] (-1,-1) rectangle (0,0);
\fill[richblue!70] (-1,-0.7)--(-0.7,-1)--(-1,-1)--cycle;
\draw (-0.5,-0.5) node[anchor=center]{$v$};
}
\end{tikzpicture} \; ,\label{eq:diluteYBE}
\end{align}
where we note that this has dashed strings instead of the solid strings used for the ghost algebra lattice model.

We also have crossing symmetry
\begin{align}
\begin{tikzpicture}[baseline={([yshift=-1mm]current bounding box.center)},scale=0.75]
{
\bulk{$u$}{(0,0)}
}
\end{tikzpicture}
\; = \;
\begin{tikzpicture}[baseline={([yshift=-1mm]current bounding box.center)},scale=0.75,rotate=90]
{
\bulk{$3\phi - u$}{(0,0)}
}
\end{tikzpicture}
\; = \;
\begin{tikzpicture}[baseline={([yshift=-1mm]current bounding box.center)},scale=0.75, rotate=180]
{
\bulk{$u$}{(0,0)}
}
\end{tikzpicture}
\; = \;
\begin{tikzpicture}[baseline={([yshift=-1mm]current bounding box.center)},scale=0.75,rotate=270]
{
\bulk{$3\phi-u$}{(0,0)}
}
\end{tikzpicture}\; , \label{eq:dilutecrossing}
\end{align}
and the local inversion relation
\begin{align}
\begin{tikzpicture}[baseline={([yshift=-1mm]current bounding box.center)},scale=0.75]
{
\bulk{$u$}{(0,0)}
\bulk{$-u$}{(2,0)}
\draw[lstring,dashed,dash phase = 1.5 pt] (1.5,0.5) to[out=45,in=135] (2.5,0.5);
\draw[lstring,dashed, dash phase = 1.5pt] (1.5,-0.5) to[out=-45,in=-135] (2.5,-0.5);
}
\end{tikzpicture}
\; =\,
f(u)\;
\begin{tikzpicture}[baseline={([yshift=-1mm]current bounding box.center)},scale=0.75]
{
\filldraw[lsq] (0,0)--(1,1)--(2,0)--(1,-1)--cycle;
\draw[lstring,dashed,dash phase=1.5pt] (0.5,0.5) to[out=-45,in=-135] (1.5,0.5);
\draw[lstring,dashed,dash phase = 1.5pt] (0.5,-0.5) to[out=45,in=135] (1.5,-0.5);
}
\end{tikzpicture}
\; , \qquad f \neq 0. \label{eq:diluteLIR}
\end{align}

\subsubsection{Boundary face operators}\label{sss:dilutebdy}
In the dilute ghost lattice loop model, the BYBE for the top boundary is
\begin{align}
\begin{tikzpicture}[baseline={([yshift=-1mm]current bounding box.center)},scale=0.7]
{
\bulk{$u-v$}{(0,0)}
\utri{$u$}{(1,1)}
\bulk{$u+v$}{(2,0)}
\utri{$v$}{(3,1)}
\draw[lstring,dashed,dash phase = -0.5pt] let \n1 = {sin(45)} in (1.5,-0.5) arc (225:315:\n1);
}
\end{tikzpicture}
\ = \ 
\begin{tikzpicture}[baseline={([yshift=-1mm]current bounding box.center)},scale=0.7]
{
\utri{$v$}{(0,0)}
\bulk{$u+v$}{(1,-1)}
\utri{$u$}{(2,0)}
\bulk{$u-v$}{(3,-1)}
\draw[lstring,dashed, dash phase =-0.5pt] let \n1 = {sin(45)} in (2.5,-1.5) arc (225:315:\n1);
}
\end{tikzpicture}\; , \label{eq:diluteBYBE}
\end{align}
and the BYBE for the bottom boundary is
\begin{align}
\begin{tikzpicture}[baseline={([yshift=-1mm]current bounding box.center)},xscale=0.7,yscale=-0.7]
{
\bulk{$u-v$}{(0,0)}
\utri{$u$}{(1,1)}
\bulk{$u+v$}{(2,0)}
\utri{$v$}{(3,1)}
\draw[lstring,dashed,dash phase = -0.5pt] let \n1 = {sin(45)} in (1.5,-0.5) arc (225:315:\n1);
}
\end{tikzpicture}
\ = \ 
\begin{tikzpicture}[baseline={([yshift=-1mm]current bounding box.center)},xscale=0.7,yscale=-0.7]
{
\utri{$v$}{(0,0)}
\bulk{$u+v$}{(1,-1)}
\utri{$u$}{(2,0)}
\bulk{$u-v$}{(3,-1)}
\draw[lstring,dashed, dash phase =-0.5pt] let \n1 = {sin(45)} in (2.5,-1.5) arc (225:315:\n1);
}
\end{tikzpicture}\; . \label{eq:diluteBYBEbot}
\end{align}
These are the same as the BYBEs \eqref{eq:denseBYBE} and \eqref{eq:denseBYBEbot} for the ghost algebra, except that the squares are connected by dashed (identity) strings, instead of solid strings. As in the fully-packed ghost model, we first seek solutions to \eqref{eq:diluteBYBE} for the top boundary face operator. By reflecting these solutions about a horizontal line and swapping parameters $\alpha_i \leftrightarrow \delta_i$, we can find the solutions to \eqref{eq:diluteBYBEbot} for the bottom boundary face operator. The BYBEs are again homogeneous, so our solutions for the boundary face operator with parameter $u$ will have overall scaling factors that are functions of $u$.

Similarly, the top boundary BYBE is also \eqref{eq:diluteBYBE} in the one-boundary dilute ghost lattice model. The top boundary BYBE is actually contained in $\dGho_2$, so we will solve it within this algebra, using diagrams without a bottom boundary. Two solutions for the dilute bottom boundary BYBE without boundary connections are given by Yung and Batchelor in \cite{BatchelorYung}, so these can be used with the top boundary face operators found in Theorem \ref{thm:diluteBYBE} to construct commuting transfer tangles for the one-boundary dilute ghost lattice model, and thus this model is also Yang-Baxter integrable. The construction is as in Section \ref{sss:dilutetransfer}, with the appropriate bottom boundary face operators.

Let
\begin{align}
\begin{tikzpicture}[baseline={([yshift=-1mm]current bounding box.center)},scale=0.7]
{
\utri{$u$}{(0,0)}
}
\end{tikzpicture}
\ &=\  
a_1(u)
\begin{tikzpicture}[baseline={([yshift=-1mm]current bounding box.center)},scale=0.7]
{
\aatop{(0,0)}
}
\end{tikzpicture}
+a_2(u)
\begin{tikzpicture}[baseline={([yshift=-1mm]current bounding box.center)},scale=0.7]
{
\emtop{(0,0)}
}
\end{tikzpicture} \nonumber
\\[2mm]
&\quad +b_1(u)
\begin{tikzpicture}[baseline={([yshift=-1mm]current bounding box.center)},scale=0.7]
{
\batop{(0,0)}
}
\end{tikzpicture}
+b_2(u)
\begin{tikzpicture}[baseline={([yshift=-1mm]current bounding box.center)},scale=0.7]
{
\bbtop{(0,0)}
}
\end{tikzpicture}
+b_3(u)
\begin{tikzpicture}[baseline={([yshift=-1mm]current bounding box.center)},scale=0.7]
{
\bctop{(0,0)}
}
\end{tikzpicture}
+b_4(u)
\begin{tikzpicture}[baseline={([yshift=-1mm]current bounding box.center)},scale=0.7]
{
\bdtop{(0,0)}
}
\end{tikzpicture} \nonumber
\\[2mm]
&\quad +b_5(u)
\begin{tikzpicture}[baseline={([yshift=-1mm]current bounding box.center)},scale=0.7]
{
\betop{(0,0)}
}
\end{tikzpicture}
+b_6(u)
\begin{tikzpicture}[baseline={([yshift=-1mm]current bounding box.center)},scale=0.7]
{
\bftop{(0,0)}
}
\end{tikzpicture}
+b_7(u)
\begin{tikzpicture}[baseline={([yshift=-1mm]current bounding box.center)},scale=0.7]
{
\bgtop{(0,0)}
}
\end{tikzpicture}
+b_8(u)
\begin{tikzpicture}[baseline={([yshift=-1mm]current bounding box.center)},scale=0.7]
{
\bhtop{(0,0)}
}
\end{tikzpicture}\; , \label{eq:dilutetri}
\end{align}
and, as we did for $\Gh_n$, let 
\begin{align}
A_1 &= \alpha_1^2 + \alpha_2^2 - 2\alpha_3^2, &A_2 &= \alpha_1\alpha_2 - \alpha_3^2.
\end{align}
Define
\begin{align}
P(\sigma,\tau; \theta) := 2\sigma \tau \cos(\theta) + \left(\sigma^2 - \tau^2\right) \sin(\theta). \label{eq:Pdef}
\end{align}

\pagebreak
\begin{theorem}\label{thm:diluteBYBE}
Let the face operator for $\dGh_n$ be given by \eqref{eq:dilutesq}. Then the BYBE \eqref{eq:diluteBYBE} admits five solutions for the boundary face operator \eqref{eq:dilutetri}. \hypertarget{solI}{Solution I} is
\begin{align}
a_1(u)&= \frac{-b(u)}{2\sin(2u)\sin(2\phi)}\bigg(\! \sin\! \left(u-\frac{\phi}{2}\right) \!(\cos(2(u-\phi))-\cos(\phi))\! \left(\mu \nu (\alpha_1+\alpha_2)+\left(\mu^2+\nu^2\right)\alpha_3\right)\nonumber \\
&\hspace{37mm}+4\cos(\phi)\sin\! \left(u-\frac{3 \phi}{2}\right)\!\!\bigg),\\[1mm]
a_2(u)&= \frac{b(u)}{2\sin(2u)\sin(2\phi)}  \bigg(\! \sin\! \left(u+\frac{\phi}{2}\right) \!(\cos(2(u-\phi))-\cos(\phi))\! \left(\mu \nu (\alpha_1+\alpha_2)+\left(\mu^2+\nu^2\right)\alpha_3\right)\nonumber \\
&\hspace{37mm}+4\cos(\phi)\sin\! \left(u+\frac{3 \phi}{2}\right)\!\!\bigg), \\[1mm]
b_1(u)&= b_3(u) = \mu\nu\,  b(u)\sin\!\left(u-\frac{\phi}{2}\right), \\
b_2(u)&= \nu^2\,  b(u)\sin\!\left(u-\frac{\phi}{2}\right), \\
b_4(u)&= \mu^2\,  b(u)\sin\!\left(u-\frac{\phi}{2}\right), \\
b_5(u)&= b_8(u) = \nu\, b(u) , \\
b_6(u)&= b_7(u) = \mu\, b(u) ,
\end{align}
where $\mu, \nu \in \C$, and $b: \C \to \C$. \hypertarget{solII}{Solution II} is
\begin{align}
a_1(u)&= \frac{-b(u)}{2\sin(2u)\sin(2\phi)}\bigg(\! \cos\!\left(u-\frac{\phi}{2}\right) \!(\cos(2(u-\phi)) + \cos(\phi)) \!\left( \mu \nu (\alpha_1+\alpha_2)+\left(\mu^2+\nu^2\right)\alpha_3\right) \nonumber \\
&\hspace{37mm}- 4\cos(\phi) \cos\!\left(u-\frac{3\phi}{2}\right)\!\!\bigg),\\[1mm]
a_2(u)&= \frac{-b(u)}{2\sin(2u)\sin(2\phi)}  \bigg(\! \cos\!\left(u+\frac{\phi}{2}\right)\! (\cos(2(u-\phi)) + \cos(\phi)) \!\left(\mu \nu (\alpha_1+\alpha_2)+\left(\mu^2+\nu^2\right)\alpha_3\right)\nonumber \\
&\hspace{37mm}- 4\cos(\phi) \cos\!\left(u+\frac{3\phi}{2}\right)\!\!\bigg) , \\[1mm]
b_1(u)&= b_3(u) = \mu\nu\,  b(u)\cos\!\left(u-\frac{\phi}{2}\right), \\
b_2(u)&= \nu^2\,  b(u)\cos\!\left(u-\frac{\phi}{2}\right), \\
b_4(u)&= \mu^2\,  b(u)\cos\!\left(\!u-\frac{\phi}{2}\right), \\
b_5(u)&= -b_8(u) = \nu\, b(u) , \\
b_6(u)&= -b_7(u) = \mu\, b(u),
\end{align}
where $\mu, \nu \in \C$, and $b: \C \to \C$. \hypertarget{solIII}{Solution III} is
\begin{align}
a_1(u) &= \frac{b(u)}{\sin(2u)} A_1 \big(\sigma\cos(u)+\tau\sin(u)\big)\big(\sigma\cos(u-\phi)+\tau\sin(u-\phi)\big),\\
a_2(u) &= \frac{b(u)}{\sin(2u)} A_1 \big(\sigma\cos(u)-\tau\sin(u)\big) \big(\sigma\cos(u-\phi)+\tau\sin(u-\phi)\big),\\
b_1(u) &= \alpha_1\, b(u)\, P(\sigma,\tau;-3\phi), \\
b_2(u) &= b_4(u) =  -\alpha_3\, b(u)\, P(\sigma,\tau; -3\phi), \\
b_3(u) &= \alpha_2\, b(u)\, P(\sigma,\tau; -3\phi), \\
b_5(u) &= b_6(u) = b_7(u) = b_8(u) = 0  ,
\end{align}
where $b:\C\to\C$, and $\sigma,\tau \in \C$ satisfy
\begin{align}
0 &= A_1\, P(\sigma,\tau;\phi) + A_2\, P(\sigma,\tau; -3\phi).
\end{align}
\hypertarget{solIV}{Solution IV} is 
\begin{align}
a_1(u) &= \frac{b(u)}{\sin(2u)}(\alpha_1-\alpha_2)\left(\mu \nu(\alpha_1+\alpha_2)+\left(\mu ^2+\nu^2\right)\alpha_3\right)\left(\sigma \cos(u) + \tau\sin(u)\right) \nonumber \\
&\hspace{20mm} \times \left(\sigma \cos(u-\phi)+\tau\sin(u-\phi)\right), \\
a_2(u) &=\frac{b(u)}{\sin(2u)}(\alpha_1-\alpha_2)\left(\mu \nu(\alpha_1+\alpha_2)+\left(\mu ^2+\nu^2\right)\alpha_3\right)\left(\sigma \cos(u) - \tau\sin(u)\right)\nonumber \\
&\hspace{20mm} \times \left(\sigma \cos(u-\phi)+\tau\sin(u-\phi)\right), \\
b_1(u) &= \mu\, b(u) \Big((\nu\alpha_2 + \mu\alpha_3)P(\sigma,\tau;\phi) + (\nu\alpha_1 + \mu \alpha_3)P(\sigma,\tau;-3\phi) \Big), \\
b_2(u) &= \nu\, b(u) \Big((\nu\alpha_2 + \mu \alpha_3)P(\sigma,\tau;\phi) + (\nu\alpha_1 + \mu \alpha_3)P(\sigma,\tau;-3\phi) \Big), \\
b_3(u) &= -\nu\, b(u) \Big((\mu \alpha_1 + \nu\alpha_3) P(\sigma,\tau; \phi) + (\mu \alpha_2+\nu\alpha_3)P(\sigma,\tau; -3\phi)\Big), \\
b_4(u) &= -\mu\, b(u) \Big((\mu \alpha_1 + \nu\alpha_3)P(\sigma,\tau; \phi)+ (\mu \alpha_2+\nu\alpha_3) P(\sigma,\tau; -3\phi)\Big) , \\
b_5(u) &= b_6(u) = b_7(u) = b_8(u) = 0,
\end{align}
where $\mu,\nu, \sigma, \tau \in \C$ and $b:\C\to\C$. \hypertarget{solV}{Solution V} is
\begin{align}
a_1(u) &= \frac{b(u)}{4\sin(2u)\cos^2(2\phi)} (\alpha_1-\alpha_2) (\sigma \cos(u)+\tau \sin(u)) (\sigma \cos(u-\phi)+\tau \sin(u-\phi)) \nonumber \\
&\hspace{10mm} \times \big((c_2\alpha_1 + c_3\alpha_3)\, P(\sigma,\tau;-3\phi) +(c_2 \alpha_2 + c_3\alpha_3)\, P(\sigma,\tau;\phi) \big), \\
a_2(u)&= \frac{b(u)}{4\sin(2u)\cos^2(2\phi)} (\alpha_1-\alpha_2) (\sigma \cos(u)-\tau \sin(u)) (\sigma \cos(u-\phi)+\tau \sin(u-\phi)) \nonumber \\
&\hspace{10mm} \times \big((c_2\alpha_1 + c_3\alpha_3)\, P(\sigma,\tau;-3\phi) + (c_2\alpha_2 + c_3\alpha_3)\, P(\sigma,\tau;\phi) \big), \\
b_1(u)&= \frac{c_2\, b(u)}{2\cos(2\phi)} P(\sigma,\tau; -\phi) \Big(\alpha_1\, P(\sigma,\tau;-3\phi)+\alpha_2\, P(\sigma,\tau;\phi) \Big), \\
b_2(u)&= -c_2\, \alpha_3\, b(u) \, P(\sigma,\tau; -\phi)^2,\\
b_3(u)&= -c_3\, \alpha_3\, b(u)\,  P(\sigma,\tau; -\phi)^2, \\
b_4(u)&= \frac{c_3\, b(u)}{2\cos(2\phi)} P(\sigma,\tau;-\phi) \Big(\alpha_1\, P(\sigma,\tau;-3\phi) +\alpha_2\, P(\sigma,\tau;\phi)\Big), \\
b_5(u) &= b_6(u) = b_7(u) = b_8(u) = 0,
\end{align}
where $b:\C\to\C$, $c_2,c_3 \in \C$ and $\sigma,\tau \in \C$ satisfy
\begin{align}
0&= A_1\, P(\sigma,\tau;\phi)\, P(\sigma,\tau;-3\phi) + A_2\, \Big(P(\sigma,\tau;\phi)^2 + P(\sigma,\tau;-3\phi)^2\Big).
\end{align}
\end{theorem}

\begin{proof}
Similar to the proof of Theorem \ref{thm:denseBYBE}, we consider the expression
\begin{align}
\begin{tikzpicture}[baseline={([yshift=-1mm]current bounding box.center)},scale=0.7]
{
\utri{$v$}{(0,0)}
\bulk{$u+v$}{(1,-1)}
\utri{$u$}{(2,0)}
\bulk{$u-v$}{(3,-1)}
\draw[lstring,dashed,dash phase = -0.5pt] let \n1 = {sin(45)} in (2.5,-1.5) arc (225:315:\n1);
}
\end{tikzpicture}
\ - \ \begin{tikzpicture}[baseline={([yshift=-1mm]current bounding box.center)},scale=0.7]
{
\bulk{$u-v$}{(0,0)}
\utri{$u$}{(1,1)}
\bulk{$u+v$}{(2,0)}
\utri{$v$}{(3,1)}
\draw[lstring,dashed,dash phase = -0.5pt] let \n1 = {sin(45)} in (1.5,-0.5) arc (225:315:\n1);
}
\end{tikzpicture}  \label{eq:diluteBYBEequiv}
\end{align}
as a linear combination of the basis diagrams, and note that the BYBE \eqref{eq:diluteBYBE} holds if and only if the coefficient of each basis diagram is zero. The exhaustive search for solutions to the BYBE is more complicated in the dilute ghost algebra than the ghost algebra, so we have separated it into cases. This is summarised in Figure \ref{fig:flowchart}. The configurations of bulk squares and boundary triangles corresponding to each diagram have been omitted from this proof, but are listed in Appendix \ref{app:configs}. Trigonometric simplifications are made frequently and without comment.

\begin{figure}
\hypertarget{flowchart}{}

\vspace{-4mm}
\begin{align*}
\begingroup
\hypersetup{hidelinks}
\begin{tikzpicture}
{
\node[minimum size=0,inner sep=0] (0) at (4,1.5) {};
\node (1) [draw, rectangle] at (0,0.3) {\hyperlink{c1}{$c_5$, $c_6$ not both zero}};
\node (11) [draw, rectangle, fill=red!20] at (-1,-0.7) {\hyperlink{c11}{$\varepsilon=+1$}};
\node (12) [draw, rectangle,fill=blue!80!cyan!20] at (1,-0.7) {\hyperlink{c12}{$\varepsilon=-1$}};
\node (2) [draw, rectangle] at (7,0.3) {\hyperlink{c2}{$c_5=c_6 = 0$}};
\node (21) [draw, rectangle] at (4,-2) {\hyperlink{c21}{$c_1$, $c_2$, $c_3$, $c_4$ not all zero}};
\node (211) [draw, rectangle] at (1.5,-5) {\hyperlink{c211}{$P(\sigma,\tau;-3\phi) \neq 0$}};
\node (2111) [draw, rectangle, fill=green!20] at (-0.5,-8) {\hyperlink{c2111}{$c_1c_3 \neq c_2c_4$}};
\node (2112) [draw, rectangle] at (5.15,-8) {\hyperlink{c2112}{$c_1c_3 = c_2c_4\vphantom{\neq}$}};
\node (21121) [draw, rectangle] at (1.75,-10.5) {\hyperlink{c21121}{$P(\sigma,\tau;-\phi) \neq 0$}};
\node (211211) [draw, rectangle,fill=yellow!30] at (0,-11.8) {\hyperlink{c211211}{$c_3$ coefficient $\neq 0$}};
\node (211212) [draw, rectangle] at (3.5,-11.8) {\hyperlink{c211212}{$c_3$ coefficient $=0$}};
\node (2112121) [draw, rectangle,fill=yellow!30] at (2.2,-12.9) {\hyperlink{c2112121}{$Q(\sigma, \tau) \neq 0$}};
\node (2112122) [draw, rectangle, fill=gray!25] at (4.8,-12.9) {\hyperlink{c2112122}{$Q(\sigma,\tau)=0 \vphantom{\neq}$}};
\node (21122) [draw, rectangle] at (9.95,-10.5) {\hyperlink{c21122}{$P(\sigma,\tau;-\phi) = 0 \vphantom{\neq}$}};
\node (211221) [draw, rectangle, fill=yellow!30] at (8.4,-11.6) {\hyperlink{c211221}{$\sigma = -\tau \tan\!\left(\frac{\phi}{2}\right)$}};
\node (211222) [draw, rectangle,fill=yellow!30] at (11.5,-11.6) {\hyperlink{c211222}{$\sigma = \tau \cot\!\left(\frac{\phi}{2}\right)$}};
\node (212) [draw, rectangle] at (8.35,-5) {\hyperlink{c212}{$P(\sigma,\tau;-3\phi) = 0 \vphantom{\neq}$}};
\node (2121) [draw, rectangle, fill=yellow!30] at (6.7,-6.1) {\hyperlink{c2121}{$\sigma = -\tau \tan\!\left(\frac{3\phi}{2}\right)$}};
\node (2122) [draw, rectangle, fill=yellow!30] at (10,-6.1) {\hyperlink{c2122}{$\sigma = \tau \cot\!\left(\frac{3\phi}{2}\right)$}};
\node (22) [draw, rectangle] at (11.05,-2) {\hyperlink{c22}{$c_1=c_2=c_3=c_4=0$}};
\node (221) [draw,rectangle, fill=red!20] at (9.4,-3) {\hyperlink{c221}{$\sigma = -\tau \tan\!\left(\frac{3\phi}{2}\right)$}};
\node (222) [draw, rectangle, fill=blue!80!cyan!20] at (12.7,-3) {\hyperlink{c222}{$\sigma = \tau \cot\!\left(\frac{3\phi}{2}\right)$}};
\node (I) [fill=red!20, anchor=west] at (11.7,1.5) {\makebox[20mm][l]{Solution I}};
\node (II) [fill=blue!80!cyan!20, anchor=west] at (11.7,1) {\makebox[20mm][l]{Solution II}};
\node (III) [fill=green!20, anchor=west] at (11.7,0.5) {\makebox[20mm][l]{Solution III}};
\node (IV) [fill=yellow!30, anchor=west] at (11.7,0) {\makebox[20mm][l]{Solution IV}};
\node (V) [fill=gray!25, anchor=west] at (11.7,-0.5) {\makebox[20mm][l]{Solution V}};
\draw (0) to (1);
\draw (0) to (2);
\draw (1) to (11);            
\draw (1) to (12);
\draw (2) to (21);
\draw (2) to (22);
\draw (21) to (211);
\draw (21) to (212);
\draw (22) to (221);
\draw (22) to (222);
\draw (211) to (2111);
\draw (211) to (2112);
\draw (212) to (2121);
\draw (212) to (2122);
\draw (2112) to (21121);
\draw (2112) to (21122);
\draw (21121) to (211211);
\draw (21121) to (211212);
\draw (21122) to (211221);
\draw (21122) to (211222);
\draw (211212) to (2112121);
\draw (211212) to (2112122);
}
\end{tikzpicture}
\endgroup
\end{align*}
\caption{Flowchart showing the cases considered when solving the BYBE for the dilute ghost algebra. This is a binary tree, with cases numbered by the sequence of left (1) and right (2) steps to reach them from the top. For example, the case ``$c_1c_3=c_2c_4$'' is numbered 2.1.1.2. Each vertex of the tree is also a hyperlink to the corresponding part of the proof. Note that $P(\sigma,\tau; \theta) = 2\sigma \tau \cos(\theta) + \left(\sigma^2 - \tau^2\right) \sin(\theta)$, and $Q(\sigma, \tau)$ is a quartic polynomial in $\sigma$ and $\tau$, defined in \eqref{eq:Qdef}.}\label{fig:flowchart}
\end{figure}
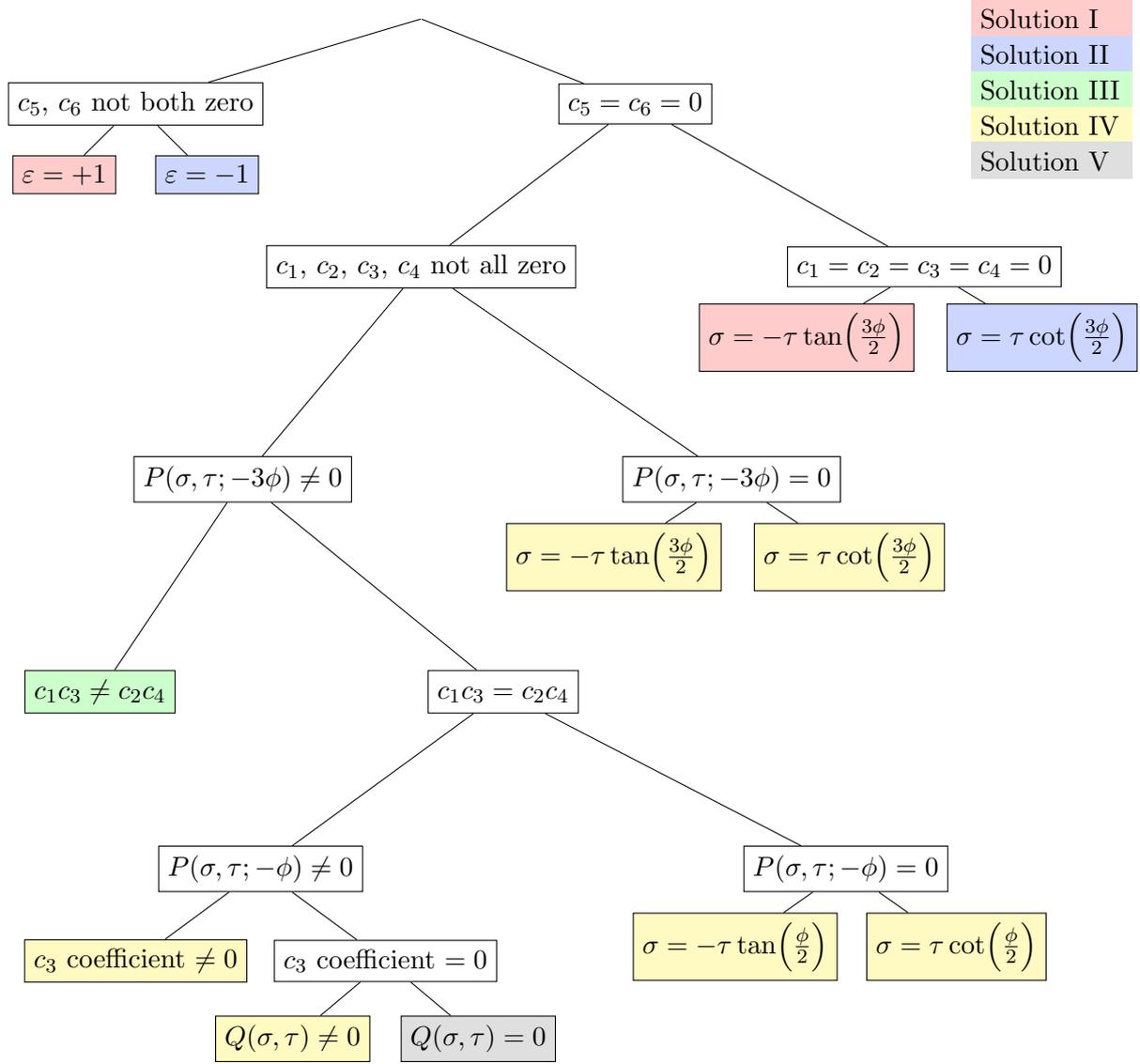

Expanding \eqref{eq:diluteBYBEequiv}, there are 110 diagrams whose coefficients are not trivially zero. We begin by considering the diagrams with four strings connected to the boundary. These have coefficients of the form
\begin{align}
\left(b_i(v)b_j(u) -b_i(u) b_j(v)  \right) \sin (u+v)\sin (u-v-3 \phi) \sin (u-v-2 \phi)\sin (u+v-\phi)
\end{align}
for $i,j \in \{1,2,3,4\}$. Analogously to the $\Gh_n$ case, this implies $b_k(u) = c_k b(u)$ for some $c_k \in \C$ and nonzero $b:\C\to\C$, for each $k \in \{1,2,3,4\}$. Applying this to all coefficients, the next four simplest coefficients are for the diagrams
\begin{align}
\begin{tikzpicture}[baseline={([yshift=-1mm]current bounding box.center)},xscale=0.5,yscale=0.5]
{
\draw[dotted] (0,0.5)--(3,0.5);
\draw [very thick](0,-1.5) -- (0,0.5);
\draw [very thick](3,-1.5)--(3,0.5);
\draw (0,0) to[out=0,in=-90] (1,0.5);
\draw (3,0) to[out=180,in=270] (2,0.5);
\draw (0,-1)--(3,-1);
}
\end{tikzpicture}\; , \qquad
\begin{tikzpicture}[baseline={([yshift=-1mm]current bounding box.center)},xscale=0.5,yscale=0.5]
{
\draw[dotted] (0,0.5)--(3,0.5);
\draw [very thick](0,-1.5) -- (0,0.5);
\draw [very thick](3,-1.5)--(3,0.5);
\draw (0,0) to[out=0,in=-90] (1,0.5);
\draw (3,0) to[out=180,in=270] (2,0.5);
\draw (0,-1)--(3,-1);
\filldraw (0.5,0.5) circle (0.08);
\filldraw (1.5,0.5) circle (0.08);
}
\end{tikzpicture}\; ,\qquad
\begin{tikzpicture}[baseline={([yshift=-1mm]current bounding box.center)},xscale=0.5,yscale=0.5]
{
\draw[dotted] (0,0.5)--(3,0.5);
\draw [very thick](0,-1.5) -- (0,0.5);
\draw [very thick](3,-1.5)--(3,0.5);
\draw (0,0) to[out=0,in=-90] (1,0.5);
\draw (3,0) to[out=180,in=270] (2,0.5);
\draw (0,-1)--(3,-1);
\filldraw (0.5,0.5) circle (0.08);
\filldraw (2.5,0.5) circle (0.08);
}
\end{tikzpicture}\; ,\qquad
\begin{tikzpicture}[baseline={([yshift=-1mm]current bounding box.center)},xscale=0.5,yscale=0.5]
{
\draw[dotted] (0,0.5)--(3,0.5);
\draw [very thick](0,-1.5) -- (0,0.5);
\draw [very thick](3,-1.5)--(3,0.5);
\draw (0,0) to[out=0,in=-90] (1,0.5);
\draw (3,0) to[out=180,in=270] (2,0.5);
\draw (0,-1)--(3,-1);
\filldraw (2.5,0.5) circle (0.08);
\filldraw (1.5,0.5) circle (0.08);
}
\end{tikzpicture}\; ,
\end{align}
and these become
\begin{align}
(b_i(u) b_j(v) -b_i(v) b_j(u)) \sin (2 \phi) \sin (u-v-3 \phi) \sin (u+v-3 \phi) \sin (u-v-2 \phi)
\end{align}
for $i \in \{5,6\}$ and $j\in\{7,8\}$. The coefficient of
\begin{align}
    \begin{tikzpicture}[baseline={([yshift=-1mm]current bounding box.center)},xscale=0.5,yscale=0.5]
{
\draw[dotted] (0,0.5)--(3,0.5);
\draw [very thick](0,-1.5) -- (0,0.5);
\draw [very thick](3,-1.5)--(3,0.5);
\draw (0,-1) to[out=0,in=-90] (1,0.5);
\draw (3,-1) to[out=180,in=270] (2,0.5);
\filldraw (0.5,0.5) circle (0.08);
\filldraw (2.5,0.5) circle (0.08);
\filldraw[fill=white] (0,0) circle (0.12);
\filldraw[fill=white] (3,0) circle (0.12);
}
\end{tikzpicture}
\end{align}
becomes
\begin{align}
&\Big(\Big(b_5(u) b_6(v)-b_7(u) b_8(v)\Big)\sin (u-v)  +\Big(b_6(v) b_8(u)-b_6(u) b_8(v)\Big)  \sin (u+v-\phi)\Big)  \nonumber \\
&\hspace{5mm} \times\sin(u+v)\sin (u-v-3 \phi)\sin(2\phi). \label{eq:M12}
\end{align}
One can check that setting these equal to zero implies that we can write $b_k(u) = c_k g(u)$, for some nonzero $g: \C \to \C$ and $k\in \{5,6,7,8\}$. Applying this, we find that \eqref{eq:M12} is zero if and only if $c_5c_6-c_7c_8=0$. From the coefficients of
\begin{align}
\begin{tikzpicture}[baseline={([yshift=-1mm]current bounding box.center)},xscale=0.5,yscale=0.5]
{
\draw[dotted] (0,0.5)--(3,0.5);
\draw [very thick](0,-1.5) -- (0,0.5);
\draw [very thick](3,-1.5)--(3,0.5);
\draw (0,-1) to[out=0,in=-90] (1,0.5);
\draw (3,-1) to[out=180,in=270] (2,0.5);
\filldraw (0.5,0.5) circle (0.08);
\filldraw (1.5,0.5) circle (0.08);
\filldraw[fill=white] (0,0) circle (0.12);
\filldraw[fill=white] (3,0) circle (0.12);
}
\end{tikzpicture}\; , \qquad
\begin{tikzpicture}[baseline={([yshift=-1mm]current bounding box.center)},xscale=0.5,yscale=0.5]
{
\draw[dotted] (0,0.5)--(3,0.5);
\draw [very thick](0,-1.5) -- (0,0.5);
\draw [very thick](3,-1.5)--(3,0.5);
\draw (0,-1) to[out=0,in=-90] (1,0.5);
\draw (3,-1) to[out=180,in=270] (2,0.5);
\filldraw (1.5,0.5) circle (0.08);
\filldraw (2.5,0.5) circle (0.08);
\filldraw[fill=white] (0,0) circle (0.12);
\filldraw[fill=white] (3,0) circle (0.12);
}
\end{tikzpicture}
\end{align}
we must have $c_5^2 - c_8^2 = 0$ and $c_6^2 -c_7^2 = 0$. It follows that $c_8 = \varepsilon\, c_5$ and $c_7 = \varepsilon\, c_6$, with $\varepsilon = \pm 1$.

Next, from the coefficients of 
\begin{align}
\begin{tikzpicture}[baseline={([yshift=-1mm]current bounding box.center)},xscale=0.5,yscale=0.5]
{
\draw[dotted] (0,0.5)--(3,0.5);
\draw [very thick](0,-1.5) -- (0,0.5);
\draw [very thick](3,-1.5)--(3,0.5);
\draw (0,0) to[out=0,in=-90] (0.75,0.5);
\draw (0,-1) to[out=0,in=-90] (1.5,0.5);
\draw (3,-1) to[out=180,in=270] (2.25,0.5);
\filldraw (1.125,0.5) circle (0.08);
\filldraw (1.875,0.5) circle (0.08);
\filldraw (0.375,0.5) circle (0.08);
\filldraw[fill=white] (3,0) circle (0.12);
}
\end{tikzpicture}\; , \qquad
\begin{tikzpicture}[baseline={([yshift=-1mm]current bounding box.center)},xscale=0.5,yscale=0.5]
{
\draw[dotted] (0,0.5)--(3,0.5);
\draw [very thick](0,-1.5) -- (0,0.5);
\draw [very thick](3,-1.5)--(3,0.5);
\draw (0,0) to[out=0,in=-90] (0.75,0.5);
\draw (0,-1) to[out=0,in=-90] (1.5,0.5);
\draw (3,-1) to[out=180,in=270] (2.25,0.5);
\filldraw (1.125,0.5) circle (0.08);
\filldraw (1.875,0.5) circle (0.08);
\filldraw (2.625,0.5) circle (0.08);
\filldraw[fill=white] (3,0) circle (0.12);
}
\end{tikzpicture}\; ,
\end{align}
we find
\begin{align}
0&= c_2 c_5 \Big(b(u) g(v) \sin (2 v-\phi) - b(v) g(u) (\varepsilon  \sin (u-v)+\sin (u+v-\phi))\!\Big), \label{eq:M26}\\
0&= c_4 c_6 \Big(b(u) g(v) \sin (2 v-\phi) - b(v) g(u) (\varepsilon  \sin (u-v)+\sin (u+v-\phi))\!\Big), \label{eq:M20}
\end{align}
so $c_2c_5 = c_4c_6 = 0$, or the expression in the larger brackets is zero.

Now, suppose $c_2=0$. Then the coefficients of
\begin{align}
\begin{tikzpicture}[baseline={([yshift=-1mm]current bounding box.center)},xscale=0.5,yscale=0.5]
{
\draw[dotted] (0,0.5)--(3,0.5);
\draw [very thick](0,-1.5) -- (0,0.5);
\draw [very thick](3,-1.5)--(3,0.5);
\draw (0,0) to[out=0,in=-90] (0.75,0.5);
\draw (0,-1) to[out=0,in=-90] (1.5,0.5);
\draw (3,-1) to[out=180,in=270] (2.25,0.5);
\filldraw (0.375,0.5) circle (0.08);
\filldraw (1.125,0.5) circle (0.08);
\filldraw (2.625,0.5) circle (0.08);
\filldraw[fill=white] (3,0) circle (0.12);
}
\end{tikzpicture} 
\; , \qquad
\begin{tikzpicture}[baseline={([yshift=-1mm]current bounding box.center)},xscale=0.5,yscale=0.5]
{
\draw[dotted] (0,0.5)--(3,0.5);
\draw [very thick](0,-1.5) -- (0,0.5);
\draw [very thick](3,-1.5)--(3,0.5);
\draw (0,0) to[out=0,in=-90] (1,0.5);
\draw (3,-1) to[out=180,in=270] (2,0.5);
\filldraw (1.5,0.5) circle (0.08);
\filldraw (0.5,0.5) circle (0.08);
\filldraw[fill=white] (0,-1) circle (0.12);
\filldraw[fill=white] (3,0) circle (0.12);
}
\end{tikzpicture}
\end{align}
become
\begin{align}
c_3 c_5 \, b(u)g(v)\sin(u-v)\sin(u+v)\sin(u-v-3\phi)\sin(u+v-3\phi) \label{eq:M16a}
\end{align}
and
\begin{align}
&2\cos(\phi) \sin (u-v) \sin (u-v-3\phi) \sin(u+v-3\phi) \nonumber \\
&\hspace{5mm}\times \Big(c_1 c_3\, \alpha_3\, b(u) b(v) \sin(\phi) + \varepsilon\, c_5^2\,  g(u) g(v) \sin (u+v+\phi)\Big), \label{eq:M25a}
\end{align}
respectively. For both of these to be zero, we need $c_5=0$. Hence $c_2=0$ implies $c_5=0$.

Similarly, using the coefficients of
\begin{align}
\begin{tikzpicture}[baseline={([yshift=-1mm]current bounding box.center)},xscale=0.5,yscale=0.5]
{
\draw[dotted] (0,0.5)--(3,0.5);
\draw [very thick](0,-1.5) -- (0,0.5);
\draw [very thick](3,-1.5)--(3,0.5);
\draw (0,0) to[out=0,in=-90] (0.75,0.5);
\draw (0,-1) to[out=0,in=-90] (1.5,0.5);
\draw (3,-1) to[out=180,in=270] (2.25,0.5);
\filldraw (1.125,0.5) circle (0.08);
\filldraw[fill=white] (3,0) circle (0.12);
}
\end{tikzpicture} 
\; , \qquad
\begin{tikzpicture}[baseline={([yshift=-1mm]current bounding box.center)},xscale=0.5,yscale=0.5]
{
\draw[dotted] (0,0.5)--(3,0.5);
\draw [very thick](0,-1.5) -- (0,0.5);
\draw [very thick](3,-1.5)--(3,0.5);
\draw (0,0) to[out=0,in=-90] (1,0.5);
\draw (3,-1) to[out=180,in=270] (2,0.5);
\filldraw (1.5,0.5) circle (0.08);
\filldraw (2.5,0.5) circle (0.08);
\filldraw[fill=white] (0,-1) circle (0.12);
\filldraw[fill=white] (3,0) circle (0.12);
}
\end{tikzpicture}
\; ,
\end{align}
we find that $c_4=0$ implies $c_6=0$. Hence we have $c_5=c_6=0$, or
\begin{align}
b(u) g(v) \sin (2 v-\phi) - b(v) g(u) \Big(\varepsilon  \sin (u-v)+\sin (u+v-\phi)\Big)=0. \label{eq:brackets}
\end{align}

\hypertarget{c1}{} \phantom{1}

\noindent\textbf{Case 1:} $c_5$, $c_6$ not both zero \hfill \hyperlink{flowchart}{Link to flowchart: $\uparrow$}

\hypertarget{c11}{} \phantom{11} 

\vspace{-3mm}

\noindent\textbf{Case 1.1:} $\varepsilon=1$ \hfill \hyperlink{flowchart}{$\uparrow$}

\noindent Rearranging \eqref{eq:brackets}, we find
\begin{align}
\frac{g(u)}{b(u)} \sin\!\left(u-\frac{\phi}{2}\right) &= \frac{g(v)}{b(v)} \sin\!\left(v-\frac{\phi}{2}\right).
\end{align}
Both sides must equal some constant $z \in \C$, so
\begin{align}
g(u) &= z\, b(u) \csc\!\left(u-\frac{\phi}{2}\right).
\end{align}
Noting that $c_5$, $c_6$ and $z$ only ever appear as $c_5 z$ and $c_6 z$, and that $z$ is nonzero because $g$ is nonzero, we may absorb $z$ into $c_5$ and $c_6$, equivalent to setting $z=1$ above.

Applying this, setting the coefficients of
\begin{align}
\begin{tikzpicture}[baseline={([yshift=-1mm]current bounding box.center)},xscale=0.5,yscale=0.5]
{
\draw[dotted] (0,0.5)--(3,0.5);
\draw [very thick](0,-1.5) -- (0,0.5);
\draw [very thick](3,-1.5)--(3,0.5);
\draw (0,0) to[out=0,in=-90] (0.75,0.5);
\draw (0,-1) to[out=0,in=-90] (1.5,0.5);
\draw (3,-1) to[out=180,in=270] (2.25,0.5);
\filldraw (0.375,0.5) circle (0.08);
\filldraw[fill=white] (3,0) circle (0.12);
}
\end{tikzpicture}
\; , \qquad
\begin{tikzpicture}[baseline={([yshift=-1mm]current bounding box.center)},xscale=0.5,yscale=0.5]
{
\draw[dotted] (0,0.5)--(3,0.5);
\draw [very thick](0,-1.5) -- (0,0.5);
\draw [very thick](3,-1.5)--(3,0.5);
\draw (0,0) to[out=0,in=-90] (0.75,0.5);
\draw (0,-1) to[out=0,in=-90] (1.5,0.5);
\draw (3,-1) to[out=180,in=270] (2.25,0.5);
\filldraw (2.625,0.5) circle (0.08);
\filldraw[fill=white] (3,0) circle (0.12);
}
\end{tikzpicture}
\; , \qquad
\begin{tikzpicture}[baseline={([yshift=-1mm]current bounding box.center)},xscale=0.5,yscale=0.5]
{
\draw[dotted] (0,0.5)--(3,0.5);
\draw [very thick](0,-1.5) -- (0,0.5);
\draw [very thick](3,-1.5)--(3,0.5);
\draw (0,0) to[out=0,in=-90] (0.75,0.5);
\draw (0,-1) to[out=0,in=-90] (1.5,0.5);
\draw (3,-1) to[out=180,in=270] (2.25,0.5);
\filldraw (1.875,0.5) circle (0.08);
\filldraw[fill=white] (3,0) circle (0.12);
}
\end{tikzpicture}
\; , \qquad
\begin{tikzpicture}[baseline={([yshift=-1mm]current bounding box.center)},xscale=0.5,yscale=0.5]
{
\draw[dotted] (0,0.5)--(3,0.5);
\draw [very thick](0,-1.5) -- (0,0.5);
\draw [very thick](3,-1.5)--(3,0.5);
\draw (0,0) to[out=0,in=-90] (0.75,0.5);
\draw (0,-1) to[out=0,in=-90] (1.5,0.5);
\draw (3,-1) to[out=180,in=270] (2.25,0.5);
\filldraw (1.125,0.5) circle (0.08);
\filldraw[fill=white] (3,0) circle (0.12);
}
\end{tikzpicture} 
\end{align}
to zero implies
\begin{align}
0 &= (c_1-c_3)c_5, \label{eq:M27}\\
0 &= (c_1-c_3)c_6, \label{eq:M19}\\
0 &= c_1c_5-c_2c_6, \label{eq:M29}\\
0 &= c_1c_6-c_4c_5. \label{eq:M30}
\end{align}
Since $c_5$ and $c_6$ are not both zero, \eqref{eq:M27} and \eqref{eq:M19} imply $c_1=c_3$. From \eqref{eq:M29}, we find $c_2=c_5\, y_1$ and $c_1=c_6\, y_1$, for some $y_1 \in \C$, and from \eqref{eq:M30}, we find $c_1= c_5\, y_2$ and $c_4=c_6\, y_2$, for some $y_2 \in \C$. Hence $c_1 = c_6\, y_1 = c_5\, y_2$, so $y_1= c_5\, \eta$ and $y_2 = c_6\, \eta$, for some $\eta \in \C$. This gives
\begin{align}
c_1=c_3 = c_5c_6\, \eta, \qquad c_2 = c_5^2\, \eta, \qquad c_4 = c_6^2\, \eta.
\end{align}
If $\eta=0$, then $c_2=c_4=0$. We showed just before \hyperlink{c1}{Case 1} that this implies $c_5=c_6=0$, but this contradicts our Case 1 assumption that $c_5$ and $c_6$ are not both zero. Hence $\eta \neq 0$.

From the coefficients of
\begin{align}
\begin{tikzpicture}[baseline={([yshift=-1mm]current bounding box.center)},xscale=0.5,yscale=0.5]
{
\draw[dotted] (0,0.5)--(3,0.5);
\draw [very thick](0,-1.5) -- (0,0.5);
\draw [very thick](3,-1.5)--(3,0.5);
\draw (0,0) to[out=0,in=-90] (0.75,0.5);
\draw (0,-1)--(3,-1);
\filldraw (0.325,0.5) circle (0.08);
\filldraw[fill=white] (3,0) circle (0.12);
}
\end{tikzpicture}
\; , \qquad \begin{tikzpicture}[baseline={([yshift=-1mm]current bounding box.center)},xscale=0.5,yscale=0.5]
{
\draw[dotted] (0,0.5)--(3,0.5);
\draw [very thick](0,-1.5) -- (0,0.5);
\draw [very thick](3,-1.5)--(3,0.5);
\draw (0,0) to[out=0,in=-90] (0.75,0.5);
\draw (0,-1)--(3,-1);
\filldraw (1.125,0.5) circle (0.08);
\filldraw[fill=white] (3,0) circle (0.12);
}
\end{tikzpicture}
\; , 
\end{align}
we find
\begin{align}
0&=-c_i \sin (u-v-3 \phi) \sin (u+v-3 \phi) \sin(2\phi)\nonumber\\
&\quad\ \  \times \bigg( b(u)b(v)\left(c_5 c_6 (\alpha_1+\alpha_2) + \left(c_5^2+ c_6^2\right)\alpha_3\right)\eta   \left(\cot\!\left(v-\frac{\phi}{2}\right)-\cot\!\left(u-\frac{\phi}{2}\right)\!\right) \nonumber \\
&\hspace{20mm}\times \sin\! \left(u-\frac{5 \phi}{2}\right)  \sin (u+v-2 \phi)\csc(2\phi)\nonumber \\
&\hspace{15mm}+b(u)\csc\!\left(u-\frac{\phi}{2}\right)\big(a_1(v)\sin (u+v-2 \phi)-a_2(v)\sin (u-v-2 \phi)\big)\nonumber \\
&\hspace{15mm}-b(v)\csc\!\left(v-\frac{\phi}{2}\right)\left(a_1(u) \sin (2 (u-\phi)) + a_2(u) \sin(2 \phi)\right) \!\!  \bigg) \label{eq:M102103}
\end{align}
for $i = 5$ and $i=6$, respectively, and from the coefficients of 
\begin{align}
\begin{tikzpicture}[baseline={([yshift=-1mm]current bounding box.center)},xscale=0.5,yscale=0.5]
{
\draw[dotted] (0,0.5)--(3,0.5);
\draw [very thick](0,-1.5) -- (0,0.5);
\draw [very thick](3,-1.5)--(3,0.5);
\draw (0,0) to[out=0,in=-90] (0.75,0.5);
\draw (0,-1)--(3,0);
\filldraw (0.325,0.5) circle (0.08);
\filldraw[fill=white] (3,-1) circle (0.12);
}
\end{tikzpicture}
\; , \qquad \begin{tikzpicture}[baseline={([yshift=-1mm]current bounding box.center)},xscale=0.5,yscale=0.5]
{
\draw[dotted] (0,0.5)--(3,0.5);
\draw [very thick](0,-1.5) -- (0,0.5);
\draw [very thick](3,-1.5)--(3,0.5);
\draw (0,0) to[out=0,in=-90] (0.75,0.5);
\draw (0,-1)--(3,0);
\filldraw (1.125,0.5) circle (0.08);
\filldraw[fill=white] (3,-1) circle (0.12);
}
\end{tikzpicture}
\; ,
\end{align}
we have
\begin{align}
0&= -c_i \sin (u-v-3 \phi) \sin (u+v-3 \phi)  \nonumber \\
&\qquad \times \bigg(b(u)b(v)\left(c_5 c_6 (\alpha_1+\alpha_2) + \left(c_5^2+c_6^2\right)\alpha_3\right) \eta \csc\!\left(u-\frac{\phi}{2}\right)  \sin (u-v) \sin (u+v-2 \phi)     \nonumber \\
&\hspace{15mm}+b(u) \csc\!\left(u-\frac{\phi}{2}\right)a_1(v) \sin (2 v)\sin(2\phi) \nonumber \\
&\hspace{15mm}-b(v) \csc\!\left(v-\frac{\phi}{2}\right) (a_1(u) \sin (u+v)+a_2(u) \sin (u-v))\sin(2\phi)\!\bigg) \label{eq:M107108}
\end{align}
for $i=5$ and $i=6$, respectively. We can rearrange each of these equations to find expressions for
\begin{align}
c_i\, b(u)b(v) \left(c_5c_6\left(\alpha_1+\alpha_2\right) + \left(c_5^2+c_6^2\right)\alpha_3\right) \eta ,
\end{align}
and then equate those expressions. Since $c_5$ and $c_6$ are not both zero, choose $i \in \{5,6\}$ such that $c_i$ is nonzero, and recall that $b$ and $\eta$ are nonzero. We can thus rearrange the resulting equation to find
\begin{align}
&\quad \frac{a_1(u)}{b(u)} \sin\! \left(u+\frac{\phi}{2}\right) \sin\! \left(u-\frac{\phi}{2}\right)+\frac{a_2(u)}{b(u)} \sin^2\!\left(u-\frac{\phi}{2}\right) \nonumber \\
&=\frac{a_1(v)}{b(v)} \sin\! \left(v+\frac{\phi}{2}\right) \sin\! \left(v-\frac{\phi}{2}\right)+\frac{a_2(v)}{b(v)} \sin ^2\!\left(v-\frac{\phi}{2}\right).
\end{align}
The left and right sides are functions of $u$ and $v$, respectively, so we may set both equal to some $\rho \in \C$. Hence we have
\begin{align}
a_1(u) &= \csc\!\left(u+\frac{\phi}{2}\right) \left(\rho\, b(u) \csc\!\left(u-\frac{\phi}{2}\right) - a_2(u)\sin\!\left(u-\frac{\phi}{2}\right)\!\right).
\end{align}
Substituting this back into \eqref{eq:M107108} and rearranging gives
\begin{align}
&\quad \frac{a_2(u)}{b(u)} \sin (2 u) \sin\! \left(u-\frac{\phi}{2}\right) \csc\!\left(u+\frac{\phi}{2}\right)-\rho  \cot\! \left(u+\frac{\phi}{2}\right) \nonumber \\
&\hspace{10mm}-\frac{\eta}{2}   \csc(2 \phi) \cos(2( u- \phi)) \left(c_5 c_6 (\alpha_1+\alpha_2) + \left(c_5^2+c_6^2\right)\alpha_3\right)\nonumber \\
&=\frac{a_2(v)}{b(v)} \sin (2 v) \sin\! \left(v-\frac{\phi}{2}\right) \csc\!\left(v+\frac{\phi}{2}\right)-\rho  \cot\! \left(v+\frac{\phi}{2} \right) \nonumber \\
&\hspace{10mm} -\frac{\eta}{2}  \csc(2 \phi) \cos(2 (v- \phi) ) \left(c_5 c_6 (\alpha_1+\alpha_2) + \left(c_5^2+c_6^2\right)\alpha_3\right).
\end{align}
We can then set both sides equal to some $\kappa \in \C$, and find that
\begin{align}
a_2(u)&=b(u) \csc(2 u) \csc\!\left(u-\frac{\phi}{2}\right) \nonumber \\
&\hspace{5mm}\times \bigg(\frac{\eta}{2}   \csc(2 \phi) \sin\! \left(u+\frac{\phi}{2}\right) \cos(2(u-\phi)) \left(c_5 c_6 (\alpha_1+\alpha_2) + \left(c_5^2+c_6^2\right)\alpha_3\right) \nonumber \\
&\hspace{15mm} +\,\kappa  \sin\! \left(u+\frac{\phi}{2}\right)+\rho  \cos\!\left(u+\frac{\phi}{2}\right)\!\!\bigg).
\end{align}

Setting the coefficients of
\begin{align}
\begin{tikzpicture}[baseline={([yshift=-1mm]current bounding box.center)},xscale=0.5,yscale=0.5]
{
\draw[dotted] (0,0.5)--(3,0.5);
\draw [very thick](0,-1.5) -- (0,0.5);
\draw [very thick](3,-1.5)--(3,0.5);
\draw (0,0) to[out=0,in=-90] (0.75,0.5);
\filldraw (0.325,0.5) circle (0.08);
\filldraw[fill=white] (0,-1) circle (0.12);
\filldraw[fill=white] (3,0) circle (0.12);
\filldraw[fill=white] (3,-1) circle (0.12);
}
\end{tikzpicture}
\; ,  \qquad 
\begin{tikzpicture}[baseline={([yshift=-1mm]current bounding box.center)},xscale=0.5,yscale=0.5]
{
\draw[dotted] (0,0.5)--(3,0.5);
\draw [very thick](0,-1.5) -- (0,0.5);
\draw [very thick](3,-1.5)--(3,0.5);
\draw (0,0) to[out=0,in=-90] (0.75,0.5);
\filldraw (1.125,0.5) circle (0.08);
\filldraw[fill=white] (0,-1) circle (0.12);
\filldraw[fill=white] (3,0) circle (0.12);
\filldraw[fill=white] (3,-1) circle (0.12);
}
\end{tikzpicture}
\; 
\end{align}
to zero, we find
\begin{align}
0 &=  c_i  \Big(\eta  \left(c_5 c_6(\alpha_1+\alpha_2) + \left(c_5^2+c_6^2\right)\alpha_3\right)+4 \kappa  \sin (\phi)-4 \rho  \cos(\phi)\Big)
\end{align}
for $i=5$ and $i=6$, respectively. Since $c_5$ and $c_6$ are not both zero,
\begin{align}
\kappa = \rho\cot\phi-\frac{\eta}{4\sin\phi}  \left(c_5 c_6(\alpha_1+\alpha_2) + \left(c_5^2+c_6^2\right)\alpha_3\right).
\end{align}

Applying this, the coefficients of 
\begin{align}
\begin{tikzpicture}[baseline={([yshift=-1mm]current bounding box.center)},xscale=0.5,yscale=0.5]
{
\draw[dotted] (0,0.5)--(3,0.5);
\draw [very thick](0,-1.5) -- (0,0.5);
\draw [very thick](3,-1.5)--(3,0.5);
\draw (0,0) to[out=0,in=-90] (0.75,0.5);
\draw (0,-1) to[out=0,in=-90] (1.5,0.5);
\draw (3,0) arc (90:270:0.5);
\filldraw (0.325,0.5) circle (0.08);
\filldraw (1.125,0.5) circle (0.08);
}
\end{tikzpicture}
\; ,  \qquad 
\begin{tikzpicture}[baseline={([yshift=-1mm]current bounding box.center)},xscale=0.5,yscale=0.5]
{
\draw[dotted] (0,0.5)--(3,0.5);
\draw [very thick](0,-1.5) -- (0,0.5);
\draw [very thick](3,-1.5)--(3,0.5);
\draw (0,0) to[out=0,in=-90] (0.75,0.5);
\draw (0,-1) to[out=0,in=-90] (1.5,0.5);
\draw (3,0) arc (90:270:0.5);
\filldraw (1.875,0.5) circle (0.08);
\filldraw (1.125,0.5) circle (0.08);
}
\end{tikzpicture}
\; 
\end{align}
become
\begin{align}
2\, c_i^2\, b(u) b(v)\, (\eta  \rho -1)\, \sin (2 \phi) \sin (u-v) \cos\!\left(v+\frac{\phi}{2}\right) \csc\!\left(v-\frac{\phi}{2}\right) \sin (u+v-3 \phi)
\end{align}
for $i=5$ and $i=6$, respectively. Since $c_5$ and $c_6$ are not both zero, setting this to zero gives $\rho = \frac{1}{\eta}$. Applying this, the coefficients of all diagrams in the expansion of \eqref{eq:diluteBYBEequiv} become zero. Rescaling by a factor of $\eta\, \sin\!\left(u-\frac{\phi}{2}\right)$, we observe that $c_5$, $c_6$ and $\eta$ only appear in this solution as $c_5\eta$ and $c_6\eta$, so we will set $\nu = c_5\eta$ and $\mu=c_6\eta$. This gives \hyperlink{solI}{Solution I}, for $b$ nonzero, and $\mu$ and $\nu$ not both zero.

\hypertarget{c12}{} \phantom{12}

\noindent\textbf{Case 1.2:} $\varepsilon=-1$ \hfill \hyperlink{flowchart}{$\uparrow$}

\nopagebreak
\noindent This case is analogous to \hyperlink{c11}{Case 1.1}, where $\varepsilon=1$, and gives \hyperlink{solII}{Solution II}, for $b$ nonzero, and $\mu$ and $\nu$ not both zero. 

\hypertarget{c2}{} \phantom{2}

\noindent\textbf{Case 2:} $c_5=c_6=0$ \hfill \hyperlink{flowchart}{$\uparrow$}

\noindent Recall that $b_k(u) = c_k b(u)$ for $k \in \{1,2,3,4\}$, and that
\begin{align}
b_5(u) &= \varepsilon\, b_8(u) = c_5\, g(u), \\
b_6(u) &= \varepsilon\, b_7(u) = c_6\, g(u), 
\end{align}
where $b$ and $g$ are nonzero, $\varepsilon = \pm 1$ and $c_i \in \C$ for $1 \leq i \leq 8$. With $c_5=c_6=0$, then, $b_5=b_6=b_7=b_8=0$.

Setting the coefficient of
\begin{align}
\begin{tikzpicture}[baseline={([yshift=-1mm]current bounding box.center)},xscale=0.5,yscale=0.5]
{
\draw[dotted] (0,0.5)--(3,0.5);
\draw [very thick](0,-1.5) -- (0,0.5);
\draw [very thick](3,-1.5)--(3,0.5);
\draw (0,0) -- (3,-1);
\filldraw[fill=white] (0,-1) circle (0.12);
\filldraw[fill=white] (3,0) circle (0.12);
}
\end{tikzpicture}
\end{align}
to zero, we find
\begin{align}
\big(a_1(u) +a_2(u)\big)\sin (u)\big(a_1(v)-a_2(v)\big) \cos(v) =  \big(a_1(v)+a_2(v)\big)\sin (v)\big(a_1(u)-a_2(u) \big)\cos(u).
\end{align}
This is equivalent to
\begin{align}
f_1(u)f_2(v) = f_1(v)f_2(u),
\end{align}
with $f_1(x) = (a_1(x)+a_2(x))\sin(x)$ and $f_2(x)=\left(a_1(x)-a_2(x)\right)\cos(x)$. It follows that $f_1$ and $f_2$ are both scalar multiples of some nonzero function $h$. Setting
\begin{align}
f_1(u) &= \sigma\, h(u) \sin(2u), \\
f_2(u) &= \tau\, h(u) \sin(2u),
\end{align}
we find
\begin{align}
a_1(u) &= h(u) \left(\sigma \cos(u) + \tau\sin(u) \right), \\
a_2(u) &= h(u) \left(\sigma \cos(u) - \tau \sin(u) \right).
\end{align}

\hypertarget{c21}{} \phantom{21}

\noindent\textbf{Case 2.1:} $c_1$, $c_2$, $c_3$, $c_4$ not all zero \hfill \hyperlink{flowchart}{$\uparrow$}

\nopagebreak
\noindent Setting the coefficient of 
\begin{align}
\begin{tikzpicture}[baseline={([yshift=-1mm]current bounding box.center)},xscale=0.5,yscale=0.5]
{
\draw[dotted] (0,0.5)--(3,0.5);
\draw [very thick](0,-1.5) -- (0,0.5);
\draw [very thick](3,-1.5)--(3,0.5);
\draw (0,0) arc (90:-90:0.5);
\filldraw[fill=white] (3,-1) circle (0.12);
\filldraw[fill=white] (3,0) circle (0.12);
}
\end{tikzpicture}
\end{align}
equal to zero, we find
\begin{align}
h(u) \sin (2 u)\, P(\sigma,\tau;-3\phi) &= b(u)(c_1\alpha_1 + c_3\alpha_2 + (c_2+c_4)\alpha_3) (\sigma \cos(u-\phi)+\tau \sin (u-\phi)) , \label{eq:3phi21}
\end{align}
where we recall $P(\sigma,\tau;\theta) = 2\sigma\tau \cos(\theta) + \left(\sigma^2-\tau^2\right)\sin(\theta)$. If $P(\sigma,\tau;-3\phi) \neq 0$, this can be rearranged to find $h$ in terms of $b$.

\hypertarget{c211}{} \phantom{211}

\noindent\textbf{Case 2.1.1:} $P(\sigma,\tau;-3\phi) \neq 0$ \hfill \hyperlink{flowchart}{$\uparrow$}

\nopagebreak
\noindent Rearranging \eqref{eq:3phi21} for $h(u)$ yields
\begin{align}
h(u) &= \frac{b(u)(c_1 \alpha_1 + c_3 \alpha_2 + (c_2+c_4)\alpha_3 ) (\sigma \cos(u-\phi)+\tau \sin (u-\phi)) }{\sin (2 u) P(\sigma,\tau;-3\phi)}.
\end{align}
Next, we consider the diagrams
\begin{align}
\begin{tikzpicture}[baseline={([yshift=-1mm]current bounding box.center)},xscale=0.5,yscale=0.5]
{
\draw[dotted] (0,0.5)--(3,0.5);
\draw [very thick](0,-1.5) -- (0,0.5);
\draw [very thick](3,-1.5)--(3,0.5);
\draw (0,0) to[out=0,in=-90] (0.75,0.5);
\draw (0,-1) to[out=0,in=-90] (1.5,0.5);
\filldraw[fill=white] (3,-1) circle (0.12);
\filldraw[fill=white] (3,0) circle (0.12);
}
\end{tikzpicture}
\; , \qquad
\begin{tikzpicture}[baseline={([yshift=-1mm]current bounding box.center)},xscale=0.5,yscale=0.5]
{
\draw[dotted] (0,0.5)--(3,0.5);
\draw [very thick](0,-1.5) -- (0,0.5);
\draw [very thick](3,-1.5)--(3,0.5);
\draw (0,0) to[out=0,in=-90] (0.75,0.5);
\draw (0,-1) to[out=0,in=-90] (1.5,0.5);
\filldraw (0.325,0.5) circle (0.08);
\filldraw (1.125,0.5) circle (0.08);
\filldraw[fill=white] (3,-1) circle (0.12);
\filldraw[fill=white] (3,0) circle (0.12);
}
\end{tikzpicture}
\; , \qquad
\begin{tikzpicture}[baseline={([yshift=-1mm]current bounding box.center)},xscale=0.5,yscale=0.5]
{
\draw[dotted] (0,0.5)--(3,0.5);
\draw [very thick](0,-1.5) -- (0,0.5);
\draw [very thick](3,-1.5)--(3,0.5);
\draw (0,0) to[out=0,in=-90] (0.75,0.5);
\draw (0,-1) to[out=0,in=-90] (1.5,0.5);
\filldraw (0.325,0.5) circle (0.08);
\filldraw (1.875,0.5) circle (0.08);
\filldraw[fill=white] (3,-1) circle (0.12);
\filldraw[fill=white] (3,0) circle (0.12);
}
\end{tikzpicture}
\; , \qquad
\begin{tikzpicture}[baseline={([yshift=-1mm]current bounding box.center)},xscale=0.5,yscale=0.5]
{
\draw[dotted] (0,0.5)--(3,0.5);
\draw [very thick](0,-1.5) -- (0,0.5);
\draw [very thick](3,-1.5)--(3,0.5);
\draw (0,0) to[out=0,in=-90] (0.75,0.5);
\draw (0,-1) to[out=0,in=-90] (1.5,0.5);
\filldraw (1.875,0.5) circle (0.08);
\filldraw (1.125,0.5) circle (0.08);
\filldraw[fill=white] (3,-1) circle (0.12);
\filldraw[fill=white] (3,0) circle (0.12);
}
\end{tikzpicture}
\; . \label{eq:diag67910}
\end{align}
For $c_1$, $c_2$, $c_3$ or $c_4$ nonzero, setting the coefficient of each of these diagrams to zero implies
\begin{align}
0&= (c_1\alpha_1 + c_3\alpha_2+(c_2+c_4)\alpha_3)\, P(\sigma,\tau;\phi) + \left(\frac{c_2c_4}{c_1} \alpha_1 + c_1 \alpha_2 + (c_2+c_4)\alpha_3\right)\! P(\sigma,\tau;-3\phi), \label{eq:211c6} \\
0&= (c_1\alpha_1 + c_3\alpha_2+(c_2+c_4)\alpha_3)\, P(\sigma,\tau;\phi) + \left(c_3 \alpha_1 + c_1 \alpha_2 + \left(c_2+\frac{c_1c_3}{c_2}\right)\!\alpha_3\right)\! P(\sigma,\tau;-3\phi), \label{eq:211c7}\\
0&= (c_1\alpha_1 + c_3\alpha_2+(c_2+c_4)\alpha_3)\, P(\sigma,\tau;\phi) + \left(c_3 \alpha_1 + \frac{c_2c_4}{c_3} \alpha_2 + (c_2+c_4)\alpha_3\right)\! P(\sigma,\tau;-3\phi), \label{eq:211c8} \\
0&= (c_1\alpha_1 + c_3\alpha_2+(c_2+c_4)\alpha_3)\, P(\sigma,\tau;\phi) + \left(c_3 \alpha_1 + c_1 \alpha_2 + \left(\frac{c_1c_3}{c_4}+c_4\right)\!\alpha_3\right)\!P(\sigma,\tau;-3\phi), \label{eq:211c9}
\end{align}
respectively. Note that these are all the same equation if $c_1c_3 = c_2c_4$.

\hypertarget{c2111}{} \phantom{2111}

\noindent\textbf{Case 2.1.1.1:} $c_1c_3 \neq c_2c_4$ \hfill \hyperlink{flowchart}{$\uparrow$}

\noindent If $c_1c_3 \neq c_2c_4$, then we cannot have $c_1c_3=0$ and $c_2c_4=0$, so both $c_1$ and $c_3$ are nonzero, or both $c_2$ and $c_4$ are nonzero.

If $c_1$ and $c_3$ are nonzero, then both \eqref{eq:211c6} and \eqref{eq:211c8} hold. Taking the difference of these equations, and using that $P(\sigma,\tau;-3\phi) \neq 0$ and $c_1c_3 \neq c_2c_4$, we find $c_3 = c_1 \frac{\alpha_2}{\alpha_1}$. Multiplying the coefficient of the second diagram in \eqref{eq:diag67910} by $c_4$, and of the fourth diagram by $c_2$, then taking the difference, we find
\begin{align}
0 &= \frac{\alpha_3}{2\alpha_1}(c_4-c_2) \left(c_2c_4 \alpha_1 - c_1^2\, \alpha_2\right)b(u)b(v) \left(\cos(2u+2v-5\phi) - \cos(\phi) \right)\sin(u-v)\sin(2\phi).
\end{align}
Observing that $c_2c_4\alpha_1 - c_1^2\, \alpha_2 = \alpha_1\left(c_2c_4 - c_1c_3\right) \neq 0$, it follows that $c_2=c_4$.

If instead $c_2$ and $c_4$ are nonzero, then \eqref{eq:211c7} and \eqref{eq:211c9} hold. Taking their difference, and using that $P(\sigma,\tau;-3\phi) \neq 0$ and $c_1c_3 \neq c_2c_4$, we find $c_2 = c_4$. Multiplying the coefficient of the first diagram in \eqref{eq:diag67910} by $c_3$, and of the third diagram by $c_1$, then taking the difference, we find
\begin{align}
0 &= \frac{1}{2}(c_3\alpha_1 - c_1 \alpha_2)\left(c_1c_3-c_2^2\right) b(u) b(v)    (\cos(2 u+2 v-5 \phi)-\cos(\phi)) \sin (u-v)\sin (2 \phi) .
\end{align}
Since $c_1c_3 - c_2^2 = c_1c_3 - c_2c_4 \neq 0$, it follows that $c_3 = c_1 \frac{\alpha_2}{\alpha_1}$.

Hence in both cases, $c_3 = c_1 \frac{\alpha_2}{\alpha_1}$ and $c_2=c_4$. Multiplying the coefficient of the first diagram in \eqref{eq:diag67910} by $c_2$, and of the second diagram by $c_1$, then taking the difference, we find
\begin{align}
0 &= \frac{1}{2\alpha_1} \left(c_2^2\,\alpha_1 -c_1^2\,\alpha_2\right)\left(c_2\alpha_1 + c_1\alpha_3\right)  b(u) b(v)    (\cos(2 u+2 v-5 \phi)-\cos(\phi)) \sin (u-v)\sin (2 \phi) .
\end{align}
Since $c_2^2\,\alpha_1 - c_1^2\,\alpha_2 = \alpha_1(c_2c_4 - c_1c_3)\neq 0$, it follows that $c_2 = -c_1 \frac{\alpha_3}{\alpha_1}$. Hence $c_1 \neq 0$, since otherwise $c_1=c_2=c_3=c_4=0$, so \eqref{eq:211c6} holds. Substituting $c_2=c_4 = -c_1 \frac{\alpha_3}{\alpha_1}$ and $c_3 = c_1 \frac{\alpha_2}{\alpha_1}$ into \eqref{eq:211c6} and dividing by $\frac{c_1}{\alpha_1} \neq 0$ yields
\begin{align}
0 &= A_1\, P(\sigma,\tau;\phi) + A_2\, P(\sigma,\tau;-3\phi).
\end{align}
If $\sigma$ and $\tau$ satisfy this quadratic relation, then all diagram coefficients are zero. After rescaling, this gives \hyperlink{solIII}{Solution III}.

\hypertarget{c2112}{} \phantom{2112}

\noindent\textbf{Case 2.1.1.2:} $c_1c_3 = c_2c_4$ \hfill \hyperlink{flowchart}{$\uparrow$}

\noindent If $c_1c_3 = c_2c_4$, equations \eqref{eq:211c6}--\eqref{eq:211c9} all become
\begin{align}
0 &=(c_1\alpha_1 + c_3\alpha_2+(c_2+c_4)\alpha_3) P(\sigma,\tau;\phi) + \left(c_3 \alpha_1 + c_1 \alpha_2 + (c_2+c_4)\alpha_3\right)P(\sigma,\tau;-3\phi). \label{eq:2112}
\end{align}
Noting that $P(\sigma,\tau;\phi)+P(\sigma,\tau;-3\phi) = 2\cos(2\phi)P(\sigma,\tau;-\phi)$, we can rearrange \eqref{eq:2112} to find $c_4$ in terms of $c_1$, $c_2$ and $c_3$, if $P(\sigma,\tau;-\phi) \neq 0$.

\hypertarget{c21121}{} \phantom{21121}

\noindent\textbf{Case 2.1.1.2.1:} $P(\sigma,\tau;-\phi) \neq 0$ \hfill \hyperlink{flowchart}{$\uparrow$}

\noindent Rearranging \eqref{eq:2112} for $c_4$ gives
\begin{align}
c_4 &= -\frac{(c_1\alpha_1+c_3\alpha_2 + c_2\alpha_3)P(\sigma,\tau;\phi) + (c_3\alpha_1 + c_1 \alpha_2 + c_2\alpha_3)P(\sigma,\tau;-3\phi)}{2\alpha_3 \cos(2\phi) P(\sigma,\tau;-\phi)}. \label{eq:21121c9}
\end{align}
Substituting this into $c_1c_3 = c_2c_4$, we find
\begin{align}
&\quad \ \ c_3 \left((c_2 \alpha_2 + c_1\alpha_3) P(\sigma,\tau; \phi) + (c_2 \alpha_1 + c_1\alpha_3)P(\sigma,\tau;-3\phi) \right) \nonumber \\
&= -c_2\left((c_1\alpha_1 + c_2 \alpha_3)P(\sigma,\tau;\phi) + (c_1\alpha_2+c_2\alpha_3)P(\sigma,\tau;-3\phi) \right). \label{eq:21121c8}
\end{align}
If the coefficient of $c_3$ here is nonzero, this can be rearranged to find $c_3$ in terms of $c_1$ and $c_2$.

\hypertarget{c211211}{} \phantom{211211}

\noindent\textbf{Case 2.1.1.2.1.1:} $c_3$ coefficient $\neq 0$ \hfill \hyperlink{flowchart}{$\uparrow$}

\noindent Rearranging \eqref{eq:21121c8} for $c_3$ gives
\begin{align}
c_3 &= - \frac{c_2\left((c_1\alpha_1 + c_2 \alpha_3)P(\sigma,\tau;\phi) + (c_1\alpha_2+c_2\alpha_3)P(\sigma,\tau;-3\phi) \right)}{(c_2 \alpha_2 + c_1\alpha_3) P(\sigma,\tau; \phi) + (c_2 \alpha_1 + c_1\alpha_3)P(\sigma,\tau;-3\phi)}.
\end{align}
Substituting this into the expression \eqref{eq:21121c9} for $c_4$ gives
\begin{align}
c_4 &= - \frac{c_1\left((c_1\alpha_1 + c_2 \alpha_3)P(\sigma,\tau;\phi) + (c_1\alpha_2+c_2\alpha_3)P(\sigma,\tau;-3\phi) \right)}{(c_2 \alpha_2 + c_1\alpha_3) P(\sigma,\tau; \phi) + (c_2 \alpha_1 + c_1\alpha_3)P(\sigma,\tau;-3\phi)}.
\end{align}

Now, recall from \hyperlink{c2}{Case 2} that
\begin{align}
a_1(u) &= h(u)\left(\sigma \cos(u) + \tau \sin(u) \right), \\
a_2(u) &= h(u) \left(\sigma \cos(u) - \tau\sin(u) \right), \\
b_k(u) &= c_k(u), &k&\in \{1,2,3,4\}, \\
b_5(u) &= b_6(u) = b_7(u) = b_8(u) = 0,
\end{align}
and from \hyperlink{c211}{Case 2.1.1} that
\begin{align}
h(u) &= \frac{b(u)\left(c_1\alpha_1 + c_3 \alpha_2 +(c_2+c_4)\alpha_3\right)\left(\sigma\cos(u-\phi) + \tau \sin (u-\phi)\right)}{\sin(2u) P(\sigma,\tau;-3\phi)}.
\end{align}
Substituting our expressions for $c_3$ and $c_4$ into these sets all diagram coefficients to zero, so this is a solution to the BYBE. Rescaling this to remove denominators and setting $c_1 = \mu$ and $c_2 = \nu$, this becomes \hyperlink{solIV}{Solution IV}, with $P(\sigma,\tau;-3\phi) \neq 0$, $P(\sigma,\tau; -\phi) \neq 0$, and
\begin{align}
(\nu \alpha_2 + \mu \alpha_3) P(\sigma,\tau; \phi) + (\nu \alpha_1 + \mu\alpha_3)P(\sigma,\tau;-3\phi) \neq 0.
\end{align}

\hypertarget{c211212}{} \phantom{211212}

\noindent\textbf{Case 2.1.1.2.1.2:} $c_3$ coefficient $=0$ \hfill \hyperlink{flowchart}{$\uparrow$}

\noindent If the coefficient of $c_3$ in \eqref{eq:21121c8} is zero, then
\begin{align}
c_1 &=- c_2 \frac{\alpha_1 P(\sigma,\tau;-3\phi) + \alpha_2P(\sigma,\tau; \phi)}{2\alpha_3\, \cos(2\phi)  P(\sigma,\tau; -\phi)},
\end{align}
where we have used that $P(\sigma,\tau; \phi)+P(\sigma,\tau; -3\phi) = 2\cos(2\phi)P(\sigma,\tau; -\phi)$. Note that we have already assumed $P(\sigma,\tau; -\phi) \neq 0$ in \hyperlink{c21121}{Case 2.1.1.2.1}. Substituting this and the expression \eqref{eq:21121c9} into our solution ansatz gives
\begin{align}
a_1(u) &= -\frac{b(u)(\alpha_1-\alpha_2)\left(\sigma \cos(u) + \tau\sin(u) \right) \left(\sigma\cos(u-\phi) + \tau\sin(u-\phi)\right)}{4 \alpha_3 \sin(2u)\cos^2(2\phi) P(\sigma,\tau;-\phi)^2} \nonumber \\
&\hspace{15mm} \times \Big(c_2\big(\alpha_1 P(\sigma,\tau; -3\phi) + \alpha_2P(\sigma,\tau; \phi) \big) + 2\,c_3\,\alpha_3 \cos(2\phi)  P(\sigma,\tau;-\phi)\Big),\\
a_2(u) &=  -\frac{b(u)(\alpha_1-\alpha_2)\left(\sigma \cos(u) - \tau\sin(u) \right) \left(\sigma\cos(u-\phi) + \tau\sin(u-\phi)\right)}{4 \alpha_3 \sin(2u)\cos^2(2\phi) P(\sigma,\tau;-\phi)^2} \nonumber \\
&\hspace{15mm} \times \Big(c_2\big(\alpha_1 P(\sigma,\tau; -3\phi) + \alpha_2P(\sigma,\tau; \phi) \big) + 2\,c_3\,\alpha_3 \cos(2\phi)  P(\sigma,\tau;-\phi)\Big), \label{eq:211212a}\\
b_1(u) &= - c_2 \, b(u)\, \frac{\alpha_1 P(\sigma,\tau;-3\phi) + \alpha_2P(\sigma,\tau; \phi)}{2\alpha_3 \cos(2\phi) P(\sigma,\tau; -\phi)^2}, \\
b_2(u) &= c_2\, b(u), \\
b_3(u) &= c_3\, b(u), \\
b_4(u) &= \frac{-b(u)}{2 \alpha_3\cos(2\phi) P(\sigma,\tau; -\phi)^2} \left(c_3 \big(\alpha_1 P(\sigma,\tau;-3\phi) + \alpha_2P(\sigma,\tau; \phi)\big) -  \frac{c_2\, Q(\sigma, \tau) }{2\alpha_3 \cos(2\phi) }\right), \label{eq:211212b9} \\
b_5(u) &= b_6(u) = b_7(u) = b_8(u) = 0, 
\end{align}
where we define
\begin{align}
Q(\sigma, \tau) := A_1 P(\sigma,\tau; \phi) P(\sigma,\tau;-3\phi) + A_2 \left(P(\sigma,\tau;\phi)^2 + P(\sigma,\tau;-3\phi)^2\right). \label{eq:Qdef}
\end{align}
This is a quartic polynomial in $\sigma$ and $\tau$. The coefficient of the diagram
\begin{align}
\begin{tikzpicture}[baseline={([yshift=-1mm]current bounding box.center)},xscale=0.5,yscale=0.5]
{
\draw[dotted] (0,0.5)--(3,0.5);
\draw [very thick](0,-1.5) -- (0,0.5);
\draw [very thick](3,-1.5)--(3,0.5);
\draw (0,0) to[out=0,in=-90] (0.75,0.5);
\draw (0,-1) to[out=0,in=-90] (1.5,0.5);
\filldraw (0.325,0.5) circle (0.08);
\filldraw (1.125,0.5) circle (0.08);
\filldraw[fill=white] (3,-1) circle (0.12);
\filldraw[fill=white] (3,0) circle (0.12);
}
\end{tikzpicture}
\end{align}
then becomes
\begin{align}
&\frac{c_2^2\, b(u)b(v) \sin(u-v)\tan(2\phi) \left(\cos(2u+2v-5\phi) - \cos(\phi)\right)}{8\alpha_3\, \cos(2\phi) P(\sigma,\tau;-\phi)^2} \ Q(\sigma, \tau),
\end{align}
and this must equal zero. Next we consider whether $Q(\sigma,\tau)$ is zero.

\hypertarget{c2112121}{} \phantom{2112121}

\noindent\textbf{Case 2.1.1.2.1.2.1:} $Q(\sigma,\tau) \neq 0$ \hfill \hyperlink{flowchart}{$\uparrow$}

\nopagebreak
\noindent If $Q(\sigma,\tau) \neq 0$, then we must have $c_2 = 0$. Substituting this into the ansatz \eqref{eq:211212a}--\eqref{eq:211212b9} sets all diagram coefficients to zero, so this is a solution. Rescaling by a factor of
\begin{align}
\frac{Q(\sigma,\tau)}{2 c_3 \alpha_3\, \cos(2\phi) P(\sigma,\tau;-\phi)},
\end{align}
this becomes \hyperlink{solIV}{Solution IV}, with $P(\sigma,\tau; -3\phi) \neq 0$, $P(\sigma,\tau; -\phi) \neq 0$, and
\begin{align}
(\nu \alpha_2 + \mu \alpha_3) P(\sigma,\tau; \phi) + (\nu \alpha_1 + \mu\alpha_3)P(\sigma,\tau;-3\phi) = 0.
\end{align}

\hypertarget{c2112122}{} \phantom{2112122}

\noindent\textbf{Case 2.1.1.2.1.2.2:} $Q(\sigma,\tau)=0$ \hfill \hyperlink{flowchart}{$\uparrow$}

\noindent Rescaling by a factor of $-\alpha_3 P(\sigma,\tau; \phi)^2$, this becomes \hyperlink{solV}{Solution V}, where we note that the coefficient of $c_2$ in $b_4$ is now zero.

\hypertarget{c21122}{} \phantom{21122}

\noindent\textbf{Case 2.1.1.2.2:} $P(\sigma,\tau;-\phi) = 0$ \hfill \hyperlink{flowchart}{$\uparrow$}

\noindent If $P(\sigma,\tau;-\phi) = 0$, then $\sigma = -\tau \tan\! \left(\frac{\phi}{2}\right)$ or $\sigma = \tau \cot\! \left(\frac{\phi}{2}\right)$.

\hypertarget{c211221}{} \phantom{211221}

\noindent\textbf{Case 2.1.1.2.2.1:} $\sigma = -\tau \tan\left(\frac{\phi}{2}\right)$ \hfill \hyperlink{flowchart}{$\uparrow$}

\noindent Substituting $\sigma = -\tau \tan \left(\frac{\phi}{2}\right)$ into \eqref{eq:2112} from \hyperlink{c2112}{Case 2.1.1.2} gives
\begin{align}
0 &= 2(c_1-c_3) (\alpha_1-\alpha_2) \,\tau^2 \tan\!\left(\frac{\phi}{2}\right) \cos(\phi).
\end{align}
Since $\tau = 0$ would contradict $P(\sigma,\tau;-3\phi) \neq 0$ (\hyperlink{c211}{Case 2.1.1}), it follows that $c_1=c_3$. We also have $c_1 c_3 = c_2c_4$ (\hyperlink{c2112}{Case 2.1.1.2}), so
\begin{align}
c_1 &= c_3 = \mu \nu \, \rho, &c_2 &= \nu^2\, \rho, &c_4 &= \mu^2\, \rho,
\end{align}
for some $\mu, \nu, \rho \in \C$. We note that $\rho \neq 0$, since $c_1$, $c_2$, $c_3$ and $c_4$ are not all zero.

This gives the solution
\begin{align}
a_1(u) &= \frac{-\rho\, b(u)}{\sin(2u)\sin(2\phi)}\left(\mu \nu (\alpha_1+\alpha_2) + \left(\mu^2 + \nu^2\right) \alpha_3\right) \cos\!\left(u-\frac{\phi}{2}\right) \cos\!\left(u-\frac{3\phi}{2}\right), \\
a_2(u) &= \frac{-\rho\, b(u)}{\sin(2u)\sin(2\phi)}\left(\mu \nu (\alpha_1+\alpha_2) + \left(\mu^2 + \nu^2\right) \alpha_3 \right) \cos\!\left(u+\frac{\phi}{2}\right) \cos\!\left(u-\frac{3\phi}{2}\right), \\
b_1(u) &= b_3(u) =\mu \nu \, \rho\, b(u) , \\
b_2(u) &= \nu^2\, \rho\, b(u) , \\
b_4(u) &= \mu^2\, \rho\, b(u), \\
b_5(u) &= b_6(u)=b_7(u)=b_8(u)=0, 
\end{align}
which is \hyperlink{solIV}{Solution IV} with $\sigma = -\tau \tan\!\left(\frac{\phi}{2}\right) \neq 0$, multiplied by $\frac{1}{4(\alpha_1-\alpha_2)}\rho\cot\!\left(\frac{\phi}{2}\right)\sec(\phi) \neq 0$.

\hypertarget{c211222}{} \phantom{211222}

\noindent\textbf{Case 2.1.1.2.2.2:} $\sigma = \tau \cot\left(\frac{\phi}{2}\right)$ \hfill \hyperlink{flowchart}{$\uparrow$}

\nopagebreak
\noindent This is analogous to \hyperlink{c211221}{Case 2.1.1.2.2.1}, and gives a nonzero scalar multiple of \hyperlink{solIV}{Solution IV} with $\sigma = \tau \cot\!\left(\frac{\phi}{2}\right) \neq 0$.

\hypertarget{c212}{} \phantom{212}

\noindent\textbf{Case 2.1.2:} $P(\sigma,\tau;-3\phi)=0$ \hfill \hyperlink{flowchart}{$\uparrow$}

\nopagebreak
\noindent If $P(\sigma,\tau; -3\phi) = 0$, then $\sigma = -\tau \tan\!\left(\frac{3\phi}{2}\right)$ or $\sigma = \tau \cot\!\left(\frac{3\phi}{2}\right)$.

\hypertarget{c2121}{} \phantom{2121}

\noindent\textbf{Case 2.1.2.1:} $\sigma = -\tau \tan\left(\frac{3\phi}{2}\right)$ \hfill \hyperlink{flowchart}{$\uparrow$}

\noindent Recalling that $c_1$, $c_2$, $c_3$ and $c_4$ are not all zero, consider the case where $c_1 \neq 0$. Then we can set the coefficient of the diagram
\begin{align}
\begin{tikzpicture}[baseline={([yshift=-1mm]current bounding box.center)},xscale=0.5,yscale=0.5]
{
\draw[dotted] (0,0.5)--(3,0.5);
\draw [very thick](0,-1.5) -- (0,0.5);
\draw [very thick](3,-1.5)--(3,0.5);
\draw (0,0) to[out=0,in=-90] (1,0.5);
\draw (3,-1) to[out=180,in=270] (2,0.5);
\filldraw[fill=white] (0,-1) circle (0.12);
\filldraw[fill=white] (3,0) circle (0.12);
}
\end{tikzpicture}
\label{eq:2121M28}
\end{align}
equal to zero, and rearrange to find
\begin{align}
&\quad \left(\frac{c_2c_4}{c_1} \alpha_1 + c_1 \alpha_2 + (c_2+c_4) \alpha_3 \right) \cot\!\left(u+\frac{3\phi}{2}\right) + \frac{\tau\, h(u)}{b(u)} \sec\! \left(\frac{3\phi}{2}\right) \csc\!\left(u+\frac{3\phi}{2}\right) \sin(2u) \nonumber \\
&= \left(\frac{c_2c_4}{c_1} \alpha_1 + c_1 \alpha_2 + (c_2+c_4) \alpha_3 \right) \cot\!\left(v+\frac{3\phi}{2}\right) + \frac{\tau\, h(v)}{b(v)} \sec\! \left(\frac{3\phi}{2}\right) \csc\!\left(v+\frac{3\phi}{2}\right) \sin(2v).
\end{align}
Since the left- and right-hand sides are functions of $u$ and $v$, respectively, it follows that both sides are equal to some constant $\kappa \in \C$. Thus
\begin{align}
\tau\, h(u) &= \frac{b(u)}{\sin(2u)} \cos\!\left(\frac{3\phi}{2}\right) \left(\! \kappa \sin\!\left(u+ \frac{3\phi}{2}\right) - \left(\frac{c_2c_4}{c_1} \alpha_1 + c_1 \alpha_2 + (c_2+c_4) \alpha_3 \!\right) \cos\!\left(u+\frac{3\phi}{2}\right)\!\right). \label{eq:2121tauh}
\end{align}
Substituting this into the coefficients of
\begin{align}
\begin{tikzpicture}[baseline={([yshift=-1mm]current bounding box.center)},xscale=0.5,yscale=0.5]
{
\draw[dotted] (0,0.5)--(3,0.5);
\draw [very thick](0,-1.5) -- (0,0.5);
\draw [very thick](3,-1.5)--(3,0.5);
\draw (0,0) to[out=0,in=-90] (1,0.5);
\draw (3,-1) to[out=180,in=270] (2,0.5);
\filldraw (0.5,0.5) circle (0.08);
\filldraw (1.5,0.5) circle (0.08);
\filldraw[fill=white] (0,-1) circle (0.12);
\filldraw[fill=white] (3,0) circle (0.12);
}
\end{tikzpicture}
\; ,\qquad
\begin{tikzpicture}[baseline={([yshift=-1mm]current bounding box.center)},xscale=0.5,yscale=0.5]
{
\draw[dotted] (0,0.5)--(3,0.5);
\draw [very thick](0,-1.5) -- (0,0.5);
\draw [very thick](3,-1.5)--(3,0.5);
\draw (0,0) to[out=0,in=-90] (1,0.5);
\draw (3,-1) to[out=180,in=270] (2,0.5);
\filldraw (0.5,0.5) circle (0.08);
\filldraw (2.5,0.5) circle (0.08);
\filldraw[fill=white] (0,-1) circle (0.12);
\filldraw[fill=white] (3,0) circle (0.12);
}
\end{tikzpicture}
\; , \qquad
\begin{tikzpicture}[baseline={([yshift=-1mm]current bounding box.center)},xscale=0.5,yscale=0.5]
{
\draw[dotted] (0,0.5)--(3,0.5);
\draw [very thick](0,-1.5) -- (0,0.5);
\draw [very thick](3,-1.5)--(3,0.5);
\draw (0,0) to[out=0,in=-90] (1,0.5);
\draw (3,-1) to[out=180,in=270] (2,0.5);
\filldraw (1.5,0.5) circle (0.08);
\filldraw (2.5,0.5) circle (0.08);
\filldraw[fill=white] (0,-1) circle (0.12);
\filldraw[fill=white] (3,0) circle (0.12);
}
\end{tikzpicture}
\label{eq:M251518}
\end{align}
and setting them equal to zero yields
\begin{align}
0 &= \frac{1}{2c_1} (c_1c_3-c_2c_4) (c_2 \alpha_1 + c_1 \alpha_3)\, b(u)b(v) \left(\cos(2v) - \cos(2(u-3\phi)) \right) \sin(u-v) \sin(2\phi), \\
0 &= \frac{-1}{2c_1} (c_1c_3-c_2c_4) (c_3 \alpha_1 + c_1 \alpha_2)\, b(u)b(v) \left(\cos(2v) - \cos(2(u-3\phi)) \right) \sin(u-v) \sin(2\phi), \\
0 &= \frac{1}{2c_1} (c_1c_3-c_2c_4) (c_4 \alpha_1 + c_1 \alpha_3)\, b(u)b(v) \left(\cos(2v) - \cos(2(u-3\phi)) \right) \sin(u-v) \sin(2\phi).
\end{align}
Hence
\begin{align}
c_1 c_3 = c_2 c_4, \label{eq:2121c6c8c7c9}
\end{align}
or
\begin{align}
c_1 &= -\alpha_1 \rho, &c_2 &= c_4 = \alpha_3 \rho, &c_3 &= -\alpha_2 \rho, \label{eq:2121rho}
\end{align}
for some $\rho \in \C$. Note that $\rho \neq 0$ because $c_1 \neq 0$. In the first case, we find
\begin{align}
\tau\, h(u) &= \frac{b(u)}{\sin(2u)} \cos\!\left(\frac{3\phi}{2}\right) \left(\! \kappa \sin\!\left(u+ \frac{3\phi}{2}\right) - \left(c_3 \alpha_1 + c_1 \alpha_2 + (c_2+c_4) \alpha_3 \right) \cos\!\left(u+\frac{3\phi}{2}\right)\!\right), \label{eq:tauhfirst}
\end{align}
with $c_1c_3=c_2c_4$, and in the second,
\begin{align}
\tau\, h(u) &= \frac{b(u)}{\sin(2u)} \cos\!\left(\frac{3\phi}{2}\right) \left(\! \kappa \sin\!\left(u+ \frac{3\phi}{2}\right) +A_2\, \rho  \cos\!\left(u+\frac{3\phi}{2}\right)\!\right). \label{eq:tauhsecond}
\end{align}

For the cases $c_2 \neq 0$, $c_3 \neq 0$ and $c_4 \neq 0$, we instead use the coefficients of the three diagrams in \eqref{eq:M251518}, respectively, to find expressions similar to \eqref{eq:2121tauh}. Substituting these into the other three diagrams (including \eqref{eq:2121M28}), we similarly find that \eqref{eq:2121c6c8c7c9} or \eqref{eq:2121rho} must hold, and get the same expressions \eqref{eq:tauhfirst} and \eqref{eq:tauhsecond} for $\tau\, h(u)$, respectively.

If \eqref{eq:2121rho} holds, the coefficient of
\begin{align}
\begin{tikzpicture}[baseline={([yshift=-1mm]current bounding box.center)},xscale=0.5,yscale=0.5]
{
\draw[dotted] (0,0.5)--(3,0.5);
\draw [very thick](0,-1.5) -- (0,0.5);
\draw [very thick](3,-1.5)--(3,0.5);
\draw (0,0) arc (90:-90:0.5);
\filldraw[fill=white] (3,-1) circle (0.12);
\filldraw[fill=white] (3,0) circle (0.12);
}
\end{tikzpicture}
\end{align}
becomes
\begin{align}
\frac{A_1\, \rho}{2} \, b(u)b(v) (\cos(2u)-\cos(2v)) \sin\!\left(u-\frac{5\phi}{2}\right) \sin(2\phi) \left(\!A_2\, \rho \cos\!\left(v+\frac{3\phi}{2}\right) + \kappa \sin\!\left(v+\frac{3\phi}{2}\right)\!\right).
\end{align}
Since $\rho \neq 0$, and $\cos\!\left(v+\frac{3\phi}{2}\right)$ and $\sin\!\left(v+\frac{3\phi}{2}\right)$ are linearly independent, this expression is nonzero. Hence we cannot have a solution with \eqref{eq:2121rho}, and so we must have $c_1c_3=c_2c_4$. This means we can write
\begin{align}
c_1 &= m_1m_2, 
&c_2 &= m_1m_3,
&c_3 &= m_3 m_4,
&c_4 &= m_2 m_4,
\end{align}
for some $m_1,m_2,m_3,m_4 \in \C$. Since $c_1$, $c_2$, $c_3$ and $c_4$ are not all zero, $m_2$ and $m_3$ are not both zero, and $m_1$ and $m_4$ are not both zero.

The coefficients of the diagrams
\begin{align}
\begin{tikzpicture}[baseline={([yshift=-1mm]current bounding box.center)},xscale=0.5,yscale=0.5]
{
\draw[dotted] (0,0.5)--(3,0.5);
\draw [very thick](0,-1.5) -- (0,0.5);
\draw [very thick](3,-1.5)--(3,0.5);
\draw (0,0) to[out=0,in=-90] (0.75,0.5);
\draw (0,-1) to[out=0,in=-90] (1.5,0.5);
\filldraw[fill=white] (3,-1) circle (0.12);
\filldraw[fill=white] (3,0) circle (0.12);
}
\end{tikzpicture}
\; , \qquad
\begin{tikzpicture}[baseline={([yshift=-1mm]current bounding box.center)},xscale=0.5,yscale=0.5]
{
\draw[dotted] (0,0.5)--(3,0.5);
\draw [very thick](0,-1.5) -- (0,0.5);
\draw [very thick](3,-1.5)--(3,0.5);
\draw (0,0) to[out=0,in=-90] (0.75,0.5);
\draw (0,-1) to[out=0,in=-90] (1.5,0.5);
\filldraw (0.325,0.5) circle (0.08);
\filldraw (1.125,0.5) circle (0.08);
\filldraw[fill=white] (3,-1) circle (0.12);
\filldraw[fill=white] (3,0) circle (0.12);
}
\end{tikzpicture}
\; , \qquad
\begin{tikzpicture}[baseline={([yshift=-1mm]current bounding box.center)},xscale=0.5,yscale=0.5]
{
\draw[dotted] (0,0.5)--(3,0.5);
\draw [very thick](0,-1.5) -- (0,0.5);
\draw [very thick](3,-1.5)--(3,0.5);
\draw (0,0) to[out=0,in=-90] (0.75,0.5);
\draw (0,-1) to[out=0,in=-90] (1.5,0.5);
\filldraw (0.325,0.5) circle (0.08);
\filldraw (1.875,0.5) circle (0.08);
\filldraw[fill=white] (3,-1) circle (0.12);
\filldraw[fill=white] (3,0) circle (0.12);
}
\end{tikzpicture}
\; , \qquad
\begin{tikzpicture}[baseline={([yshift=-1mm]current bounding box.center)},xscale=0.5,yscale=0.5]
{
\draw[dotted] (0,0.5)--(3,0.5);
\draw [very thick](0,-1.5) -- (0,0.5);
\draw [very thick](3,-1.5)--(3,0.5);
\draw (0,0) to[out=0,in=-90] (0.75,0.5);
\draw (0,-1) to[out=0,in=-90] (1.5,0.5);
\filldraw (1.875,0.5) circle (0.08);
\filldraw (1.125,0.5) circle (0.08);
\filldraw[fill=white] (3,-1) circle (0.12);
\filldraw[fill=white] (3,0) circle (0.12);
}
\end{tikzpicture}
 \label{eq:2121M67910}
\end{align}
become
\begin{align}
&c_i\, b(u)b(v)(\cos(2u)-\cos(2v)) \sin(2\phi) \cos\!\left(u-\frac{5\phi}{2}\right)\sin\!\left(v+\frac{3\phi}{2}\right) \nonumber \\
&\hspace{5mm} \times \bigg( \kappa \sin(4\phi) - (m_3 m_4 \alpha_1 + m_1m_2\alpha_2 + (m_1m_3+m_2 m_4)\alpha_3) \cos(4\phi)  \nonumber \\
&\hspace{15mm}- \frac{1}{2}(m_1m_2 \alpha_1 + m_3 m_4\alpha_2 + (m_1m_3+m_2 m_4)\alpha_3) \left( \tan\!\left(u-\frac{5\phi}{2}\right) \cot\!\left(v+\frac{3\phi}{2}\right)+1 \right)\!\! \bigg),
\end{align}
for $i=1,2,3,4$, respectively. Since $c_i$ is nonzero for some $i \in \{1,2,3,4\}$, setting these coefficients to zero implies the last set of brackets is zero. Since the constant function $1$ and $\tan\!\left(u-\frac{5\phi}{2}\right) \cot\!\left(v+\frac{3\phi}{2}\right)+1$ are linearly independent, the coefficients of both of these functions must be zero. Hence we have
\begin{align}
\kappa = (m_3 m_4 \alpha_1 + m_1m_2\alpha_2 + (m_1m_3+m_2 m_4)\alpha_3) \cot(4\phi), \label{eq:2121kappa}
\end{align}
and 
\begin{align}
0 &= m_1m_2 \alpha_1 + m_3 m_4\alpha_2 + (m_1m_3+m_2 m_4)\alpha_3. \label{eq:2121const}
\end{align}
Since $m_2$ and $m_3$ are not both zero, \eqref{eq:2121const} implies $m_4\alpha_2 + m_1\alpha_3 = -m_2\, z$ and $m_1\alpha_1 + m_4 \alpha_3 = m_3 \, z$, for some $z \in \C$. Hence
\begin{align}
m_1 &= \frac{(m_3 \alpha_2 + m_2 \alpha_3)z}{A_2}, &m_4 &= - \frac{(m_2\alpha_1 + m_3\alpha_3)z}{A_2}.
\end{align}
Applying this and \eqref{eq:2121kappa} sets all diagram coefficients to zero. Noting that $z \neq 0$, because $m_1$ and $m_4$ are not both zero, we can rescale this solution by a factor of $A_2/z$ and set $m_2 = \mu$ and $m_3 = \nu$ to produce \hyperlink{solIV}{Solution IV}, with $\sigma = -\tau \tan\!\left(\frac{3\phi}{2}\right)$.

\hypertarget{c2122}{} \phantom{2122}

\noindent\textbf{Case 2.1.2.2:} $\sigma = \tau \cot\left(\frac{3\phi}{2}\right)$ \hfill \hyperlink{flowchart}{$\uparrow$}

\noindent This case is analogous to \hyperlink{c2121}{Case 2.1.2.1}, and yields a nonzero scalar multiple of \hyperlink{solIV}{Solution IV} with $\sigma = \tau \cot\!\left(\frac{3\phi}{2}\right)$.

\hypertarget{c22}{} \phantom{22}

\noindent\textbf{Case 2.2:} $c_1=c_2=c_3=c_4=0$ \hfill \hyperlink{flowchart}{$\uparrow$}

\nopagebreak
\noindent Recall from Case \hyperlink{c2}{2} that 
\begin{align}
a_1(u) &= h(u) \left(\sigma \cos(u) + \tau\sin(u) \right), \\
a_2(u) &= h(u) \left(\sigma \cos(u) - \tau \sin(u) \right), \\
b_k(u) &= c_kb(u), &k \in \{1,2,3,4\}, \\
b_5(u) &= b_6(u) = b_7(u)=b_8(u) = 0, 
\end{align}
where $b$ and $h$ are nonzero, and $\sigma, \tau, c_k \in \C$ for each $k \in \{1,2,3,4\}$. With $c_1=c_2=c_3=c_4=0$, the coefficient of
\begin{align}
\begin{tikzpicture}[baseline={([yshift=-1mm]current bounding box.center)},xscale=0.5,yscale=0.5]
{
\draw[dotted] (0,0.5)--(3,0.5);
\draw [very thick](0,-1.5) -- (0,0.5);
\draw [very thick](3,-1.5)--(3,0.5);
\draw (0,0) arc (90:-90:0.5);
\filldraw[fill=white] (3,-1) circle (0.12);
\filldraw[fill=white] (3,0) circle (0.12);
}
\end{tikzpicture}
\end{align}
becomes
\begin{align}
\frac{1}{2} h(u) h(v) \sin (2 u) \sin (2 v) \sin (2 \phi) (\cos(2 u)-\cos(2 v)) P(\sigma,\tau;-3\phi).
\end{align}
For this to be zero for all $u,v$, we must have $P(\sigma,\tau;-3\phi)=0$, and therefore $\sigma = - \tau \tan\left(\frac{3\phi}{2}\right)$ or $\sigma = \tau\cot\left(\frac{3\phi}{2}\right)$. After some rescaling, these yield the solutions in Cases \hyperlink{c221}{2.2.1} and \hyperlink{c222}{2.2.2}, respectively. These coincide with the $\dTL_n$ BYBE solutions found by Yung and Batchelor in \cite{BatchelorYung}.

\hypertarget{c221}{} \phantom{221}
\vspace{-5mm}

\noindent\textbf{Case 2.2.1:} $\sigma = - \tau \tan\left(\frac{3\phi}{2}\right)$ \hfill \hyperlink{flowchart}{$\uparrow$}

\noindent This gives the solution
\vspace{-2mm}
\begin{align}
a_1(u) &= h(u)\sin\left(u-\frac{3\phi}{2}\right), \\
a_2(u) &= -h(u) \sin\left(u+ \frac{3\phi}{2}\right), \\
b_1(u)&=b_2(u)=b_3(u)=b_4(u)=b_5(u)=b_6(u)=b_7(u)=b_8(u) = 0.
\end{align}
Up to scaling, this is \hyperlink{solI}{Solution I} with $\mu=\nu=0$. Note that allowing $h=0$ would give the zero solution, and that Solution I becomes the zero solution when $b=0$, so Solution I has now been fully accounted for.

\hypertarget{c222}{} \phantom{222}

\noindent\textbf{Case 2.2.2:} $\sigma = \tau \cot\left(\frac{3\phi}{2}\right)$ \hfill \hyperlink{flowchart}{$\uparrow$}

\noindent This gives the solution
\vspace{-2mm}
\begin{align}
a_1(u) &= h(u)\cos\left(u-\frac{3\phi}{2}\right), \\
a_2(u) &= h(u) \cos\left(u+ \frac{3\phi}{2}\right), \\
b_1(u)&=b_2(u)=b_3(u)=b_4(u)=b_5(u)=b_6(u)=b_7(u)=b_8(u) = 0.
\end{align}
Up to scaling, this is \hyperlink{solII}{Solution II} with $\mu=\nu=0$. Note that allowing $h=0$ would give the zero solution, and that Solution II becomes the zero solution when $b=0$, so Solution II has now been fully accounted for.
\end{proof}

\subsubsection{Transfer tangle}\label{sss:dilutetransfer}
As we did for the fully-packed model, we construct the transfer tangle
\begin{align}
    T(u) = \begin{tikzpicture}[baseline={([yshift=-1mm]current bounding box.center)},scale=1]
{
\draw[lstring,dashed,dash phase=1pt] (0.5,-1)--(0.5,-0.5);
\draw[lstring,dashed,dash phase=1pt] (0.5,-5)--(0.5,-5.5);
\draw[lstring,dashed,dash phase=1pt] (1.5,-1)--(1.5,-0.5);
\draw[lstring,dashed,dash phase=1pt] (1.5,-5)--(1.5,-5.5);
\facetop{(0,0)}{$3\phi -u$}
\faceop{(0,-1)}{$u$}
\faceop{(1,-1)}{$3\phi - u$}
\faceop{(0,-2)}{$u$}
\faceop{(1,-2)}{$3\phi-u$}
\vdotsq{(0,-3)}
\vdotsq{(1,-3)}
\faceop{(0,-4)}{$u$}
\faceop{(1,-4)}{$3\phi - u$}
\facebot{(0,-6)}{$u$}
}
\end{tikzpicture}
\end{align}
out of bulk and boundary face operators that satisfy the YBE, the local inversion relation, crossing symmetry and the BYBEs. We similarly have $T(u)T(v) = T(v)T(u)$ for all $u,v \in \C$, as in \cite[\S 2.4]{BehrendPearceOBrien}. Note that the boundary face operators are linked to the bulk face operators with dashed identity strings, instead of the solid strings used in the fully-packed case.

\section{Discussion}
In this paper, we constructed the ghost algebra $\Gh_n$ and the dilute ghost algebra $\dGh_n$. These are two-boundary generalisations of the TL and dilute TL algebras, respectively, with three important properties. First, the algebras are associative. Second, their diagram bases include diagrams with an odd number of strings connected to one or both boundaries. Third, associated with each boundary, they have a unital or non-unital subalgebra, respectively, that is isomorphic to the one-boundary TL algebra with two distinct boundary parameters. 

In Section \ref{s:intro}, we considered the algebra obtained by allowing diagrams with an odd number of strings connected to each boundary, and applying the multiplication rules from $\BBTL_n$. It was shown that such an algebra cannot have all three of our desired properties simultaneously. This was resolved by introducing bookkeeping devices called ghosts on the boundaries of our diagrams. The parity of the string endpoints on each boundary is now determined by counting ghosts as well as strings, and boundary arcs leave behind a ghost at each endpoint when removed. Associativity is then ensured by requiring the sum of the number of strings and ghosts on each boundary to be even. 

We then introduced lattice loop models associated with the ghost algebra and the dilute ghost algebra. The bulk of the lattices is the same as the TL and dilute TL lattice models, respectively, so we used existing bulk face operators from these models that are known to satisfy the YBE, local inversion relations, and have crossing symmetry. The boundary face operators, however, needed to incorporate the new ghosts. Given the bulk face operators, we classified all boundary face operators satisfying the BYBEs for generic values of the algebra parameters. From these face operators, we built commuting families of transfer tangles, using the construction from \cite{BehrendPearceOBrien}. Some of our BYBE solutions for the ghost and dilute ghost algebras have multiple degrees of freedom, and potentially nonzero coefficients on all possible boundary triangles. It would thus be interesting to compute the Hamiltonians arising from these transfer tangles, and study their spectra, to learn whether these new boundary conditions lead to new physical behaviour.

The ghost algebra and the dilute ghost algebra are constructed to be cellular with respect to the anti-involution given by reflecting basis diagrams about a vertical line. Graham and Lehrer's paper \cite{GL} introduces the definition of cellular algebras, and also proves a number of results about their representation theory, provided they are finite-dimensional and defined over a field. In particular, the classification of finite-dimensional irreducible modules amounts to determining whether the discriminants of certain bilinear forms are zero. The bilinear forms in question are defined on \textit{cell modules}; for $\Gh_n$ and $\dGh_n$, these modules are spanned by the $\Gh_n$- and $\dGh_n$-half-diagrams, respectively, with a fixed number of defects, as defined in Appendix \ref{app:dims}. The action is defined similarly to the usual diagram multiplication. Computing these discriminants for all $n$ and all defect numbers is highly non-trivial, given how fast the dimensions of these modules grow, but the author hopes to achieve it in future work.

Finally, in Appendix \ref{app:nc}, we considered generalisations of the ghost algebra and the dilute ghost algebra which are no longer cellular with respect to the reflection anti-involution. These are still associative and have parity-dependent boundary parameters, but can be further generalised. Indeed, the only purpose of the ghosts in our diagrams is to keep track of which strings should be numbered even or odd during multiplication. Thus, each string endpoint on each boundary could be labelled odd or even, and then we could remove the ghosts entirely. Similar algebras can then be constructed with any number of labels that may be applied to each boundary connection, and parameters assigned to boundary arcs with each possible pairing of labels. One could then find necessary relations on the parameters for these algebras to be cellular with respect to different anti-involutions, or determine whether their more general boundary conditions lead to even more interesting structure in the corresponding lattice models.

Another natural extension of this work would be to find a set of generators and relations for the ghost algebra and the dilute ghost algebra. We believe that the ghost algebra is generated by the usual TL generators $e_i$, $i=1, \dots, n-1$, and the boundary diagrams
\begin{align}
\begin{tikzpicture}[baseline={([yshift=-1mm]current bounding box.center)},scale=0.45]
{
\draw[dotted] (0,0.5)--(3,0.5);
\draw[dotted] (0,-5.5)--(3,-5.5);
\draw [very thick] (0,-5.5)--(0,0.5);
\draw [very thick] (3,-5.5)--(3,0.5);
\draw (0,0) to[out=0,in=-90] (1,0.5);
\draw (3,0) to[out=180,in=270] (2,0.5);
\draw (0,-1)--(3,-1);
\draw (0,-5)--(3,-5);
\node at (1.5,-2.8) [anchor=center] {$\vdots$};
}
\end{tikzpicture}\; ,
&&
\begin{tikzpicture}[baseline={([yshift=-1mm]current bounding box.center)},xscale=0.45,yscale=-0.45]
{
\draw[dotted] (0,0.5)--(3,0.5);
\draw[dotted] (0,-5.5)--(3,-5.5);
\draw [very thick] (0,-5.5)--(0,0.5);
\draw [very thick] (3,-5.5)--(3,0.5);
\draw (0,0) to[out=0,in=-90] (1,0.5);
\draw (3,0) to[out=180,in=270] (2,0.5);
\draw (0,-1)--(3,-1);
\draw (0,-5)--(3,-5);
\node at (1.5,-3.3) [anchor=center] {$\vdots$};
}
\end{tikzpicture}\; ,
&&
\begin{tikzpicture}[baseline={([yshift=-1mm]current bounding box.center)},xscale=0.45,yscale=0.45]
{
\draw[dotted] (0,0.5)--(3,0.5);
\draw[dotted] (0,-5.5)--(3,-5.5);
\draw [very thick] (0,-5.5)--(0,0.5);
\draw [very thick] (3,-5.5)--(3,0.5);
\draw (0,0) to[out=0,in=-90] (1,0.5);
\draw (0,-1) to[out=0,in=180] (3,0);
\draw (0,-5) to[out=0,in=180] (3,-4);
\draw (3,-5) to[out=180,in=90] (2,-5.5);
\node at (1.5,-2.3) [anchor=center] {$\vdots$};
\filldraw (2,0.5) circle (0.08);
\filldraw (1,-5.5) circle (0.08);
}
\end{tikzpicture}\; ,
&&
\begin{tikzpicture}[baseline={([yshift=-1mm]current bounding box.center)},xscale=-0.45,yscale=0.45]
{
\draw[dotted] (0,0.5)--(3,0.5);
\draw[dotted] (0,-5.5)--(3,-5.5);
\draw [very thick] (0,-5.5)--(0,0.5);
\draw [very thick] (3,-5.5)--(3,0.5);
\draw (0,0) to[out=0,in=-90] (1,0.5);
\draw (0,-1) to[out=0,in=180] (3,0);
\draw (0,-5) to[out=0,in=180] (3,-4);
\draw (3,-5) to[out=180,in=90] (2,-5.5);
\node at (1.5,-2.3) [anchor=center] {$\vdots$};
\filldraw (2,0.5) circle (0.08);
\filldraw (1,-5.5) circle (0.08);
}
\end{tikzpicture}\; ,
\end{align}
with each of the four possible arrangements of ghosts. Some of the techniques used to find a presentation for the symplectic blob algebra in \cite{SympBlobPres} could be used for the ghost algebra, though the generators of the last two types above present an additional complication. A presentation for the dilute ghost algebra would be much more tedious, given the difficulty of obtaining a presentation for the dilute TL algebra alone.

\subsection*{Acknowledgements}
The author is supported by the Australian Government Research Training Program Stipend and Tuition Fee Offset. The author would like to thank J\o{}rgen Rasmussen and Jon Links for their helpful feedback on drafts of this paper.

\appendix
\section{Dimension proofs} \label{app:dims}
In this section, we find the dimensions of the algebras $\Gh_n$, $\dGh_n$, $\Gho_n$ and $\dGho_n$ introduced in this paper. Table \ref{tab:ghostdims} lists these dimensions for small $n$. 

We first introduce some terminology. A string attached to both sides of the rectangle in a diagram is called a \textit{throughline}. Suppose we cut all of the throughlines in a $\Gh_n$-diagram and discard the right-hand side, along with any strings connected to it. If we ensure that both boundaries still have an even number of strings and ghosts by adding or removing a ghost on the far right as necessary, the result is called a \textit{$\Gh_n$-half-diagram}. For example, cutting the $\Gh_6$-diagrams 
\begin{align}
\begin{tikzpicture}[baseline={([yshift=-1mm]current bounding box.center)},xscale=0.5,yscale=0.5]
{
\draw[dotted] (0,0.5)--(3,0.5);
\draw[dotted] (0,-5.5)--(3,-5.5);
\draw [very thick](0,-5.5) -- (0,0.5);
\draw [very thick](3,-5.5)--(3,0.5);
\draw (0,0) arc (-90:0:0.5);
\draw (0,-1) arc (90:-90:0.5);
\draw (0,-3) to[out=0,in=-90] (1.3,0.5);
\draw (0,-4) to[out=0,in=180] (3,-1);
\draw (0,-5) to[out=0,in=180] (3,-2);
\draw (3,0) arc (270:180:0.5);
\draw (3,-3) to[out=180,in=90] (1.7,-5.5);
\draw (3,-4) arc (90:270:0.5);
\filldraw (0.9,0.5) circle (0.08);
\filldraw (0.85,-5.5) circle (0.08);
}
\end{tikzpicture}\; ,
\quad\ \ 
\begin{tikzpicture}[baseline={([yshift=-1mm]current bounding box.center)},xscale=0.5,yscale=0.5]
{
\draw[dotted] (0,0.5)--(3,0.5);
\draw[dotted] (0,-5.5)--(3,-5.5);
\draw [very thick](0,-5.5) -- (0,0.5);
\draw [very thick](3,-5.5)--(3,0.5);
\draw (0,0) to[out=0,in=180] (3,-4);
\draw (0,-1) to[out=0,in=0] (0,-4);
\draw (0,-2) arc (90:-90:0.5);
\draw (0,-5) -- (3,-5);
\draw (3,0) arc (90:270:0.5);
\draw (3,-2) arc (90:270:0.5);
}
\end{tikzpicture}\; ,
\quad\ \ 
\begin{tikzpicture}[baseline={([yshift=-1mm]current bounding box.center)},xscale=0.5,yscale=0.5]
{
\draw[dotted] (0,0.5)--(3,0.5);
\draw[dotted] (0,-5.5)--(3,-5.5);
\draw [very thick](0,-5.5) -- (0,0.5);
\draw [very thick](3,-5.5)--(3,0.5);
\draw (0,0) arc (90:0:1.65 and 5.5);
\draw (0,-1) arc (90:0:1.375 and 4.5);
\draw (0,-2) arc (90:0:1.1 and 3.5);
\draw (0,-3) arc (90:0:0.825 and 2.5);
\draw (0,-4) arc (90:0:0.55 and 1.5);
\draw (0,-5) arc (90:0:0.27 and 0.5);
\draw (3,-5) arc (270:180:1.65 and 5.5);
\draw (3,-4) arc (270:180:1.375 and 4.5);
\draw (3,-3) arc (270:180:1.1 and 3.5);
\draw (3,-2) arc (270:180:0.825 and 2.5);
\draw (3,-1) arc (270:180:0.55 and 1.5);
\draw (3,-0) arc (270:180:0.27 and 0.5);
}
\end{tikzpicture}\; ,
\quad\ \ 
\begin{tikzpicture}[baseline={([yshift=-1mm]current bounding box.center)},xscale=0.5,yscale=0.5]
{
\draw[dotted] (0,0.5)--(3,0.5);
\draw[dotted] (0,-5.5)--(3,-5.5);
\draw [very thick](0,-5.5) -- (0,0.5);
\draw [very thick](3,-5.5)--(3,0.5);
\draw (0,0) arc (-90:0:0.8 and 0.5);
\draw (0,-1) to[out=0,in=180] (3,0);
\draw (0,-2) to[out=0,in=180] (3,-1);
\draw (0,-3) to[out=0,in=180] (3,-2);
\draw (0,-4) to[out=0,in=180] (3,-3);
\draw (0,-5) to[out=0,in=180] (3,-4);
\draw (3,-5) arc (90:180:0.8 and 0.5);
\filldraw (1.9,0.5) circle (0.08);
\filldraw (1.1,-5.5) circle (0.08);
}
\end{tikzpicture}
\end{align}
in half produces the $\Gh_6$-half-diagrams
\begin{align}
\begin{tikzpicture}[baseline={([yshift=-1mm]current bounding box.center)},xscale=0.5,yscale=0.5]
{
\draw[dotted] (0,0.5)--(2,0.5);
\draw[dotted] (0,-5.5)--(2,-5.5);
\draw [very thick](0,-5.5) -- (0,0.5);
\draw (0,0) arc (-90:0:0.5);
\draw (0,-1) arc (90:-90:0.5);
\draw (0,-3) to[out=0,in=-90] (1.3,0.5);
\draw (0,-4) -- (2,-4);
\draw (0,-5) -- (2,-5);
\filldraw (0.9,0.5) circle (0.08);
\filldraw (1.7,0.5) circle (0.08);
}
\end{tikzpicture}\; ,
\qquad
\begin{tikzpicture}[baseline={([yshift=-1mm]current bounding box.center)},xscale=0.5,yscale=0.5]
{
\draw[dotted] (0,0.5)--(2,0.5);
\draw[dotted] (0,-5.5)--(2,-5.5);
\draw [very thick](0,-5.5) -- (0,0.5);
\draw (0,0) -- (2,0);
\draw (0,-1) to[out=0,in=0] (0,-4);
\draw (0,-2) arc (90:-90:0.5);
\draw (0,-5) -- (2,-5);
}
\end{tikzpicture}\; ,
\qquad
\begin{tikzpicture}[baseline={([yshift=-1mm]current bounding box.center)},xscale=0.5,yscale=0.5]
{
\draw[dotted] (0,0.5)--(2,0.5);
\draw[dotted] (0,-5.5)--(2,-5.5);
\draw [very thick](0,-5.5) -- (0,0.5);
\draw (0,0) arc (90:0:1.65 and 5.5);
\draw (0,-1) arc (90:0:1.375 and 4.5);
\draw (0,-2) arc (90:0:1.1 and 3.5);
\draw (0,-3) arc (90:0:0.825 and 2.5);
\draw (0,-4) arc (90:0:0.55 and 1.5);
\draw (0,-5) arc (90:0:0.27 and 0.5);
}
\end{tikzpicture}\; ,
\qquad
\begin{tikzpicture}[baseline={([yshift=-1mm]current bounding box.center)},xscale=0.5,yscale=0.5]
{
\draw[dotted] (0,0.5)--(2,0.5);
\draw[dotted] (0,-5.5)--(2,-5.5);
\draw [very thick](0,-5.5) -- (0,0.5);
\draw (0,0) arc (-90:0:0.8 and 0.5);
\draw (0,-1) -- (2,-1);
\draw (0,-2) -- (2,-2);
\draw (0,-3) -- (2,-3);
\draw (0,-4) -- (2,-4);
\draw (0,-5) -- (2,-5);
\filldraw (1.4,0.5) circle (0.08);
}
\end{tikzpicture}\; ,
\label{eq:eghalfdiags}
\end{align}
respectively. The cut throughlines in these half-diagrams are called \textit{defects}. A string connecting a node to another node is called a \textit{link}, and a string connecting a node to either boundary is called a \textit{boundary link}. Note that with these definitions, a boundary link is not a link.

We could also make a $\Gh_n$-half-diagram by cutting the throughlines and discarding the left half of the rectangle and its strings instead, reflecting the result across a vertical line, and adding or removing ghosts in the rightmost domains as necessary. Hence each $\Gh_n$-diagram with $d$ throughlines can be cut to produce a pair of $\Gh_n$-half-diagrams with $d$ defects. It is clear that the resulting map, from the set of $\Gh_n$-diagrams with $d$ throughlines, to the set of pairs of $\Gh_n$-half-diagrams with $d$ defects, is injective.

Any pair of $\Gh_n$-half-diagrams with $d$ defects can be put together to make a $\Gh_n$-diagram with $d$ throughlines. Given $\Gh_n$-half-diagrams $x$ and $y$, each with $d$ defects, reflect $y$ about a vertical line, and connect its defects to those of $x$ from top to bottom. At each boundary, the rightmost domain of $x$ has been joined to the leftmost domain of the reflected $y$, so the number of ghosts in the resulting domain is the sum of the number of ghosts in each of those domains, simplified modulo 2. Denote the result by $\out{x}{y}$. This gives an injective map $\out{\cdot}{\cdot}$ from the set of pairs of $\Gh_n$-half-diagrams with $d$ defects, to the set of $\Gh_n$-diagrams with $d$ throughlines, which is the inverse of the above map, and so these sets are in bijection.

For example, with $(n,d) = (6,3)$, we could have
\begin{align}
x\,=\ \begin{tikzpicture}[baseline={([yshift=-1mm]current bounding box.center)},xscale=0.5,yscale=0.5]
{
\draw[dotted] (0,0.5)--(2,0.5);
\draw[dotted] (0,-5.5)--(2,-5.5);
\draw [very thick](0,-5.5) -- (0,0.5);
\draw (0,0) arc (-90:0:0.6 and 0.5);
\draw (0,-1) to[out=0,in=-90] (1.2,0.5);
\draw (0,-2) -- (2,-2);
\draw (0,-3) -- (2,-3);
\draw (0,-4) -- (2,-4);
\draw (0,-5) arc (90:0:0.6 and 0.5);
\filldraw (0.3,0.5) circle (0.08);
\filldraw (1.6,0.5) circle (0.08);
\filldraw (1.3,-5.5) circle (0.08);
}
\end{tikzpicture}\; ,
\qquad \qquad
y\,=\ \begin{tikzpicture}[baseline={([yshift=-1mm]current bounding box.center)},xscale=0.5,yscale=0.5]
{
\draw[dotted] (0,0.5)--(2,0.5);
\draw[dotted] (0,-5.5)--(2,-5.5);
\draw [very thick](0,-5.5) -- (0,0.5);
\draw (0,0) arc (-90:0:0.6 and 0.5);
\draw (0,-1) -- (2,-1);
\draw (0,-2) -- (2,-2);
\draw (0,-3) arc (90:-90:0.5);
\draw (0,-5) -- (2,-5);
\filldraw (1.3,0.5) circle (0.08);
}
\end{tikzpicture}\; ,
\qquad \qquad 
\out{x}{y}\, =\ \begin{tikzpicture}[baseline={([yshift=-1mm]current bounding box.center)},xscale=0.5,yscale=0.5]
{
\draw[dotted] (0,0.5)--(3,0.5);
\draw[dotted] (0,-5.5)--(3,-5.5);
\draw [very thick](0,-5.5) -- (0,0.5);
\draw [very thick](3,-5.5) -- (3,0.5);
\draw (0,0) arc (-90:0:0.6 and 0.5);
\draw (0,-1) to[out=0,in=-90] (1.2,0.5);
\draw (0,-2) to[out=0,in=180] (3,-1);
\draw (0,-3) to[out=0,in=180] (3,-2);
\draw (0,-4) to[out=0,in=180] (3,-5);
\draw (0,-5) arc (90:0:0.6 and 0.5);
\draw (3,0) arc (270:180:0.6 and 0.5);
\draw (3,-3) arc (90:270:0.5);
\filldraw (0.3,0.5) circle (0.08);
\filldraw (1.6,0.5) circle (0.08);
\filldraw (2,0.5) circle (0.08);
\filldraw (1.8,-5.5) circle (0.08);
}
\end{tikzpicture}
\ =\ \begin{tikzpicture}[baseline={([yshift=-1mm]current bounding box.center)},xscale=0.5,yscale=0.5]
{
\draw[dotted] (0,0.5)--(3,0.5);
\draw[dotted] (0,-5.5)--(3,-5.5);
\draw [very thick](0,-5.5) -- (0,0.5);
\draw [very thick](3,-5.5) -- (3,0.5);
\draw (0,0) arc (-90:0:0.6 and 0.5);
\draw (0,-1) to[out=0,in=-90] (1.2,0.5);
\draw (0,-2) to[out=0,in=180] (3,-1);
\draw (0,-3) to[out=0,in=180] (3,-2);
\draw (0,-4) to[out=0,in=180] (3,-5);
\draw (0,-5) arc (90:0:0.6 and 0.5);
\draw (3,0) arc (270:180:0.6 and 0.5);
\draw (3,-3) arc (90:270:0.5);
\filldraw (0.3,0.5) circle (0.08);
\filldraw (1.8,-5.5) circle (0.08);
}
\end{tikzpicture}\; .
\label{eq:cutglueex}
\end{align}

The idea of cutting diagrams with $d$ throughlines in half to produce half-diagrams with $d$ defects can also be used for $\TL_n$, $\BTL_n$, $\Gho_n$, $\dGh_n$ and $\dGho_n$, and in each case, the map from pairs of half-diagrams with $d$ defects to diagrams with $d$ throughlines is injective.

\begin{proposition}\label{prop:ghostdim}
For each $n \in \N$, 
\begin{align}
\dim\Gh_n = \sum_{d=0}^n \left(\sum_{j=0}^\floor{\frac{n-d}{2}} 2^{n-2j-d}(n-2j-d+1)\left(\binom{n}{j}-\binom{n}{j-1}\right) \right)^2. \label{eq:ghostdim}
\end{align}
\end{proposition}

\begin{proof}
Recall (e.g. from \cite{RSA}) that the number of distinct $\TL_n$-half-diagrams (i.e. $\Gh_n$-half-diagrams with no boundary links) with $j$ links is equal to
\begin{align}
\binom{n}{j} - \binom{n}{j-1}, \label{eq:TLstddim}
\end{align}
where $0 \leq j \leq \floor{\frac{n-d}{2}}$. The number of defects in such a half-diagram is then $n-2j$. We can then convert this $\TL_n$-half-diagram to a $\Gh_n$-half-diagram with $d$ defects as follows. Fix $k$ such that $0 \leq k \leq n-2j-d$, and connect the top $k$ defects to the top boundary, and the bottom $n-2j-d-k$ defects to the bottom boundary. We then have a total of $n-2j-d$ boundary links, which means we have $n-2j-d+2$ domains. Each domain can be empty or have a ghost, but we must have an even total number of strings and ghosts on each boundary. Hence we can freely choose whether to put a ghost in all but one of the domains on each boundary, but the last one is determined by parity. This means there are $2^{n-2j-d}$ ways to populate the $n-2j-d$ free domains with ghosts, and thus $2^{n-2j-d}\left(\binom{n}{j}-\binom{n}{j-1}\right)$ $\Gh_n$-half-diagrams with $d$ defects and $n-2j$ boundary links, of which $k$ are top boundary links.

Observing that this is independent of $k$ and summing over $0 \leq k \leq n-2j-d$ for the number of top boundary links, it follows that the number of $\Gh_n$-half diagrams with $d$ defects and $j$ links is
\begin{align}
2^{n-2j-d}(n-2j-d+1)\left(\binom{n}{j}-\binom{n}{j-1}\right). \label{eq:ghostdimj}
\end{align}
Summing over the number of links, $j$, with $0 \leq j \leq \floor{\frac{n-d}{2}}$, we find that number of $\Gh_n$-half-diagrams with $d$ defects is
\begin{align}
\sum_{j=0}^\floor{\frac{n-d}{2}}2^{n-2j-d}(n-2j-d+1)\left(\binom{n}{j}-\binom{n}{j-1}\right).\label{eq:ghoststddimprop}
\end{align}
Then, since we know each $\Gh_n$-diagram with $d$ throughlines can be constructed from a unique pair of $\Gh_n$-half-diagrams with $d$ defects, we can square this to find the number of pairs of such half-diagrams, and sum over $d$ to find the number of $\Gh_n$-diagrams.
\end{proof}

\begin{proposition} \label{prop:dimdghost}
For $n \in \N$,
\begin{align}
   \dim\dGh_n = \sum_{d=0}^n \left(\sum_{v=0}^{n-d} \binom{n}{v} \sum_{j=0}^\floor{\frac{n-v-d}{2}} 2^{n-v-2j-d}\left(n-v-2j-d+1\right) \left(\binom{n-v}{j} - \binom{n-v}{j-1}\right)\right)^2.
\end{align}
\end{proposition}

\begin{proof}
From Proposition \ref{prop:ghostdim}, the number of $\Gh_n$-half-diagrams with $d$ defects is
\begin{align}
    \sum_{j=0}^\floor{\frac{n-d}{2}} 2^{n-2j-d}\left(n-2j-d+1\right) \left(\binom{n}{j} - \binom{n}{j-1}\right).
\end{align}
Each $\dGh_n$-half-diagram with $d$ defects can be constructed by selecting $v$ of $n$ nodes to be empty, and drawing a $\Gh_n$-half-diagram with $d$ defects on the remaining $n-v$ nodes. There are $\binom{n}{v}$ ways to choose the empty nodes, for $0\leq v \leq n-d$, and thus
\begin{align}
    \sum_{v=0}^{n-d} \binom{n}{v} \sum_{j=0}^\floor{\frac{n-v-d}{2}} 2^{n-v-2j-d}\left(n-v-2j-d+1\right) \left(\binom{n-v}{j} - \binom{n-v}{j-1}\right)
\end{align}
distinct $\dGh_n$-half-diagrams with $d$ defects.

Each $\dGh_n$-diagram with $d$ throughlines can be constructed from a unique pair of $\dGh_n$-half-diagrams with $d$ defects, so the number of $\dGh_n$-diagrams is
\begin{align}
    \sum_{d=0}^n \left(\sum_{v=0}^{n-d} \binom{n}{v} \sum_{j=0}^\floor{\frac{n-v-d}{2}} 2^{n-v-2j-d}\left(n-v-2j-d+1\right) \left(\binom{n-v}{j} - \binom{n-v}{j-1}\right)\right)^2.
\end{align}
\end{proof}

\begin{proposition}\label{prop:dimobghost}
For $n \in \N$,
\begin{align}
    \dim\Gho_n = \sum_{d=0}^n \left(\sum_{j=0}^\floor{\frac{n-d}{2}} 2^{n-2j-d} \left(\binom{n}{j} - \binom{n}{j-1}\right)\right)^2.
\end{align}
\end{proposition}

\begin{proof}
The argument is analogous to Proposition \ref{prop:ghostdim}, except that defects in the $\TL_n$-half-diagrams can now only be connected to the top boundary. This means the factor of $(n-2j-d+1)$, associated with choosing how many defects connect to each boundary, disappears.
\end{proof}

\begin{proposition}\label{prop:dimdobghost}
For $n \in \N$,
\begin{align}
    \dim\dGho_n = \sum_{d=0}^n \left(\sum_{v=0}^{n-d} \binom{n}{v} \sum_{j=0}^\floor{\frac{n-v-d}{2}} 2^{n-v-2j-d} \left(\binom{n-v}{j} - \binom{n-v}{j-1}\right)\right)^2.
\end{align}
\end{proposition}

\begin{proof}
The argument is analogous to Proposition \ref{prop:dimdghost}, except we use $\Gho_{n-v}$-half-diagrams with $d$ defects, instead of $\Gh_{n-v}$-half-diagrams with $d$ defects. The number of such half-diagrams is the sum over $j$ in the expression in Proposition \ref{prop:dimobghost}, but with $n-v$ in place of $n$.
\end{proof}

\begin{table}[H]
\centering
\caption{Table of dimensions of $\Gho_n$, $\Gh_n$, $\dGho_n$ and $\dGh_n$ for small $n$.}\label{tab:ghostdims}
\vspace{2mm}
\begin{tabular}{|l|r|r|r|r|}
\hline
$n$  & $\dim\Gho_n$  & $\dim\Gh_n$    & $\dim\dGho_n$   & $\dim\dGh_n$     \\ \hline
1  & 5        & 17         & 10          & 26            \\ \hline
2  & 30       & 186        & 117         & 521           \\ \hline
3  & 185      & 1,813       & 1,407        & 9,355          \\ \hline
4  & 1,150     & 16,102      & 17,083       & 156,947        \\ \hline
5  & 7,170     & 135,866     & 208,284      & 2,514,932       \\ \hline
6  & 44,760    & 1,099,276    & 2,544,751     & 38,968,815      \\ \hline
7  & 279,585   & 8,639,133    & 31,125,138    & 588,475,298     \\ \hline
8  & 1,746,870  & 66,258,526   & 380,928,795   & 8,706,799,523    \\ \hline
9  & 10,916,150 & 498,701,470  & 4,663,705,782  & 126,690,947,758  \\ \hline
10 & 68,219,860 & 3,693,607,300 & 57,109,857,519 & 1,818,028,127,339 \\ \hline
\end{tabular}
\end{table}

\section{Associativity} \label{app:assoc}
In this section, we prove that the dilute ghost algebra $\dGh_n$ is associative. Since $\dGho_n$, $\Gh_n$ and $\Gho_n$ are (unital, non-unital and non-unital) subalgebras of $\dGh_n$, it follows that they are also associative. We start by introducing some terminology to describe diagram multiplication more formally, and then prove associativity in Theorem \ref{thm:assoc} and Corollary \ref{cor:assoc}.

A \textit{$\dGh_n$-pseudo-diagram} consists of finitely many non-crossing strings drawn within the usual two-boundary rectangle of $2n$ nodes, and finitely many ghosts on each boundary, subject to the following requirements. Each node must have at most one string endpoint attached to it. Each string endpoint can be attached to a node or a boundary, or be left unattached. Strings may also form loops. Ghosts cannot be drawn on the ends of strings, and the number of ghosts plus the number of string endpoints attached to each boundary must be even. Two $\dGh_n$-pseudo-diagrams are considered equal if their strings are equivalent under a continuous planar deformation that preserves the rectangle, up to sliding but not reordering ghosts and string endpoints, and their corresponding domains have the same number of ghosts, modulo 2.

For example,
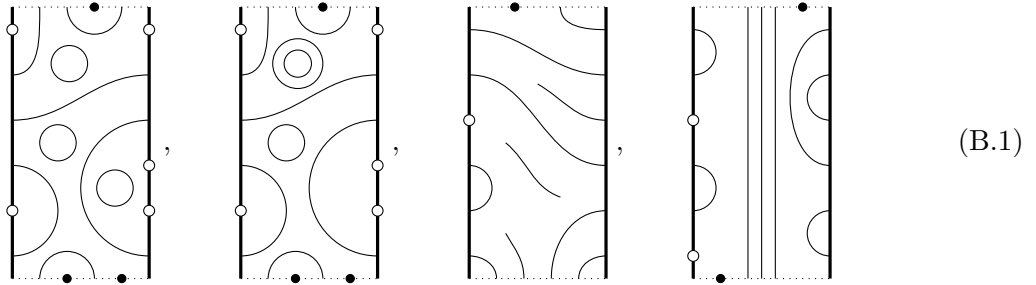
\begin{align}
\begin{tikzpicture}[baseline={([yshift=-1mm]current bounding box.center)},xscale=0.6,yscale=0.6]
{
\draw[very thick] (0,0.5)--(0,-5.5);
\draw[very thick] (3,0.5)--(3,-5.5);
\draw[dotted] (0,0.5)--(3,0.5);
\draw[dotted] (0,-5.5)--(3,-5.5);
\draw (0,-1) to[out=0,in=-90] (0.6,0.5);
\draw (0,-2) to[out=0,in=180] (3,-1);
\draw (0,-3) arc (90:-90:1);
\draw (3,-2) arc (90:270:1.5);
\draw (1.2,0.5) arc(180:360:0.6);
\draw (0.6,-5.5) arc(180:0:0.6);
\draw (1,-2.5) circle (0.4);
\draw (2.25,-3.5) circle (0.4);
\draw (1.25,-0.75) circle (0.4);
\filldraw[fill=white] (0,0) circle (0.12);
\filldraw[fill=white] (0,-4) circle (0.12);
\filldraw[fill=white] (3,0) circle (0.12);
\filldraw[fill=white] (3,-3) circle (0.12);
\filldraw[fill=white] (3,-4) circle (0.12);
\filldraw (1.8,0.5) circle (0.09);
\filldraw (1.2,-5.5) circle (0.09);
\filldraw (2.4,-5.5) circle (0.09);
}
\end{tikzpicture}
\; , \qquad
\begin{tikzpicture}[baseline={([yshift=-1mm]current bounding box.center)},xscale=0.6,yscale=0.6]
{
\draw[very thick] (0,0.5)--(0,-5.5);
\draw[very thick] (3,0.5)--(3,-5.5);
\draw[dotted] (0,0.5)--(3,0.5);
\draw[dotted] (0,-5.5)--(3,-5.5);
\draw (0,-1) to[out=0,in=-90] (0.6,0.5);
\draw (0,-2) to[out=0,in=180] (3,-1);
\draw (0,-3) arc (90:-90:1);
\draw (3,-2) arc (90:270:1.5);
\draw (1.2,0.5) arc(180:360:0.6);
\draw (0.6,-5.5) arc(180:0:0.6);
\draw (1,-2.5) circle (0.4);
\draw (1.25,-0.75) circle (0.55);
\draw (1.25,-0.75) circle (0.3);
\filldraw[fill=white] (0,0) circle (0.12);
\filldraw[fill=white] (0,-4) circle (0.12);
\filldraw[fill=white] (3,0) circle (0.12);
\filldraw[fill=white] (3,-3) circle (0.12);
\filldraw[fill=white] (3,-4) circle (0.12);
\filldraw (1.8,0.5) circle (0.09);
\filldraw (1.2,-5.5) circle (0.09);
\filldraw (2.4,-5.5) circle (0.09);
}
\end{tikzpicture}
\; , \qquad
\begin{tikzpicture}[baseline={([yshift=-1mm]current bounding box.center)},xscale=0.6,yscale=0.6]
{
\draw[white] (1,-5.5) circle (0.09);
\draw[very thick] (0,0.5)--(0,-5.5);
\draw[very thick] (3,0.5)--(3,-5.5);
\draw[dotted] (0,0.5)--(3,0.5);
\draw[dotted] (0,-5.5)--(3,-5.5);
\draw (0,0) to[out=0,in=180] (3,-1);
\draw (0,-1) to[out=0,in=180] (3,-3);
\draw (0,-3) arc (90:-90:0.5);
\draw (0,-5) to[out=0,in=90] (0.6,-5.5);
\draw (3,0) to[out=180,in=270] (2,0.5);
\draw (3,-2) to[out=180,in=330] (1.5,-1.2);
\draw (3,-4) to[out=180,in=90] (1.8,-5.5);
\draw (3,-5) to[out=180,in=90] (2.4,-5.5);
\draw (1.2,-5.5) to[out=90,in=300] (0.8,-4.5);
\draw (0.8,-2.5) to[out=-40,in=160] (2,-3.7);
\filldraw[fill=white] (0,-2) circle (0.12);
\filldraw (1,0.5) circle (0.09);
}
\end{tikzpicture}
\; , \qquad
\begin{tikzpicture}[baseline={([yshift=-1mm]current bounding box.center)},xscale=0.6,yscale=0.6]
{
\draw[very thick] (0,0.5)--(0,-5.5);
\draw[very thick] (3,0.5)--(3,-5.5);
\draw[dotted] (0,0.5)--(3,0.5);
\draw[dotted] (0,-5.5)--(3,-5.5);
\draw (0,0) arc(90:-90:0.5);
\draw (0,-3) arc (90:-90:0.5);
\draw (3,0) to[out=180,in=180] (3,-3);
\draw (3,-1) arc(90:270:0.5);
\draw (3,-4) arc(90:270:0.5);
\draw (1.2,0.5)--(1.2,-5.5);
\draw (1.5,0.5)--(1.5,-5.5);
\draw (1.8,0.5)--(1.8,-5.5);
\filldraw[fill=white] (0,-2) circle (0.12);
\filldraw[fill=white] (0,-5) circle (0.12);
\filldraw (2.4,0.5) circle (0.09);
\filldraw (0.6,-5.5) circle (0.09);
}
\end{tikzpicture}
\end{align}
are all $\dGh_6$-pseudo-diagrams. Note that the first two are not equal because planar deformations do not allow us to pull a loop over other strings.

Since the number of ghosts in each domain is only considered modulo 2, each $\dGh_n$-pseudo-diagram can be drawn with at most one ghost in each domain. Further, we can number the string endpoints and ghosts along each boundary, left-to-right, starting from 1. If two drawings of $\dGh_n$-pseudo-diagrams are equal, then each string endpoint on the boundary might be assigned a different number in each drawing, but the number of string endpoints to its left is the same, and the number of ghosts in each domain to its left is the same modulo 2. Hence each string endpoint on the boundary has a well-defined parity.

Let $\circ$ be the binary operation of concatenation on $\dGh_n$-pseudo-diagrams, where the overlapping middle line from the two rectangles is removed. It is clear that this is associative. For example, if
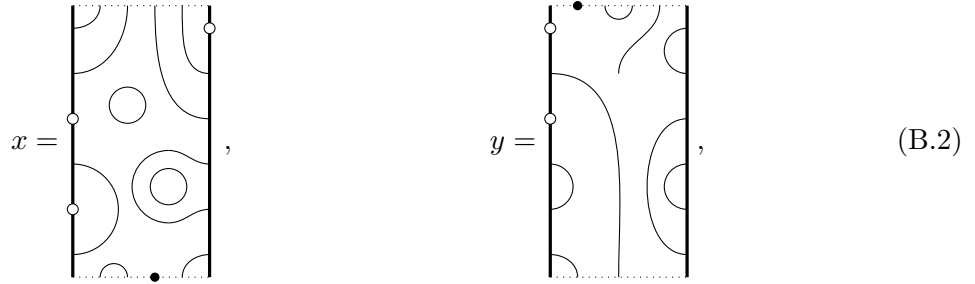
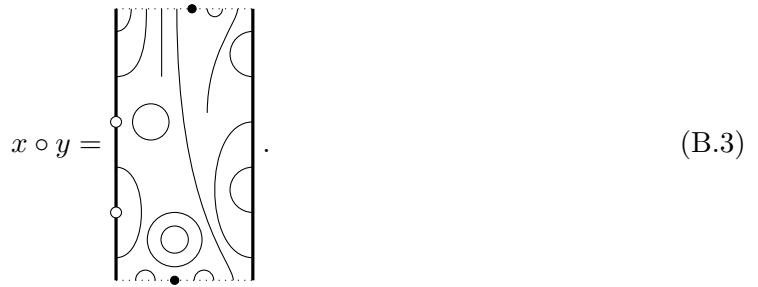
\begin{align}
x&=\begin{tikzpicture}[baseline={([yshift=-1mm]current bounding box.center)},xscale=0.6,yscale=0.6]
{
\draw[white] (1,0.5) circle (0.09);
\draw[very thick] (0,0.5)--(0,-5.5);
\draw[very thick] (3,0.5)--(3,-5.5);
\draw[dotted] (0,0.5)--(3,0.5);
\draw[dotted] (0,-5.5)--(3,-5.5);
\draw (0,0) to[out=0,in=-90] (0.6,0.5);
\draw (0,-1) to[out=0,in=-90] (1.2,0.5);
\draw (0,-3) arc (90:-90:1);
\draw (3,-1) to[out=180,in=270] (2.4,0.5);
\draw (3,-2) to[out=180,in=270] (1.8,0.5);
\draw (3,-5) to[out=180,in=90] (2.4,-5.5);
\draw (1.2,-1.7) circle (0.4);
\draw (2.1,-3.5) circle (0.4);
\draw (2.1,-2.7) arc(90:270:0.8);
\draw (3,-3) to[out=180,in=0] (2.1,-2.7);
\draw (3,-4) to[out=180,in=0] (2.1,-4.3);
\draw (0.6,-5.5) arc(180:0:0.3);
\filldraw[fill=white] (0,-2) circle (0.12);
\filldraw[fill=white] (0,-4) circle (0.12);
\filldraw[fill=white] (3,0) circle (0.12);
\filldraw (1.8,-5.5) circle (0.09);
}
\end{tikzpicture}
\; , &y&=
\begin{tikzpicture}[baseline={([yshift=-1mm]current bounding box.center)},xscale=0.6,yscale=0.6]
{
\draw[white] (1,-5.5) circle (0.09);
\draw[very thick] (0,0.5)--(0,-5.5);
\draw[very thick] (3,0.5)--(3,-5.5);
\draw[dotted] (0,0.5)--(3,0.5);
\draw[dotted] (0,-5.5)--(3,-5.5);
\draw (0,-1) to[out=0,in=90] (1.5,-5.5);
\draw (0,-3) arc (90:-90:0.5);
\draw (0,-5) to[out=0,in=90] (0.6,-5.5);
\draw (3,0) arc(90:270:0.5);
\draw (3,-2) to[out=180,in=180] (3,-5);
\draw (3,-3) arc(90:270:0.5);
\draw (1.2,0.5) arc(180:360:0.3);
\draw (2.4,0.5)to[out=-90,in=90] (1.5,-1);
\filldraw[fill=white] (0,0) circle (0.12);
\filldraw[fill=white] (0,-2) circle (0.12);
\filldraw (0.6,0.5) circle (0.09);
}
\end{tikzpicture}
\; ,
\end{align}
then
\begin{align}
x\circ y = \begin{tikzpicture}[baseline={([yshift=-1mm]current bounding box.center)},xscale=0.5714,yscale=0.6]
{
\draw[very thick] (0,0.5)--(0,-5.5);
\draw[very thick] (3.15,0.5)--(3.15,-5.5);
\draw[dotted] (0,0.5)--(3.15,0.5);
\draw[dotted] (0,-5.5)--(3.15,-5.5);
\draw (0,0) to[out=0,in=-90] (0.35,0.5);
\draw (0,-1) to[out=0,in=-90] (0.7,0.5);
\draw (0,-3) to[out=0,in=0] (0,-5);
\draw (0.8,-2) circle (0.42 and 0.4);
\draw (1.35,-4.6) circle (0.315 and 0.3);
\draw (1.35,-4.6) circle (0.63 and 0.6);
\draw (0.45,-5.5) arc(180:0:0.225);
\filldraw[fill=white] (0,-2) circle (0.126 and 0.12);
\filldraw[fill=white] (0,-4) circle (0.126 and 0.12);
\filldraw (1.35,-5.5) circle (0.0945 and 0.09);
\draw (3.15,0) arc(90:270:0.525 and 0.5);
\draw (3.15,-2) to[out=180,in=180] (3.15,-5);
\draw (3.15,-3) arc(90:270:0.525 and 0.5);
\draw (2.1,0.5) arc(180:360:0.175);
\draw (2.8,0.5) ..controls (2.8,0.2) and (2.1,-0.5).. (2.1,-1.8);
\filldraw (1.75,0.5) circle (0.0945 and 0.09);
\draw (1.8,-5.5) arc(180:0:0.225);
\draw (1.4,0.5) ..controls (1.4,-4) and (2.7,-5.1).. (2.7,-5.5);
\draw (1.05,0.5)--(1.05,-1);
}
\end{tikzpicture}
\; .
\end{align}
Note that $\circ$ is concatenation only, so even though a string from $x$ meets an empty node from $y$, $x \circ y$ is nonzero. 

Define a function $\chi$ on $\dGh_n$-pseudo-diagrams by
\begin{align}
\chi(x) := \prod_{\rho \,\in \{\beta,\alpha_1,\alpha_2,\alpha_3,\gamma_{12},\gamma_3,\delta_1,\delta_2,\delta_3\}} \rho^{\#(\rho,x)},
\end{align}
where $\#(\rho,x)$ is the number of strings in $x$ that would be assigned the parameter $\rho$ during multiplication in the dilute ghost algebra, as listed in Table \ref{tab:params}. For example, with $x$ and $y$ as above,
\begin{align}
\chi(x)&=\beta^2\delta_1, &\chi(y) &= \alpha_2, &\chi(x\circ y) &= \beta^3 \alpha_2\gamma_3 \delta_1\delta_2.
\end{align}

Define a function $\zeta$ on $\dGh_n$-pseudo-diagrams by $\zeta(x) = 1$ if each string endpoint in $x$ is attached to a node or boundary, and $\zeta(x)=0$ otherwise. Note that loops do not have endpoints, so $\zeta(x)=1$ does not imply that $x$ has no loops. Indeed, with $x$ and $y$ as above, $\zeta(x) = 1$, $\zeta(y)=0$ and $\zeta(x\circ y) = 0$.

For each $\dGh_n$-pseudo-diagram $x$, let $\e{x}$ be the $\dGh_n$-diagram obtained by removing all strings from $x$ except those that are attached to a node at one end, and a node or boundary at the other. A ghost must be left at the position of each removed string endpoint on the boundary, though we may choose to draw the resulting diagram with at most one ghost in each domain. With $x$ and $y$ as above,
\begin{align}
\e{x} &= \begin{tikzpicture}[baseline={([yshift=-1mm]current bounding box.center)},xscale=0.6,yscale=0.6]
{
\draw[white] (1,0.5) circle (0.09);
\draw[very thick] (0,0.5)--(0,-5.5);
\draw[very thick] (3,0.5)--(3,-5.5);
\draw[dotted] (0,0.5)--(3,0.5);
\draw[dotted] (0,-5.5)--(3,-5.5);
\draw (0,0) to[out=0,in=-90] (0.6,0.5);
\draw (0,-1) to[out=0,in=-90] (1.2,0.5);
\draw (0,-3) arc (90:-90:1);
\draw (3,-1) to[out=180,in=270] (2.4,0.5);
\draw (3,-2) to[out=180,in=270] (1.8,0.5);
\draw (3,-5) to[out=180,in=90] (2.4,-5.5);
\draw (3,-3) arc(90:270:0.5);
\filldraw[fill=white] (0,-2) circle (0.12);
\filldraw[fill=white] (0,-4) circle (0.12);
\filldraw[fill=white] (3,0) circle (0.12);
\filldraw (1.8,-5.5) circle (0.09);
}
\end{tikzpicture}
\; , &\e{y}&=\begin{tikzpicture}[baseline={([yshift=-1mm]current bounding box.center)},xscale=0.6,yscale=0.6]
{
\draw[very thick] (0,0.5)--(0,-5.5);
\draw[very thick] (3,0.5)--(3,-5.5);
\draw[dotted] (0,0.5)--(3,0.5);
\draw[dotted] (0,-5.5)--(3,-5.5);
\draw (0,-1) to[out=0,in=90] (1.5,-5.5);
\draw (0,-3) arc (90:-90:0.5);
\draw (0,-5) to[out=0,in=90] (0.6,-5.5);
\draw (3,0) arc(90:270:0.5);
\draw (3,-2) to[out=180,in=180] (3,-5);
\draw (3,-3) arc(90:270:0.5);
\filldraw[fill=white] (0,0) circle (0.12);
\filldraw[fill=white] (0,-2) circle (0.12);
}
\end{tikzpicture}
\; , &\e{x\circ y} &=\begin{tikzpicture}[baseline={([yshift=-1mm]current bounding box.center)},xscale=0.6,yscale=0.6]
{
\draw[very thick] (0,0.5)--(0,-5.5);
\draw[very thick] (3,0.5)--(3,-5.5);
\draw[dotted] (0,0.5)--(3,0.5);
\draw[dotted] (0,-5.5)--(3,-5.5);
\draw (0,0) to[out=0,in=-90] (0.6,0.5);
\draw (0,-1) to[out=0,in=-90] (1.2,0.5);
\draw (0,-3) to[out=0,in=0] (0,-5);
\filldraw[fill=white] (0,-2) circle (0.12);
\filldraw[fill=white] (0,-4) circle (0.12);
\draw (3,0) arc(90:270:0.5);
\draw (3,-2) to[out=180,in=180] (3,-5);
\draw (3,-3) arc(90:270:0.5);
}
\end{tikzpicture}
\; .
\end{align}

Then for $\dGh_n$-diagrams $x$ and $y$, we have
\begin{align}
xy = \chi(x\circ y)\, \zeta(x\circ y)\e{x\circ y}. \label{eq:assoccirc}
\end{align}

\begin{theorem}\label{thm:assoc}
The dilute ghost algebra $\dGh_n$ is associative.
\end{theorem}

\begin{proof}
It suffices to show that multiplication of $\dGh_n$-diagrams is associative. Let $x$, $y$ and $z$ be $\dGh_n$ diagrams. From \eqref{eq:assoccirc}, we have
\begin{align}
(xy)z &= \chi(x\circ y)\, \zeta(x\circ y)\e{x\circ y} z \\
&= \chi(x\circ y)\chi(\e{x\circ y} \circ z) \,\zeta(x\circ y)\zeta(\e{x\circ y} \circ z) \e{\e{x\circ y}\circ z} \\
&= \chi(x\circ y)\chi(\e{x\circ y} \circ z) \,\zeta(x\circ y)\zeta(\e{x\circ y} \circ z) \e{x\circ y\circ z},
\end{align}
where in the last line we have used the straightforward equality $\e{\e{x\circ y}\circ z}=\e{x\circ y\circ z}$. Analogously,
\begin{align}
x(yz) &= \chi(y\circ z)\chi(x\circ\e{y\circ z})\, \zeta (y \circ z) \zeta(x\circ\e{y\circ z}) \e{x\circ y \circ z}.
\end{align}
We will proceed by showing that
\begin{align}
\chi(x\circ y)\chi(\e{x\circ y} \circ z)\ =\ \chi(x\circ y \circ z) \ =\  \chi(y\circ z)\chi(x\circ\e{y\circ z}) \label{eq:chi}
\end{align}
and
\begin{align}
\zeta(x\circ y)\zeta(\e{x\circ y} \circ z) \ =\ \zeta(x\circ y \circ z) \ =\ \zeta (y \circ z) \zeta(x\circ\e{y\circ z}). \label{eq:zeta}
\end{align}

To show \eqref{eq:chi}, recall that $\chi$ is determined by the number of loops and each kind of boundary arc in its $\dGh_n$-pseudo-diagram argument. Since $x$, $y$ and $z$ are $\dGh_n$-diagrams, not just pseudo-diagrams, they have none of these features independently. There are three ways such objects can be formed from $x$, $y$ and $z$ in these products:
\begin{enumerate}
    \item from strings in $x$ and $y$, but not $z$,
    \item from strings in $y$ and $z$, but not $x$,
    \item from strings in $x$, $y$ and $z$.
\end{enumerate}
Each of these objects contributes a factor of its associated parameter to $\chi(x\circ y \circ z)$. Each object of type 1 contributes a factor of its associated parameter to $\chi(x\circ y)$ only, while each object of type 2 or 3 contributes to $\chi(\e{x\circ y}\circ z)$. Indeed, any objects of type 1 are removed by $\e{\cdot}$ before the concatenation of $\e{x\circ y}$ and $z$ occurs, and no objects of type 2 or 3 can be formed until this happens, since $z$ is not yet included. It is most important to note here that the parameter assigned to each object is the same in each case. For loops this is obvious, but for boundary arcs we need to check that the parity of each endpoint is unchanged. Indeed, within a $\dGh_n$-pseudo-diagram, the parity of each string endpoint on the boundary is well-defined. Since each $\dGh_n$-pseudo-diagram has an even number of ghosts plus string endpoints on each boundary, concatenation changes the numbering by an even number, and so does not affect parity. Therefore $\chi(\e{x\circ y}\circ z) = \chi(x\circ y \circ z)$. Similarly, each object of type 2 contributes a factor of its associated parameter to $\chi(y\circ z)$, while each object of type 1 or 3 contributes to $\chi(x\circ\e{y\circ z})$, so $\chi(x\circ\e{y\circ z}) = \chi(x\circ y \circ z)$.

The argument for \eqref{eq:zeta} is analogous, since $\zeta$ is determined by the number of strings with at least one unattached endpoint. The only notable difference is that the three types of objects for $\zeta$ are formed from the strings or empty nodes of the relevant diagrams, instead of just the strings.
\end{proof}

\begin{corollary}\label{cor:assoc}
The ghost algebra $\Gh_n$, the one-boundary ghost algebra $\Gho_n$, and the one-boundary dilute ghost algebra $\dGho_n$ are associative.
\end{corollary}
\begin{proof}
Since $\Gh_n$, $\Gho_n$ and $\dGho_n$ are non-unital, non-unital and unital subalgebras of $\dGh_n$, respectively, they inherit associativity from $\dGh_n$.
\end{proof}

\section{Cellularity}\label{app:cellular}
In this section, we show that $\Gho_n$, $\Gh_n$, $\dGho_n$ and $\dGh_n$ are cellular algebras, as defined by Graham and Lehrer in \cite[\S 1]{GL}. We reproduce the definition for convenience. We use the terms half-diagram, throughline and defect defined in Appendix \ref{app:dims}, as well as the injective map from pairs of half-diagrams to diagrams. We also use the notation $\chi$, $\zeta$, $\circ$ and $\e{\cdot}$ from Appendix \ref{app:assoc}.

Let $A$ be an associative unital algebra over a commutative ring $R$ with identity. Then $A$ is a \textit{cellular algebra} if it has a \textit{cell datum} $(\Lambda,M,C,\ast)$, where
\begin{enumerate}
\item $\Lambda$ is a partially ordered set, and for each $\lambda \in \Lambda$, $M(\lambda)$ is a finite set such that $C: \bigsqcup_{\lambda\in\Lambda} M(\lambda) \times M(\lambda) \to A$ is an injective map with image an $R$-basis of $A$. 
\item If $\lambda \in \Lambda$ and $S,T \in M(\lambda)$, write $C_{S,T}^\lambda := C(S,T)$. Then $\ast$ is an $R$-linear anti-involution of $A$ such that $\left(C_{S,T}^\lambda\right)^\ast = C_{T,S}^\lambda$.
\item If $\lambda \in \Lambda$ and $S,T \in M(\lambda)$, then for any element $a \in A$, we have 
\begin{align}
aC_{S,T}^\lambda \equiv \sum_{S' \in M(\lambda)} r_a(S',S) C^\lambda_{S',T} \qquad \mod A(<\lambda), \label{eq:cellr}
\end{align}
where $r_a(S',S) \in R$ is independent of $T$, and where $A(<\lambda)$ is the $R$-submodule of $A$ generated by $\left\lbrace \left. C^\mu_{S^\mu,T^\mu} \ \right\rvert\  \mu < \lambda;\ S^\mu, T^\mu \in M(\mu)\right\rbrace$.
\end{enumerate}

\begin{theorem}\label{thm:cellular}
The algebras $\Gho_n$, $\Gh_n$, $\dGho_n$ and $\dGh_n$ are cellular.
\end{theorem}

\begin{proof}
Let $R = \C$ and let $A \in \left\lbrace\Gho_n, \Gh_n, \dGho_n, \dGh_n\right\rbrace$. Let $\Lambda = \{0,1, \dots, n\}$ with the usual partial order $\leq$. For each $\lambda \in \Lambda$, let $M(\lambda)$ be the set of all $A$-half-diagrams with $\lambda$ defects. For $\lambda \in \Lambda$ and $S,T \in M(\lambda)$, let $C^\lambda_{S,T}$ be the $A$-diagram with $\lambda$ defects obtained by reflecting $T$ about a vertical line and linking its defects to those of $S$.

Then from Appendix \ref{app:dims}, $C: \bigsqcup_{\lambda \in \Lambda} M(\lambda) \times M(\lambda) \to A$ is injective, and its image is the set of all $A$-diagrams, which is indeed a $\C$-basis for $A$.

Define $\ast:A\to A$ as the linear map that takes each $A$-diagram to its reflection about a vertical line. This is clearly bijective. To see that this is an anti-involution, consider the strings and ghosts on each boundary of each $A$-diagram. During multiplication these are numbered, and boundary arcs assigned parameters according to the parities of their endpoints. Recalling that the sum of the number of strings and ghosts on each boundary is even, the parity assigned to each string or ghost is swapped under the reflection $\ast$, but so is the left-to-right order. Hence, $\ast$ maps each type of boundary arc to itself, except those that have both ends on the same boundary with the same parity, and top-to-bottom boundary arcs. These exceptional boundary arcs are assigned the same parameter as their images, and this ensures that $\ast$ is an anti-involution. It follows from the definition of $C$ that $\left(C_{S,T}^\lambda\right)^\ast = C_{T,S}^\lambda$ for all $S,T \in M(\lambda)$ and $\lambda \in \Lambda$.

Finally, for part 3, consider the product of two $A$-diagrams $C_{U,V}^\nu$ and $C_{S,T}^\lambda$, where $\nu,\lambda \in \Lambda$, $U,V \in M(\nu)$ and $S,T \in M(\lambda)$. Note that $C_{U,V}^\nu C_{S,T}^\lambda$ is a scalar multiple of an $A$-diagram with at most $\min\{\lambda,\nu\}$ throughlines. Thus if $\nu < \lambda$, then $C_{U,V}^\nu C_{S,T}^\lambda \equiv 0 \mod A(<\lambda)$, so we can take $r_{C_{U,V}^\nu}(S,S')=0$ for all $S' \in M(\lambda)$ in this case. If $\nu \geq \lambda$, it is possible that the $A$-diagram $\e{C_{U,V}^\nu \circ C_{S,T}^\lambda}$, obtained by concatenating these diagrams and removing strings not connected to either side, has fewer than $\lambda$ throughlines. Since $C_{U,V}^\nu C_{S,T}^\lambda$ is a scalar multiple of this diagram, it is congruent to $0 \mod A(<\lambda)$, and thus we can take $r_{C_{U,V}^\nu}(S,S')=0$ for all $S' \in M(\lambda)$ in this case as well. If $\nu \geq \lambda$ and $\e{C_{U,V}^\nu \circ C_{S,T}^\lambda}$ has $\lambda$ throughlines, then since 
\begin{align}
C_{U,V}^\nu C_{S,T}^\lambda = \chi\left(C_{U,V}^\nu \circ C_{S,T}^\lambda\right)\zeta\left(C_{U,V}^\nu \circ C_{S,T}^\lambda\right)\e{C_{U,V}^\nu \circ C_{S,T}^\lambda},
\end{align}
we can take
\begin{align}
r_{C_{U,V}^\nu}(S,S') = \begin{cases}
\chi\left(C_{U,V}^\nu \circ C_{S,T}^\lambda\right)\zeta\left(C_{U,V}^\nu \circ C_{S,T}^\lambda\right), & \text{if }S' = \mathcal{L}\left(\e{C_{U,V}^\nu \circ C_{S,T}^\lambda} \right), \\
0, &\text{otherwise},
\end{cases}
\end{align}
where $\mathcal{L}(x)$ is the left $A$-half-diagram obtained by cutting the $A$-diagram $x$ in half along a vertical line.

Observe that the boundary arcs and loops formed in the concatenation $C_{U,V}^\nu \circ C_{S,T}^\lambda$ can only use strings from the $A$-half-diagrams $V$ and $S$, so $\chi\left(C_{U,V}^\nu \circ C_{S,T}^\lambda \right)$ is a function of $V$ and $S$ only. A string meets an empty node in this concatenation if and only if a string from $V$ meets an empty node of $S$ or vice versa, so $\zeta\left(C_{U,V}^\nu \circ C_{S,T}^\lambda \right)$ is also a function of $V$ and $S$ only. Note also that $\mathcal{L}\left(C_{U,V}^\nu C_{S,T}^\lambda \right)$ does not depend on $T$ because the only strings from $T$ that may be connected to the left side of the resulting diagram are its throughlines, and these are cut to produce $\mathcal{L}\left(C_{U,V}^\nu C_{S,T}^\lambda \right)$. Hence in all cases, $r_{C_{U,V}^\nu}$ is independent of $T$.

For each $a \in A$, we can write
\begin{align}
a = \sum_{\nu \in \Lambda}\ \  \sum_{U,V \in M(\nu)} f^\nu_{U,V}C_{U,V}^\nu,
\end{align}
for some $f^\nu_{U,V} \in \C$, and then
\begin{align}
r_a(S',S) = \sum_{\nu \in \Lambda}\ \  \sum_{U,V \in M(\nu)} f^\nu_{U,V}r_{C_{U,V}^\nu}(S',S).
\end{align}
Since $r_{C_{U,V}^\nu}$ does not depend on $T$, it follows that $r_a$ does not depend on $T$. Thus part 3 is satisfied, and $A$ is a cellular algebra.
\end{proof}

We note that the algebras $\Gho_n$, $\Gh_n$, $\dGho_n$ and $\dGh_n$ are finite-dimensional over the field $\C$, so the representation theory results from \cite[\S 3]{GL} apply to them.

\section{Generalisations} \label{app:nc}

In Appendix \ref{app:cellular}, we showed that the algebras $\Gho_n$, $\Gh_n$, $\dGho_n$ and $\dGh_n$ are cellular, using the anti-involution given by reflecting their basis diagrams about a vertical line. The algebras $\TL_n$, $\dTL_n$, $\BTL_n$ and $\BBTL_n$ are also cellular with respect to this reflection, though $\BBTL_n$ is infinite-dimensional, so the results from \cite[\S 3]{GL} do not hold. However, the proof of cellularity relies on some boundary arcs of different types being assigned the same parameter; see the arcs assigned $\alpha_3$, $\delta_3$, $\gamma_{12}$ and $\gamma_3$ in Table \ref{tab:params}. At the cost of cellularity with respect to reflection and half-diagrams, we can construct \textit{generalised} versions of $\Gh_n$ and $\dGh_n$, denoted by $\ncGh_n$ and $\ncdGh_n$, by using the same diagram bases and multiplication rules, but distinct parameters for each boundary arc type, as given in Table \ref{tab:ncparams}. The \textit{generalised one-boundary ghost algebra} $\ncGho_n$ and the \textit{generalised one-boundary dilute ghost algebra} $\ncdGho_n$ are the subalgebras of $\ncGh_n$ and $\ncdGh_n$, respectively, spanned by diagrams with no strings connected to the bottom boundary. The associativity of the generalised algebras can be proven as in Theorem \ref{thm:assoc} and Corollary \ref{cor:assoc}, with the obvious modifications to the function $\chi$.

\begin{table}[H]
\centering
\caption{Table of boundary arcs and their associated parameters in $\ncGh_n$ and $\ncdGh_n$.}\label{tab:ncparams}
\vspace{2mm}
\begin{tabular}{|c|c||c|c||c|c|}
\hline
Parameter & Boundary arc & Parameter & Boundary arc & Parameter & Boundary arc\\
\hline
$\alpha_1$ & $\begin{tikzpicture}[baseline={([yshift=-1mm]current bounding box.center)},xscale=0.5,yscale=0.5]
{
\draw[dotted] (0,0.5)--(2,0.5);
\draw (0.2,0.5) arc (-180:0:0.8);
\node at (0.2,0.5) [anchor=south] {\scriptsize odd};
\node at (1.8,0.5) [anchor=south] {\scriptsize even};
\node at (1,-0.5) {};
}
\end{tikzpicture}$
&$\delta_1$ &$\begin{tikzpicture}[baseline={([yshift=-1mm]current bounding box.center)},xscale=0.5,yscale=0.5]
{
\draw[dotted] (0,0)--(2,0);
\draw (0.2,0) arc (180:0:0.8);
\node at (0.2,-0.9) [anchor=south] {\scriptsize odd};
\node at (1.8,-0.9) [anchor=south] {\scriptsize even};
\node at (1,1) {};
}
\end{tikzpicture}$
&$\gamma_{1}$ & $\begin{tikzpicture}[baseline={([yshift=5mm]current bounding box.south)},xscale=0.5,yscale=0.5]
{
\draw[dotted] (0,0)--(2,0);
\draw[dotted] (0,1)--(2,1);
\draw (1,1)--(1,0);
\node at (1,1) [anchor=south] {\scriptsize odd};
\node at (1,-0.9) [anchor=south] {\scriptsize even};
}
\end{tikzpicture}$\\
\hline
$\alpha_2$ &$\begin{tikzpicture}[baseline={([yshift=-1mm]current bounding box.center)},xscale=0.5,yscale=0.5]
{
\draw[dotted] (0,0.5)--(2,0.5);
\draw (0.2,0.5) arc (-180:0:0.8);
\node at (0.2,0.5) [anchor=south] {\scriptsize even};
\node at (1.8,0.5) [anchor=south] {\scriptsize odd};
\node at (1,-0.5) {};
}
\end{tikzpicture}$ 
&$\delta_2$& $\begin{tikzpicture}[baseline={([yshift=-1mm]current bounding box.center)},xscale=0.5,yscale=0.5]
{
\draw[dotted] (0,0)--(2,0);
\draw (0.2,0) arc (180:0:0.8);
\node at (0.2,-0.9) [anchor=south] {\scriptsize even};
\node at (1.8,-0.9) [anchor=south] {\scriptsize odd};
\node at (1,1) {};
}
\end{tikzpicture}$
&$\gamma_{2}$ & $\begin{tikzpicture}[baseline={([yshift=5mm]current bounding box.south)},xscale=0.5,yscale=0.5]
{
\draw[dotted] (0,0)--(2,0);
\draw[dotted] (0,1)--(2,1);
\draw (1,1)--(1,0);
\node at (1,1) [anchor=south] {\scriptsize even};
\node at (1,-0.9) [anchor=south] {\scriptsize odd};
}
\end{tikzpicture}$\\
\hline
$\alpha_3$ & $\begin{tikzpicture}[baseline={([yshift=4mm]current bounding box.south)},xscale=0.5,yscale=0.5]
{
\draw[dotted] (0,0.5)--(2,0.5);
\draw (0.2,0.5) arc (-180:0:0.8);
\node at (0.2,0.5) [anchor=south] {\scriptsize odd};
\node at (1.8,0.5) [anchor=south] {\scriptsize odd};
\node at (1,-0.5) {};
}
\end{tikzpicture}$
&$\delta_3$ &$\begin{tikzpicture}[baseline={([yshift=-1mm]current bounding box.center)},xscale=0.5,yscale=0.5]
{
\draw[dotted] (0,0)--(2,0);
\draw (0.2,0) arc (180:0:0.8);
\node at (0.2,-0.9) [anchor=south] {\scriptsize odd};
\node at (1.8,-0.9) [anchor=south] {\scriptsize odd};
\node at (1,1) {};
}
\end{tikzpicture}$
&$\gamma_3$ & $\begin{tikzpicture}[baseline={([yshift=5mm]current bounding box.south)},xscale=0.5,yscale=0.5]
{
\draw[dotted] (0,0)--(2,0);
\draw[dotted] (0,1)--(2,1);
\draw (1,1)--(1,0);
\node at (1,1) [anchor=south] {\scriptsize odd};
\node at (1,-0.9) [anchor=south] {\scriptsize odd};
}
\end{tikzpicture}$\\
\hline
$\alpha_4$ & $\begin{tikzpicture}[baseline={([yshift=4mm]current bounding box.south)},xscale=0.5,yscale=0.5]
{
\draw[dotted] (0,0.5)--(2,0.5);
\draw (0.2,0.5) arc (-180:0:0.8);
\node at (0.2,0.5) [anchor=south] {\scriptsize even};
\node at (1.8,0.5) [anchor=south] {\scriptsize even};
\node at (1,-0.5) {};
}
\end{tikzpicture}$
&$\delta_4$ &$\begin{tikzpicture}[baseline={([yshift=-1mm]current bounding box.center)},xscale=0.5,yscale=0.5]
{
\draw[dotted] (0,0)--(2,0);
\draw (0.2,0) arc (180:0:0.8);
\node at (0.2,-0.9) [anchor=south] {\scriptsize even};
\node at (1.8,-0.9) [anchor=south] {\scriptsize even};
\node at (1,1) {};
}
\end{tikzpicture}$
&$\gamma_4$ & $\begin{tikzpicture}[baseline={([yshift=5mm]current bounding box.south)},xscale=0.5,yscale=0.5]
{
\draw[dotted] (0,0)--(2,0);
\draw[dotted] (0,1)--(2,1);
\draw (1,1)--(1,0);
\node at (1,1) [anchor=south] {\scriptsize even};
\node at (1,-0.9) [anchor=south] {\scriptsize even};
}
\end{tikzpicture}$\\
\hline
\end{tabular}
\end{table}

One can recover the original algebras from their generalised counterparts by setting $\alpha_4=\alpha_3$, $\delta_4=\delta_3$, $\gamma_1=\gamma_2=\gamma_{12}$ and $\gamma_4=\gamma_3$.

The solutions to the top boundary BYBE for $\ncGh_n$ and $\ncdGh_n$ (and therefore $\ncGho_n$ and $\ncdGho_n$) are exactly the same as the solutions for $\Gh_n$ and $\dGh_n$, except:
\begin{itemize}
    \item $A_1 = \alpha_1^2 + \alpha_2^2 - 2\alpha_3\alpha_4$ and $A_2 = \alpha_1\alpha_2 - \alpha_3\alpha_4$, wherever these appear.
    \item In $\Gh_n$ Solution III, we still have $b_2(u) = \alpha_3 b(u)$ but now $b_4(u) = \alpha_4 b(u)$.
    \item In $\dGh_n$ Solutions I, II and IV, the factor $\left(\mu\nu(\alpha_1 + \alpha_2) + \left(\mu^2 + \nu^2\right) \alpha_3\right)$ in $a_1(u)$ and $a_2(u)$ becomes $\left(\mu\nu(\alpha_1 + \alpha_2) + \mu^2 \alpha_3 + \nu^2 \alpha_4\right)$.
    \item In $\dGh_n$ Solution III, $b_2(u) = -\alpha_3 b(u) P(\sigma,\tau; -3\phi)$ but now $b_4(u) = -\alpha_4 b(u) P(\sigma,\tau; -3\phi)$.
\end{itemize}
The proofs are analogous to the proofs of Theorems \ref{thm:denseBYBE} and \ref{thm:diluteBYBE}, and the bottom boundary BYBE solutions are found by reflecting the top boundary BYBE solutions about a horizontal line, and swapping the parameters $\alpha_i \leftrightarrow \delta_i$.

\section{Configurations for Theorem \ref{thm:diluteBYBE}}\label{app:configs}
This section lists the configurations of bulk squares and boundary triangles corresponding to each diagram in the proof of Theorem \ref{thm:diluteBYBE}, in order of appearance. Diagrams with the same strings but a different arrangement of ghosts have been omitted; the configurations are also the same up to the placement of ghosts. Similarly, where configurations differ only by ghosts, only one is drawn, and the number of configurations with such strings is noted.
\begin{align*}
&\begin{tikzpicture}[baseline={([yshift=-1mm]current bounding box.center)},scale=0.55]
{
\draw[dotted] (0,0.5)--(3,0.5);
\draw [very thick](0,-1.5) -- (0,0.5);
\draw [very thick](3,-1.5)--(3,0.5);
\draw (0,0) to[out=0,in=-90] (0.6,0.5);
\draw (0,-1) to[out=0,in=-90] (1.2,0.5);
\draw (3,0) to[out=180,in=270] (2.4,0.5);
\draw (3,-1) to[out=180,in=270] (1.8,0.5);
}
\end{tikzpicture} \; :
&&\begin{tikzpicture}[baseline={([yshift=-1mm]current bounding box.center)},scale=0.55]
{
\idsq{(1,-1)}
\batri{(2,0)}
\esq{(3,-1)}
\batri{(4,0)}
\draw[lstring] let \n1 = {sin(45)} in (2.5,-1.5) arc (225:315:\n1);
\foreach \x in {2.5,3.5}
{\foreach \y in {-0.5,-1.5}
{\fill[lgh] (\x,\y) circle (1.6 pt);}}
\fill[lgh] (4.5,-0.5) circle (1.6 pt);
}
\end{tikzpicture}
\qquad
\begin{tikzpicture}[baseline={([yshift=-1mm]current bounding box.center)},scale=0.55]
{
\batri{(0,0)}
\esq{(1,-1)}
\batri{(2,0)}
\idsq{(3,-1)}
\draw[lstring] let \n1 = {sin(45)} in (2.5,-1.5) arc (225:315:\n1);
\foreach \x in {2.5,3.5}
{\foreach \y in {-0.5,-1.5}
{\fill[lgh] (\x,\y) circle (1.6 pt);}}
\fill[lgh] (1.5,-0.5) circle (1.6 pt);
}
\end{tikzpicture}
\\[9mm]
&\begin{tikzpicture}[baseline={([yshift=-1mm]current bounding box.center)},scale=0.55]
{
\draw[dotted] (0,0.5)--(3,0.5);
\draw [very thick](0,-1.5) -- (0,0.5);
\draw [very thick](3,-1.5)--(3,0.5);
\draw (0,0) to[out=0,in=-90] (0.6,0.5);
\draw (0,-1) -- (3,-1);
\draw (3,0) to[out=180,in=270] (2.4,0.5);
}
\end{tikzpicture} \; :
&&\begin{tikzpicture}[baseline={([yshift=-1mm]current bounding box.center)},scale=0.55]
{
\idsq{(1,-1)}
\batri{(2,0)}
\idsq{(3,-1)}
\aatri{(4,0)}
\draw[lstring] let \n1 = {sin(45)} in (2.5,-1.5) arc (225:315:\n1);
\foreach \x in {2.5,3.5}
{\foreach \y in {-0.5,-1.5}
{\fill[lgh] (\x,\y) circle (1.6 pt);}}
\fill[lgh] (4.5,-0.5) circle (1.6 pt);
}
\end{tikzpicture}
\quad
\begin{tikzpicture}[baseline={([yshift=-1mm]current bounding box.center)},scale=0.55]
{
\idsq{(1,-1)}
\bftri{(2,0)}
\bsq{(3,-1)}
\bhtri{(4,0)}
\draw[lstring] let \n1 = {sin(45)} in (2.5,-1.5) arc (225:315:\n1);
\foreach \x in {2.5,3.5}
{\foreach \y in {-0.5,-1.5}
{\fill[lgh] (\x,\y) circle (1.6 pt);}}
\fill[lgh] (4.5,-0.5) circle (1.6 pt);
}
\end{tikzpicture}
\quad
\left[
\begin{tikzpicture}[baseline={([yshift=-1mm]current bounding box.center)},scale=0.55]
{
\idsq{(1,-1)}
\batri{(2,0)}
\idsq{(3,-1)}
\batri{(4,0)}
\draw[lstring] let \n1 = {sin(45)} in (2.5,-1.5) arc (225:315:\n1);
\foreach \x in {2.5,3.5}
{\foreach \y in {-0.5,-1.5}
{\fill[lgh] (\x,\y) circle (1.6 pt);}}
\fill[lgh] (4.5,-0.5) circle (1.6 pt);
}
\end{tikzpicture}
\right]_{\times 4}
\begin{tikzpicture}[baseline={([yshift=-1mm]current bounding box.center)},scale=0.55]
{
\idsq{(1,-1)}
\aatri{(2,0)}
\idsq{(3,-1)}
\batri{(4,0)}
\draw[lstring] let \n1 = {sin(45)} in (2.5,-1.5) arc (225:315:\n1);
\foreach \x in {2.5,3.5}
{\foreach \y in {-0.5,-1.5}
{\fill[lgh] (\x,\y) circle (1.6 pt);}}
\fill[lgh] (4.5,-0.5) circle (1.6 pt);
}
\end{tikzpicture}
\\[2mm]
& & &
\begin{tikzpicture}[baseline={([yshift=-1mm]current bounding box.center)},scale=0.55]
{
\batri{(0,0)}
\idsq{(1,-1)}
\aatri{(2,0)}
\idsq{(3,-1)}
\draw[lstring] let \n1 = {sin(45)} in (2.5,-1.5) arc (225:315:\n1);
\foreach \x in {2.5,3.5}
{\foreach \y in {-0.5,-1.5}
{\fill[lgh] (\x,\y) circle (1.6 pt);}}
\fill[lgh] (1.5,-0.5) circle (1.6 pt);
}
\end{tikzpicture}
\quad
\begin{tikzpicture}[baseline={([yshift=-1mm]current bounding box.center)},scale=0.55]
{
\bftri{(0,0)}
\bsq{(1,-1)}
\bhtri{(2,0)}
\idsq{(3,-1)}
\draw[lstring] let \n1 = {sin(45)} in (2.5,-1.5) arc (225:315:\n1);
\foreach \x in {2.5,3.5}
{\foreach \y in {-0.5,-1.5}
{\fill[lgh] (\x,\y) circle (1.6 pt);}}
\fill[lgh] (1.5,-0.5) circle (1.6 pt);
}
\end{tikzpicture}
\quad
\left[
\begin{tikzpicture}[baseline={([yshift=-1mm]current bounding box.center)},scale=0.55]
{
\batri{(0,0)}
\idsq{(1,-1)}
\batri{(2,0)}
\idsq{(3,-1)}
\draw[lstring] let \n1 = {sin(45)} in (2.5,-1.5) arc (225:315:\n1);
\foreach \x in {2.5,3.5}
{\foreach \y in {-0.5,-1.5}
{\fill[lgh] (\x,\y) circle (1.6 pt);}}
\fill[lgh] (1.5,-0.5) circle (1.6 pt);
}
\end{tikzpicture}
\right]_{\times 4}
\begin{tikzpicture}[baseline={([yshift=-1mm]current bounding box.center)},scale=0.55]
{
\aatri{(0,0)}
\idsq{(1,-1)}
\batri{(2,0)}
\idsq{(3,-1)}
\draw[lstring] let \n1 = {sin(45)} in (2.5,-1.5) arc (225:315:\n1);
\foreach \x in {2.5,3.5}
{\foreach \y in {-0.5,-1.5}
{\fill[lgh] (\x,\y) circle (1.6 pt);}}
\fill[lgh] (1.5,-0.5) circle (1.6 pt);
}
\end{tikzpicture}
\\[9mm]
&\begin{tikzpicture}[baseline={([yshift=-1mm]current bounding box.center)},scale=0.55]
{
\draw[dotted] (0,0.5)--(3,0.5);
\draw [very thick](0,-1.5) -- (0,0.5);
\draw [very thick](3,-1.5)--(3,0.5);
\draw (0,-1) to[out=0,in=-90] (1.2,0.5);
\draw (3,-1) to[out=180,in=270] (1.8,0.5);
\filldraw (0.6,0.5) circle (0.08);
\filldraw (2.4,0.5) circle (0.08);
\filldraw[fill=white] (0,0) circle (0.12);
\filldraw[fill=white] (3,0) circle (0.12);
}
\end{tikzpicture} \; :
& & 
\begin{tikzpicture}[baseline={([yshift=-1mm]current bounding box.center)},scale=0.55]
{
\usq{(1,-1)}
\bctri{(2,0)}
\dsq{(3,-1)}
\emtri{(4,0)}
\foreach \x in {2.5,3.5}
{\foreach \y in {-0.5,-1.5}
{\fill[lgh] (\x,\y) circle (1.6 pt);}}
\fill[lgh] (4.5,-0.5) circle (1.6 pt);
}
\end{tikzpicture}
\quad
\begin{tikzpicture}[baseline={([yshift=-1mm]current bounding box.center)},scale=0.55]
{
\usq{(1,-1)}
\betri{(2,0)}
\rsq{(3,-1)}
\bftri{(4,0)}
\foreach \x in {2.5,3.5}
{\foreach \y in {-0.5,-1.5}
{\fill[lgh] (\x,\y) circle (1.6 pt);}}
\fill[lgh] (4.5,-0.5) circle (1.6 pt);
}
\end{tikzpicture}
\quad
\begin{tikzpicture}[baseline={([yshift=-1mm]current bounding box.center)},scale=0.55]
{
\bsq{(1,-1)}
\bhtri{(2,0)}
\esq{(3,-1)}
\bftri{(4,0)}
\draw[lstring] let \n1 = {sin(45)} in (2.5,-1.5) arc (225:315:\n1);
\foreach \x in {2.5,3.5}
{\foreach \y in {-0.5,-1.5}
{\fill[lgh] (\x,\y) circle (1.6 pt);}}
\fill[lgh] (4.5,-0.5) circle (1.6 pt);
}
\end{tikzpicture}
\\[2mm]
& & & 
\begin{tikzpicture}[baseline={([yshift=-1mm]current bounding box.center)},scale=0.55]
{
\bhtri{(0,0)}
\esq{(1,-1)}
\bftri{(2,0)}
\bsq{(3,-1)}
\draw[lstring] let \n1 = {sin(45)} in (2.5,-1.5) arc (225:315:\n1);
\foreach \x in {2.5,3.5}
{\foreach \y in {-0.5,-1.5}
{\fill[lgh] (\x,\y) circle (1.6 pt);}}
\fill[lgh] (1.5,-0.5) circle (1.6 pt);
}
\end{tikzpicture}
\quad
\begin{tikzpicture}[baseline={([yshift=-1mm]current bounding box.center)},scale=0.55]
{
\bhtri{(0,0)}
\lsq{(1,-1)}
\bgtri{(2,0)}
\dsq{(3,-1)}
\foreach \x in {2.5,3.5}
{\foreach \y in {-0.5,-1.5}
{\fill[lgh] (\x,\y) circle (1.6 pt);}}
\fill[lgh] (1.5,-0.5) circle (1.6 pt);
}
\end{tikzpicture}
\quad
\begin{tikzpicture}[baseline={([yshift=-1mm]current bounding box.center)},scale=0.55]
{
\emtri{(0,0)}
\usq{(1,-1)}
\bctri{(2,0)}
\dsq{(3,-1)}
\foreach \x in {2.5,3.5}
{\foreach \y in {-0.5,-1.5}
{\fill[lgh] (\x,\y) circle (1.6 pt);}}
\fill[lgh] (1.5,-0.5) circle (1.6 pt);
}
\end{tikzpicture}
\\[9mm]
&\begin{tikzpicture}[baseline={([yshift=-1mm]current bounding box.center)},scale=0.55]
{
\draw[dotted] (0,0.5)--(3,0.5);
\draw [very thick](0,-1.5) -- (0,0.5);
\draw [very thick](3,-1.5)--(3,0.5);
\draw (0,0) to[out=0,in=-90] (0.75,0.5);
\draw (0,-1) to[out=0,in=-90] (1.5,0.5);
\draw (3,-1) to[out=180,in=270] (2.25,0.5);
\filldraw (0.375,0.5) circle (0.08);
\filldraw (1.125,0.5) circle (0.08);
\filldraw (1.875,0.5) circle (0.08);
\filldraw[fill=white] (3,0) circle (0.12);
}
\end{tikzpicture} \; :
& & 
\begin{tikzpicture}[baseline={([yshift=-1mm]current bounding box.center)},scale=0.55]
{
\idsq{(-1,-1)}
\bbtri{(0,0)}
\esq{(1,-1)}
\betri{(2,0)}
\draw[lstring] let \n1 = {sin(45)} in (0.5,-1.5) arc (225:315:\n1);
\foreach \x in {0.5,1.5}
{\foreach \y in {-0.5,-1.5}
{\fill[lgh] (\x,\y) circle (1.6 pt);}}
\fill[lgh] (2.5,-0.5) circle (1.6 pt);
}
\end{tikzpicture}
\qquad
\begin{tikzpicture}[baseline={([yshift=-1mm]current bounding box.center)},scale=0.55]
{
\bbtri{(0,0)}
\esq{(1,-1)}
\betri{(2,0)}
\bsq{(3,-1)}
\draw[lstring] let \n1 = {sin(45)} in (2.5,-1.5) arc (225:315:\n1);
\foreach \x in {2.5,3.5}
{\foreach \y in {-0.5,-1.5}
{\fill[lgh] (\x,\y) circle (1.6 pt);}}
\fill[lgh] (1.5,-0.5) circle (1.6 pt);
}
\end{tikzpicture}
\quad
\begin{tikzpicture}[baseline={([yshift=-1mm]current bounding box.center)},scale=0.55]
{
\bbtri{(0,0)}
\lsq{(1,-1)}
\bhtri{(2,0)}
\dsq{(3,-1)}
\foreach \x in {2.5,3.5}
{\foreach \y in {-0.5,-1.5}
{\fill[lgh] (\x,\y) circle (1.6 pt);}}
\fill[lgh] (1.5,-0.5) circle (1.6 pt);
}
\end{tikzpicture}
\quad
\begin{tikzpicture}[baseline={([yshift=-1mm]current bounding box.center)},scale=0.55]
{
\betri{(0,0)}
\usq{(1,-1)}
\bbtri{(2,0)}
\dsq{(3,-1)}
\foreach \x in {2.5,3.5}
{\foreach \y in {-0.5,-1.5}
{\fill[lgh] (\x,\y) circle (1.6 pt);}}
\fill[lgh] (1.5,-0.5) circle (1.6 pt);
}
\end{tikzpicture}
\end{align*}

\begin{align*}
&\begin{tikzpicture}[baseline={([yshift=-1mm]current bounding box.center)},scale=0.55]
{
\draw[dotted] (0,0.5)--(3,0.5);
\draw [very thick](0,-1.5) -- (0,0.5);
\draw [very thick](3,-1.5)--(3,0.5);
\draw (0,0) to[out=0,in=-90] (1.2,0.5);
\draw (3,-1) to[out=180,in=270] (2.4,0.5);
\filldraw (0.6,0.5) circle (0.08);
\filldraw (1.8,0.5) circle (0.08);
\filldraw[fill=white] (0,-1) circle (0.12);
\filldraw[fill=white] (3,0) circle (0.12);
}
\end{tikzpicture} \; :
& & 
\begin{tikzpicture}[baseline={([yshift=-1mm]current bounding box.center)},scale=0.55]
{
\tsq{(-1,-1)}
\bbtri{(0,0)}
\dsq{(1,-1)}
\emtri{(2,0)}
\foreach \x in {0.5,1.5}
{\foreach \y in {-0.5,-1.5}
{\fill[lgh] (\x,\y) circle (1.6 pt);}}
\fill[lgh] (2.5,-0.5) circle (1.6 pt);
}
\end{tikzpicture}
\quad
\begin{tikzpicture}[baseline={([yshift=-1mm]current bounding box.center)},scale=0.55]
{
\tsq{(-1,-1)}
\betri{(0,0)}
\rsq{(1,-1)}
\betri{(2,0)}
\foreach \x in {0.5,1.5}
{\foreach \y in {-0.5,-1.5}
{\fill[lgh] (\x,\y) circle (1.6 pt);}}
\fill[lgh] (2.5,-0.5) circle (1.6 pt);
}
\end{tikzpicture}
\quad
\begin{tikzpicture}[baseline={([yshift=-1mm]current bounding box.center)},scale=0.55]
{
\dsq{(-1,-1)}
\bhtri{(0,0)}
\esq{(1,-1)}
\betri{(2,0)}
\draw[lstring] let \n1 = {sin(45)} in (0.5,-1.5) arc (225:315:\n1);
\foreach \x in {0.5,1.5}
{\foreach \y in {-0.5,-1.5}
{\fill[lgh] (\x,\y) circle (1.6 pt);}}
\fill[lgh] (2.5,-0.5) circle (1.6 pt);
}
\end{tikzpicture}
\\[2mm]
& & &
\begin{tikzpicture}[baseline={([yshift=-1mm]current bounding box.center)},scale=0.55]
{
\betri{(0,0)}
\rsq{(1,-1)}
\betri{(2,0)}
\bsq{(3,-1)}
\draw[lstring] let \n1 = {sin(45)} in (2.5,-1.5) arc (225:315:\n1);
\foreach \x in {2.5,3.5}
{\foreach \y in {-0.5,-1.5}
{\fill[lgh] (\x,\y) circle (1.6 pt);}}
\fill[lgh] (1.5,-0.5) circle (1.6 pt);
}
\end{tikzpicture}
\quad 
\begin{tikzpicture}[baseline={([yshift=-1mm]current bounding box.center)},scale=0.55]
{
\betri{(0,0)}
\emsq{(1,-1)}
\bhtri{(2,0)}
\dsq{(3,-1)}
\foreach \x in {2.5,3.5}
{\foreach \y in {-0.5,-1.5}
{\fill[lgh] (\x,\y) circle (1.6 pt);}}
\fill[lgh] (1.5,-0.5) circle (1.6 pt);
}
\end{tikzpicture}
\quad 
\left[
\begin{tikzpicture}[baseline={([yshift=-1mm]current bounding box.center)},scale=0.55]
{
\bbtri{(0,0)}
\tsq{(1,-1)}
\batri{(2,0)}
\dsq{(3,-1)}
\foreach \x in {2.5,3.5}
{\foreach \y in {-0.5,-1.5}
{\fill[lgh] (\x,\y) circle (1.6 pt);}}
\fill[lgh] (1.5,-0.5) circle (1.6 pt);
}
\end{tikzpicture} \right]_{\times 4}
\begin{tikzpicture}[baseline={([yshift=-1mm]current bounding box.center)},scale=0.55]
{
\bbtri{(0,0)}
\dsq{(1,-1)}
\emtri{(2,0)}
\bsq{(3,-1)}
\draw[lstring] let \n1 = {sin(45)} in (2.5,-1.5) arc (225:315:\n1);
\foreach \x in {2.5,3.5}
{\foreach \y in {-0.5,-1.5}
{\fill[lgh] (\x,\y) circle (1.6 pt);}}
\fill[lgh] (1.5,-0.5) circle (1.6 pt);
}
\end{tikzpicture}
\\[2mm]
& & &
\begin{tikzpicture}[baseline={([yshift=-1mm]current bounding box.center)},scale=0.55]
{
\bbtri{(0,0)}
\tsq{(1,-1)}
\aatri{(2,0)}
\dsq{(3,-1)}
\foreach \x in {2.5,3.5}
{\foreach \y in {-0.5,-1.5}
{\fill[lgh] (\x,\y) circle (1.6 pt);}}
\fill[lgh] (1.5,-0.5) circle (1.6 pt);
}
\end{tikzpicture}
\quad 
\begin{tikzpicture}[baseline={([yshift=-1mm]current bounding box.center)},scale=0.55]
{
\aatri{(0,0)}
\tsq{(1,-1)}
\bbtri{(2,0)}
\dsq{(3,-1)}
\foreach \x in {2.5,3.5}
{\foreach \y in {-0.5,-1.5}
{\fill[lgh] (\x,\y) circle (1.6 pt);}}
\fill[lgh] (1.5,-0.5) circle (1.6 pt);
}
\end{tikzpicture}
\\[9mm]
&\begin{tikzpicture}[baseline={([yshift=-1mm]current bounding box.center)},scale=0.55]
{
\draw[dotted] (0,0.5)--(3,0.5);
\draw [very thick](0,-1.5) -- (0,0.5);
\draw [very thick](3,-1.5)--(3,0.5);
\draw (0,0) to[out=0,in=-90] (1,0.5);
\draw (0,-1) -- (3,-1);
\filldraw (0.5,0.5) circle (0.08);
\filldraw[fill=white] (3,0) circle (0.12);
}
\end{tikzpicture} \; :
& & 
\begin{tikzpicture}[baseline={([yshift=-1mm]current bounding box.center)},scale=0.55]
{
\idsq{(-1,-1)}
\betri{(0,0)}
\bsq{(1,-1)}
\emtri{(2,0)}
\draw[lstring] let \n1 = {sin(45)} in (0.5,-1.5) arc (225:315:\n1);
\foreach \x in {0.5,1.5}
{\foreach \y in {-0.5,-1.5}
{\fill[lgh] (\x,\y) circle (1.6 pt);}}
\fill[lgh] (2.5,-0.5) circle (1.6 pt);
}
\end{tikzpicture}
\quad
\left[
\begin{tikzpicture}[baseline={([yshift=-1mm]current bounding box.center)},scale=0.55]
{
\idsq{(-1,-1)}
\bbtri{(0,0)}
\idsq{(1,-1)}
\bftri{(2,0)}
\draw[lstring] let \n1 = {sin(45)} in (0.5,-1.5) arc (225:315:\n1);
\foreach \x in {0.5,1.5}
{\foreach \y in {-0.5,-1.5}
{\fill[lgh] (\x,\y) circle (1.6 pt);}}
\fill[lgh] (2.5,-0.5) circle (1.6 pt);
}
\end{tikzpicture}
\right]_{\times 4}
\begin{tikzpicture}[baseline={([yshift=-1mm]current bounding box.center)},scale=0.55]
{
\idsq{(-1,-1)}
\aatri{(0,0)}
\idsq{(1,-1)}
\betri{(2,0)}
\draw[lstring] let \n1 = {sin(45)} in (0.5,-1.5) arc (225:315:\n1);
\foreach \x in {0.5,1.5}
{\foreach \y in {-0.5,-1.5}
{\fill[lgh] (\x,\y) circle (1.6 pt);}}
\fill[lgh] (2.5,-0.5) circle (1.6 pt);
}
\end{tikzpicture}
\\[2mm]
& & &
\begin{tikzpicture}[baseline={([yshift=-1mm]current bounding box.center)},scale=0.55]
{
\betri{(0,0)}
\bsq{(1,-1)}
\emtri{(2,0)}
\bsq{(3,-1)}
\draw[lstring] let \n1 = {sin(45)} in (2.5,-1.5) arc (225:315:\n1);
\foreach \x in {2.5,3.5}
{\foreach \y in {-0.5,-1.5}
{\fill[lgh] (\x,\y) circle (1.6 pt);}}
\fill[lgh] (1.5,-0.5) circle (1.6 pt);
}
\end{tikzpicture}
\quad
\left[
\begin{tikzpicture}[baseline={([yshift=-1mm]current bounding box.center)},scale=0.55]
{
\bbtri{(0,0)}
\idsq{(1,-1)}
\bftri{(2,0)}
\bsq{(3,-1)}
\draw[lstring] let \n1 = {sin(45)} in (2.5,-1.5) arc (225:315:\n1);
\foreach \x in {2.5,3.5}
{\foreach \y in {-0.5,-1.5}
{\fill[lgh] (\x,\y) circle (1.6 pt);}}
\fill[lgh] (1.5,-0.5) circle (1.6 pt);
}
\end{tikzpicture}
\right]_{\times 4}
\begin{tikzpicture}[baseline={([yshift=-1mm]current bounding box.center)},scale=0.55]
{
\betri{(0,0)}
\usq{(1,-1)}
\aatri{(2,0)}
\dsq{(3,-1)}
\foreach \x in {2.5,3.5}
{\foreach \y in {-0.5,-1.5}
{\fill[lgh] (\x,\y) circle (1.6 pt);}}
\fill[lgh] (1.5,-0.5) circle (1.6 pt);
}
\end{tikzpicture}
\quad
\begin{tikzpicture}[baseline={([yshift=-1mm]current bounding box.center)},scale=0.55]
{
\aatri{(0,0)}
\idsq{(1,-1)}
\betri{(2,0)}
\bsq{(3,-1)}
\draw[lstring] let \n1 = {sin(45)} in (2.5,-1.5) arc (225:315:\n1);
\foreach \x in {2.5,3.5}
{\foreach \y in {-0.5,-1.5}
{\fill[lgh] (\x,\y) circle (1.6 pt);}}
\fill[lgh] (1.5,-0.5) circle (1.6 pt);
}
\end{tikzpicture}
\\[9mm]
&\begin{tikzpicture}[baseline={([yshift=-1mm]current bounding box.center)},scale=0.55]
{
\draw[dotted] (0,0.5)--(3,0.5);
\draw [very thick](0,-1.5) -- (0,0.5);
\draw [very thick](3,-1.5)--(3,0.5);
\draw (0,0) to[out=0,in=-90] (1,0.5);
\draw (0,-1) -- (3,0);
\filldraw (0.5,0.5) circle (0.08);
\filldraw[fill=white] (3,-1) circle (0.12);
}
\end{tikzpicture} \; :
& & 
\begin{tikzpicture}[baseline={([yshift=-1mm]current bounding box.center)},scale=0.55]
{
\idsq{(-1,-1)}
\betri{(0,0)}
\usq{(1,-1)}
\aatri{(2,0)}
\draw[lstring] let \n1 = {sin(45)} in (0.5,-1.5) arc (225:315:\n1);
\foreach \x in {0.5,1.5}
{\foreach \y in {-0.5,-1.5}
{\fill[lgh] (\x,\y) circle (1.6 pt);}}
\fill[lgh] (2.5,-0.5) circle (1.6 pt);
}
\end{tikzpicture}
\\[2mm]
& & &
\begin{tikzpicture}[baseline={([yshift=-1mm]current bounding box.center)},scale=0.55]
{
\betri{(0,0)}
\usq{(1,-1)}
\aatri{(2,0)}
\tsq{(3,-1)}
\foreach \x in {2.5,3.5}
{\foreach \y in {-0.5,-1.5}
{\fill[lgh] (\x,\y) circle (1.6 pt);}}
\fill[lgh] (1.5,-0.5) circle (1.6 pt);
}
\end{tikzpicture}
\quad
\begin{tikzpicture}[baseline={([yshift=-1mm]current bounding box.center)},scale=0.55]
{
\betri{(0,0)}
\bsq{(1,-1)}
\emtri{(2,0)}
\usq{(3,-1)}
\draw[lstring] let \n1 = {sin(45)} in (2.5,-1.5) arc (225:315:\n1);
\foreach \x in {2.5,3.5}
{\foreach \y in {-0.5,-1.5}
{\fill[lgh] (\x,\y) circle (1.6 pt);}}
\fill[lgh] (1.5,-0.5) circle (1.6 pt);
}
\end{tikzpicture}
\quad
\left[
\begin{tikzpicture}[baseline={([yshift=-1mm]current bounding box.center)},scale=0.55]
{
\bbtri{(0,0)}
\idsq{(1,-1)}
\bftri{(2,0)}
\usq{(3,-1)}
\draw[lstring] let \n1 = {sin(45)} in (2.5,-1.5) arc (225:315:\n1);
\foreach \x in {2.5,3.5}
{\foreach \y in {-0.5,-1.5}
{\fill[lgh] (\x,\y) circle (1.6 pt);}}
\fill[lgh] (1.5,-0.5) circle (1.6 pt);
}
\end{tikzpicture}
\right]_{\times 4}
\begin{tikzpicture}[baseline={([yshift=-1mm]current bounding box.center)},scale=0.55]
{
\aatri{(0,0)}
\idsq{(1,-1)}
\betri{(2,0)}
\usq{(3,-1)}
\draw[lstring] let \n1 = {sin(45)} in (2.5,-1.5) arc (225:315:\n1);
\foreach \x in {2.5,3.5}
{\foreach \y in {-0.5,-1.5}
{\fill[lgh] (\x,\y) circle (1.6 pt);}}
\fill[lgh] (1.5,-0.5) circle (1.6 pt);
}
\end{tikzpicture}
\\[9mm]
&\begin{tikzpicture}[baseline={([yshift=-1mm]current bounding box.center)},scale=0.55]
{
\draw[dotted] (0,0.5)--(3,0.5);
\draw [very thick](0,-1.5) -- (0,0.5);
\draw [very thick](3,-1.5)--(3,0.5);
\draw (0,0) to[out=0,in=-90] (1,0.5);
\filldraw (0.5,0.5) circle (0.08);
\filldraw[fill=white] (0,-1) circle (0.12);
\filldraw[fill=white] (3,0) circle (0.12);
\filldraw[fill=white] (3,-1) circle (0.12);
}
\end{tikzpicture} \; :
& & 
\begin{tikzpicture}[baseline={([yshift=-1mm]current bounding box.center)},scale=0.55]
{
\tsq{(-1,-1)}
\betri{(0,0)}
\emsq{(1,-1)}
\emtri{(2,0)}
\foreach \x in {0.5,1.5}
{\foreach \y in {-0.5,-1.5}
{\fill[lgh] (\x,\y) circle (1.6 pt);}}
\fill[lgh] (2.5,-0.5) circle (1.6 pt);
}
\end{tikzpicture}
\quad
\left[
\begin{tikzpicture}[baseline={([yshift=-1mm]current bounding box.center)},scale=0.55]
{
\tsq{(-1,-1)}
\bbtri{(0,0)}
\tsq{(1,-1)}
\bftri{(2,0)}
\foreach \x in {0.5,1.5}
{\foreach \y in {-0.5,-1.5}
{\fill[lgh] (\x,\y) circle (1.6 pt);}}
\fill[lgh] (2.5,-0.5) circle (1.6 pt);
}
\end{tikzpicture}
\right]_{\times 4}
\begin{tikzpicture}[baseline={([yshift=-1mm]current bounding box.center)},scale=0.55]
{
\tsq{(-1,-1)}
\aatri{(0,0)}
\tsq{(1,-1)}
\betri{(2,0)}
\foreach \x in {0.5,1.5}
{\foreach \y in {-0.5,-1.5}
{\fill[lgh] (\x,\y) circle (1.6 pt);}}
\fill[lgh] (2.5,-0.5) circle (1.6 pt);
}
\end{tikzpicture}
\quad
\begin{tikzpicture}[baseline={([yshift=-1mm]current bounding box.center)},scale=0.55]
{
\dsq{(-1,-1)}
\bhtri{(0,0)}
\lsq{(1,-1)}
\emtri{(2,0)}
\draw[lstring] let \n1 = {sin(45)} in (0.5,-1.5) arc (225:315:\n1);
\foreach \x in {0.5,1.5}
{\foreach \y in {-0.5,-1.5}
{\fill[lgh] (\x,\y) circle (1.6 pt);}}
\fill[lgh] (2.5,-0.5) circle (1.6 pt);
}
\end{tikzpicture}
\\[2mm]
& & &
\begin{tikzpicture}[baseline={([yshift=-1mm]current bounding box.center)},scale=0.55]
{
\dsq{(-1,-1)}
\emtri{(0,0)}
\usq{(1,-1)}
\betri{(2,0)}
\draw[lstring] let \n1 = {sin(45)} in (0.5,-1.5) arc (225:315:\n1);
\foreach \x in {0.5,1.5}
{\foreach \y in {-0.5,-1.5}
{\fill[lgh] (\x,\y) circle (1.6 pt);}}
\fill[lgh] (2.5,-0.5) circle (1.6 pt);
}
\end{tikzpicture}
\quad
\begin{tikzpicture}[baseline={([yshift=-1mm]current bounding box.center)},scale=0.55]
{
\betri{(0,0)}
\emsq{(1,-1)}
\emtri{(2,0)}
\emsq{(3,-1)}
\foreach \x in {2.5,3.5}
{\foreach \y in {-0.5,-1.5}
{\fill[lgh] (\x,\y) circle (1.6 pt);}}
\fill[lgh] (1.5,-0.5) circle (1.6 pt);
}
\end{tikzpicture}
\ 
\left[
\begin{tikzpicture}[baseline={([yshift=-1mm]current bounding box.center)},scale=0.55]
{
\bbtri{(0,0)}
\tsq{(1,-1)}
\bftri{(2,0)}
\emsq{(3,-1)}
\foreach \x in {2.5,3.5}
{\foreach \y in {-0.5,-1.5}
{\fill[lgh] (\x,\y) circle (1.6 pt);}}
\fill[lgh] (1.5,-0.5) circle (1.6 pt);
}
\end{tikzpicture}
\right]_{\times 4}
\left[
\begin{tikzpicture}[baseline={([yshift=-1mm]current bounding box.center)},scale=0.55]
{
\bbtri{(0,0)}
\dsq{(1,-1)}
\bhtri{(2,0)}
\lsq{(3,-1)}
\draw[lstring] let \n1 = {sin(45)} in (2.5,-1.5) arc (225:315:\n1);
\foreach \x in {2.5,3.5}
{\foreach \y in {-0.5,-1.5}
{\fill[lgh] (\x,\y) circle (1.6 pt);}}
\fill[lgh] (1.5,-0.5) circle (1.6 pt);
}
\end{tikzpicture}
\right]_{\times 4}
\\[2mm]
& & &
\left[
\begin{tikzpicture}[baseline={([yshift=-1mm]current bounding box.center)},scale=0.55]
{
\betri{(0,0)}
\rsq{(1,-1)}
\batri{(2,0)}
\lsq{(3,-1)}
\draw[lstring] let \n1 = {sin(45)} in (2.5,-1.5) arc (225:315:\n1);
\foreach \x in {2.5,3.5}
{\foreach \y in {-0.5,-1.5}
{\fill[lgh] (\x,\y) circle (1.6 pt);}}
\fill[lgh] (1.5,-0.5) circle (1.6 pt);
}
\end{tikzpicture}
\right]_{\times 4}
\begin{tikzpicture}[baseline={([yshift=-1mm]current bounding box.center)},scale=0.55]
{
\betri{(0,0)}
\rsq{(1,-1)}
\aatri{(2,0)}
\lsq{(3,-1)}
\draw[lstring] let \n1 = {sin(45)} in (2.5,-1.5) arc (225:315:\n1);
\foreach \x in {2.5,3.5}
{\foreach \y in {-0.5,-1.5}
{\fill[lgh] (\x,\y) circle (1.6 pt);}}
\fill[lgh] (1.5,-0.5) circle (1.6 pt);
}
\end{tikzpicture}
\quad
\begin{tikzpicture}[baseline={([yshift=-1mm]current bounding box.center)},scale=0.55]
{
\aatri{(0,0)}
\tsq{(1,-1)}
\betri{(2,0)}
\emsq{(3,-1)}
\foreach \x in {2.5,3.5}
{\foreach \y in {-0.5,-1.5}
{\fill[lgh] (\x,\y) circle (1.6 pt);}}
\fill[lgh] (1.5,-0.5) circle (1.6 pt);
}
\end{tikzpicture}
\quad
\begin{tikzpicture}[baseline={([yshift=-1mm]current bounding box.center)},scale=0.55]
{
\aatri{(0,0)}
\dsq{(1,-1)}
\bhtri{(2,0)}
\lsq{(3,-1)}
\draw[lstring] let \n1 = {sin(45)} in (2.5,-1.5) arc (225:315:\n1);
\foreach \x in {2.5,3.5}
{\foreach \y in {-0.5,-1.5}
{\fill[lgh] (\x,\y) circle (1.6 pt);}}
\fill[lgh] (1.5,-0.5) circle (1.6 pt);
}
\end{tikzpicture}
\\[9mm]
&\begin{tikzpicture}[baseline={([yshift=-1mm]current bounding box.center)},scale=0.55]
{
\draw[dotted] (0,0.5)--(3,0.5);
\draw [very thick](0,-1.5) -- (0,0.5);
\draw [very thick](3,-1.5)--(3,0.5);
\draw (0,0) to[out=0,in=-90] (0.75,0.5);
\draw (0,-1) to[out=0,in=-90] (1.5,0.5);
\draw (3,0) arc (90:270:0.5);
\filldraw (0.375,0.5) circle (0.08);
\filldraw (1.125,0.5) circle (0.08);
}
\end{tikzpicture} \; :
& & 
\begin{tikzpicture}[baseline={([yshift=-1mm]current bounding box.center)},scale=0.55]
{
\idsq{(-1,-1)}
\bbtri{(0,0)}
\esq{(1,-1)}
\aatri{(2,0)}
\draw[lstring] let \n1 = {sin(45)} in (0.5,-1.5) arc (225:315:\n1);
\foreach \x in {0.5,1.5}
{\foreach \y in {-0.5,-1.5}
{\fill[lgh] (\x,\y) circle (1.6 pt);}}
\fill[lgh] (2.5,-0.5) circle (1.6 pt);
}
\end{tikzpicture}
\quad
\begin{tikzpicture}[baseline={([yshift=-1mm]current bounding box.center)},scale=0.55]
{
\bbtri{(0,0)}
\lsq{(1,-1)}
\emtri{(2,0)}
\rsq{(3,-1)}
\foreach \x in {2.5,3.5}
{\foreach \y in {-0.5,-1.5}
{\fill[lgh] (\x,\y) circle (1.6 pt);}}
\fill[lgh] (1.5,-0.5) circle (1.6 pt);
}
\end{tikzpicture}
\quad
\begin{tikzpicture}[baseline={([yshift=-1mm]current bounding box.center)},scale=0.55]
{
\bbtri{(0,0)}
\esq{(1,-1)}
\aatri{(2,0)}
\esq{(3,-1)}
\draw[lstring] let \n1 = {sin(45)} in (2.5,-1.5) arc (225:315:\n1);
\foreach \x in {2.5,3.5}
{\foreach \y in {-0.5,-1.5}
{\fill[lgh] (\x,\y) circle (1.6 pt);}}
\fill[lgh] (1.5,-0.5) circle (1.6 pt);
}
\end{tikzpicture}
\quad\left[
\begin{tikzpicture}[baseline={([yshift=-1mm]current bounding box.center)},scale=0.55]
{
\bbtri{(0,0)}
\esq{(1,-1)}
\batri{(2,0)}
\esq{(3,-1)}
\draw[lstring] let \n1 = {sin(45)} in (2.5,-1.5) arc (225:315:\n1);
\foreach \x in {2.5,3.5}
{\foreach \y in {-0.5,-1.5}
{\fill[lgh] (\x,\y) circle (1.6 pt);}}
\fill[lgh] (1.5,-0.5) circle (1.6 pt);
}
\end{tikzpicture}
\right]_{\times 4}
\\[2mm]
& & &
\begin{tikzpicture}[baseline={([yshift=-1mm]current bounding box.center)},scale=0.55]
{
\betri{(0,0)}
\usq{(1,-1)}
\betri{(2,0)}
\rsq{(3,-1)}
\foreach \x in {2.5,3.5}
{\foreach \y in {-0.5,-1.5}
{\fill[lgh] (\x,\y) circle (1.6 pt);}}
\fill[lgh] (1.5,-0.5) circle (1.6 pt);
}
\end{tikzpicture}
\quad
\begin{tikzpicture}[baseline={([yshift=-1mm]current bounding box.center)},scale=0.55]
{
\betri{(0,0)}
\bsq{(1,-1)}
\bhtri{(2,0)}
\esq{(3,-1)}
\draw[lstring] let \n1 = {sin(45)} in (2.5,-1.5) arc (225:315:\n1);
\foreach \x in {2.5,3.5}
{\foreach \y in {-0.5,-1.5}
{\fill[lgh] (\x,\y) circle (1.6 pt);}}
\fill[lgh] (1.5,-0.5) circle (1.6 pt);
}
\end{tikzpicture}
\quad
\left[
\begin{tikzpicture}[baseline={([yshift=-1mm]current bounding box.center)},scale=0.55]
{
\bbtri{(0,0)}
\idsq{(1,-1)}
\batri{(2,0)}
\esq{(3,-1)}
\draw[lstring] let \n1 = {sin(45)} in (2.5,-1.5) arc (225:315:\n1);
\foreach \x in {2.5,3.5}
{\foreach \y in {-0.5,-1.5}
{\fill[lgh] (\x,\y) circle (1.6 pt);}}
\fill[lgh] (1.5,-0.5) circle (1.6 pt);
}
\end{tikzpicture}
\right]_{\times 4}
\begin{tikzpicture}[baseline={([yshift=-1mm]current bounding box.center)},scale=0.55]
{
\aatri{(0,0)}
\idsq{(1,-1)}
\bbtri{(2,0)}
\esq{(3,-1)}
\draw[lstring] let \n1 = {sin(45)} in (2.5,-1.5) arc (225:315:\n1);
\foreach \x in {2.5,3.5}
{\foreach \y in {-0.5,-1.5}
{\fill[lgh] (\x,\y) circle (1.6 pt);}}
\fill[lgh] (1.5,-0.5) circle (1.6 pt);
}
\end{tikzpicture}
\\[2mm]
& & & 
\begin{tikzpicture}[baseline={([yshift=-1mm]current bounding box.center)},scale=0.55]
{
\bbtri{(0,0)}
\idsq{(1,-1)}
\aatri{(2,0)}
\esq{(3,-1)}
\draw[lstring] let \n1 = {sin(45)} in (2.5,-1.5) arc (225:315:\n1);
\foreach \x in {2.5,3.5}
{\foreach \y in {-0.5,-1.5}
{\fill[lgh] (\x,\y) circle (1.6 pt);}}
\fill[lgh] (1.5,-0.5) circle (1.6 pt);
}
\end{tikzpicture}
\quad
\begin{tikzpicture}[baseline={([yshift=-1mm]current bounding box.center)},scale=0.55]
{
\bbtri{(0,0)}
\esq{(1,-1)}
\aatri{(2,0)}
\idsq{(3,-1)}
\draw[lstring] let \n1 = {sin(45)} in (2.5,-1.5) arc (225:315:\n1);
\foreach \x in {2.5,3.5}
{\foreach \y in {-0.5,-1.5}
{\fill[lgh] (\x,\y) circle (1.6 pt);}}
\fill[lgh] (1.5,-0.5) circle (1.6 pt);
}
\end{tikzpicture}
\end{align*}
\begin{align*}
&\begin{tikzpicture}[baseline={([yshift=-1mm]current bounding box.center)},scale=0.55]
{
\draw[dotted] (0,0.5)--(3,0.5);
\draw [very thick](0,-1.5) -- (0,0.5);
\draw [very thick](3,-1.5)--(3,0.5);
\draw (0,0) -- (3,-1);
\filldraw[fill=white] (3,0) circle (0.12);
\filldraw[fill=white] (0,-1) circle (0.12);
}
\end{tikzpicture} \; :
& & 
\begin{tikzpicture}[baseline={([yshift=-1mm]current bounding box.center)},scale=0.55]
{
\tsq{(-1,-1)}
\aatri{(0,0)}
\dsq{(1,-1)}
\emtri{(2,0)}
\foreach \x in {0.5,1.5}
{\foreach \y in {-0.5,-1.5}
{\fill[lgh] (\x,\y) circle (1.6 pt);}}
\fill[lgh] (2.5,-0.5) circle (1.6 pt);
}
\end{tikzpicture}
\quad 
\begin{tikzpicture}[baseline={([yshift=-1mm]current bounding box.center)},scale=0.55]
{
\dsq{(-1,-1)}
\emtri{(0,0)}
\bsq{(1,-1)}
\emtri{(2,0)}
\draw[lstring] let \n1 = {sin(45)} in (0.5,-1.5) arc (225:315:\n1);
\foreach \x in {0.5,1.5}
{\foreach \y in {-0.5,-1.5}
{\fill[lgh] (\x,\y) circle (1.6 pt);}}
\fill[lgh] (2.5,-0.5) circle (1.6 pt);
}
\end{tikzpicture}
\quad
\left[
\begin{tikzpicture}[baseline={([yshift=-1mm]current bounding box.center)},scale=0.55]
{
\dsq{(-1,-1)}
\bhtri{(0,0)}
\idsq{(1,-1)}
\bftri{(2,0)}
\draw[lstring] let \n1 = {sin(45)} in (0.5,-1.5) arc (225:315:\n1);
\foreach \x in {0.5,1.5}
{\foreach \y in {-0.5,-1.5}
{\fill[lgh] (\x,\y) circle (1.6 pt);}}
\fill[lgh] (2.5,-0.5) circle (1.6 pt);
}
\end{tikzpicture}
\right]_{\times 4}
\\[2mm]
& & &
\begin{tikzpicture}[baseline={([yshift=-1mm]current bounding box.center)},scale=0.55]
{
\aatri{(0,0)}
\tsq{(1,-1)}
\aatri{(2,0)}
\dsq{(3,-1)}
\foreach \x in {2.5,3.5}
{\foreach \y in {-0.5,-1.5}
{\fill[lgh] (\x,\y) circle (1.6 pt);}}
\fill[lgh] (1.5,-0.5) circle (1.6 pt);
}
\end{tikzpicture}
\quad
\begin{tikzpicture}[baseline={([yshift=-1mm]current bounding box.center)},scale=0.55]
{
\aatri{(0,0)}
\dsq{(1,-1)}
\emtri{(2,0)}
\bsq{(3,-1)}
\draw[lstring] let \n1 = {sin(45)} in (2.5,-1.5) arc (225:315:\n1);
\foreach \x in {2.5,3.5}
{\foreach \y in {-0.5,-1.5}
{\fill[lgh] (\x,\y) circle (1.6 pt);}}
\fill[lgh] (1.5,-0.5) circle (1.6 pt);
}
\end{tikzpicture}
\\[9mm]
&\begin{tikzpicture}[baseline={([yshift=-1mm]current bounding box.center)},scale=0.55]
{
\draw[dotted] (0,0.5)--(3,0.5);
\draw [very thick](0,-1.5) -- (0,0.5);
\draw [very thick](3,-1.5)--(3,0.5);
\draw (0,0) arc (90:-90:0.5);
\filldraw[fill=white] (3,0) circle (0.12);
\filldraw[fill=white] (3,-1) circle (0.12);
}
\end{tikzpicture} \; :
& & 
\begin{tikzpicture}[baseline={([yshift=-1mm]current bounding box.center)},scale=0.55]
{
\lsq{(-1,-1)}
\emtri{(0,0)}
\emsq{(1,-1)}
\emtri{(2,0)}
\foreach \x in {0.5,1.5}
{\foreach \y in {-0.5,-1.5}
{\fill[lgh] (\x,\y) circle (1.6 pt);}}
\fill[lgh] (2.5,-0.5) circle (1.6 pt);
}
\end{tikzpicture}
\quad 
\begin{tikzpicture}[baseline={([yshift=-1mm]current bounding box.center)},scale=0.55]
{
\esq{(-1,-1)}
\aatri{(0,0)}
\lsq{(1,-1)}
\emtri{(2,0)}
\draw[lstring] let \n1 = {sin(45)} in (0.5,-1.5) arc (225:315:\n1);
\foreach \x in {0.5,1.5}
{\foreach \y in {-0.5,-1.5}
{\fill[lgh] (\x,\y) circle (1.6 pt);}}
\fill[lgh] (2.5,-0.5) circle (1.6 pt);
}
\end{tikzpicture}
\ 
\left[\begin{tikzpicture}[baseline={([yshift=-1mm]current bounding box.center)},scale=0.55]
{
\esq{(-1,-1)}
\batri{(0,0)}
\lsq{(1,-1)}
\emtri{(2,0)}
\draw[lstring] let \n1 = {sin(45)} in (0.5,-1.5) arc (225:315:\n1);
\foreach \x in {0.5,1.5}
{\foreach \y in {-0.5,-1.5}
{\fill[lgh] (\x,\y) circle (1.6 pt);}}
\fill[lgh] (2.5,-0.5) circle (1.6 pt);
}
\end{tikzpicture}
\right]_{\times 4}
\left[
\begin{tikzpicture}[baseline={([yshift=-1mm]current bounding box.center)},scale=0.55]
{
\esq{(-1,-1)}
\bftri{(0,0)}
\usq{(1,-1)}
\bftri{(2,0)}
\draw[lstring] let \n1 = {sin(45)} in (0.5,-1.5) arc (225:315:\n1);
\foreach \x in {0.5,1.5}
{\foreach \y in {-0.5,-1.5}
{\fill[lgh] (\x,\y) circle (1.6 pt);}}
\fill[lgh] (2.5,-0.5) circle (1.6 pt);
}
\end{tikzpicture}
\right]_{\times 4}
\\[2mm]
& & &
\left[
\begin{tikzpicture}[baseline={([yshift=-1mm]current bounding box.center)},scale=0.55]
{
\lsq{(-1,-1)}
\bhtri{(0,0)}
\tsq{(1,-1)}
\bftri{(2,0)}
\foreach \x in {0.5,1.5}
{\foreach \y in {-0.5,-1.5}
{\fill[lgh] (\x,\y) circle (1.6 pt);}}
\fill[lgh] (2.5,-0.5) circle (1.6 pt);
}
\end{tikzpicture}
\right]_{\times 4}
\begin{tikzpicture}[baseline={([yshift=-1mm]current bounding box.center)},scale=0.55]
{
\idsq{(-1,-1)}
\aatri{(0,0)}
\lsq{(1,-1)}
\emtri{(2,0)}
\draw[lstring] let \n1 = {sin(45)} in (0.5,-1.5) arc (225:315:\n1);
\foreach \x in {0.5,1.5}
{\foreach \y in {-0.5,-1.5}
{\fill[lgh] (\x,\y) circle (1.6 pt);}}
\fill[lgh] (2.5,-0.5) circle (1.6 pt);
}
\end{tikzpicture}
\\[2mm]
& & & 
\begin{tikzpicture}[baseline={([yshift=-1mm]current bounding box.center)},scale=0.55]
{
\aatri{(0,0)}
\lsq{(1,-1)}
\emtri{(2,0)}
\emsq{(3,-1)}
\foreach \x in {2.5,3.5}
{\foreach \y in {-0.5,-1.5}
{\fill[lgh] (\x,\y) circle (1.6 pt);}}
\fill[lgh] (1.5,-0.5) circle (1.6 pt);
}
\end{tikzpicture}
\quad 
\begin{tikzpicture}[baseline={([yshift=-1mm]current bounding box.center)},scale=0.55]
{
\aatri{(0,0)}
\esq{(1,-1)}
\aatri{(2,0)}
\lsq{(3,-1)}
\draw[lstring] let \n1 = {sin(45)} in (2.5,-1.5) arc (225:315:\n1);
\foreach \x in {2.5,3.5}
{\foreach \y in {-0.5,-1.5}
{\fill[lgh] (\x,\y) circle (1.6 pt);}}
\fill[lgh] (1.5,-0.5) circle (1.6 pt);
}
\end{tikzpicture}
\ \left[
\begin{tikzpicture}[baseline={([yshift=-1mm]current bounding box.center)},scale=0.55]
{
\aatri{(0,0)}
\esq{(1,-1)}
\batri{(2,0)}
\lsq{(3,-1)}
\draw[lstring] let \n1 = {sin(45)} in (2.5,-1.5) arc (225:315:\n1);
\foreach \x in {2.5,3.5}
{\foreach \y in {-0.5,-1.5}
{\fill[lgh] (\x,\y) circle (1.6 pt);}}
\fill[lgh] (1.5,-0.5) circle (1.6 pt);
}
\end{tikzpicture}
\right]_{\times 4}
\ 
\begin{tikzpicture}[baseline={([yshift=-1mm]current bounding box.center)},scale=0.55]
{
\aatri{(0,0)}
\idsq{(1,-1)}
\aatri{(2,0)}
\lsq{(3,-1)}
\draw[lstring] let \n1 = {sin(45)} in (2.5,-1.5) arc (225:315:\n1);
\foreach \x in {2.5,3.5}
{\foreach \y in {-0.5,-1.5}
{\fill[lgh] (\x,\y) circle (1.6 pt);}}
\fill[lgh] (1.5,-0.5) circle (1.6 pt);
}
\end{tikzpicture}
\\[9mm]
&\begin{tikzpicture}[baseline={([yshift=-1mm]current bounding box.center)},scale=0.55]
{
\draw[dotted] (0,0.5)--(3,0.5);
\draw [very thick](0,-1.5) -- (0,0.5);
\draw [very thick](3,-1.5)--(3,0.5);
\draw (0,0) to[out=0,in=-90] (0.75,0.5);
\draw (0,-1) to[out=0,in=-90] (1.5,0.5);
\filldraw[fill=white] (3,0) circle (0.12);
\filldraw[fill=white] (3,-1) circle (0.12);
}
\end{tikzpicture} \; :
& & 
\begin{tikzpicture}[baseline={([yshift=-1mm]current bounding box.center)},scale=0.55]
{
\idsq{(-1,-1)}
\batri{(0,0)}
\lsq{(1,-1)}
\emtri{(2,0)}
\draw[lstring] let \n1 = {sin(45)} in (0.5,-1.5) arc (225:315:\n1);
\foreach \x in {0.5,1.5}
{\foreach \y in {-0.5,-1.5}
{\fill[lgh] (\x,\y) circle (1.6 pt);}}
\fill[lgh] (2.5,-0.5) circle (1.6 pt);
}
\end{tikzpicture}
\quad 
\begin{tikzpicture}[baseline={([yshift=-1mm]current bounding box.center)},scale=0.55]
{
\idsq{(-1,-1)}
\bftri{(0,0)}
\usq{(1,-1)}
\betri{(2,0)}
\draw[lstring] let \n1 = {sin(45)} in (0.5,-1.5) arc (225:315:\n1);
\foreach \x in {0.5,1.5}
{\foreach \y in {-0.5,-1.5}
{\fill[lgh] (\x,\y) circle (1.6 pt);}}
\fill[lgh] (2.5,-0.5) circle (1.6 pt);
}
\end{tikzpicture}
\\[2mm]
& & &
\begin{tikzpicture}[baseline={([yshift=-1mm]current bounding box.center)},scale=0.55]
{
\batri{(0,0)}
\lsq{(1,-1)}
\emtri{(2,0)}
\emsq{(3,-1)}
\foreach \x in {2.5,3.5}
{\foreach \y in {-0.5,-1.5}
{\fill[lgh] (\x,\y) circle (1.6 pt);}}
\fill[lgh] (1.5,-0.5) circle (1.6 pt);
}
\end{tikzpicture}
\quad
\begin{tikzpicture}[baseline={([yshift=-1mm]current bounding box.center)},scale=0.55]
{
\batri{(0,0)}
\esq{(1,-1)}
\aatri{(2,0)}
\lsq{(3,-1)}
\draw[lstring] let \n1 = {sin(45)} in (2.5,-1.5) arc (225:315:\n1);
\foreach \x in {2.5,3.5}
{\foreach \y in {-0.5,-1.5}
{\fill[lgh] (\x,\y) circle (1.6 pt);}}
\fill[lgh] (1.5,-0.5) circle (1.6 pt);
}
\end{tikzpicture}
\ 
\left[\begin{tikzpicture}[baseline={([yshift=-1mm]current bounding box.center)},scale=0.55]
{
\batri{(0,0)}
\esq{(1,-1)}
\batri{(2,0)}
\lsq{(3,-1)}
\draw[lstring] let \n1 = {sin(45)} in (2.5,-1.5) arc (225:315:\n1);
\foreach \x in {2.5,3.5}
{\foreach \y in {-0.5,-1.5}
{\fill[lgh] (\x,\y) circle (1.6 pt);}}
\fill[lgh] (1.5,-0.5) circle (1.6 pt);
}
\end{tikzpicture}
\right]_{\times 4}
\begin{tikzpicture}[baseline={([yshift=-1mm]current bounding box.center)},scale=0.55]
{
\batri{(0,0)}
\idsq{(1,-1)}
\aatri{(2,0)}
\lsq{(3,-1)}
\draw[lstring] let \n1 = {sin(45)} in (2.5,-1.5) arc (225:315:\n1);
\foreach \x in {2.5,3.5}
{\foreach \y in {-0.5,-1.5}
{\fill[lgh] (\x,\y) circle (1.6 pt);}}
\fill[lgh] (1.5,-0.5) circle (1.6 pt);
}
\end{tikzpicture}
\\[2mm]
& & &
\begin{tikzpicture}[baseline={([yshift=-1mm]current bounding box.center)},scale=0.55]
{
\bftri{(0,0)}
\usq{(1,-1)}
\betri{(2,0)}
\emsq{(3,-1)}
\foreach \x in {2.5,3.5}
{\foreach \y in {-0.5,-1.5}
{\fill[lgh] (\x,\y) circle (1.6 pt);}}
\fill[lgh] (1.5,-0.5) circle (1.6 pt);
}
\end{tikzpicture}
\quad
\begin{tikzpicture}[baseline={([yshift=-1mm]current bounding box.center)},scale=0.55]
{
\bftri{(0,0)}
\bsq{(1,-1)}
\bhtri{(2,0)}
\lsq{(3,-1)}
\draw[lstring] let \n1 = {sin(45)} in (2.5,-1.5) arc (225:315:\n1);
\foreach \x in {2.5,3.5}
{\foreach \y in {-0.5,-1.5}
{\fill[lgh] (\x,\y) circle (1.6 pt);}}
\fill[lgh] (1.5,-0.5) circle (1.6 pt);
}
\end{tikzpicture}
\ 
\left[\begin{tikzpicture}[baseline={([yshift=-1mm]current bounding box.center)},scale=0.55]
{
\batri{(0,0)}
\idsq{(1,-1)}
\batri{(2,0)}
\lsq{(3,-1)}
\draw[lstring] let \n1 = {sin(45)} in (2.5,-1.5) arc (225:315:\n1);
\foreach \x in {2.5,3.5}
{\foreach \y in {-0.5,-1.5}
{\fill[lgh] (\x,\y) circle (1.6 pt);}}
\fill[lgh] (1.5,-0.5) circle (1.6 pt);
}
\end{tikzpicture}
\right]_{\times 4}
\begin{tikzpicture}[baseline={([yshift=-1mm]current bounding box.center)},scale=0.55]
{
\aatri{(0,0)}
\idsq{(1,-1)}
\batri{(2,0)}
\lsq{(3,-1)}
\draw[lstring] let \n1 = {sin(45)} in (2.5,-1.5) arc (225:315:\n1);
\foreach \x in {2.5,3.5}
{\foreach \y in {-0.5,-1.5}
{\fill[lgh] (\x,\y) circle (1.6 pt);}}
\fill[lgh] (1.5,-0.5) circle (1.6 pt);
}
\end{tikzpicture}
\end{align*}

\vspace{1mm}
\noindent \textbf{Madeline Nurcombe} (\texttt{m.nurcombe@uq.net.au})

\noindent \textit{School of Mathematics and Physics, University of Queensland}

\noindent \textit{St Lucia, Brisbane, Queensland, 4072, Australia}

\end{document}